\PassOptionsToPackage{unicode}{hyperref}
\PassOptionsToPackage{hyphens}{url}
\PassOptionsToPackage{dvipsnames,svgnames,x11names}{xcolor}

\documentclass[12pt]{article}
\usepackage{setspace}

\usepackage{siunitx}  
\usepackage{amsmath,amssymb}
\usepackage{mathtools}
\usepackage{enumitem}
\usepackage{soul}
\usepackage{xcolor}
\setulcolor{red}    
\setstcolor{red}    
\usepackage{amsthm}
\usepackage{cancel}
\usepackage{stackrel}
\usepackage{multirow}
\usepackage{iftex}
\usepackage{tikz}
\usepackage{titlesec}
\usetikzlibrary{shapes.geometric, arrows, positioning, calc, shadows.blur}\usetikzlibrary{positioning}
\allowdisplaybreaks[4]
\ifPDFTeX
  \usepackage[T1]{fontenc}
  \usepackage[utf8]{inputenc}
  \usepackage{textcomp} % provide euro and other symbols
\else % if luatex or xetex
  \usepackage{unicode-math}
  \defaultfontfeatures{Scale=MatchLowercase}
  \defaultfontfeatures[\rmfamily]{Ligatures=TeX,Scale=1}
\fi

\PassOptionsToPackage{normalem}{ulem}
\usepackage{ulem}

\providecolor{lyxadded}{rgb}{0,0,1}
\providecolor{lyxdeleted}{rgb}{1,0,0}

\providecommand{\tabularnewline}{\\}

\DeclareRobustCommand{\mklyxadded}[1]{\textcolor{lyxadded}\bgroup#1\egroup}
\DeclareRobustCommand{\mklyxdeleted}[1]{\textcolor{lyxdeleted}\bgroup\mklyxsout{#1}\egroup}
\DeclareRobustCommand{\mklyxsout}[1]{\ifx\\#1\else\sout{#1}\fi}

\usepackage{lmodern}
\ifPDFTeX\else  
\fi

\IfFileExists{upquote.sty}{\usepackage{upquote}}{}
\IfFileExists{microtype.sty}{% use microtype if available
    \usepackage[]{microtype}
    \UseMicrotypeSet[protrusion]{basicmath} % disable protrusion for tt fonts
}{}

\makeatletter
\@ifundefined{KOMAClassName}{% if non-KOMA class
    \IfFileExists{parskip.sty}{%
        \usepackage{parskip}
    }{% else
        \setlength{\parindent}{0pt}
        \setlength{\parskip}{6pt plus 2pt minus 1pt}}
}{% if KOMA class
    \KOMAoptions{parskip=half}}
\makeatother

\makeatletter
\ifx\paragraph\undefined\else
\let\oldparagraph\paragraph
\renewcommand{\paragraph}{
    \@ifstar
    \xxxParagraphStar
    \xxxParagraphNoStar
}
\newcommand{\xxxParagraphStar}[1]{\oldparagraph*{#1}\mbox{}}
\newcommand{\xxxParagraphNoStar}[1]{\oldparagraph{#1}\mbox{}}
\fi
\ifx\subparagraph\undefined\else
\let\oldsubparagraph\subparagraph
\renewcommand{\subparagraph}{
    \@ifstar
    \xxxSubParagraphStar
    \xxxSubParagraphNoStar
}
\newcommand{\xxxSubParagraphStar}[1]{\oldsubparagraph*{#1}\mbox{}}
\newcommand{\xxxSubParagraphNoStar}[1]{\oldsubparagraph{#1}\mbox{}}
\fi
\makeatother

\usepackage{longtable,booktabs,array}
\usepackage{calc} % for calculating minipage widths

\usepackage{etoolbox}
\makeatletter
\patchcmd\longtable{\par}{\if@noskipsec\mbox{}\fi\par}{}{}
\makeatother

\IfFileExists{footnotehyper.sty}{\usepackage{footnotehyper}}{\usepackage{footnote}}
\makesavenoteenv{longtable}
\usepackage{graphicx}
\makeatletter
\def\maxwidth{\ifdim\Gin@nat@width>\linewidth\linewidth\else\Gin@nat@width\fi}
\def\maxheight{\ifdim\Gin@nat@height>\textheight\textheight\else\Gin@nat@height\fi}
\makeatother

\setkeys{Gin}{width=\maxwidth,height=\maxheight,keepaspectratio}

\makeatletter
\def\fps@figure{htbp}
\makeatother

\makeatletter
\@ifpackageloaded{caption}{}{\usepackage{caption}}
\AtBeginDocument{%
    \ifdefined\contentsname
    \renewcommand*\contentsname{Table of contents}
    \else
    \newcommand\contentsname{Table of contents}
    \fi
    \ifdefined\listfigurename
    \renewcommand*\listfigurename{List of Figures}
    \else
    \newcommand\listfigurename{List of Figures}
    \fi
    \ifdefined\listtablename
    \renewcommand*\listtablename{List of Tables}
    \else
    \newcommand\listtablename{List of Tables}
    \fi
    \ifdefined\figurename
    \renewcommand*\figurename{Figure}
    \else
    \newcommand\figurename{Figure}
    \fi
    \ifdefined\tablename
    \renewcommand*\tablename{Table}
    \else
    \newcommand\tablename{Table}
    \fi
}
\@ifpackageloaded{float}{}{\usepackage{float}}
\floatstyle{ruled}
\@ifundefined{c@chapter}{\newfloat{codelisting}{h}{lop}}{\newfloat{codelisting}{h}{lop}[chapter]}
\floatname{codelisting}{Listing}

\makeatother
\makeatletter
\@ifpackageloaded{caption}{}{\usepackage{caption}}
\@ifpackageloaded{subcaption}{}{\usepackage{subcaption}}
\makeatother

\ifLuaTeX
\usepackage{selnolig}  % disable illegal ligatures
\fi

\usepackage[sort]{natbib}
\usepackage{bookmark}
\usepackage{dsfont}
\IfFileExists{xurl.sty}{\usepackage{xurl}}{} % add URL line breaks if available
\hypersetup{
    colorlinks=true,
    linkcolor={blue},
    filecolor={Maroon},
    citecolor={Blue},
    urlcolor={Blue},
    pdfcreator={LaTeX via pandoc}}

\usepackage{threeparttable}
\usepackage{tabularray}
\usepackage{verbatim}

\global\long\def\tr{\mathrm{tr}}%
\global\long\def\R{\mathbb{R}}%
\global\long\def\1{\mathds{1}}%
\global\long\def\tp{\overset{p}{\to}}%
\global\long\def\tod{\overset{d}{\to}}%
\global\long\def\e{\mathbb{E}}%
\global\long\def\ep{\epsilon}%
\global\long\def\diag{\text{diag}}%
\global\long\def\var{\mathrm{Var}}%
\global\long\def\bs#1{\boldsymbol{#1}}%
\global\long\def\trans{\top}%
\global\long\def\vec{\mathrm{vec}}%

\numberwithin{equation}{section}
\theoremstyle{plain}

\theoremstyle{definition}

\theoremstyle{plain}
\newtheorem{thm}{\protect\theoremname}
\theoremstyle{remark}

\theoremstyle{plain}
\newtheorem{lem}{\protect\lemmaname}
\theoremstyle{plain}

\theoremstyle{plain}
\newtheorem{assumption}{\protect\assumptionname}

\theoremstyle{definition}

\providecommand{\definitionname}{Definition}

\providecommand{\lemmaname}{Lemma}

\providecommand{\propositionname}{Proposition}
\providecommand{\remarkname}{Remark}

\providecommand{\theoremname}{Theorem}

\renewcommand\qedsymbol{\ensuremath{\blacksquare}}

\makeatother
\providecommand{\algorithmname}{Algorithm}
\providecommand{\assumptionname}{Assumption}
\providecommand{\definitionname}{Definition}
\providecommand{\lemmaname}{Lemma}
\providecommand{\remarkname}{Remark}
\providecommand{\theoremname}{Theorem}

\usepackage[sectionbib]{bibunits}
\defaultbibliographystyle{agsm} 
\defaultbibliography{r.bib} % name of the .bib file 

\newcommand{\anon}{1}

\begin{document}

\def\spacingset#1{\renewcommand{\baselinestretch}%
{#1}\small\normalsize} \spacingset{1}

%%%%%%%%%%%%%%%%%%%%%%%%%%%%%%%%%%%%%%%%%%%%%%%%%%%%%%%%%%%%%%%%%%%%%%%%%%%%%%

\if1\anon
	{
		\title{\bf Integrating Heterogeneous Information in Randomized Experiments: A Unified Calibration Framework\thanks{The authors are listed in alphabetical order. The complete source code for the simulations and empirical applications is hosted on GitHub at \url{https://github.com/zeqiwu1202/Replication-for-CAR-calibration}.}}
		\author{Wei Ma\thanks{Email address: \protect\href{mailto:mawei@ruc.edu.cn}{mawei@ruc.edu.cn}.} \hspace{1em}  Zeqi Wu\thanks{Corresponding author. Email address: \protect\href{mailto:wuzeqi@ruc.edu.cn}{wuzeqi@ruc.edu.cn}.}\hspace{1em}  Zheng Zhang\thanks{Email address: \protect\href{mailto:zhengzhang@ruc.edu.cn}{zhengzhang@ruc.edu.cn}.}\\
			Institute of Statistics and Big Data, Renmin
			University of China
		}
		\date{}
		\maketitle
	} \fi

\if0\anon
{
	\bigskip
	\bigskip
	\bigskip
	\begin{spacing}{1.5}
	
	\begin{center}
		{\LARGE\bf Calibration for Treatment Effect Estimation in Randomized Controlled Trials: a Unified Approach for Covariates Adjustment and Information Borrowing}
	\end{center}

\end{spacing}
	\medskip
} \fi

\begin{bibunit}

\vspace{-1em}  
%\bigskip     For JASA, Uncomment
\begin{abstract}
   In modern randomized experiments, large-scale data collection increasingly yields rich baseline covariates and auxiliary information from multiple sources. Such information offers opportunities for more precise treatment effect estimation, but it also raises the challenge of integrating heterogeneous information coherently without compromising validity. Covariate-adaptive randomization (CAR) is widely used to improve covariate balance at the design stage, but it typically balances only a small set of covariates used to form strata, making covariate adjustment at the analysis stage essential for more efficient estimation of treatment effects. Beyond standard covariate adjustment, it is often desirable to incorporate auxiliary information, including cross-stratum information, predictions from various machine learning models, and external data from historical trials or real-world sources. While this auxiliary information is widely available, existing covariate adjustment methods under CAR primarily exploit within-stratum covariates and do not provide a coherent mechanism for integrating it. We propose a unified calibration framework that integrates such information through an information proxy vector and calibration weights defined by a convex optimization problem. The resulting estimator recovers many recent covariate adjustment procedures as special cases while providing a systematic mechanism for both internal and external information borrowing within a single framework. We establish large-sample validity and a no-harm efficiency guarantee, showing that incorporating additional information sources cannot increase asymptotic variance, and we extend the theory to settings in which both the number of strata and the number of information sources grow with the sample size. Simulation studies and an empirical analysis of a field experiment on savings behavior in Uganda and Malawi demonstrate the strong finite-sample performance and practical utility of our method.
\end{abstract}
\noindent%
\vfill

\newpage
%\spacingset{1.8} % DON'T change the spacing! for JASA
\doublespacing

\section{Introduction}

In modern randomized experiments, ensuring balance in baseline covariates
across treatment groups is critical for reducing bias and improving
the credibility of the trial results. Covariate-adaptive randomization
(CAR) methods, such as stratified biased coin randomization (\citealp{efron1971Forcing}),
stratified block randomization (\citealp{zelen1974Randomization}),
and minimization (\citealp{pocock1975Sequential,taves1974Minimization}),
are widely used to achieve this balance during the design stage. In practice, however, CAR is typically implemented using only a small set of covariates to form strata, so balance is not directly enforced for many other baseline covariates. Moreover, some important pre-treatment covariates may only be observed after randomization, for instance when they are collected together with the outcome variable (\citealp{bai2024Covariate}).
Consequently, covariate adjustment at the statistical analysis stage
plays a crucial complementary role. By incorporating baseline covariates
into the statistical analysis, covariate adjustment methods correct
for residual imbalances, thereby improving the precision and efficiency
of treatment effect estimates and strengthening the validity of the
trial.

The adjustment of baseline covariates in the statistical analysis
stage has been studied for a long time (see, e.g., \citealp{tsiatis2008Covariate,zhang2008Improving,lin2013Agnostic}
and the references therein). Under the CAR design, \citet{ma2022Regression,ye2022Inference,gu2023RegressionBased}
adjusted for additional covariates using linear regression, while
\citet{liu2023Lassoadjusted} applied Lasso (\citealp{tibshirani1996Regression}).
These approaches resulted in average treatment effect (ATE) estimators
that are more efficient than the naive difference-in-means estimator
and the ordinary least squares estimator (regressing outcomes on strata
indicators) proposed by \citet{bugni2018Inference} and \citet{bugni2019Inference}.
This holds true even in the absence of strong evidence for a linear
relationship between covariates and outcomes. To further enhance efficiency,
recent work by \citet{rafi2023Efficient}, \citet{tu2024Unified}
and \citet{bannick2025General} introduced unified frameworks for
nonlinear covariate adjustment. Relying on the augmented inverse probability
weighting (AIPW, \citealp{robins1994Estimation}), these frameworks
incorporate machine learning techniques, such as random forests and
deep neural networks, to adjust for covariates. These machine learning
methods can be effective in capturing complex nonlinear relationships
between covariates and outcomes.

However, much of the existing covariate adjustment literature can
be interpreted as focusing on a particular form of internal information
borrowing, namely using baseline covariates from the current trial,
typically within each stratum, to improve efficiency. This focus leaves
relatively little room for other practically important forms of information
borrowing, and it limits the extent to which standard adjustment methods
can integrate heterogeneous information in a systematic way. Internally,
efficiency can often be improved by borrowing information across strata
when the outcome--covariate relationship is stable, and by combining
predictions from multiple machine learning methods when no single
learner is uniformly reliable. Externally, there is increasing interest
in leveraging historical trials and real-world data to support analyses
of concurrent trials, especially when sample sizes are constrained
by cost, ethical considerations, or recruitment challenges \citep{FDA2019RareDiseases,gu2024incorporatingexternaldataanalyzing}.
Existing AIPW-based nonlinear adjustment frameworks (\citealp{tu2024Unified,rafi2023Efficient,bannick2025General}),
however, are not designed to integrate these heterogeneous information
sources, as they typically rely on a single nuisance estimate and
lack a systematic mechanism for combining multiple internal predictors
or external information sources. To fill this gap, we propose a unified
calibration framework for integrating heterogeneous information that
accommodates both internal and external borrowing, thereby enabling
estimators that jointly exploit covariates, cross-stratum information,
and auxiliary data sources.

A central methodological feature of our approach is the use of calibration
weights under CAR designs. Although calibration has been studied in
survey sampling (\citealp{deville1992Calibration,kwon2025Debiased}),
missing data problems (\citealp{qin2007EmpiricalLikelihoodBased,tan2014Secondorder}),
and observational studies (\citealp{chan2016Globally}), its analysis under CAR designs raises distinct theoretical issues. Unlike the aforementioned
settings where samples are typically independent and identically distributed
(i.i.d.), CAR designs induce complex dependence structures among treatment
assignments within strata. We address this dependence through a conditional
asymptotic argument. Specifically, we condition on the realized stratum
indicators and treatment assignments, treat them as fixed, and establish
large-sample results using conditional laws of large numbers and conditional
central limit theorems. This approach provides a tractable framework
for inference under CAR and yields proof techniques that extend to
regimes with a growing number of strata and an increasing number of
information sources, which may be of independent interest. We summarize
the main contributions as follows.
\begin{itemize}
\item \textbf{A unified calibration framework.} We introduce a calibration-based
framework for estimation and inference under CAR. This framework is
unified in three aspects. First, it provides a common formulation
that recovers many recent covariate adjustment procedures
as special cases (e.g., \citealp{bugni2019Inference,cohen2024Noharm,ma2022Regression,tu2024Unified,bannick2025General,ye2022Inference,ye2023Better,liu2023Lassoadjusted}).
Second, it places internal and external information borrowing within
a single architecture, yielding a systematic approach to integrating
heterogeneous information sources. Third, it applies broadly across
CAR schemes satisfying Assumption~\ref{assu:treatment assignment},
so that the resulting inference procedure is not tied to a particular
randomization method.
\item \textbf{Flexible and robust information borrowing.} We develop practical
constructions of the information proxy vector that accommodate a wide
range of internal and external information sources. Internally, the
framework can borrow across strata and aggregate heterogeneous machine
learning predictions. Externally, it can incorporate information from
historical trials and real-world data. Importantly, our framework
is model-agnostic regarding information sources. We refer to this
property as \textit{robustness}, meaning that the validity of our
statistical inference holds even if the utilized information is biased
or generated by inaccurate models.
\item \textbf{General inference theory under CAR.} We provide a rigorous
theoretical foundation, proving that our estimator is asymptotically
normal with a consistently estimable variance. We establish a \textit{no-harm}
efficiency guarantee, ensuring that utilizing additional information
sources improves, or at worst maintains, estimation efficiency. Furthermore,
we develop proof techniques tailored to CAR-induced dependence, distinct
from existing arguments (e.g., \citealp{bugni2018Inference,bugni2019Inference,ma2022Regression,liu2023Lassoadjusted}).
These techniques accommodate a growing number of strata and an increasing
number of information sources, offering tools that may be of independent
interest in related problems.
\end{itemize}
The paper is organized as follows. Figure~\ref{fig:Roadmap} provides
an overview. Section~\ref{sec:Basic-framework-and} introduces the
setting and our unified calibration framework. Section~\ref{sec:Discussions-of-the}
discusses practical strategies for constructing the information proxy
vector $\bs{\xi}_{n}$, including approaches that borrow auxiliary
information from both internal and external sources. Section~\ref{sec:Asymptotic-properties}
establishes the large-sample properties of the proposed estimator,
including asymptotic normality, consistent variance estimation, and
efficiency comparisons. Section~\ref{sec:Extensions} presents theoretical
extensions, including settings with diverging numbers of strata and
a growing dimension of $\bs{\xi}_{n}$, as well as general discrepancy
measures. Finally, Sections~\ref{sec:Simulation-studies} and \ref{sec:Empirical-application}
assess finite-sample performance via simulation and illustrate the
method using experimental data from \citet{dupas2018Bankinga}. All
proofs are provided in the Appendix (Supplementary Material, \citealp{ma2026Integrating}).

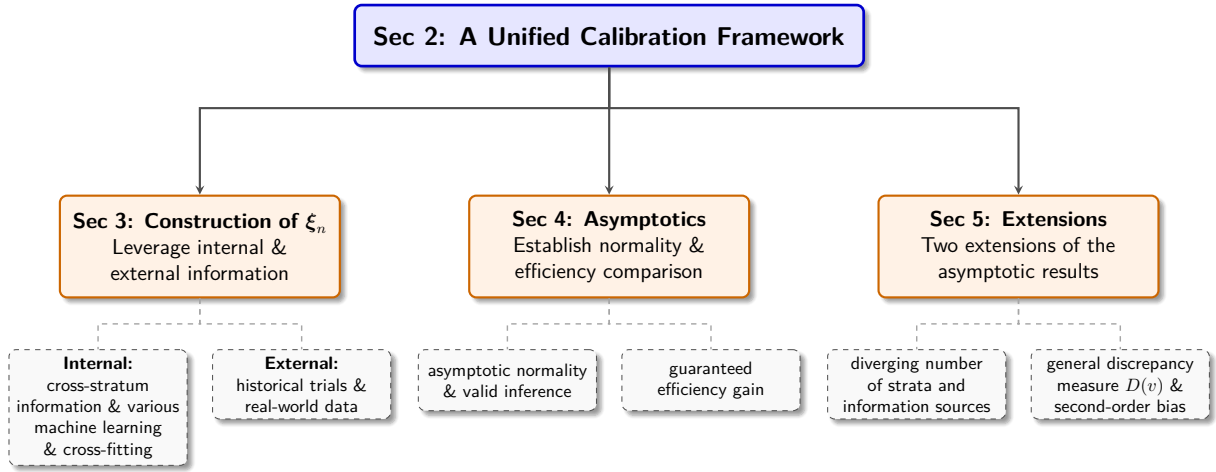
\begin{figure}[!tbh]
\resizebox{\textwidth}{!}{

\tikzset{
    % --- Define basic styles ---
    base/.style = {
        rectangle,
        rounded corners,
        draw=black,
        text centered,
        font=\sffamily, % Use sans-serif font, remove if standard serif is preferred
        blur shadow={shadow blur steps=5}
    },
    % Top level node style (Section 2)
    topnode/.style = {
        base,
        fill=blue!10,
        draw=blue!80!black,
        line width=1.5pt,
        minimum width=10cm,
        minimum height=1.2cm, % Adjusted height after removing subtitle
        font=\sffamily\large\bfseries
    },
    % Middle level node style (Sections 3, 4, 5)
    midnode/.style = {
        base,
        fill=orange!10,
        draw=orange!80!black,
        line width=1.2pt,
        minimum width=5.5cm,
        minimum height=2cm,
        text width=5.2cm % Force text wrap for longer titles
    },
    % Bottom level detail node style
    subnode/.style = {
        base,
        fill=gray!5,
        draw=gray!60!black,
        dashed, % Dashed border indicates extension/detail status
        minimum width=3.5cm,
        minimum height=1.2cm,
        text width=3.2cm,
        font=\sffamily\footnotesize
    },
    % Connector line style (Main arrows)
    connector/.style = {
        ->,
        >=stealth,
        line width=1.2pt,
        color=black!70
    },
    % Sub-connector line style (Dashed lines to details)
    subconnector/.style = {
        -,
        dashed,
        line width=0.8pt,
        color=gray!70
    }
}

\begin{tikzpicture}[node distance=2cm and 2.5cm]

    % --- Place Nodes ---

    % Level 1: Top Node (Section 2)
    \node (sec2) [topnode] {Sec 2: A Unified Calibration Framework};

    % Level 2: Middle Nodes (Sections 3, 4, 5)
    % Strategy: Place the middle node (Sec 4) first, then position 3 and 5 relative to it.
    \node (sec4) [midnode, below=2.5cm of sec2] {\textbf{Sec 4: Asymptotics}\\ Establish normality \& efficiency comparison};
    \node (sec3) [midnode, left=of sec4] {\textbf{Sec 3: Construction of $\bs{\xi}_{n}$}\\ Leverage internal \& external information};
    \node (sec5) [midnode, right=of sec4] {\textbf{Sec 5: Extensions}\\ Two extensions of the asymptotic results};

    % Level 3: Sub Nodes (Details)
    % Sub-nodes for Section 3
    % Using xshift to manually space them out horizontally below the parent node
    \node (sub3a) [subnode, below=1cm of sec3, xshift=-2cm] {\textbf{Internal:} \\ cross-stratum information \& various machine learning \& cross-fitting};
    \node (sub3b) [subnode, below=1cm of sec3, xshift=2cm] {\textbf{External:} \\ historical trials \& real-world data};

    % Sub-nodes for Section 4
    \node (sub4a) [subnode, below=1cm of sec4, xshift=-2cm] {asymptotic normality \& valid inference};
    \node (sub4b) [subnode, below=1cm of sec4, xshift=2cm] {guaranteed \\ efficiency gain};

    % Sub-nodes for Section 5
    \node (sub5a) [subnode, below=1cm of sec5, xshift=-2cm] {diverging number of strata and information sources};
    \node (sub5b) [subnode, below=1cm of sec5, xshift=2cm] {general discrepancy measure $D(v)$ \& second-order bias};

    % --- Draw Connectors ---

    % Main connectors (from Sec 2 to 3, 4, 5)
    % Using path construction (-|) to create right-angle arrows
    \draw [connector] (sec2.south) -- +(0,-0.8) -| (sec3.north);
    \draw [connector] (sec2.south) -- (sec4.north);
    \draw [connector] (sec2.south) -- +(0,-0.8) -| (sec5.north);

    % Sub-connectors (detail connections)
    % Using temporary coordinates (++(0,-0.5)) to create a clean branching look
    \draw [subconnector] (sec3.south) -- ++(0,-0.5) -| (sub3a.north);
    \draw [subconnector] (sec3.south) -- ++(0,-0.5) -| (sub3b.north);

    \draw [subconnector] (sec4.south) -- ++(0,-0.5) -| (sub4a.north);
    \draw [subconnector] (sec4.south) -- ++(0,-0.5) -| (sub4b.north);

    \draw [subconnector] (sec5.south) -- ++(0,-0.5) -| (sub5a.north);
    \draw [subconnector] (sec5.south) -- ++(0,-0.5) -| (sub5b.north);

\end{tikzpicture}

}

\caption{Roadmap of the proposed unified calibration framework and theoretical
analysis. \label{fig:Roadmap}}
\end{figure}

\textit{Notation.} We maintain the following notation conventions
throughout the paper. For any column vector $\bs x=\left(x_{1},x_{2},\ldots,x_{d}\right)^{\top}\in\mathbb{R}^{d}$,
where $\mathbb{R}^{d}$ is the $d$-dimensional Euclidean space, $\left\Vert \bs x\right\Vert =\left(\bs x^{\top}\bs x\right)^{1/2}$
denotes its Euclidean norm. For any matrix $A=\left(a_{ij}\right)_{n\times m}$,
$\left\Vert A\right\Vert $ denotes its maximum singular value, i.e.,
the operator norm, $\left\Vert A\right\Vert _{F}=\sqrt{\tr\left(AA^{\trans}\right)}$
denotes it Frobenius norm, and $A^{+}$ denotes its Moore-Penrose
inverse. For two positive non-random sequence $a_{n},b_{n}$ and random
vector sequence $X_{n}$, $X_{n}=o_{P}(a_{n})$ means $P\left(\left\Vert X_{n}\right\Vert >a_{n}\epsilon\right)\to0$
as $n\to\infty$ for any $\epsilon>0$ and $X_{n}=O_{P}(a_{n})$ means
for any $\epsilon>0$, there exists a constant $M>0$ such that $\limsup_{n\to\infty}P\left(\left\Vert X_{n}\right\Vert \geq a_{n}M\right)<\epsilon$.
The notation $\1(\cdot)$ denotes the indicator function, which takes
the value 1 if the condition inside the parentheses is true and 0
otherwise.

\section{A unified calibration framework\label{sec:Basic-framework-and}}

\subsection{Preliminaries}

Let $A_{i}\in\{0,1\}$ ($i=1,\ldots,n$) denote the treatment assignment
indicator, where $A_{i}=1$ indicates that the $i$-th unit is assigned
to the treatment group and $A_{i}=0$ otherwise. We assume the assignments
$\{A_{i}\}_{i=1}^{n}$ are generated via a CAR design satisfying Assumption~\ref{assu:treatment assignment}.
Consequently, the variables $A_{i}$ are typically not i.i.d. Adopting
the potential outcomes framework (\citealp{imbens2015Causal}), we
define $Y_{i}(1)$ and $Y_{i}(0)$ as the potential outcomes under
treatment and control, respectively. The observed outcome $Y_{i}$
is determined by $Y_{i}=A_{i}Y_{i}(1)+(1-A_{i})Y_{i}(0)$.

The experimental design stratifies units into $K$ strata, with $B_{i}\in\{1,\ldots,K\}$
indicating the stratum of unit $i$. For notational convenience, let
$[k]=\{i:B_{i}=k\}$ denote the set of indices for units belonging
to stratum $k$. To ensure non-empty strata, we assume positive assignment
probabilities: $p_{[k]}=P(B_{i}=k)>0$ for all $k\in\{1,\ldots,K\}$
and $i\in\{1,\ldots,n\}$. The target treatment allocation proportion
stratum $k$ is $\pi_{[k]}=P(A_{i}=1\mid B_{i}=k)\in(0,1)$. Each
unit has a $p$-dimensional baseline covariate vector $\bs X_{i}=(X_{i1},\ldots,X_{ip})^{\top}\in\mathcal{X}\subset\R^{p}$,
which may be either low- or high-dimensional. We assume that the covariate
vector $\bs X_{i}$ contains the stratum indicator $B_{i}$, but to
highlight the stratum indicator $B_{i}$, we sometimes use the notation
$(\bs X_{i},B_{i})$.

Let subscripts 1 and 0 denote treatment and control groups, respectively.
The treatment group contains $n_{1}=\sum_{i=1}^{n}A_{i}$ units and
the control group contains $n_{0}=\sum_{i=1}^{n}(1-A_{i})$ units.
For stratum-specific quantities, we use subscript $[k]$: let $n_{[k]}=\sum_{i\in[k]}1$
denote the stratum size, with $n_{1[k]}=\sum_{i\in[k]}A_{i}$ and
$n_{0[k]}=\sum_{i\in[k]}(1-A_{i})$ representing treated and control
units in stratum $k$, respectively. The stratum proportion and treatment
allocation proportion are defined as $p_{n[k]}=n_{[k]}/n$ and $\pi_{n[k]}=n_{1[k]}/n_{[k]}$,
correspondingly.

Our parameter of interest is the population average treatment effect
(ATE): 
\[
\tau=\e[Y_{i}(1)-Y_{i}(0)].
\]
Under CAR, the ATE parameter $\tau$ can be consistently estimated
by aggregating the treatment effect estimates from each stratum. A
simple estimator for this is the stratified difference-in-means estimator
(\citealp{bugni2019Inference,ma2022Regression}):
\[
\widehat{\tau}_{\mathrm{sdim}}:=\sum_{k=1}^{K}p_{n[k]}\left(\overline{Y}_{1[k]}-\overline{Y}_{0[k]}\right),
\]
where $\overline{Y}_{a[k]}:=\frac{1}{n_{a[k]}}\sum_{i\in[k]}\1(A_{i}=a)Y_{i}$
for $a\in\{0,1\}$ and $k=1,\ldots,K$. To improve estimation efficiency,
recent literature has proposed nonlinear covariate adjustment methods
(\citealp{tu2024Unified,rafi2023Efficient,bannick2025General}). These
methods typically use the augmented inverse probability weighting
(AIPW, \citealp{robins1994Estimation}) estimator and rely on estimating
the conditional mean function $h_{a[k]}^{*}(\bs X_{i})=\e\left[Y_{i}(a)\mid\bs X_{i},B_{i}=k\right]$.
However, AIPW-based adjustment methods are limited in that they cannot
combine different estimates of $h_{a[k]}^{*}(\bs X_{i})$. For instance,
these estimates may leverage internal information via cross-stratum
borrowing and heterogeneous machine learning predictions, or incorporate
external information from historical trials, real-world data, and
functional forms suggested by domain experts (see Section~\ref{sec:Discussions-of-the}
for a detailed discussion). In the following, we will introduce a
unified adjustment framework that can combine these diverse estimates
of $h_{a[k]}^{*}(\bs X_{i})$, leading to more efficient and robust
ATE estimators.

\subsection{Construction of the calibration estimator}

Let $D(v):\mathbb{R}\to\R$ be a twice continuously differentiable
and strictly convex function that measures the discrepancy from $v$
to 1, e.g., $D(v)=(v-1)^{2}/2$ and $D(v)=v-\log v$. Suppose $\bs{\xi}_{n}:\mathcal{X}\to\R^{d}$
($n\geq1$) is a sequence of $\R^{d}$-valued (possibly) random functions
of $\bs X_{i}$, which we refer to as the \textit{information proxy
vector}. Typically, the elements of $\bs{\xi}_{n}$ may be different
estimates of the conditional mean function $h_{a[k]}^{*}(\bs X_{i})$.
Based on $\bs{\xi}_{n}$, we propose the calibration estimator: 
\begin{equation}
\widehat{\tau}_{\mathrm{cal}}:=\widehat{\tau}_{\mathrm{sdim}}+\frac{1}{n}\sum_{i=1}^{n}\widehat{w}_{i}r_{i},\label{eq:cal estimator}
\end{equation}
where $r_{i}:=\sum_{k=1}^{K}\bigl\{\frac{A_{i}}{\pi_{n[k]}}(Y_{i}-\overline{Y}_{1[k]})-\frac{1-A_{i}}{1-\pi_{n[k]}}(Y_{i}-\overline{Y}_{0[k]})\bigr\}\1(B_{i}=k)$
and the calibration weights $\widehat{w}_{i}$'s ($i=1,\ldots,n$)
solve the calibration problem:
\begin{equation}
\begin{cases}
(\widehat{w}_{1},\ldots,\widehat{w}_{n})=\min_{w_{i},1\leq i\leq n}\sum_{i=1}^{n}D(w_{i})\ \text{subject to }\\
\frac{1}{n}\sum_{i=1}^{n}w_{i}\left\{ A_{i}-\pi_{n[k]}\right\} \1(B_{i}=k)\left\{ \bs{\xi}_{n}(\bs X_{i})-\overline{\bs{\xi}}_{n[k]}\right\} =0, & \forall k=1,\ldots,K.
\end{cases}\label{eq:cal opt problem}
\end{equation}
In (\ref{eq:cal opt problem}), $\overline{\bs{\xi}}_{n[k]}:=\frac{1}{n_{[k]}}\sum_{i\in[k]}\bs{\xi}_{n}(\bs X_{i})$
is defined as the stratum-specific sample mean for $\bs{\xi}_{n}(\bs X_{i})$.
The calibration problem (\ref{eq:cal opt problem}) is a convex optimization
problem involving $dK$ linear constraints. This can be efficiently
solved using standard convex optimization software via its dual formulation
(see \citealp{boyd2004Convex}). The calibration estimator (\ref{eq:cal estimator})
consists of two components: (i) the stratified difference-in-means
estimator $\widehat{\tau}_{\mathrm{sdim}}$ and (ii) a correction
term constructed from weighted residuals, where the weights are determined
by the calibration problem (\ref{eq:cal opt problem}). Analogous
to linear regression analysis, the residuals $r_{i}=\sum_{k=1}^{K}\bigl\{\frac{A_{i}}{\pi_{n[k]}}(Y_{i}-\overline{Y}_{1[k]})-\frac{1-A_{i}}{1-\pi_{n[k]}}(Y_{i}-\overline{Y}_{0[k]})\bigr\}\1(B_{i}=k)$,
$i=1,\ldots,n$, represent the part of $\sum_{k=1}^{K}\bigl\{\frac{A_{i}}{\pi_{n[k]}}Y_{i}-\frac{1-A_{i}}{1-\pi_{n[k]}}Y_{i}\bigr\}\1(B_{i}=k)$
that remains unexplained by the stratum mean. Geometrically, the residuals
$r_{i}$ arise from projecting $(\sum_{k=1}^{K}\bigl\{\frac{A_{i}}{\pi_{n[k]}}Y_{i}-\frac{1-A_{i}}{1-\pi_{n[k]}}Y_{i}\bigr\}\1(B_{i}=k):1\leq i\leq n)$
onto the orthogonal complement of the space spanned by $(1,...,1)$,
which inherently implies $\sum_{i=1}^{n}r_{i}=0$. The calibration
problem (\ref{eq:cal opt problem}) enforces the balance of $\bs{\xi}_{n}(\bs X_{i})-\overline{\bs{\xi}}_{n[k]}$
across different treatment groups within strata when the observed
samples are weighted by $\widehat{w}_{i}$, $i=1,\ldots,n$, thereby
naturally incorporating the information contained in $\bs{\xi}_{n}$
into the calibration weights $\widehat{w}_{i}$'s. This incorporation
of information allows the calibration weights $\widehat{w}_{i}$ to
further explain the variability in $\sum_{k=1}^{K}\bigl\{\frac{A_{i}}{\pi_{n[k]}}Y_{i}-\frac{1-A_{i}}{1-\pi_{n[k]}}Y_{i}\bigr\}\1(B_{i}=k)$.
Consequently, the weighted average of the residuals, $\frac{1}{n}\sum_{i=1}^{n}\widehat{w}_{i}r_{i}$,
acts as a correction term representing the part explained by the calibration
weights. If $\widehat{w}_{i}=1$ for all $i=1,\ldots,n$, which occurs,
for instance, when $\bs{\xi}_{n}(\bs X)-\overline{\bs{\xi}}_{n[k]}$
is already balanced by the randomization procedure, then $\widehat{\tau}_{\mathrm{cal}}$
reduces to $\widehat{\tau}_{\mathrm{sdim}}$.

Our calibration estimator (\ref{eq:cal estimator}) provides a unified
information integration framework. With an appropriate choice of the
information proxy vector $\boldsymbol{\xi}_{n}$, it recovers many
existing covariate adjustment methods, which can be viewed as a special
case of internal information borrowing (e.g., \citealp{liu2023Lassoadjusted,cohen2024Noharm,tu2024Unified,bannick2025General,ye2022Inference,ye2023Better,gu2024incorporatingexternaldataanalyzing,bugni2019Inference,ma2022Regression}).
For instance, when $\boldsymbol{\xi}_{n}$ is specified as a Lasso-based
estimate, (\ref{eq:cal estimator}) becomes analogous to the Lasso-adjusted
estimator of \citet{liu2023Lassoadjusted}. Beyond internal adjustment,
(\ref{eq:cal estimator}) also accommodates external information proxies,
in parallel to external-borrowing approaches such as \citet{gu2024incorporatingexternaldataanalyzing},
and thus enables new estimators that jointly leverage internal and
external information.

A key property of our calibration estimator (\ref{eq:cal estimator})
is its invariance under affine transformations of $\bs{\xi}_{n}(\bs X)$.
Specifically, replacing $\bs{\xi}_{n}(\bs X)$ with $\mathbf{Q}\bs{\xi}_{n}(\bs X)+\bs q$
in the optimization problem (\ref{eq:cal opt problem}) results in
the same calibration estimator $\widehat{\tau}_{\mathrm{cal}}$, where
$\mathbf{Q}\in\mathbb{R}^{d\times d}$ is any invertible $d\times d$
matrix and $\bs q\in\mathbb{R}^{d}$ is any $d$-dimensional vector.
If $\bs{\xi}_{n}(\bs X)$ serves as an estimate of the conditional
mean function $h_{a[k]}^{*}(\bs X_{i})$, this property allows for
misspecification up to an affine transformation, offering a flexibility
that is absent in standard AIPW-based covariate adjustment methods.

\section{Strategies for constructing the information proxy $\protect\bs{\xi}_{n}$\label{sec:Discussions-of-the}}

The efficacy and flexibility of the proposed unified framework relies
on the construction of $\bs{\xi}_{n}$, which serves as a proxy for
auxiliary information. Accordingly, we categorize the strategies for
determining $\bs{\xi}_{n}$ into internal and external information
borrowing. Sections~\ref{subsec:Stratum-information-borrowing}--\ref{subsec:cross-fitting}
focus on internal strategies, encompassing cross-stratum information
borrowing, the aggregation of different machine learning predictions,
and cross-fitting techniques, whereas Section~\ref{subsec:-from-external}
discusses the incorporation of external information such as historical
trials and real-world data.

\subsection{Cross-stratum information borrowing\label{subsec:Stratum-information-borrowing}}

Suppose $\widehat{h}_{a[k]}(\cdot)$, where $a\in\{0,1\}$ and $1\leq k\leq K$,
are the estimators for $h_{a[k]}^{*}(\cdot)$. According to Theorem~\ref{thm:asymptotic properties of tau_cal}
(to be established in Section~\ref{sec:Asymptotic-properties}),
if $\widehat{h}_{a[k]}(\cdot)$ is consistent to $h_{a[k]}^{*}(\cdot)$,
then taking $\bs{\xi}_{n}(\bs X)=\left(\sum_{k=1}^{K}\widehat{h}_{1[k]}(\bs X)\1(B=k),\sum_{k=1}^{K}\widehat{h}_{0[k]}(\bs X)\1(B=k)\right)^{\trans}$
in (\ref{eq:cal estimator}) will lead to a semiparametrically efficient
estimator for $\tau$. However, it is important to note that this
approach ensures each stratum performs covariate adjustment using
only the information available within that stratum itself.\textcolor{red}{{}
}However, borrowing information across strata is often advantageous,
particularly when the relationship between covariates $\bs X$ and
potential outcomes $(Y(1),Y(0))$ is stable across different strata.
Our framework accommodates this naturally by taking $\bs{\xi}_{n}(\bs X)=\left(\widehat{h}_{1[k]}(\bs X),\widehat{h}_{0[k]}(\bs X):1\leq k\le K\right)^{\trans}$
in (\ref{eq:cal estimator}). This specification allows each stratum
to use information from all strata, leading to a more efficient calibration
estimator.

\subsection{Integration of heterogeneous machine learning predictions}

In practice, a variety of machine learning methods can be used to
estimate the conditional mean function $h_{a[k]}^{*}(\bs X)$, such
as deep neural networks (\citealp{lecun2015deep,jiao2023Deep,farrell2021Deepa}),
random forests (\citealp{breiman2001random,wager2018estimation}),
Lasso (\citealp{tibshirani1996Regression}), and others. Our framework
provides a natural approach to combine these machine learning estimates.
Suppose $\widehat{h}_{a[k]}^{\text{rf}}(\bs X)$ and $\widehat{h}_{a[k]}^{\text{nn}}(\bs X)$
are the estimates of $h_{a[k]}^{*}(\bs X)$ based on random forests
and deep neural networks, respectively. We can then let $\bs{\xi}_{n}(\bs X)=\left(\widehat{h}_{a[k]}^{\text{rf}}(\bs X),\widehat{h}_{a[k]}^{\text{nn}}(\bs X):a\in\{0,1\},1\leq k\le K\right)^{\trans}$
in (\ref{eq:cal estimator}) to derive a calibration estimator for
$\tau$. According to Theorem \ref{thm:efficiency comparison thm}
(to be established in Section~\ref{sec:Asymptotic-properties}),
this estimator will be more efficient than those based solely on random
forests or neural networks when a single machine learning method fails
to fully capture the conditional mean function $h_{a[k]}^{*}(\bs X)$.
This improvement enhances the efficiency of covariate adjustment methods.

\subsection{Implementation via cross-fitting and sample splitting\label{subsec:cross-fitting}}

When estimating $\bs{\xi}_{n}$ via machine learning methods, it is
often desirable to rely on an independent sample in order to mitigate
overfitting and to make Assumption \ref{assu:conditions on =00005Cxi}
more plausible. Such independence can be achieved through sample-splitting.
However, sample-splitting typically reduces efficiency. To recover
efficiency, we adopt the cross-fitting technique (see \citealp{chernozhukov2018Double,tu2024Unified,bannick2025General,rafi2023Efficient}).
Specifically, we follow the sample-splitting procedure outlined in
\citet[Section 4]{rafi2023Efficient}, which ensures that each fold
contains data from every stratum and treatment arm. Suppose the units
$\{1,\ldots,n\}$ are partitioned into two equal folds, $I_{0}$ and
$I_{1}$. The cross-fitted calibration estimator is then defined as
$\widehat{\tau}_{\mathrm{cal}}^{\mathrm{CF}}:=\frac{1}{2}\widehat{\tau}_{\mathrm{cal}}^{(0)}+\frac{1}{2}\widehat{\tau}_{\mathrm{cal}}^{(1)}$,
where 
\begin{align*}
\widehat{\tau}_{\mathrm{cal}}^{(\iota)} & =\sum_{k=1}^{K}p_{n[k]}^{(\iota)}\left(\overline{Y}_{1[k]}^{(\iota)}-\overline{Y}_{0[k]}^{(\iota)}\right)\\
 & \quad+\frac{1}{\left|I_{\iota}\right|}\sum_{i\in I_{\iota}}\widehat{w}_{i}\sum_{k=1}^{K}\left\{ \frac{A_{i}}{\pi_{n[k]}^{(\iota)}}\left(Y_{i}-\overline{Y}_{1[k]}^{(\iota)}\right)-\frac{1-A_{i}}{1-\pi_{n[k]}^{(\iota)}}\left(Y_{i}-\overline{Y}_{0[k]}^{(\iota)}\right)\right\} \1(B_{i}=k)
\end{align*}
with $\left|I_{\iota}\right|=\sum_{i\in I_{\iota}}1$, $p_{n[k]}^{(\iota)}=\sum_{i\in I_{\iota}}\1(B_{i}=k)/\left|I_{\iota}\right|$,
$\overline{Y}_{a[k]}^{(\iota)}=\sum_{i\in I_{\iota}}\1(A_{i}=a,B_{i}=k)Y_{i}/\sum_{i\in I_{\iota}}\1(A_{i}=a,B_{i}=k)$,
$\pi_{n[k]}^{(\iota)}=\sum_{i\in I_{\iota}}\1(A_{i}=a,B_{i}=k)/\sum_{i\in I_{\iota}}\1(B_{i}=k)$
for $\iota\in\{0,1\}$. The weights $\widehat{w}_{i}$'s ($i\in I_{\iota}$)
are obtained by solving the following calibration problem:
\[
\begin{cases}
(\widehat{w}_{i})_{i\in I_{\iota}}=\min_{w_{i}:i\in I_{\iota}}\sum_{i\in I_{\iota}}D(w_{i})\ \text{subject to }\\
\sum_{i\in I_{\iota}}w_{i}\left\{ A_{i}-\pi_{n[k]}^{(\iota)}\right\} \1(B_{i}=k)\left\{ \bs{\xi}_{n}^{(1-\iota)}(\bs X_{i})-\overline{\bs{\xi}}_{n[k]}^{(1-\iota)}\right\} =0, & \forall k=1,\ldots,K,
\end{cases}
\]
where $\bs{\xi}_{n}^{(1-\iota)}$ is measurable with respect to $\{(Y_{i},\bs X_{i},A_{i}):i\in I_{\iota}\}$
and $\overline{\bs{\xi}}_{n[k]}^{(1-\iota)}:=\sum_{i\in I_{\iota}}\bs{\xi}_{n}(\bs X_{i})\1(B_{i}=k)/\sum_{i\in I_{\iota}}\1(B_{i}=k)$
denotes the stratum-specific sample mean for $\bs{\xi}_{n}^{(1-\iota)}(\bs X)$.
The asymptotic properties of $\widehat{\tau}_{\mathrm{cal}}^{\mathrm{CF}}$
can be established by combining the proof of Theorem~\ref{thm:asymptotic properties of tau_cal}
with the arguments from the proof of \citet[Theorem 4.1]{rafi2023Efficient}
and we omit the details here. We will evaluate the finite-sample performance
of this cross-fitted estimator $\widehat{\tau}_{\mathrm{cal}}^{\mathrm{CF}}$
in our simulation studies.

\subsection{Leveraging historical and real-world data\label{subsec:-from-external}}

Beyond internal strategies, constructing $\bs{\xi}_{n}$ using external
information offers a powerful way to enhance efficiency. In this subsection,
we focus on two primary sources: historical clinical trials and real-world
data. While these sources provide valuable information, they often
differ distributionally from the current trial. A key distinction
of our framework is its ability to leverage such heterogeneous external
data without requiring restrictive similarity assumptions.

First, historical clinical trials often provide readily available
data from settings similar to the current study. Since treatments
often vary between trials, existing literature has predominantly focused
on leveraging historical control data to enhance the precision of
estimates in the current trial (\citealp{pocock1975Sequential,callegaro2023Historical}).
These methods usually rely on the assumption that the historical and
current control arms are comparable. In contrast, a distinct advantage
of our unified framework is its assumption-lean nature: we impose
no assumptions on the validity or direct transferability of the borrowed
information. Consequently, we are not restricted to borrowing solely
from historical control arms; we can flexibly incorporate information
from historical studies involving treatments that may differ from
those in the current trial.

Second, beyond historical clinical trials, real-world data represents
another vast information source. Real-world data encompasses data
generated from routine healthcare delivery and patient monitoring,
including electronic health records (EHRs), disease registries, and
increasingly, data from wearable devices. Unlike tightly controlled
clinical trials, observational studies based on real-world data typically
involve much larger and more diverse patient populations, which can
significantly strengthen the statistical power of the analysis.

Crucially, both historical trial data and real-world data may exhibit
a ``covariate shift'', i.e., the distribution of covariates $\bs X$
in these external sources often differs from that of the current trial.
However, a common phenomenon is that while the marginal distribution
of $\bs X$ changes, the conditional distribution of the potential
outcomes $(Y(1),Y(0))$ given the covariates $\bs X$ remains comparable.
This stability is important because the optimal choice of $\bs{\xi}_{n}(\bs X)$,
which ensures that the resulting estimator $\widehat{\tau}_{\mathrm{cal}}$
is semiparametrically efficient, depends solely on the conditional
distribution of $(Y(1),Y(0))$ given $\bs X$ (see Theorem~\ref{thm:asymptotic properties of tau_cal}).
In practice, we can leverage this by estimating the conditional mean
functions $h_{1[k]}^{*}(\bs X)$ and $h_{0[k]}^{*}(\bs X)$ ($1\leq k\leq K$)
from these external datasets and setting $\bs{\xi}_{n}(\bs X)$ as
the vector of these estimates. Our Theorem~\ref{thm:efficiency comparison thm}
guarantees that incorporating these estimates from the external datasets
will result in a calibration estimator with no greater asymptotic
variance. Consequently, our framework provides a robust strategy for
improving estimation efficiency without the risk of ``negative transfer''.
Unlike existing methods that rely on Bayesian frameworks (\citealp{hobbs2011Hierarchical,ibrahim2015Power})
or Trans-Lasso (\citealp{gu2024incorporatingexternaldataanalyzing}),
where the transfer of information depends on the prior distributions
and/or model assumptions, our approach is entirely model-free. Notably,
it imposes no similarity constraints between the external data and
the target data-generating process, thereby enhancing flexibility
and robustness against real-world data complexities.

\section{Asymptotic properties\label{sec:Asymptotic-properties}}

We make the following assumptions.
\begin{assumption}
\label{assu:independent sampling}$\bs W_{i}=(Y_{i}(1),Y_{i}(0),\bs X_{i}^{\trans},B_{i})^{\trans},$
$i=1,\ldots,n$, are i.i.d. samples from the population distribution
of $\bs W=(Y(1),Y(0),\bs X^{\trans},B)^{\trans}$, and we denote $\bs W^{(n)}=\{\bs W_{1},\ldots,\bs W_{n}\}$.
Besides, $\sup_{1\leq k\leq K}\left[\left|Y_{i}(a)\right|^{2+\ep}\mid B_{i}=k\right]<\infty$
for both $a=0,1$, where $0<\ep\leq1$ is a constant.
\end{assumption}
\begin{assumption}
[Treatment assignment]\label{assu:treatment assignment}Let $A^{(n)}=\{A_{1},\ldots,A_{n}\}$
and $B^{(n)}=\{B_{1},\ldots,B_{n}\}$. The treatment assignment mechanism
satisfies the following conditions:
\begin{enumerate}[label=(\arabic*)]
\item conditional on the stratum indicators $B^{(n)}$, the treatment assignments
$A^{(n)}$ are independent of $\bs W^{(n)}$, i.e., $\bs W^{(n)}\perp A^{(n)}\mid B^{(n)}$;
\item $\sup_{1\leq k\leq K}\left|\pi_{n[k]}-\pi_{[k]}\right|\tp0$ as $n\to\infty$,
where $\pi_{[k]}$, $k=1,\ldots,K$, satisfy 
\[
0<\inf_{n\geq1}\inf_{1\leq k\leq K}\pi_{[k]}\leq\sup_{n\geq1}\sup_{1\leq k\leq K}\pi_{[k]}<1.
\]
Moreover, $\lim_{n\to\infty}P(\inf_{1\leq k\leq K}n_{[k]}\geq d+2)=1$,
where $d$ denotes the dimension of $\bs{\xi}_{n}$.
\end{enumerate}
\end{assumption}
\begin{assumption}
\label{assu:conditions on =00005Cxi}$d$ and $K$ are fixed numbers.
The sequence of (possibly) random functions $\bs{\xi}_{n}$ satisfies
the following conditions:
\begin{enumerate}[label=(\arabic*)]
\item There exists a sequence of non-stochastic functions $\bs{\xi}_{n}^{*}:\mathcal{X}\to\R^{d}$
($n\geq1$) such that
\begin{enumerate}
\item $\sup_{n\geq1}\e\left[\left\Vert \bs{\xi}_{n}^{*}(\bs X_{i})\right\Vert ^{2+\epsilon}\right]<\infty$,
where $\epsilon$ is defined in Assumption \ref{assu:independent sampling};
\item the minimal non-zero singular value of $\e\left[\widetilde{\bs{\xi}}_{n}^{*}(\bs X_{i})\widetilde{\bs{\xi}}_{n}^{*}(\bs X_{i})^{\trans}\mid B_{i}=k\right]$
is larger than some constant $c>0$ uniformly over $n\geq1$ and $k=1,\ldots,K$;
\item for every $k=1,\ldots,K$,
\[
\frac{1}{n_{1[k]}}\sum_{i\in[k]}A_{i}\left\{ \bs{\xi}_{n}(\bs X_{i})-\bs{\xi}_{n}^{*}(\bs X_{i})\right\} -\frac{1}{n_{0[k]}}\sum_{i\in[k]}\left(1-A_{i}\right)\left\{ \bs{\xi}_{n}(\bs X_{i})-\bs{\xi}_{n}^{*}(\bs X_{i})\right\} =o_{P}(n^{-1/2})
\]
and $\frac{1}{n_{[k]}}\sum_{i\in[k]}\left\Vert \bs{\xi}_{n}(\bs X_{i})-\bs{\xi}_{n}^{*}(\bs X_{i})\right\Vert ^{2}=o_{P}(1)$\textup{;}
\item $\liminf_{n\to\infty}(\varsigma_{\widetilde{Y}}^{2}-\varsigma_{\widetilde{Y}\mid\widetilde{\bs{\xi}}_{n}^{*}}^{2})>0$,
where $\varsigma_{\widetilde{Y}}^{2}$ and $\varsigma_{\widetilde{Y}\mid\widetilde{\bs{\xi}}_{n}^{*}}^{2}$
are defined in Appendix~\ref{sec:Proofs-for-the};
\end{enumerate}
\item For every $k=1,\ldots,K$, 
\[
\left\Vert \left\{ \frac{1}{n}\sum_{i\in[k]}(A_{i}-\pi_{n[k]})^{2}\left\{ \bs{\xi}_{n}(\bs X_{i})-\overline{\bs{\xi}}_{n[k]}\right\} \left\{ \bs{\xi}_{n}(\bs X_{i})-\overline{\bs{\xi}}_{n[k]}\right\} ^{\trans}\right\} ^{+}\right\Vert =O_{P}(1).
\]
\end{enumerate}
\end{assumption}
In this section, we focus on the case where the number of strata,
$K$, is fixed, and leave the case where $K$ is diverging to Section~\ref{subsec:d diverges}.
Assumptions~\ref{assu:independent sampling} and \ref{assu:treatment assignment}
are standard in the literature on statistical inference under CAR
(see, e.g., \citealp{bugni2018Inference,bannick2025General,bugni2019Inference,ma2022Regression,tu2024Unified,jiang2023Regressionadjusted}
and references therein). Assumption \ref{assu:treatment assignment}
is satisfied for many randomization methods, such as stratified block
randomization \citep{zelen1974Randomization}, and Pocock and Simon’s
minimization \citep{pocock1975Sequential}. Assumption \ref{assu:conditions on =00005Cxi}(1)(a)
is a mild moment assumption on $\bs{\xi}_{n}^{*}$, the probability
limit of $\bs{\xi}_{n}$. Assumption \ref{assu:conditions on =00005Cxi}(1)(b)
restricts that the minimal non-zero singular value of $\e\left[\widetilde{\bs{\xi}}_{n}^{*}(\bs X_{i})\widetilde{\bs{\xi}}_{n}^{*}(\bs X_{i})^{\trans}\mid B_{i}=k\right]$
is bounded away from zero. Importantly, this assumption allows the
matrix $\e\left[\widetilde{\bs{\xi}}_{n}^{*}(\bs X_{i})\widetilde{\bs{\xi}}_{n}^{*}(\bs X_{i})^{\trans}\mid B_{i}=k\right]$
to be singular. This is reasonable, as in many applications, the components
of $\bs{\xi}_{n}$ may share common information, leading to cases
where the matrix $\e\left[\widetilde{\bs{\xi}}_{n}^{*}(\bs X_{i})\widetilde{\bs{\xi}}_{n}^{*}(\bs X_{i})^{\trans}\mid B_{i}=k\right]$
is indeed singular. Assumption \ref{assu:conditions on =00005Cxi}(1)(c)
requires that $\bs{\xi}_{n}$ converges to its probability limit $\bs{\xi}_{n}^{*}$.
This assumption is similar to \citet[Assumption 4]{tu2024Unified},
\citet[Assumption 2]{bannick2025General} and \citet[Assumption 3]{jiang2023Regressionadjusted}.
When $\bs{\xi}_{n}$ is derived from linear regression, it simplifies
to \citet[Assumption 3]{liu2023Lassoadjusted}. When $\bs{\xi}_{n}$
is derived from local linear kernel regression, this assumption can
be verified by \citet[Theorem 3]{tu2024Unified}. If $\bs{\xi}_{n}$
is a general machine learning estimator that depends on the entire
observed dataset, then verifying Assumption \ref{assu:conditions on =00005Cxi}(1)(c)
usually requires that the function class containing $\bs{\xi}_{n}$
satisfies the Donsker's condition (see \citealp[Assumption 3]{bannick2025General}),
which may be violated if $\bs{\xi}_{n}$ is derived from complex machine
learning algorithms. To bypass the Donsker's condition, one can use
sample-splitting and/or cross-fitting (\citealp{rafi2023Efficient,tu2024Unified,chernozhukov2018Double})
techniques as discussed in Section~\ref{subsec:cross-fitting}. By
using cross-fitting, Assumption~\ref{assu:conditions on =00005Cxi}(1)(c)
can be replaced by a mean square error convergence condition (see
\citealp[Assumption 4.1]{rafi2023Efficient} and \citealp[Assumption 6]{tu2024Unified}),
which is satisfied for many machine learning methods (see, e.g., \citealp{farrell2021Deepa,chi2022Asymptotic,wager2018estimation,jiao2023Deep}).
Assumption \ref{assu:conditions on =00005Cxi}(1)(d) ensures that
the outcome $Y$ cannot be fully explained by $\bs{\xi}(\bs X)$.
Assumption~\ref{assu:conditions on =00005Cxi}(2) is a technical
requirement that ensures the estimated Moore--Penrose inverse remains
well-behaved and does not diverge. This condition is necessary to
address the discontinuity of the Moore--Penrose inverse operation
(see e.g., \citealp[Theorem 3.1]{stewart1977Perturbation}), as Assumptions~\ref{assu:independent sampling}--\ref{assu:treatment assignment}
and \ref{assu:conditions on =00005Cxi}(1)(a)--(c) alone do not guarantee
the convergence of the inverse. However, in the standard case where
the limiting matrix $\e\left[\widetilde{\bs{\xi}}_{n}^{*}(\bs X_{i})\widetilde{\bs{\xi}}_{n}^{*}(\bs X_{i})^{\trans}\mid B_{i}=k\right]$
is non-singular, the Moore--Penrose inverse coincides with the standard
matrix inverse. Since standard inversion is a continuous operation,
Assumption~\ref{assu:conditions on =00005Cxi}(2) is then automatically
implied by the preceding assumptions.

Recall that $h_{a[k]}^{*}(\bs X)=\e\left[Y(a)\mid\bs X,B=k\right]$
and let $\widetilde{h}_{a[k]}^{*}(\bs X):=h_{a[k]}^{*}(\bs X)-\e\left[h_{a[k]}^{*}(\bs X)\mid B=k\right]$
and $\widetilde{\bs{\xi}}_{n}^{*}(\bs X):=\bs{\xi}_{n}^{*}(\bs X)-\e\left[\bs{\xi}_{n}^{*}(\bs X)\mid B=k\right]$.
We have the following theorem.
\begin{thm}
\label{thm:asymptotic properties of tau_cal}Suppose that Assumptions
\ref{assu:independent sampling}--\ref{assu:conditions on =00005Cxi}
hold and $D(v)=(v-1)^{2}/2$. Then 
\[
\frac{\sqrt{n}\left(\widehat{\tau}_{\mathrm{cal}}-\tau\right)}{\sqrt{\varsigma_{H}^{2}+\varsigma_{\widetilde{Y}}^{2}-\varsigma_{\widetilde{Y}\mid\widetilde{\bs{\xi}}_{n}^{*}}^{2}}}\tod N(0,1)\text{ and }\widehat{\varsigma}_{H}^{2}+\widehat{\varsigma}_{\widetilde{Y}}^{2}-\widehat{\varsigma}_{\widetilde{Y}\mid\widetilde{\bs{\xi}}_{n}^{*}}^{2}=\varsigma_{H}^{2}+\varsigma_{\widetilde{Y}}^{2}-\varsigma_{\widetilde{Y}\mid\widetilde{\bs{\xi}}_{n}^{*}}^{2}+o_{P}(1),
\]
where the definitions of the asymptotic variances $\varsigma_{H}^{2}$,
$\varsigma_{\widetilde{Y}}^{2}$, $\varsigma_{\widetilde{Y}\mid\widetilde{\bs{\xi}}_{n}^{*}}^{2}$
and their estimates $\widehat{\varsigma}_{H}^{2},\widehat{\varsigma}_{\widetilde{Y}}^{2}$
and $\widehat{\varsigma}_{\widetilde{Y}\mid\widetilde{\bs{\xi}}_{n}^{*}}^{2}$
can be found in Appendix~\ref{sec:Proofs-for-the}. Moreover, if
for every $k=1,\ldots,K$, there exists a non-stochastic vector $\bs{\alpha}_{k}\in\R^{d}$
such that 
\[
\left\{ \sqrt{\frac{1-\pi_{[k]}}{\pi_{[k]}}}\widetilde{h}_{1[k]}^{*}(\bs X)+\sqrt{\frac{\pi_{[k]}}{1-\pi_{[k]}}}\widetilde{h}_{0[k]}^{*}(\bs X)\right\} \1(B=k)=\bs{\alpha}_{k}^{\trans}\widetilde{\bs{\xi}}_{n}^{*}(\bs X)\1(B=k),
\]
then 
\[
\frac{\sqrt{n}\left(\widehat{\tau}_{\mathrm{cal}}-\tau\right)}{\sqrt{\varsigma_{H}^{2}+\varsigma_{\widetilde{Y}}^{2}-\varsigma_{\widetilde{Y}\mid\widetilde{h}^{*}}^{2}}}\tod N(0,1),
\]
and the asymptotic variance $\varsigma_{H}^{2}+\varsigma_{\widetilde{Y}}^{2}-\varsigma_{\widetilde{Y}\mid\widetilde{h}^{*}}^{2}$
matches the semiparametric efficiency bound developed in \citet[Theorem 3.1]{rafi2023Efficient}.
\end{thm}
In Theorem \ref{thm:asymptotic properties of tau_cal}, we focus on
the case that $D(v)=(v-1)^{2}/2$; the general $D(v)$ will be handled
in Section \ref{sec:General-discrepancy-measure}. Theorem \ref{thm:asymptotic properties of tau_cal}
shows that our calibration estimator is asymptotically normally distributed
and its asymptotic variance can be consistently estimated. Moreover,
as long as there exists a linear combination of $\bs{\xi}_{n}^{*}(\bs X)\1(B=k)$
that equals $\bigl\{\sqrt{\frac{1-\pi_{[k]}}{\pi_{[k]}}}h_{1[k]}^{*}(\bs X)+\sqrt{\frac{\pi_{[k]}}{1-\pi_{[k]}}}h_{0[k]}^{*}(\bs X)\bigr\}\1(B=k)$,
our calibration estimator is semiparametric efficient. To the best
of our knowledge, this is a new condition for achieving the efficiency
bound in the literature. Existing work (e.g., \citealp{tu2024Unified})
requires that both $h_{1[k]}^{*}(\bs X)$ and $h_{0[k]}^{*}(\bs X)$
be consistently estimated. In contrast, we only require that the linear
combination of these functions be consistently estimated, which is
weaker. This condition aligns with \citet[Remark 2]{bai2022Optimality},
which states that optimal stratification can be achieved by knowing
the function $\left\{ \sqrt{\frac{1-\pi_{[k]}}{\pi_{[k]}}}h_{1[k]}^{*}(\bs X)+\sqrt{\frac{\pi_{[k]}}{1-\pi_{[k]}}}h_{0[k]}^{*}(\bs X)\right\} $.
In this context, the optimal stratification refers to the stratification
method under which the difference-in-means estimator has the smallest
mean squared error. Since the components of $\bs{\xi}_{n}(\bs X)$
can be chosen as different estimates of $\bigl\{\sqrt{\frac{1-\pi_{[k]}}{\pi_{[k]}}}h_{1[k]}^{*}(\bs X)+\sqrt{\frac{\pi_{[k]}}{1-\pi_{[k]}}}h_{0[k]}^{*}(\bs X)\bigr\}$
from various pre-specified models, the calibration estimator is \textit{multiply
efficient}: it remains semiparametric efficient as long as at least
one of these models is correctly specified. In contrast, AIPW-based
covariate adjustment methods lack this kind of multiple efficiency.

Theorem \ref{thm:asymptotic properties of tau_cal} implies that the
calibration estimator (\ref{eq:cal estimator}) possesses two theoretical
advantages. First, Assumption \ref{assu:treatment assignment} imposes
no restriction on the choice of randomization scheme, which may include
simple randomization, stratified block randomization, or minimization.
Since the asymptotic distribution of the calibration estimator is
invariant to the randomization method, it satisfies the property referred
to in the literature as ``universal applicability'' (\citealp{bannick2025General,ye2023Better}).
Second, because $\varsigma_{\widetilde{Y}\mid\widetilde{\bs{\xi}}_{n}^{*}}^{2}\geq0$,
Theorem~\ref{thm:asymptotic properties of tau_cal}, together with
\citet[Proposition 3]{ma2022Regression}, implies that $\widehat{\tau}_{\mathrm{cal}}$
is always at least as efficient as the stratified difference-in-means
estimator $\widehat{\tau}_{\mathrm{sdim}}$, which corresponds to
the special case where $\widehat{w}_{i}=1$ for all $i=1,\ldots,n$
in (\ref{eq:cal estimator}). Hence, the calibration estimator guarantees
efficiency gains. More generally, we have the following efficiency
comparison result.
\begin{thm}
\label{thm:efficiency comparison thm}Rewrite the estimator $\widehat{\tau}_{\mathrm{cal}}$
in (\ref{eq:cal estimator}) as $\widehat{\tau}_{\mathrm{cal}}(\bs{\xi}_{n})$.
Let $\mathbf{\Lambda}$ be a $\widetilde{d}\times d$ non-stochastic
matrix, where $\widetilde{d}\geq1$ is fixed. Suppose that Assumptions
\ref{assu:independent sampling}--\ref{assu:conditions on =00005Cxi}
hold for both $\bs{\xi}_{n}$ and $\mathbf{\Lambda}\bs{\xi}_{n}$.
Then $\widehat{\tau}_{\mathrm{cal}}(\bs{\xi}_{n})$ and $\widehat{\tau}_{\mathrm{cal}}(\mathbf{\Lambda}\bs{\xi}_{n})$
satisfy
\[
\frac{\sqrt{n}\left(\widehat{\tau}_{\mathrm{cal}}(\bs{\xi}_{n})-\tau\right)}{\sqrt{\varsigma_{H}^{2}+\varsigma_{\widetilde{Y}}^{2}-\varsigma_{\widetilde{Y}\mid\widetilde{\bs{\xi}}_{n}^{*}}^{2}}}\tod N(0,1)\text{ and  }\frac{\sqrt{n}\left(\widehat{\tau}_{\mathrm{cal}}(\mathbf{\Lambda}\bs{\xi}_{n})-\tau\right)}{\sqrt{\varsigma_{H}^{2}+\varsigma_{\widetilde{Y}}^{2}-\varsigma_{\widetilde{Y}\mid\mathbf{\Lambda}\widetilde{\bs{\xi}}_{n}^{*}}^{2}}}\tod N(0,1),
\]
respectively. Furthermore, the asymptotic variances satisfy $\varsigma_{H}^{2}+\varsigma_{\widetilde{Y}}^{2}-\varsigma_{\widetilde{Y}\mid\widetilde{\bs{\xi}}_{n}^{*}}^{2}\leq\varsigma_{H}^{2}+\varsigma_{\widetilde{Y}}^{2}-\varsigma_{\widetilde{Y}\mid\mathbf{\Lambda}\widetilde{\bs{\xi}}_{n}^{*}}^{2}.$
\end{thm}
Theorem \ref{thm:efficiency comparison thm} implies that incorporating
additional elements into $\bs{\xi}_{n}$ satisfies a no-harm property;
i.e., it is guaranteed to improve or at least maintain efficiency.
Consequently, from a theoretical perspective, it appears optimal to
include as many elements as possible in $\bs{\xi}_{n}$. However,
as the dimension of $\bs{\xi}_{n}$ becomes large, it is no longer
appropriate to treat the dimension $d$ as fixed. We will address
the scenario where the dimension of $\bs{\xi}_{n}$ diverges in Section~\ref{subsec:d diverges}.

\section{Extensions\label{sec:Extensions}}

\subsection{The dimension of $\protect\bs{\xi}_{n}(\cdot)$ and the number of
strata are diverging\label{subsec:d diverges}}

In this section, we allow the dimension of $\bs{\xi}_{n}(\cdot)$
to grow with $n$, that is, $d=d_{n}\to\infty$ as $n\to\infty$.
Additionally, we also allow the number of strata to increase with
$n$, i.e., $K=K_{n}\to\infty$ as $n\to\infty$, which is commonly
encountered in many applications. We begin by stating an assumption.
\begin{assumption}
\label{assu:conditions on =00005Cxi-diverging xi}The sequence of
random functions $\bs{\xi}_{n}(\cdot)$ satisfies the following conditions.
\begin{enumerate}[label=(\arabic*)]
\item There exists a sequence of non-stochastic functions $\bs{\xi}_{n}^{*}(\cdot):\mathcal{X}\to\R^{d}$
($n\geq1$) such that
\begin{enumerate}
\item $K^{2}r_{n}^{2}/n\to0$, $K^{2}\zeta_{n}^{2}r_{n}/n\to0$ and $\max\left\{ \zeta_{n}^{2}\log(2Kd),r_{n}\log^{2}(2Kd)\right\} /(n\inf_{1\leq k\leq K}p_{[k]})\to0$
as $n\to\infty$, where $r_{n}:=\max_{1\leq k\leq K}\mathrm{rank}\left\{ \e\left[\bs{\xi}_{n}^{*}(\bs X_{i})\bs{\xi}_{n}^{*}(\bs X_{i})^{\trans}\mid B_{i}=k\right]\right\} $
and $\zeta_{n}:=\sup_{\bs x\in\mathcal{X}}\left\Vert \bs{\xi}_{n}^{*}(\bs x)\right\Vert $;
\item the maximal singular value of $\e\left[\bs{\xi}_{n}^{*}(\bs X_{i})\bs{\xi}_{n}^{*}(\bs X_{i})^{\trans}\mid B_{i}=k\right]$
is smaller than some finite constant $C>0$ uniformly over $n\geq1$
and $k=1,\ldots,K$; the minimal non-zero singular value of $\e\left[\widetilde{\bs{\xi}}_{n}^{*}(\bs X_{i})\widetilde{\bs{\xi}}_{n}^{*}(\bs X_{i})^{\trans}\mid B_{i}=k\right]$
is larger than some constant $c>0$ uniformly over $n\geq1$ and $k=1,\ldots,K$;
\item it holds that 
\begin{align*}
 & \sup_{1\leq k\leq K}\left\Vert \frac{1}{n_{1[k]}}\sum_{i\in[k]}A_{i}\left\{ \bs{\xi}_{n}(\bs X_{i})-\bs{\xi}_{n}^{*}(\bs X_{i})\right\} -\frac{1}{n_{0[k]}}\sum_{i\in[k]}\left(1-A_{i}\right)\left\{ \bs{\xi}_{n}(\bs X_{i})-\bs{\xi}_{n}^{*}(\bs X_{i})\right\} \right\Vert \\
= & o_{P}(n^{-1/2})
\end{align*}
and $\sup_{1\leq k\leq K}\frac{1}{n_{[k]}}\sum_{i\in[k]}\left\Vert \bs{\xi}_{n}(\bs X_{i})-\bs{\xi}_{n}^{*}(\bs X_{i})\right\Vert ^{2}=O_{P}(n^{-1/2})$\textup{;}
\item $\liminf_{n\to\infty}\left(\varsigma_{\widetilde{Y}}^{2}-\varsigma_{\widetilde{Y}\mid\widetilde{\bs{\xi}}_{n}^{*}}^{2}\right)>0$.
\end{enumerate}
\item $\sup_{1\leq k\leq K}\left\Vert \left\{ \frac{1}{n_{[k]}}\sum_{i\in[k]}(A_{i}-\pi_{n[k]})^{2}\left\{ \bs{\xi}_{n}(\bs X_{i})-\overline{\bs{\xi}}_{n[k]}\right\} \left\{ \bs{\xi}_{n}(\bs X_{i})-\overline{\bs{\xi}}_{n[k]}\right\} ^{\trans}\right\} ^{+}\right\Vert =O_{P}(1)$.
\end{enumerate}
\end{assumption}
Assumption \ref{assu:conditions on =00005Cxi-diverging xi} is very
similar to Assumption \ref{assu:conditions on =00005Cxi}, except
that we allow both the dimension of $\bs{\xi}_{n}$ and the number
of strata to diverge. Assumption \ref{assu:conditions on =00005Cxi-diverging xi}(1)(a)
restricts the growth rate of $d$ and $K$. Suppose that the matrix
$\e\left[\bs{\xi}_{n}^{*}(\bs X_{i})\bs{\xi}_{n}^{*}(\bs X_{i})^{\trans}\mid B_{i}=k\right]$
is non-singular and the elements in $\bs{\xi}_{n}^{*}(\bs x)$ are
uniformly bounded, we have $r_{n}=d$ and $\zeta_{n}=\sqrt{d}$. Consider
the case that the strata have equal size, i.e., $p_{[k]}=1/K$. Then
a sufficient condition for Assumption \ref{assu:conditions on =00005Cxi-diverging xi}(1)(a)
is $Kd=o(\sqrt{n})$, meaning that we allow the product of the dimension
of $\bs{\xi}_{n}$ and the number of strata to grow, but at a rate
slower than $\sqrt{n}$. When $K$ is fixed, this condition is consistent
with a similar result for the regression-adjusted ATE estimator presented
by \citet{lei2021Regression} under complete randomization and finite-population
asymptotics. To further relax the growth rate constraint on $d$, \citet{lei2021Regression}, \citet{lu2025Debiased}, and \citet{gu2025assumptionleancovariateadjustmentcovariate} proposed debiased regression-adjusted estimators that allow $d$ to grow faster than $\sqrt{n}$. In contrast, \citet{jiang2025Adjustments} achieved the same goal by assuming a correctly specified linear relationship between the potential outcomes and covariates. It is also possible
to debias the calibration estimator (\ref{eq:cal estimator}) in a
way that would relax Assumption~\ref{assu:conditions on =00005Cxi-diverging xi}(1)(a).
However, addressing this is beyond the scope of the present paper
and will be left for future research. Assumption \ref{assu:conditions on =00005Cxi-diverging xi}(1)(b)
requires that largest and smallest non-zero singular values of $\e\left[\widetilde{\bs{\xi}}_{n}^{*}(\bs X_{i})\widetilde{\bs{\xi}}_{n}^{*}(\bs X_{i})^{\trans}\mid B_{i}=k\right]$
are bounded away from zero and infinity, while still allowing this
matrix to be singular. Assumption \ref{assu:conditions on =00005Cxi-diverging xi}(1)(c)
strengthens Assumption \ref{assu:conditions on =00005Cxi}(1)(c).
Under cross-fitting, Assumption \ref{assu:conditions on =00005Cxi-diverging xi}(1)(c)
requires that $\bs{\xi}_{n}$ converges to its probability limit $\bs{\xi}_{n}^{*}$
at rate $n^{-1/4}$ uniformly over $k=1,\ldots,K$. Such a rate is
achievable for many modern machine learning methods (see, e.g., \citealp{farrell2021Deepa,jiao2023Deep}).
Assumption \ref{assu:conditions on =00005Cxi-diverging xi}(2) parallels
Assumption \ref{assu:conditions on =00005Cxi}(2), except that it
is required to hold uniformly over $k=1,\ldots,K$. We are now ready
to state the following theorem.
\begin{thm}
\label{thm:asymptotic properties of tau_cal-diverging xi}Suppose
that Assumptions \ref{assu:independent sampling}--\ref{assu:treatment assignment}
and \ref{assu:conditions on =00005Cxi-diverging xi} hold and $D(v)=(v-1)^{2}/2$.
Then 
\[
\frac{\sqrt{n}\left(\widehat{\tau}_{\mathrm{cal}}-\tau\right)}{\sqrt{\varsigma_{H}^{2}+\varsigma_{\widetilde{Y}}^{2}-\varsigma_{\widetilde{Y}\mid\widetilde{\bs{\xi}}_{n}^{*}}^{2}}}\tod N(0,1)\text{ and }\widehat{\varsigma}_{H}^{2}+\widehat{\varsigma}_{\widetilde{Y}}^{2}-\widehat{\varsigma}_{\widetilde{Y}\mid\widetilde{\bs{\xi}}_{n}^{*}}^{2}=\varsigma_{H}^{2}+\varsigma_{\widetilde{Y}}^{2}-\varsigma_{\widetilde{Y}\mid\widetilde{\bs{\xi}}_{n}^{*}}^{2}+o_{P}(1).
\]
Moreover, when $K$ is a fixed number, if for every $k=1,\ldots,K$,
there exists a non-stochastic vector $\bs{\alpha}_{k}\in\R^{d}$ such
that 
\[
\left\{ \sqrt{\frac{1-\pi_{[k]}}{\pi_{[k]}}}\widetilde{h}_{1[k]}^{*}(\bs X)+\sqrt{\frac{\pi_{[k]}}{1-\pi_{[k]}}}\widetilde{h}_{0[k]}^{*}(\bs X)\right\} \1(B=k)=\bs{\alpha}_{k}^{\trans}\widetilde{\bs{\xi}}_{n}^{*}(\bs X)\1(B=k),
\]
then 
\[
\frac{\sqrt{n}\left(\widehat{\tau}_{\mathrm{cal}}-\tau\right)}{\sqrt{\varsigma_{H}^{2}+\varsigma_{\widetilde{Y}}^{2}-\varsigma_{\widetilde{Y}\mid\widetilde{h}^{*}}^{2}}}\tod N(0,1),
\]
and the asymptotic variance $\varsigma_{H}^{2}+\varsigma_{\widetilde{Y}}^{2}-\varsigma_{\widetilde{Y}\mid\widetilde{h}^{*}}^{2}$
matches the semiparametric efficiency bound developed in \citet[Theorem 3.1]{rafi2023Efficient}.
\end{thm}
Theorem \ref{thm:asymptotic properties of tau_cal-diverging xi} shows
that, even when both the dimension of $\bs{\xi}_{n}$ and the number
of strata diverge, the calibration estimator (\ref{eq:cal estimator})
remains asymptotically normal. It is important to note that, because
we allow the number of strata to increase with the sample size, Assumption
\ref{assu:treatment assignment} becomes non-trivial. \citet{xin2024inferencecovariateadaptiverandomizationstrata}
studied inference under CAR with a growing number of strata, and our
Assumption \ref{assu:treatment assignment} corresponds directly to
their Assumptions (B1)--(B3). Moreover, \citet{xin2024inferencecovariateadaptiverandomizationstrata}
derived more primitive conditions under which Assumption \ref{assu:treatment assignment}
holds in the cases of simple randomization, stratified adaptive biased-coin
randomization and stratified block randomization. For further discussion
of Assumption \ref{assu:treatment assignment}, we refer readers to
\citet[Section 3]{xin2024inferencecovariateadaptiverandomizationstrata}.

\subsection{General discrepancy measure $D(v)$\label{sec:General-discrepancy-measure}}

In Theorem \ref{thm:asymptotic properties of tau_cal}, our analysis
was limited to the quadratic discrepancy measure $D(v)=(v-1)^{2}/2$.
In this section, we extend the analysis to a general class of discrepancy
measures $D(v)$, which can be theoretically shown to exhibit smaller
second-order bias. To proceed, we impose the following mild regularity
condition on $D(v)$.

\begin{assumption}
\label{assu:conditions on D(v)}Let $D^{\prime}$ be the derivative
of $D$ and $(D^{\prime})^{-1}$ is the inverse function of $D^{\prime}$.
Define $\rho(v):=D\left\{ (D^{\prime})^{-1}(-v)\right\} +v\cdot(D^{\prime})^{-1}(-v)$.
We assume $\rho(v)$ is concave and three times continuously differentiable,
$\rho^{\prime}(0)=1$, $-\infty<\rho^{\prime\prime}(0)<0$ and there
exist constants $\delta>0$ and $0\leq C_{\rho}<\infty$ such that
$\left|\rho^{\prime\prime\prime}(v)-\rho^{\prime\prime\prime}(0)\right|\leq C_{\rho}\left|v\right|$
for all $\left|v\right|\leq\delta$.
\end{assumption}
Assumption \ref{assu:conditions on D(v)} holds for a wide range of
commonly used discrepancy measures, and Table~\ref{tab:Different-choices-of D(v)}
provides several popular examples. The case that $D(v)=v\log v-v$
is related to the exponential tilting estimator studied by \citet{kitamura1997Informationtheoretic},
and $D(v)=v-\log v$ is related to empirical likelihood estimator
studied by \citet{qin1994Empirical}.

\begin{table}[!tbh]
\caption{\label{tab:Different-choices-of D(v)}Different choices of $D(v)$
and their corresponding $\rho(v)$.}

\centering{}%
\begin{tabular}{ccccc}
\toprule 
$D(v)$ & $\rho(v)$ & $\rho^{\prime}(v)$ & $\rho^{\prime\prime}(0)$ & $\rho^{\prime\prime\prime}(0)$\tabularnewline
\midrule
\midrule 
$(v-1)^{2}/2$ & $-v^{2}/2+v$ & $-v+1$ & $-1$ & $0$\tabularnewline
\midrule 
$v\log v-v$ & $-e^{-v}$ & $e^{-v}$ & $-1$ & $1$\tabularnewline
\midrule 
$v-\log v$ & $1+\log(1+v)$ & $\frac{1}{1+v}$ & $-1$ & $2$\tabularnewline
\bottomrule
\end{tabular}
\end{table}

Since our goal is to characterize the second-order bias of the calibration
estimator, we require a strengthened version of Assumption \ref{assu:conditions on =00005Cxi}.
\begin{assumption}
\label{assu:conditions on =00005Cxi-general D(v)}
\begin{enumerate}[label=(\arabic*)]
\item The sequence of random functions $\bs{\xi}_{n}(\cdot)$ satisfies
$\frac{1}{n}\sum_{i=1}^{n}\left\Vert \bs{\xi}_{n}(\bs X_{i})\right\Vert ^{4}=O_{P}(1)$;
\item The non-stochastic functions $\bs{\xi}_{n}^{*}(\cdot):\mathcal{X}\to\R^{d}$
($n\geq1$) in Assumption~\ref{assu:conditions on =00005Cxi} satisfy
the following conditions: $\sup_{n\geq1}\e\left[\left\Vert \bs{\xi}_{n}^{*}(\bs X_{i})\right\Vert ^{4}\right]<\infty$,
and for each $k=1,\ldots,K$, 
\[
\left\Vert \frac{1}{n_{1[k]}}\sum_{i\in[k]}A_{i}\left\{ \bs{\xi}_{n}(\bs X_{i})-\bs{\xi}_{n}^{*}(\bs X_{i})\right\} -\frac{1}{n_{0[k]}}\sum_{i\in[k]}\left(1-A_{i}\right)\left\{ \bs{\xi}_{n}(\bs X_{i})-\bs{\xi}_{n}^{*}(\bs X_{i})\right\} \right\Vert =o_{P}(n^{-1}),
\]
and $\frac{1}{n_{[k]}}\sum_{i\in[k]}\left\Vert \bs{\xi}_{n}(\bs X_{i})-\bs{\xi}_{n}^{*}(\bs X_{i})\right\Vert ^{2}=O_{P}(\Delta_{n}^{2})$\textup{,}
where $\Delta_{n}\to0$ as $n\to\infty$ is a sequence of real numbers.
\end{enumerate}
\end{assumption}
We have the following theorem.
\begin{thm}
\label{thm:asymptotic properties of tau_cal-general D(v)}Suppose
that Assumptions \ref{assu:independent sampling}--\ref{assu:conditions on =00005Cxi}
and \ref{assu:conditions on D(v)} hold. If $\sup_{n\geq1}\sup_{1\leq i\leq n}\e\left[\left\Vert \bs{\xi}_{n}(\bs X_{i})\right\Vert ^{2+\epsilon}\right]<\infty$
for some $\epsilon>0$ and $\e\left[\widetilde{\bs{\xi}}_{n}^{*}(\bs X_{i})\widetilde{\bs{\xi}}_{n}^{*}(\bs X_{i})^{\trans}\mid B_{i}=k\right]$
is non-singular for every $k=1,\ldots,K$, then the conclusion of
Theorem \ref{thm:asymptotic properties of tau_cal} holds.

If, in addition, Assumption \ref{assu:conditions on =00005Cxi-general D(v)}
holds and $\e\left[Y_{i}(a)^{4}\right]<\infty$ for all $a=0,1$,
then 
\[
\widehat{\tau}_{\mathrm{cal}}-\tau=\frac{1}{n}\sum_{i=1}^{n}\psi_{1,i}+\frac{1}{n}\psi_{2}+o_{P}(n^{-1})+O_{P}(\Delta_{n}n^{-1/2}),
\]
where $\psi_{1,i}=\sum_{k=1}^{K}\left\{ \left(\frac{A_{i}}{\pi_{n[k]}}-\frac{1-A_{i}}{1-\pi_{n[k]}}\right)\1(B_{i}=k)\cdot Y_{i}-\tau-\bs{\beta}_{[k],\mathcal{C}_{n}}^{\trans}\bs{\Xi}_{i,[k]}^{*}\right\} $
with $\e\left[\frac{1}{n}\sum_{i=1}^{n}\psi_{1,i}\right]=0$ and $\e\left[\psi_{2}\right]=\left\{ \frac{\rho^{\prime\prime\prime}(0)}{2\rho^{\prime\prime}(0)^{2}}-1\right\} \sum_{k=1}^{K}\tr\left\{ \e\left[\left(\mathbf{\Sigma}_{[k]}^{\mathcal{C}_{n}}\right)^{-1}\mathbf{\Sigma}_{\bs{\Xi}\bs{\Xi}\epsilon[k]}^{\mathcal{C}_{n}}\right]\right\} $\textup{.}
The definitions of $\bs{\beta}_{[k],\mathcal{C}_{n}}^{\trans}\bs{\Xi}_{i,[k]}^{*}$
and $\left(\mathbf{\Sigma}_{[k]}^{\mathcal{C}_{n}}\right)^{-1}\mathbf{\Sigma}_{\bs{\Xi}\bs{\Xi}\epsilon[k]}^{\mathcal{C}_{n}}$
are provided in  Appendix \ref{sec:Proofs-for-the}.
\end{thm}
The first conclusion of Theorem \ref{thm:asymptotic properties of tau_cal-general D(v)}
implies that, provided that Assumption \ref{assu:conditions on D(v)}
holds, different choices of $D(v)$ lead to calibration estimators
with the same asymptotic distribution. We also note that this conclusion
requires the non-singularity of $\e\left[\widetilde{\bs{\xi}}_{n}^{*}(\bs X_{i})\widetilde{\bs{\xi}}_{n}^{*}(\bs X_{i})^{\trans}\mid B_{i}=k\right]$
in order to ensure that (\ref{eq:cal opt problem}) admits a unique
solution asymptotically. In the special case where $D(v)=(v-1)^{2}/2$,
such a condition is not necessary. This is because, under the quadratic
discrepancy, the weights $\widehat{w}_{i},i=1,\ldots,n$, admit a
closed-form expression. Even if \eqref{eq:cal opt problem} admits
infinitely many solutions, one can always select the solution obtained
via the Moore-Penrose inverse (see (\ref{eq:def of lambdahat}) in
the proof of Theorem~\ref{thm:asymptotic properties of tau_cal-general D(v)}
in Appendix~\ref{subsec:Proof-of-Theorem}). In contrast, when $D(v)$
is a general discrepancy measure, the weights $\widehat{w}_{i}$ no
longer have an explicit form. In this case, if (\ref{eq:cal opt problem})
admits infinitely many solutions, it is unclear which solution an
optimization algorithm would converge to. In practice, a regularization
term could be introduced to enforce the uniqueness of the solution,
although this is beyond the scope of the present paper. Additionally,
if $\e\left[\widetilde{\bs{\xi}}_{n}^{*}(\bs X_{i})\widetilde{\bs{\xi}}_{n}^{*}(\bs X_{i})^{\trans}\mid B_{i}=k\right]$
is singular, one could apply singular value decomposition (SVD) to
eliminate the redundant components of $\bs{\xi}_{n}$ and use the
reduced version of $\bs{\xi}_{n}$ to compute the calibration estimator
(\ref{eq:cal estimator}). We leave the rigorous theoretical development
of this strategy for future investigation.

The second conclusion of Theorem \ref{thm:asymptotic properties of tau_cal-general D(v)}
provides a characterization of the second-order bias under the CAR
design, which, to the best of our knowledge, is a novel contribution
to the literature. Unlike existing studies on second-order bias (\citealp{newey2004Higher,tan2014Secondorder}),
the observed samples under the CAR design are not i.i.d. The proof
of this result relies on a conditional argument: we first examine
the asymptotic expansion of the calibration estimator conditional
on the treatment and stratum indicators, and then apply the law of
iterated expectation. From Theorem~\ref{thm:asymptotic properties of tau_cal-general D(v)},
we see that when $\Delta_{n}=o(n^{-1/2})$, taking $D(v)=v-\log v$
gives the calibration estimator with zero second-order bias $\e\left[\psi_{2}\right]$.
This aligns with the findings in i.i.d. settings, where empirical
likelihood-based estimators exhibit smaller second-order bias (\citealp{newey2004Higher,tan2014Secondorder}).

\section{Simulation studies\label{sec:Simulation-studies}}

We evaluate the performance of our calibration estimators using Monte
Carlo simulations. For $a\in\{0,1\}$ and $1\leq i\leq n$, the potential
outcomes are generated as follows:
\[
Y_{i}(a)=g_{a}(\bs X_{i})+\epsilon_{a,i}\ \ i=1,\ldots,n,\ a\in\{0,1\},
\]
where $\bs X_{i}$, $\epsilon_{a,i}$ and $g_{a}(\cdot)$ will be
specified in each model and the triplet $(\bs X_{i},\epsilon_{0,i},\epsilon_{1,i})$
for $1\leq i\leq n$ is i.i.d.

We present simulation results for the estimators under three randomization
methods: simple randomization, stratified block randomization, and
minimization (\citealp{pocock1975Sequential}). The sample size $n$
varies across 500, 1000, and 2000, and the number of covariates is
set to be $p=30$. For stratified block randomization, we use a block
size of 6. In the minimization method, a biased-coin probability of
0.75 and equal weights are employed. The cross-fitting technique is
applied with two folds as demonstrated in Section~\ref{subsec:cross-fitting}.
Unless otherwise specified, the calibration estimator is obtained
by setting $D(v)=(v-1)^{2}/2$. We compare nine estimators: \textbf{(i)}
\texttt{cal\_rf}: For each stratum $k$, we use random forest to estimate
the conditional mean function $g_{0}(\cdot)$ and $g_{1}(\cdot)$,
denoting their estimates as $\widehat{g}_{0k}^{\text{rf}}(\cdot)$
and $\widehat{g}_{1k}^{\text{rf}}(\cdot)$. Then \texttt{cal\_rf }is
obtained by taking $\bs{\xi}_{n}(\bs X_{i})=(\sum_{k=1}^{K}\widehat{g}_{0k}^{\text{rf}}(\bs X_{i})\1(B_{i}=k),\sum_{k=1}^{K}\widehat{g}_{1k}^{\text{rf}}(\bs X_{i})\1(B_{i}=k))^{\trans}$
in our calibration estimator (\ref{eq:cal estimator}); \textbf{(ii)}
\texttt{cal\_nn}: Defined similarly to \texttt{cal\_rf}, but with
the random forest estimates replaced by neural network estimates $\widehat{g}_{0k}^{\text{rf}}(\cdot)$
and $\widehat{g}_{1k}^{\text{rf}}(\cdot)$; \textbf{(iii)} \texttt{cal\_rfnn}:
This estimator is obtained by combining the random forest and neural
network estimates. It is obtained by taking $\bs{\xi}_{n}(\bs X_{i})=(\sum_{k=1}^{K}\widehat{g}_{0k}^{\sharp}(\bs X_{i})\1(B_{i}=k),\sum_{k=1}^{K}\widehat{g}_{1k}^{\sharp}(\bs X_{i})\1(B_{i}=k):\sharp\in\{\text{rf},\text{nn}\})^{\trans}$
in our calibration estimator (\ref{eq:cal estimator}); \textbf{(iv)}
\texttt{cal\_rflin}: Defined similarly to \texttt{cal\_rfnn}, but
obtained by combining the random forest and linear regression estimates.
\textbf{(v)}\texttt{ cal\_rf\_g}: Take $\bs{\xi}_{n}(\bs X_{i})=(\widehat{g}_{0k}^{\text{rf}}(\bs X_{i}),\widehat{g}_{1k}^{\text{rf}}(\bs X_{i}):1\leq k\leq K)^{\trans}$
and use it in our calibration estimator (\ref{eq:cal estimator});
\textbf{(vi)} \texttt{cal\_nn\_g}: Defined similarly to \texttt{cal\_rf\_g},
but with neural network estimates; \textbf{(vii)} \texttt{cal\_lin\_EL}:
Defined similarly to \texttt{cal\_rf}, but replacing the random forest
estimates with linear regression estimates and taking $D(v)=v-\log v$;\textbf{
(viii)} \texttt{aipw\_rf}: We use AIPW-base method (\citealp{tu2024Unified})
with the random forest estimates $\widehat{g}_{0k}^{\text{rf}}(\bs X_{i})$
and $\widehat{g}_{1k}^{\text{rf}}(\bs X_{i})$ to adjust the covariates;
\textbf{(ix)} \texttt{aipw\_nn}: Defined similarly to \texttt{aipw\_rf},
but replacing the random forest estimates with neural network estimates
$\widehat{g}_{0k}^{\text{nn}}(\cdot)$ and $\widehat{g}_{1k}^{\text{nn}}(\cdot)$;
\textbf{(x)} \texttt{aipw\_lin}: Defined similarly to \texttt{aipw\_rf},
but replacing the random forest estimates with linear regression estimates;
\textbf{(xi)} \texttt{sdim}: The stratified difference-in-means estimator
$\widehat{\tau}_{\mathrm{sdim}}$. For each estimator, we report the
following metrics based on 300 replications: absolute bias, empirical
standard deviation (SD), average estimated standard error (SE), and
the empirical coverage probability (CP) of the 95\% confidence intervals.

\textbf{Model 1.} Model 1 imposes linear models for the conditional
mean functions $g_{0}(\bs X)$ and $g_{1}(\bs X)$. In this model,
we set
\begin{align*}
g_{0}(\bs X_{i}) & =\mu_{0}+\sum_{j=1}^{4}\beta_{0j}X_{ij}\ \text{ and }\ g_{1}(\bs X_{i})=\mu_{1}+\sum_{j=1}^{4}\beta_{1j}X_{ij},
\end{align*}
with $\mu_{0}=1$, $\mu_{1}=4$, $(\beta_{01},\ldots,\beta_{04})=(75,35,125,80)$,
and $(\beta_{11},\ldots,\beta_{14})=(100,80,60,40)$. Additionally,
the variables are specified as follows: $\epsilon_{0,i}\sim N(0,1)$,
$\epsilon_{1,i}\sim N(0,9)$, $X_{i1}\sim\text{Beta}(3,4)$, $X_{i2}\sim\text{Uniform}(-2,2)$,
$X_{i3}$ takes values in $\{-1,1\}$ with equal probability, $X_{i4}$
takes values in $\{3,5\}$ with probabilities 0.6 and 0.4, respectively.
These variables are independent of one another. The remaining variables
$X_{i5},\dots,X_{ip}$ are independent of $X_{i1},\dots,X_{i4}$ and
follow a multivariate normal distribution with zero mean and a covariance
matrix where all off-diagonal elements are 0.2, while the diagonal
elements are 1. The randomization variable is an additional variable
taking values in $\{1,2,3,4\}$ with probabilities 0.2, 0.3, 0.3,
and 0.2, respectively, and is independent of $X_{ij}$ for $j=1,\dots,p$.
The simulation results for Model 1 are presented in Table~\ref{tab:Model1}.
Since both $g_{0}(\bs X_{i})$ and $g_{1}(\bs X_{i})$ are linear
in the covariates $X_{ij}$, the linear regression-adjusted estimator,
\texttt{aipw\_lin}, performs optimally in large samples ($n=1000,2000$),
as expected. The performance of \texttt{cal\_rflin} is very close
to that of \texttt{aipw\_lin}, which aligns with the conclusion of
Theorem \ref{thm:efficiency comparison thm}. However, in smaller
samples ($n=500$), the \texttt{aipw\_lin} estimator does not perform
the best, as linear regression is sensitive to outliers. In contrast,
\texttt{cal\_rflin} maintains stable performance, suggesting that
incorporating different estimates of $g_{0}(\bs X_{i})$ and $g_{1}(\bs X_{i})$
into the calibration estimator enhances its robustness. 
\begin{table}[!tbh]
\caption{\label{tab:Model1}The comparison of the performance of different
estimators under Model 1.}

\resizebox{\textwidth}{!}{%
\begin{threeparttable}
\begin{centering}
\begin{tabular}{clrrrrrrrrrrrr}
\toprule 
\multirow{2}{*}{$n$} & \multirow{2}{*}{Estimator} & \multicolumn{4}{c}{Simple Rand.} & \multicolumn{4}{c}{Stratified Block Rand.} & \multicolumn{4}{c}{Minimization}\tabularnewline
\cmidrule{3-14}
 &  & Bias & SD & SE & CP & Bias & SD & SE & CP & Bias & SD & SE & CP\tabularnewline
\midrule 
\multirow{11}{*}{500} & \texttt{cal\_rf} & 0.12 & 6.85 & 7.39 & 0.973 & 0.12 & 6.74 & 7.31 & 0.960 & 0.36 & 6.98 & 7.33 & 0.953\tabularnewline
 & \texttt{cal\_nn} & 0.33 & 11.38 & 11.03 & 0.960 & 0.65 & 11.01 & 11.05 & 0.950 & 0.63 & 11.47 & 11.00 & 0.933\tabularnewline
 & \texttt{cal\_rfnn} & 0.47 & 6.90 & 7.33 & 0.960 & 0.18 & 6.90 & 7.24 & 0.960 & 0.12 & 7.03 & 7.27 & 0.950\tabularnewline
 & \texttt{cal\_rflin} & 0.08 & 4.36 & 5.23 & 0.983 & 0.35 & 4.78 & 5.19 & 0.967 & 0.37 & 4.60 & 5.19 & 0.973\tabularnewline
 & \texttt{cal\_rf\_g} & 0.13 & 5.83 & 6.08 & 0.967 & 0.48 & 6.01 & 6.07 & 0.960 & 0.01 & 5.68 & 6.09 & 0.960\tabularnewline
 & \texttt{cal\_nn\_g} & 0.48 & 9.23 & 9.10 & 0.960 & 0.79 & 9.41 & 9.15 & 0.943 & 0.31 & 9.31 & 9.10 & 0.950\tabularnewline
 & \texttt{cal\_lin\_EL} & 0.10 & 4.25 & 5.26 & 0.980 & 0.32 & 4.49 & 5.21 & 0.983 & 0.14 & 4.39 & 5.22 & 0.987\tabularnewline
 & \texttt{aipw\_rf} & 0.31 & 9.03 & 8.91 & 0.957 & 0.07 & 8.74 & 8.88 & 0.940 & 0.65 & 9.13 & 8.88 & 0.943\tabularnewline
 & \texttt{aipw\_nn} & 0.29 & 14.17 & 13.55 & 0.923 & 0.74 & 14.45 & 13.54 & 0.910 & 1.22 & 13.98 & 13.52 & 0.950\tabularnewline
 & \texttt{aipw\_lin} & 3.22 & 50.59 & 14.23 & 0.953 & 0.69 & 25.12 & 9.04 & 0.967 & 1.78 & 19.45 & 7.30 & 0.953\tabularnewline
 & \texttt{sdim} & 0.48 & 12.35 & 12.34 & 0.953 & 0.06 & 11.71 & 12.31 & 0.943 & 0.93 & 12.61 & 12.28 & 0.927\tabularnewline
\midrule 
\multirow{11}{*}{1000} & \texttt{cal\_rf} & 0.21 & 4.00 & 4.13 & 0.957 & 0.06 & 3.70 & 4.12 & 0.960 & 0.09 & 3.85 & 4.12 & 0.963\tabularnewline
 & \texttt{cal\_nn} & 0.10 & 6.39 & 6.48 & 0.963 & 0.15 & 6.42 & 6.48 & 0.940 & 0.38 & 6.14 & 6.44 & 0.960\tabularnewline
 & \texttt{cal\_rfnn} & 0.05 & 3.93 & 4.06 & 0.963 & 0.06 & 3.78 & 4.05 & 0.960 & 0.17 & 3.90 & 4.05 & 0.973\tabularnewline
 & \texttt{cal\_rflin} & 0.17 & 2.94 & 3.28 & 0.973 & 0.24 & 2.84 & 3.27 & 0.973 & 0.05 & 3.14 & 3.27 & 0.957\tabularnewline
 & \texttt{cal\_rf\_g} & 0.24 & 3.68 & 3.66 & 0.940 & 0.05 & 3.45 & 3.67 & 0.963 & 0.08 & 3.67 & 3.66 & 0.950\tabularnewline
 & \texttt{cal\_nn\_g} & 0.40 & 4.48 & 4.73 & 0.967 & 0.18 & 4.83 & 4.76 & 0.950 & 0.55 & 4.74 & 4.74 & 0.950\tabularnewline
 & \texttt{cal\_lin\_EL} & 0.12 & 2.87 & 3.27 & 0.973 & 0.29 & 2.76 & 3.27 & 0.977 & 0.05 & 3.06 & 3.27 & 0.970\tabularnewline
 & \texttt{aipw\_rf} & 0.07 & 5.11 & 5.24 & 0.967 & 0.04 & 5.25 & 5.23 & 0.947 & 0.16 & 5.27 & 5.24 & 0.950\tabularnewline
 & \texttt{aipw\_nn} & 0.17 & 7.96 & 7.51 & 0.943 & 0.26 & 7.52 & 7.53 & 0.953 & 0.36 & 7.38 & 7.44 & 0.973\tabularnewline
 & \texttt{aipw\_lin} & 0.08 & 2.81 & 2.90 & 0.943 & 0.32 & 2.74 & 2.91 & 0.960 & 0.05 & 3.03 & 2.91 & 0.940\tabularnewline
 & \texttt{sdim} & 0.11 & 8.27 & 8.71 & 0.953 & 0.10 & 8.70 & 8.67 & 0.950 & 0.27 & 8.61 & 8.67 & 0.953\tabularnewline
\midrule 
\multirow{11}{*}{2000} & \texttt{cal\_rf} & 0.10 & 2.49 & 2.53 & 0.960 & 0.02 & 2.32 & 2.53 & 0.960 & 0.11 & 2.50 & 2.53 & 0.967\tabularnewline
 & \texttt{cal\_nn} & 0.27 & 2.50 & 2.68 & 0.960 & 0.01 & 2.62 & 2.68 & 0.963 & 0.05 & 2.64 & 2.67 & 0.950\tabularnewline
 & \texttt{cal\_rfnn} & 0.17 & 2.23 & 2.33 & 0.960 & 0.04 & 2.17 & 2.33 & 0.963 & 0.04 & 2.32 & 2.33 & 0.953\tabularnewline
 & \texttt{cal\_rflin} & 0.12 & 2.04 & 2.19 & 0.967 & 0.04 & 2.06 & 2.19 & 0.973 & 0.02 & 2.11 & 2.19 & 0.970\tabularnewline
 & \texttt{cal\_rf\_g} & 0.13 & 2.37 & 2.38 & 0.950 & 0.07 & 2.22 & 2.38 & 0.957 & 0.14 & 2.33 & 2.38 & 0.953\tabularnewline
 & \texttt{cal\_nn\_g} & 0.28 & 2.09 & 2.23 & 0.957 & 0.15 & 2.19 & 2.24 & 0.957 & 0.05 & 2.18 & 2.23 & 0.967\tabularnewline
 & \texttt{cal\_lin\_EL} & 0.16 & 2.01 & 2.19 & 0.970 & 0.00 & 2.03 & 2.19 & 0.970 & 0.02 & 2.08 & 2.19 & 0.973\tabularnewline
 & \texttt{aipw\_rf} & 0.12 & 3.04 & 3.05 & 0.957 & 0.15 & 3.00 & 3.05 & 0.953 & 0.20 & 3.10 & 3.05 & 0.947\tabularnewline
 & \texttt{aipw\_nn} & 0.29 & 2.77 & 2.91 & 0.967 & 0.04 & 2.93 & 2.89 & 0.927 & 0.06 & 2.81 & 2.88 & 0.953\tabularnewline
 & \texttt{aipw\_lin} & 0.16 & 1.99 & 2.06 & 0.967 & 0.00 & 2.02 & 2.06 & 0.957 & 0.03 & 2.06 & 2.05 & 0.963\tabularnewline
 & \texttt{sdim} & 0.22 & 5.87 & 6.13 & 0.967 & 0.41 & 6.28 & 6.13 & 0.940 & 0.37 & 6.12 & 6.13 & 0.953\tabularnewline
\bottomrule
\end{tabular}
\par\end{centering}
\begin{tablenotes}[flushleft]
      \footnotesize 
      \item \textit{Abbreviations:} Rand., Randomization; SD, standard deviation, SE: standard error; CP, coverage probability. 
    \end{tablenotes}
  \end{threeparttable}
}
\end{table}

\textbf{Model 2.} Model 2 imposes additive but nonlinear models for
the conditional mean functions $g_{0}(\bs X)$ and $g_{1}(\bs X)$.
In this model, we set
\begin{align*}
g_{0}(\bs X_{i}) & =\mu_{0}+\beta_{01}\log(X_{i1}+1)+\beta_{02}X_{i2}^{2}+\beta_{03}\exp(X_{i3})+\beta_{04}/(X_{i4}+3)\\
g_{1}(\bs X_{i}) & =\mu_{1}+\beta_{11}\exp(X_{i1}+2)+\beta_{12}/(X_{i1}+1)+\beta_{13}X_{i2}^{2},
\end{align*}
with $\mu_{0}=-3$, $\mu_{1}=0$, $(\beta_{01},\ldots,\beta_{04})=(10,24,15,20)$,
and $(\beta_{11},\beta_{12},\beta_{13})=(20,17,10)$. Additionally,
the variables are specified as follows: $\epsilon_{0,i}\sim N(0,1)$,
$\epsilon_{1,i}\sim N(0,9)$, $X_{i1}\sim\text{Beta}(3,4)$, $X_{i2}\sim\text{Uniform}(-2,2)$
with these two variables being independent of each other. The additional
covariates $X_{i3},\dots,X_{ip}$ are first generated as in Model
1. Then we randomly select $\left\lfloor p/3\right\rfloor $ covariates
from the additional covariates and multiply them by either $X_{i1}$
or $X_{i2}$ with equal probability to form the final additional covariates.
The randomization variable is an additional variable taking values
in $\{1,2,3,4\}$ with probabilities 0.2, 0.3, 0.3, and 0.2, respectively,
and is independent of $X_{ij}$ for $j=1,\dots,p$. The simulation
results for Model 2 are presented in Table~\ref{tab:Model2}. The
results indicate that random forests-based calibration estimators
(\texttt{cal\_rf}, \texttt{cal\_rfnn}, \texttt{cal\_rflm}, \texttt{cal\_rf\_g})
consistently outperform other methods across different randomization
schemes and sample sizes. The AIPW-based estimators (\texttt{aipw\_rf}
and \texttt{aipw\_nn}) tend to perform worse than the calibration
estimators (\texttt{cal\_rf} and \texttt{cal\_nn}), especially when
sample size is small ($n=500$). The empirical coverage probabilities
indicate that, in most cases, the estimators yield reliable 95\% confidence
intervals.

\begin{table}[!tbh]
\caption{\label{tab:Model2}The comparison of the performance of different
estimators under Model 2.}

\resizebox{\textwidth}{!}{%
\begin{threeparttable}
\begin{centering}
\begin{tabular}{clrrrrrrrrrrrr}
\toprule 
\multirow{2}{*}{$n$} & \multirow{2}{*}{Estimator} & \multicolumn{4}{c}{Simple Rand.} & \multicolumn{4}{c}{Stratified Block Rand.} & \multicolumn{4}{c}{Minimization}\tabularnewline
\cmidrule{3-14}
 &  & Bias & SD & SE & CP & Bias & SD & SE & CP & Bias & SD & SE & CP\tabularnewline
\midrule 
\multirow{11}{*}{500} & \texttt{cal\_rf} & 0.01 & 2.09 & 2.30 & 0.980 & 0.19 & 2.22 & 2.30 & 0.957 & 0.28 & 2.25 & 2.29 & 0.967\tabularnewline
 & \texttt{cal\_nn} & 0.16 & 3.09 & 3.14 & 0.957 & 0.14 & 3.19 & 3.14 & 0.940 & 0.10 & 3.08 & 3.13 & 0.960\tabularnewline
 & \texttt{cal\_rfnn} & 0.01 & 2.14 & 2.30 & 0.977 & 0.20 & 2.25 & 2.30 & 0.960 & 0.25 & 2.26 & 2.29 & 0.957\tabularnewline
 & \texttt{cal\_rflin} & 0.22 & 2.11 & 2.26 & 0.977 & 0.01 & 2.23 & 2.26 & 0.960 & 0.11 & 2.23 & 2.26 & 0.960\tabularnewline
 & \texttt{cal\_rf\_g} & 0.10 & 2.00 & 2.16 & 0.957 & 0.19 & 2.25 & 2.16 & 0.937 & 0.15 & 2.18 & 2.16 & 0.950\tabularnewline
 & \texttt{cal\_nn\_g} & 0.12 & 3.23 & 3.09 & 0.943 & 0.19 & 3.21 & 3.09 & 0.937 & 0.05 & 3.13 & 3.08 & 0.957\tabularnewline
 & \texttt{cal\_lin\_EL} & 0.33 & 2.68 & 2.78 & 0.950 & 0.19 & 2.73 & 2.77 & 0.960 & 0.10 & 2.69 & 2.76 & 0.963\tabularnewline
 & \texttt{aipw\_rf} & 0.20 & 2.35 & 2.42 & 0.950 & 0.04 & 2.50 & 2.42 & 0.940 & 0.09 & 2.46 & 2.42 & 0.957\tabularnewline
 & \texttt{aipw\_nn} & 0.87 & 5.98 & 5.99 & 0.963 & 0.47 & 6.02 & 5.92 & 0.957 & 0.25 & 5.94 & 5.87 & 0.937\tabularnewline
 & \texttt{aipw\_lin} & 2.91 & 196.74 & 43.13 & 0.953 & 3.67 & 142.28 & 37.24 & 0.963 & 6.10 & 95.39 & 33.45 & 0.937\tabularnewline
 & \texttt{sdim} & 0.22 & 3.07 & 3.12 & 0.950 & 0.14 & 3.19 & 3.12 & 0.940 & 0.08 & 3.06 & 3.12 & 0.970\tabularnewline
\midrule 
\multirow{11}{*}{1000} & \texttt{cal\_rf} & 0.22 & 1.47 & 1.51 & 0.957 & 0.15 & 1.48 & 1.50 & 0.950 & 0.29 & 1.54 & 1.50 & 0.943\tabularnewline
 & \texttt{cal\_nn} & 0.04 & 2.24 & 2.15 & 0.960 & 0.09 & 2.16 & 2.13 & 0.963 & 0.27 & 2.15 & 2.14 & 0.950\tabularnewline
 & \texttt{cal\_rfnn} & 0.22 & 1.49 & 1.50 & 0.953 & 0.15 & 1.50 & 1.50 & 0.947 & 0.25 & 1.55 & 1.50 & 0.940\tabularnewline
 & \texttt{cal\_rflin} & 0.08 & 1.44 & 1.47 & 0.960 & 0.08 & 1.43 & 1.47 & 0.957 & 0.15 & 1.53 & 1.46 & 0.937\tabularnewline
 & \texttt{cal\_rf\_g} & 0.17 & 1.46 & 1.45 & 0.940 & 0.12 & 1.46 & 1.45 & 0.947 & 0.25 & 1.50 & 1.45 & 0.937\tabularnewline
 & \texttt{cal\_nn\_g} & 0.02 & 2.17 & 2.02 & 0.950 & 0.11 & 2.08 & 2.00 & 0.937 & 0.30 & 2.04 & 2.02 & 0.943\tabularnewline
 & \texttt{cal\_lin\_EL} & 0.07 & 1.74 & 1.67 & 0.947 & 0.18 & 1.70 & 1.67 & 0.940 & 0.25 & 1.66 & 1.67 & 0.947\tabularnewline
 & \texttt{aipw\_rf} & 0.07 & 1.62 & 1.59 & 0.947 & 0.06 & 1.60 & 1.59 & 0.943 & 0.26 & 1.63 & 1.59 & 0.943\tabularnewline
 & \texttt{aipw\_nn} & 0.17 & 3.36 & 3.22 & 0.943 & 0.05 & 3.40 & 3.22 & 0.950 & 0.27 & 3.28 & 3.26 & 0.960\tabularnewline
 & \texttt{aipw\_lin} & 0.11 & 1.75 & 1.67 & 0.950 & 0.19 & 1.73 & 1.66 & 0.947 & 0.30 & 1.67 & 1.66 & 0.947\tabularnewline
 & \texttt{sdim} & 0.01 & 2.25 & 2.20 & 0.950 & 0.07 & 2.27 & 2.20 & 0.933 & 0.31 & 2.20 & 2.20 & 0.960\tabularnewline
\midrule 
\multirow{11}{*}{2000} & \texttt{cal\_rf} & 0.09 & 1.00 & 1.01 & 0.943 & 0.05 & 1.04 & 1.01 & 0.937 & 0.09 & 0.97 & 1.01 & 0.953\tabularnewline
 & \texttt{cal\_nn} & 0.03 & 1.26 & 1.27 & 0.950 & 0.13 & 1.31 & 1.27 & 0.950 & 0.05 & 1.25 & 1.27 & 0.950\tabularnewline
 & \texttt{cal\_rfnn} & 0.07 & 1.01 & 1.01 & 0.943 & 0.07 & 1.04 & 1.01 & 0.923 & 0.07 & 0.97 & 1.01 & 0.953\tabularnewline
 & \texttt{cal\_rflin} & 0.01 & 1.00 & 0.99 & 0.943 & 0.14 & 1.02 & 0.99 & 0.923 & 0.00 & 0.96 & 0.99 & 0.957\tabularnewline
 & \texttt{cal\_rf\_g} & 0.05 & 1.00 & 0.99 & 0.950 & 0.08 & 1.04 & 0.99 & 0.927 & 0.07 & 0.97 & 0.99 & 0.950\tabularnewline
 & \texttt{cal\_nn\_g} & 0.03 & 1.10 & 1.07 & 0.957 & 0.13 & 1.08 & 1.07 & 0.930 & 0.04 & 1.08 & 1.07 & 0.953\tabularnewline
 & \texttt{cal\_lin\_EL} & 0.01 & 1.15 & 1.13 & 0.943 & 0.12 & 1.14 & 1.13 & 0.943 & 0.03 & 1.13 & 1.13 & 0.963\tabularnewline
 & \texttt{aipw\_rf} & 0.01 & 1.07 & 1.06 & 0.947 & 0.12 & 1.09 & 1.05 & 0.933 & 0.05 & 1.06 & 1.06 & 0.960\tabularnewline
 & \texttt{aipw\_nn} & 0.06 & 1.38 & 1.44 & 0.957 & 0.13 & 1.57 & 1.43 & 0.923 & 0.12 & 1.49 & 1.43 & 0.940\tabularnewline
 & \texttt{aipw\_lin} & 0.00 & 1.14 & 1.12 & 0.930 & 0.11 & 1.14 & 1.12 & 0.937 & 0.02 & 1.11 & 1.12 & 0.960\tabularnewline
 & \texttt{sdim} & 0.03 & 1.58 & 1.55 & 0.943 & 0.15 & 1.58 & 1.55 & 0.930 & 0.12 & 1.64 & 1.55 & 0.933\tabularnewline
\bottomrule
\end{tabular}
\par\end{centering}
\begin{tablenotes}[flushleft]
      \footnotesize 
      \item \textit{Abbreviations:} Rand., Randomization; SD, standard deviation, SE: standard error; CP, coverage probability. 
    \end{tablenotes}
  \end{threeparttable}
}
\end{table}

\textbf{Model 3.} Model 3 imposes non-additive and nonlinear models
for the conditional mean functions $g_{0}(\bs X)$ and $g_{1}(\bs X)$.
In this model, we set
\begin{align*}
g_{0}(\bs X_{i}) & =\mu_{0}+\beta_{01}X_{i1}X_{i2}/(X_{i1}+X_{i2}+2)+\beta_{02}X_{i1}^{2}(X_{i2}+X_{i3})\\
g_{1}(\bs X_{i}) & =\mu_{1}+\beta_{11}(X_{i2}+X_{i4})+\beta_{12}X_{i2}^{2}/\exp(X_{i1}+2),
\end{align*}
with $\mu_{0}=5$, $\mu_{1}=2$, $(\beta_{01},\beta_{02})=(42,83)$,
and $(\beta_{11},\beta_{12})=(30,75)$. Additionally, the variables
are specified as follows: $\epsilon_{0,i}\sim\text{t}(2)$, $\epsilon_{1,i}\sim3\times\text{t}(2)$,
$X_{i1}\sim\text{Beta}(3,4)$, $X_{i2}\sim\text{Uniform}(-2,2)$,
$X_{i3}\sim N(0,1)$, $X_{i4}\sim\text{Uniform}(0,2)$, respectively.
These variables are independent of one another. The remaining variables
$X_{i5},\dots,X_{ip}$ are independent of $X_{i1},\dots,X_{i4}$ and
follow a multivariate normal distribution with zero mean and a symmetric
Toeplitz covariance matrix where the first row is a geometric sequence
with initial value 1 and common ratio 0.5. The randomization variable
is an additional variable taking values in $\{1,2\}$ with probabilities
0.4, and 0.6, respectively, and is independent of $X_{ij}$ for $j=1,\dots,p$.
The simulation results for Model 3 are presented in Table~\ref{tab:Model3}.
The results indicate that the random forests-based calibration estimators
(\texttt{cal\_rf}, \texttt{cal\_rfnn}, \texttt{cal\_rflm}, \texttt{cal\_rf\_g})
perform the best across different randomization methods and sample
sizes. As the sample size increases, the calibration estimators converge
toward better performance with lower SD, while maintaining correct
empirical coverage probabilities. However, the \texttt{sdim} estimator
consistently lags behind the calibration estimators in terms of SD,
especially in smaller sample sizes. This finding is consistent with
Theorem \ref{thm:asymptotic properties of tau_cal}, which demonstrates
that the calibration estimator is always more efficient than the \texttt{sdim}
estimator.

\begin{table}[!tbh]
\caption{\label{tab:Model3}The comparison of the performance of different
estimators under Model 3.}

\resizebox{\textwidth}{!}{%
\begin{threeparttable}
\begin{centering}
\begin{tabular}{clrrrrrrrrrrrr}
\toprule 
\multirow{2}{*}{$n$} & \multirow{2}{*}{Estimator} & \multicolumn{4}{c}{Simple Rand.} & \multicolumn{4}{c}{Stratified Block Rand.} & \multicolumn{4}{c}{Minimization}\tabularnewline
\cmidrule{3-14}
 &  & Bias & SD & SE & CP & Bias & SD & SE & CP & Bias & SD & SE & CP\tabularnewline
\midrule 
\multirow{11}{*}{500} & \texttt{cal\_rf} & 0.13 & 2.80 & 2.68 & 0.933 & 0.01 & 2.40 & 2.61 & 0.980 & 0.03 & 2.41 & 2.65 & 0.967\tabularnewline
 & \texttt{cal\_nn} & 0.17 & 3.62 & 3.43 & 0.943 & 0.08 & 3.20 & 3.35 & 0.960 & 0.08 & 3.29 & 3.38 & 0.960\tabularnewline
 & \texttt{cal\_rfnn} & 0.09 & 2.83 & 2.67 & 0.940 & 0.00 & 2.45 & 2.61 & 0.967 & 0.01 & 2.41 & 2.65 & 0.960\tabularnewline
 & \texttt{cal\_rflin} & 0.08 & 2.80 & 2.63 & 0.930 & 0.01 & 2.36 & 2.56 & 0.967 & 0.02 & 2.42 & 2.61 & 0.957\tabularnewline
 & \texttt{cal\_rf\_g} & 0.16 & 2.74 & 2.58 & 0.927 & 0.01 & 2.27 & 2.52 & 0.980 & 0.11 & 2.35 & 2.56 & 0.957\tabularnewline
 & \texttt{cal\_nn\_g} & 0.08 & 3.23 & 3.16 & 0.933 & 0.08 & 2.94 & 3.08 & 0.953 & 0.08 & 2.93 & 3.12 & 0.960\tabularnewline
 & \texttt{cal\_lin\_EL} & 0.09 & 3.00 & 2.82 & 0.923 & 0.10 & 2.51 & 2.76 & 0.970 & 0.03 & 2.66 & 2.80 & 0.967\tabularnewline
 & \texttt{aipw\_rf} & 0.04 & 3.08 & 2.84 & 0.937 & 0.16 & 2.60 & 2.78 & 0.960 & 0.16 & 2.67 & 2.82 & 0.953\tabularnewline
 & \texttt{aipw\_nn} & 0.11 & 5.09 & 4.88 & 0.943 & 0.12 & 4.72 & 4.62 & 0.947 & 0.23 & 4.93 & 4.79 & 0.947\tabularnewline
 & \texttt{aipw\_lin} & 0.13 & 3.16 & 3.01 & 0.937 & 0.12 & 2.79 & 2.89 & 0.930 & 0.02 & 2.83 & 2.95 & 0.977\tabularnewline
 & \texttt{sdim} & 0.19 & 4.31 & 4.00 & 0.930 & 0.23 & 3.77 & 3.93 & 0.950 & 0.26 & 3.87 & 3.97 & 0.960\tabularnewline
\midrule 
\multirow{11}{*}{1000} & \texttt{cal\_rf} & 0.03 & 1.87 & 1.79 & 0.947 & 0.06 & 1.82 & 1.80 & 0.940 & 0.07 & 1.66 & 1.78 & 0.957\tabularnewline
 & \texttt{cal\_nn} & 0.09 & 2.31 & 2.18 & 0.937 & 0.02 & 2.15 & 2.20 & 0.947 & 0.02 & 2.20 & 2.18 & 0.960\tabularnewline
 & \texttt{cal\_rfnn} & 0.06 & 1.89 & 1.79 & 0.937 & 0.08 & 1.82 & 1.79 & 0.953 & 0.08 & 1.68 & 1.78 & 0.960\tabularnewline
 & \texttt{cal\_rflin} & 0.06 & 1.88 & 1.76 & 0.943 & 0.05 & 1.82 & 1.77 & 0.943 & 0.08 & 1.69 & 1.75 & 0.953\tabularnewline
 & \texttt{cal\_rf\_g} & 0.01 & 1.87 & 1.75 & 0.933 & 0.06 & 1.76 & 1.75 & 0.957 & 0.11 & 1.67 & 1.74 & 0.953\tabularnewline
 & \texttt{cal\_nn\_g} & 0.04 & 2.12 & 2.01 & 0.950 & 0.04 & 1.94 & 2.02 & 0.963 & 0.15 & 2.03 & 2.01 & 0.963\tabularnewline
 & \texttt{cal\_lin\_EL} & 0.02 & 1.96 & 1.82 & 0.943 & 0.05 & 1.87 & 1.82 & 0.953 & 0.04 & 1.76 & 1.81 & 0.967\tabularnewline
 & \texttt{aipw\_rf} & 0.00 & 2.04 & 1.87 & 0.943 & 0.01 & 1.91 & 1.88 & 0.947 & 0.01 & 1.81 & 1.86 & 0.957\tabularnewline
 & \texttt{aipw\_nn} & 0.20 & 2.98 & 2.66 & 0.920 & 0.01 & 2.97 & 2.70 & 0.950 & 0.05 & 2.76 & 2.66 & 0.950\tabularnewline
 & \texttt{aipw\_lin} & 0.00 & 1.98 & 1.83 & 0.940 & 0.08 & 1.89 & 1.83 & 0.947 & 0.04 & 1.79 & 1.82 & 0.953\tabularnewline
 & \texttt{sdim} & 0.07 & 2.69 & 2.81 & 0.967 & 0.05 & 2.77 & 2.81 & 0.950 & 0.00 & 2.77 & 2.80 & 0.967\tabularnewline
\midrule 
\multirow{11}{*}{2000} & \texttt{cal\_rf} & 0.10 & 1.31 & 1.22 & 0.943 & 0.01 & 1.20 & 1.21 & 0.947 & 0.06 & 1.41 & 1.24 & 0.937\tabularnewline
 & \texttt{cal\_nn} & 0.08 & 1.42 & 1.36 & 0.947 & 0.08 & 1.37 & 1.36 & 0.957 & 0.11 & 1.59 & 1.39 & 0.943\tabularnewline
 & \texttt{cal\_rfnn} & 0.10 & 1.29 & 1.21 & 0.943 & 0.03 & 1.21 & 1.21 & 0.947 & 0.09 & 1.43 & 1.24 & 0.940\tabularnewline
 & \texttt{cal\_rflin} & 0.10 & 1.29 & 1.20 & 0.940 & 0.01 & 1.20 & 1.20 & 0.943 & 0.06 & 1.41 & 1.23 & 0.940\tabularnewline
 & \texttt{cal\_rf\_g} & 0.12 & 1.28 & 1.19 & 0.950 & 0.03 & 1.16 & 1.19 & 0.950 & 0.08 & 1.40 & 1.22 & 0.943\tabularnewline
 & \texttt{cal\_nn\_g} & 0.10 & 1.34 & 1.28 & 0.940 & 0.07 & 1.30 & 1.28 & 0.953 & 0.13 & 1.55 & 1.31 & 0.940\tabularnewline
 & \texttt{cal\_lin\_EL} & 0.07 & 1.34 & 1.24 & 0.930 & 0.04 & 1.24 & 1.23 & 0.957 & 0.02 & 1.46 & 1.26 & 0.953\tabularnewline
 & \texttt{aipw\_rf} & 0.05 & 1.36 & 1.25 & 0.930 & 0.03 & 1.26 & 1.25 & 0.950 & 0.01 & 1.46 & 1.27 & 0.940\tabularnewline
 & \texttt{aipw\_nn} & 0.08 & 1.54 & 1.45 & 0.930 & 0.05 & 1.48 & 1.45 & 0.963 & 0.06 & 1.77 & 1.49 & 0.933\tabularnewline
 & \texttt{aipw\_lin} & 0.07 & 1.31 & 1.23 & 0.933 & 0.03 & 1.25 & 1.23 & 0.953 & 0.04 & 1.57 & 1.26 & 0.930\tabularnewline
 & \texttt{sdim} & 0.07 & 2.14 & 1.98 & 0.927 & 0.05 & 2.06 & 1.98 & 0.943 & 0.01 & 2.17 & 2.00 & 0.947\tabularnewline
\bottomrule
\end{tabular}
\par\end{centering}
\begin{tablenotes}[flushleft]
      \footnotesize 
      \item \textit{Abbreviations:} Rand., Randomization; SD, standard deviation, SE: standard error; CP, coverage probability. 
    \end{tablenotes}
  \end{threeparttable}
}
\end{table}

Models 1-3 assume the homogeneity of the conditional mean functions
$g_{0}(\bs X)$ and $g_{1}(\bs X)$ across different strata. In Appendix~\ref{sec:Additional-simulation-results},
we also examine a model that accounts for the heterogeneity of the
conditional mean functions $g_{0}(\bs X)$ and $g_{1}(\bs X)$ across
different strata. Our findings show that the results are very similar
to those obtained under Models 1--3.

\section{Empirical application\label{sec:Empirical-application}}

In this section, we apply our calibration method to the experimental
data from \citet{dupas2018Bankinga}, which conducted covariate-adaptive
randomized experiments to assess the impact of subsidized bank account
access ($A_{i}$) on total savings ($Y_{i}$) for individuals across
three countries: Uganda, Malawi, and Chile. \citet{dupas2018Bankinga}
estimated both the average treatment effects (ATEs) and the quantile
treatment effects (QTEs) of the subsidy. More recently, \citet{jiang2023Regressionadjusted}
applied a regression-adjusted estimator to this dataset to analyze
the QTE in Uganda.

In this section, we aim to estimate the ATEs of subsidized access
to bank accounts in Uganda and Malawi. The Uganda sample includes
2,159 observations, stratified into 41 strata based on gender, occupation,
and bank branch, following a stratified block randomization design.
The Malawi sample includes 2,108 observations, stratified into 78
strata based on occupation, gender, marital status, literacy, and
whether the respondent was from the household or market, also following
a stratified block  randomization design. To avoid overly small strata,
we exclude those containing fewer than six samples. This results in
a final Uganda sample of 2,115 observations across 37 strata, and
a Malawi sample of 1,987 observations across 67 strata. In both countries,
half of the households were randomly assigned to receive the bank
account subsidy, while the other half served as the control group.
All monetary variables are winsorized at the ninety-ninth percentile
to mitigate the influence of outliers. After the randomization and
the intervention, \citet{dupas2018Bankinga} conducted 3 rounds of
follow-up surveys in Uganda and Malawi. In line with \citet{jiang2023Regressionadjusted},
we focus on the first-round follow-up survey to assess the impact
of the bank account subsidy on total savings.

Following \citet{dupas2018Bankinga}, we consider one baseline covariate
$X$: the baseline value of total savings. When estimating the ATE
of subsidized bank account access in Uganda, we leverage data from
Malawi to obtain certain components of $\bs{\xi}_{n}(X)$ used in
our calibration estimator, and vice versa when estimating the ATE
for Malawi. For both countries, we consider the following estimators:
(i) \texttt{cal\_X}: this estimator is obtained by taking $\bs{\xi}_{n}(X)=X$
in (\ref{eq:cal estimator}); (ii) \texttt{cal\_X}$^{\beta}$: this
estimator is obtained by taking $\bs{\xi}_{n}(X)=(X+1)^{\beta}$ in
(\ref{eq:cal estimator}); (iii) \texttt{cal\_X\_X}$^{\beta}$: this
estimator is obtained by taking $\bs{\xi}_{n}(X)=(X,(X+1)^{\beta})^{\trans}$
in (\ref{eq:cal estimator}); (iv) \texttt{cal\_info\_X}: this estimator
is obtained by taking $\bs{\xi}_{n}(X)=(X,\widehat{g}_{\text{rf}}^{\text{info}}(X))^{\trans}$
in (\ref{eq:cal estimator}), where $\widehat{g}_{\text{rf}}^{\text{info}}(X)$
is the random forest estimator of $\e\left[Y\mid X\right]$ using
the data from the other country; (v) \texttt{cal\_info\_X}$^{\beta}$:
this estimator is obtained by taking $\bs{\xi}_{n}(X)=((X+1)^{\beta},\widehat{g}_{\text{rf}}^{\text{info}}(X))^{\trans}$
in (\ref{eq:cal estimator}); (vi) \texttt{cal\_info\_X\_X}$^{\beta}$:
this estimator is obtained by taking $\bs{\xi}_{n}(X)=(X,(X+1)^{\beta},\widehat{g}_{\text{rf}}^{\text{info}}(X))^{\trans}$
in (\ref{eq:cal estimator}); (vii) \texttt{sdim}: the stratified
difference-in-means estimator $\widehat{\tau}_{\mathrm{sdim}}$. We
include the term $(X+1)^{\beta}$ because an approximately linear
relationship is observed between $\log(Y+1)$ and $\log(X+1)$ in
both countries. We fit linear regressions of the form $\log(Y+1)=\alpha+\beta\log(X+1)+\epsilon$,
yielding an estimated $\beta$ of 0.481 for Uganda and 0.408 for Malawi.
Figure~\ref{fig:real data CI} and Table~\ref{tab:empirical-result}
present the 95\% confidence intervals and point estimates obtained
using these methods.

\begin{figure}[!tbh]
\begin{centering}
\includegraphics[width=0.85\textwidth]{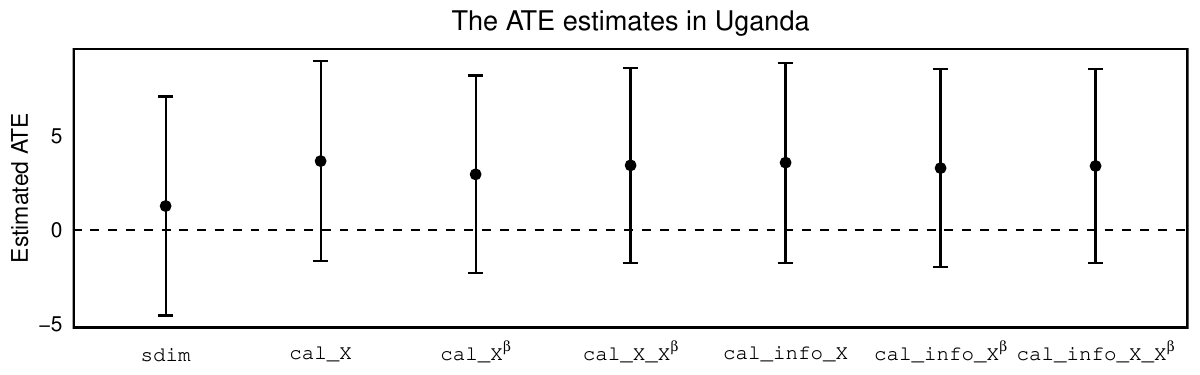}
\par\end{centering}
\begin{centering}
\includegraphics[width=0.85\textwidth]{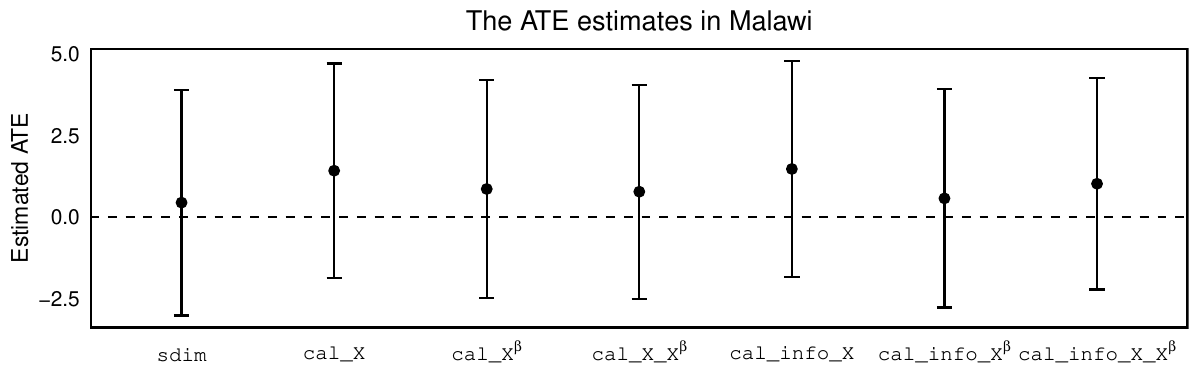}
\par\end{centering}
\caption{\label{fig:real data CI}The 95\% confidence intervals for the ATE
of the bank account subsidy on total household savings in Uganda and
Malawi, respectively. Top panel: Uganda; bottom panel: Malawi. The
estimators whose names contain ``\texttt{X}'' (or ``\texttt{X}$^{\beta}$'')
indicate that we add the component $X$ (or $(X+1)^{\beta}$) to $\protect\bs{\xi}(X)$
to obtain the calibration estimator. The estimators whose names contain ``\texttt{info}'' indicate that we incorporate information from the
other country to form one component of $\protect\bs{\xi}(X)$.}
\end{figure}

\begin{table}[!tbh]
\caption{\label{tab:empirical-result}The ATE estimates of the bank account
subsidy effect on total household savings in Uganda and Malawi, with
standard errors reported in parentheses.}

\resizebox{\textwidth}{!}{%
\begin{threeparttable}
\begin{centering}
\begin{tabular}{cccccccc}
\toprule 
 & \texttt{sdim} & \texttt{cal\_X} & \texttt{cal\_X}$^{\beta}$ & \texttt{cal\_X\_X}$^{\beta}$ & \texttt{cal\_info\_X} & \texttt{cal\_info\_X}$^{\beta}$ & \texttt{cal\_info\_X\_X}$^{\beta}$\tabularnewline
\midrule
\midrule
\multirow{2}{*}{Uganda} & 1.289 & 3.687 & 2.979 & 3.459 & 3.604 & 3.313 & 3.426\tabularnewline
 & (2.980) & (2.726) & (2.691) & (2.661) & (2.722) & (2.697) & (2.645)\tabularnewline
\multirow{2}{*}{Malawi} & 0.452 & 1.431 & 0.870 & 0.787 & 1.484 & 0.582 & 1.031\tabularnewline
 & (1.766) & (1.679) & (1.702) & (1.674) & (1.690) & (1.708) & (1.654)\tabularnewline
\bottomrule
\end{tabular}
\par\end{centering}
\begin{tablenotes}[flushleft]
 \footnotesize 
 \item \textit{Note:} The estimators whose names contain ``\texttt{X}`` (or ``\texttt{X}$^{\beta}$'') indicate that we add the component $X$ (or $(X+1)^{\beta}$) to $\bs{\xi}(X)$ to obtain the calibration estimator. The estimators whose names contain ``\texttt{info}'' indicate that we incorporate information from the other country to form one component of $\bs{\xi}(X)$.
    \end{tablenotes}
  \end{threeparttable}
}
\end{table}

The results presented in Figure~\ref{fig:real data CI} and Table~\ref{tab:empirical-result}
lead to two key observations. First, in line with the theoretical
results (Theorem \ref{thm:efficiency comparison thm}), the standard
errors for the \texttt{cal\_info\_X\_X}$^{\beta}$ estimator are the
lowest among all estimators in both Uganda and Malawi. For example,
the standard errors of \texttt{cal\_info\_X\_X}$^{\beta}$ ATE estimates
are 11.2\% and 6.3\% smaller than those of the stratified difference-in-means
ATE estimates in Uganda and Malawi, respectively. Second, in both
countries, all ATE estimates are statistically insignificant, suggesting
that expanding access to basic bank accounts does not lead to a significant
increase in total savings on average. This finding is consistent with
the results of \citet{dupas2018Bankinga}.

%%%%%%%%%%%%%%%%%%%%%

\putbib 
\end{bibunit}

\newpage
\setcounter{page}{1}
{\setstretch{1.7} 
\begin{center}
    {\bf \Large 
    Supplementary Materials for ``Integrating Heterogeneous Information in Randomized Experiments: A Unified Calibration Framework''}
\end{center}
\bigskip
}

\begin{center}
\large
	Wei Ma \hspace{1em} Zeqi Wu \hspace{1em} Zheng Zhang \\ 
    Institute of Statistics and Big Data, Renmin University of China
\end{center}

\doublespacing

\bigskip
This supplementary material includes appendices containing additional simulation results and proofs for the main paper.
\appendix

\renewcommand{\thesection}{\Alph{section}}
\titleformat{\section}
{\normalfont\Large\bfseries}
{Appendix \thesection}{1em}{}

\begin{bibunit}
\section{Some useful lemmas}
\begin{lem}
\label{lem:LLN}Suppose that Assumptions \ref{assu:independent sampling}-\ref{assu:treatment assignment}
hold and $\frac{\log(2K)}{n\inf_{1\leq k\leq K}p_{[k]}}\to0$ as $n\to\infty$.
Then 
\[
\left|\frac{n_{[k]}}{n}-p_{[k]}\right|=\left|\frac{1}{n}\sum_{i=1}^{n}\1(B_{i}=k)-p_{[k]}\right|\leq\sqrt{p_{[k]}}O_{P}\left(\sqrt{\frac{\log(2K)}{n}}\right)=p_{[k]}o_{P}(1)
\]
and 
\[
\left|\frac{n}{n_{[k]}}-\frac{1}{p_{[k]}}\right|=\frac{1}{p_{[k]}}o_{P}(1),
\]
where the random variables $O_{P}\left(\sqrt{\frac{\log(2K)}{n}}\right)$
and $o_{P}(1)$ do not depend on $k$. In addition, we also have 
\[
\left|\frac{1}{n}\sum_{i=1}^{n}A_{i}\1(B_{i}=k)-\pi_{[k]}p_{[k]}\right|\leq p_{[k]}o_{P}(1),
\]
where the random variable $o_{P}(1)$ does not depend on $k$.
\end{lem}
\begin{proof}
By Chernoff’s inequality (see \citet[Exercise 2.3.5]{vershynin2018HighDimensional}),
we have 
\begin{align*}
P\left(\frac{\left|\frac{1}{n}\sum_{i=1}^{n}\1(B_{i}=k)-p_{[k]}\right|}{\sqrt{p_{[k]}}}\geq\delta\right) & \leq2\exp\left(-c\delta^{2}n\right)
\end{align*}
for any $0<\delta<\sqrt{p_{[k]}}$. Then, it follows from the union
bound that 
\begin{align*}
P\left(\sup_{1\leq k\leq K}\frac{\left|\frac{1}{n}\sum_{i=1}^{n}\1(B_{i}=k)-p_{[k]}\right|}{\sqrt{p_{[k]}}}\geq\delta\right) & \leq2K\exp\left(-c\delta^{2}n\right)
\end{align*}
for any $0<\delta<\min_{1\leq k\leq K}\sqrt{p_{[k]}}$, which leads
to 
\begin{align*}
P\left(\sup_{1\leq k\leq K}\frac{\left|\frac{1}{n}\sum_{i=1}^{n}\1(B_{i}=k)-p_{[k]}\right|}{\sqrt{p_{[k]}}}\geq\delta\sqrt{\frac{\log(2K)}{n}}\right) & \leq\exp\left(\log(2K)\left\{ 1-c\delta^{2}\right\} \right)
\end{align*}
for any $0<\delta<\sqrt{\frac{n}{\log(2K)}}\min_{1\leq k\leq K}\sqrt{p_{[k]}}$.
Let $0<\epsilon<2$ be any number. Since $\frac{\log(2K)}{n\inf_{1\leq k\leq K}p_{[k]}}\to0$,
there exists an integer $N_{0}$ such that 
\[
\sqrt{\frac{1-\log\left(\frac{\epsilon}{2}\right)}{c}}<\sqrt{\frac{n}{\log(2K)}}\min_{1\leq k\leq K}\sqrt{p_{[k]}}
\]
for all $n\geq N_{0}$. Let $\delta_{\epsilon}=\sqrt{\frac{1-\log\left(\frac{\epsilon}{2}\right)}{c}}$,
then we have $1-c\delta_{\epsilon}^{2}=\log\left(\frac{\epsilon}{2}\right)<0$
and
\begin{align*}
 & P\left(\sup_{1\leq k\leq K}\frac{\left|\frac{1}{n}\sum_{i=1}^{n}\1(B_{i}=k)-p_{[k]}\right|}{\sqrt{p_{[k]}}}\geq\delta_{\epsilon}\sqrt{\frac{\log(2K)}{n}}\right)\\
\leq & \exp\left(\log(2K)\left\{ 1-c\delta_{\epsilon}^{2}\right\} \right)\leq\exp\left(\log(2)\left\{ 1-c\delta_{\epsilon}^{2}\right\} \right)=\epsilon.
\end{align*}
This implies that 
\[
\sup_{1\leq k\leq K}\frac{\left|\frac{1}{n}\sum_{i=1}^{n}\1(B_{i}=k)-p_{[k]}\right|}{\sqrt{p_{[k]}}}=O_{P}\left(\sqrt{\frac{\log(2K)}{n}}\right)
\]
and thus
\[
\left|\frac{n_{[k]}}{n}-p_{[k]}\right|=\left|\frac{1}{n}\sum_{i=1}^{n}\1(B_{i}=k)-p_{[k]}\right|\leq\sqrt{p_{[k]}}O_{P}\left(\sqrt{\frac{\log(2K)}{n}}\right),
\]
where the random variable $O_{P}\left(\sqrt{\frac{\log(2K)}{n}}\right)$
does not depend on $k$. Note that 
\begin{align*}
 & \frac{1}{\sqrt{p_{[k]}}}O_{P}\left(\sqrt{\frac{\log(2K)}{n}}\right)\leq\frac{1}{\min_{1\leq k\leq K}\sqrt{p_{[k]}}}O_{P}\left(\sqrt{\frac{\log(2K)}{n}}\right)\\
= & O_{P}\left(\sqrt{\frac{\log(2K)}{\min_{1\leq k\leq K}\sqrt{p_{[k]}}n}}\right)=o_{P}(1).
\end{align*}
We also have 
\[
\sqrt{p_{[k]}}O_{P}\left(\sqrt{\frac{\log(2K)}{n}}\right)=p_{[k]}\frac{1}{\sqrt{p_{[k]}}}O_{P}\left(\sqrt{\frac{\log(2K)}{n}}\right)=p_{[k]}o_{P}(1).
\]

Furthermore, we have 
\begin{align*}
\left|\frac{n}{n_{[k]}}-\frac{1}{p_{[k]}}\right| & =\frac{1}{p_{[k]}}\left|\frac{\frac{n_{[k]}}{n}-p_{[k]}}{\frac{n_{[k]}}{n}}\right|=\frac{1}{p_{[k]}}\left|\frac{p_{[k]}o_{P}(1)}{p_{[k]}+p_{[k]}o_{P}(1)}\right|=\frac{1}{p_{[k]}}\left|\frac{o_{P}(1)}{1+o_{P}(1)}\right|=\frac{1}{p_{[k]}}o_{P}(1)
\end{align*}
where the random variable $o_{P}(1)$ does not depend on $k$ and
the second step follows from $\left|\frac{n_{[k]}}{n}-p_{[k]}\right|=p_{[k]}o_{P}(1)$.

The last conclusion follows from
\begin{align*}
 & \left|\frac{1}{n}\sum_{i=1}^{n}A_{i}\1(B_{i}=k)-\pi_{[k]}p_{[k]}\right|=\left|\frac{1}{n}\sum_{i=1}^{n}\1(B_{i}=k)\frac{\sum_{i=1}^{n}A_{i}\1(B_{i}=k)}{\sum_{i=1}^{n}\1(B_{i}=k)}-\pi_{[k]}p_{[k]}\right|\\
\leq & \left|\frac{1}{n}\sum_{i=1}^{n}\1(B_{i}=k)-p_{[k]}\right|\frac{\sum_{i=1}^{n}A_{i}\1(B_{i}=k)}{\sum_{i=1}^{n}\1(B_{i}=k)}+\left|\frac{\sum_{i=1}^{n}A_{i}\1(B_{i}=k)}{\sum_{i=1}^{n}\1(B_{i}=k)}-\pi_{[k]}\right|p_{[k]}\\
\leq & p_{[k]}o_{P}(1),
\end{align*}
where the last step follows from $\left|\frac{n_{[k]}}{n}-p_{[k]}\right|\leq p_{[k]}o_{P}(1)$
and Assumption \ref{assu:treatment assignment}.
\end{proof}
\begin{lem}
\label{lem:uniform LLN}Suppose that Assumptions \ref{assu:independent sampling}-\ref{assu:treatment assignment}
and \ref{assu:conditions on =00005Cxi-diverging xi} hold. Then for
every $k=1,\ldots,K$, it holds that (i)
\[
\left\Vert \frac{1}{n}\sum_{i=1}^{n}\1(B_{i}=k)\bs{\xi}_{n}^{*}(\bs X_{i})-\e\left[\1(B_{i}=k)\bs{\xi}_{n}^{*}(\bs X_{i})\right]\right\Vert \leq\sqrt{p_{[k]}}O_{P}\left(\sqrt{\frac{r_{n}\log(2Kd)}{n}}\right)
\]
where the random variable $O_{P}\left(\sqrt{\frac{r_{n}\log(2Kd)}{n}}\right)$
does not depend on $k$; (ii)
\[
\left\Vert \frac{1}{n}\sum_{i=1}^{n}A_{i}\1(B_{i}=k)\widetilde{\bs{\xi}}_{n}^{*}(\bs X_{i})\right\Vert \leq\sqrt{p_{[k]}}O_{P}\left(\sqrt{\frac{r_{n}\log(2Kd)}{n}}\right)
\]
where the random variable $O_{P}\left(\sqrt{\frac{r_{n}\log(2Kd)}{n}}\right)$
does not depend on $k$; (iii) 
\[
\left\Vert \frac{1}{n}\sum_{i=1}^{n}\1(B_{i}=k)\bs{\xi}_{n}^{*}(\bs X_{i})\bs{\xi}_{n}^{*}(\bs X_{i})^{\trans}-\e\left[\1(B_{i}=k)\bs{\xi}_{n}^{*}(\bs X_{i})\bs{\xi}_{n}^{*}(\bs X_{i})^{\trans}\right]\right\Vert \leq p_{[k]}o_{P}(1),
\]
where the random variable $o_{P}(1)$ does not depend on $k$; and
(iv)
\[
\left\Vert \frac{1}{n}\sum_{i=1}^{n}\bs{\Xi}_{i,[k]}^{*}\bs{\Xi}_{i,[k]}^{*\trans}-\mathbf{\Sigma}_{[k]}^{\mathcal{C}_{n}}\right\Vert \leq p_{[k]}o_{P}(1),
\]
where the random variable $o_{P}(1)$ does not depend on $k$;
\end{lem}
\begin{proof}
\textbf{(i)} Let 
\[
\bs Z_{i}=\1(B_{i}=k)\bs{\xi}_{n}^{*}(\bs X_{i})-\e\left[\1(B_{i}=k)\bs{\xi}_{n}^{*}(\bs X_{i})\right]
\]
and 
\[
\bs S_{n}=\sum_{i=1}^{n}\bs Z_{i}.
\]
Then by Assumption \ref{assu:conditions on =00005Cxi-diverging xi}
it holds that 
\[
\left\Vert \bs Z_{i}\right\Vert \leq2\zeta_{n}
\]
for all $1\leq i\leq n$ and 
\begin{align*}
V & :=\max\left\{ \left\Vert \e\left[\bs S_{n}\bs S_{n}^{\trans}\right]\right\Vert ,\left\Vert \e\left[\bs S_{n}^{\trans}\bs S_{n}\right]\right\Vert \right\} =\max\left\{ \left\Vert \sum_{i=1}^{n}\e\left[\bs Z_{i}\bs Z_{i}^{\trans}\right]\right\Vert ,\left\Vert \sum_{i=1}^{n}\e\left[\bs Z_{i}^{\trans}\bs Z_{i}\right]\right\Vert \right\} \\
 & \leq\sum_{i=1}^{n}\e\left[\left\Vert \bs Z_{i}\right\Vert ^{2}\right]\leq\sum_{i=1}^{n}\e\left[\1(B_{i}=k)\left\Vert \bs{\xi}_{n}^{*}(\bs X_{i})\right\Vert ^{2}\right]=np_{[k]}\tr\left\{ \e\left[\bs{\xi}_{n}^{*}(\bs X_{i})\bs{\xi}_{n}^{*}(\bs X_{i})^{\trans}\mid B_{i}=k\right]\right\} \\
 & \leq Cnp_{[k]}r_{n}.
\end{align*}
It follows from the Bernstain's inequality for matrices (see \citet[Theorem 1.6.2]{tropp2015Introduction})
that
\[
P\left(\left\Vert \bs S_{n}\right\Vert \geq t\right)\leq(d+1)\exp\left(-\frac{t^{2}/2}{V+2\zeta_{n}t/3}\right)\leq(d+1)\exp\left(-\frac{t^{2}/2}{Cnp_{[k]}r_{n}+2\zeta_{n}t/3}\right)
\]
for all $t\geq0$. Thus, 
\begin{align*}
 & P\left(\frac{\left\Vert \frac{1}{n}\sum_{i=1}^{n}\1(B_{i}=k)\bs{\xi}_{n}^{*}(\bs X_{i})-\e\left[\1(B_{i}=k)\bs{\xi}_{n}^{*}(\bs X_{i})\right]\right\Vert }{\sqrt{p_{[k]}}}\geq t\right)\\
\leq & (d+1)\exp\left(-\frac{np_{[k]}t^{2}/2}{Cp_{[k]}r_{n}+2\zeta_{n}\sqrt{p_{[k]}}t/3}\right)\\
\leq & (d+1)\exp\left(-\frac{nt^{2}}{4Cr_{n}}\right)
\end{align*}
for all $0\leq t\leq3Cr_{n}\sqrt{p_{[k]}}/(2\zeta_{n})$. By the union
bound we have 
\begin{align*}
 & P\left(\sup_{1\leq k\leq K}\frac{\left\Vert \frac{1}{n}\sum_{i=1}^{n}\1(B_{i}=k)\bs{\xi}_{n}^{*}(\bs X_{i})-\e\left[\1(B_{i}=k)\bs{\xi}_{n}^{*}(\bs X_{i})\right]\right\Vert }{\sqrt{p_{[k]}}}\geq\delta\sqrt{\frac{r_{n}\log(K(d+1))}{n}}\right)\\
\leq & K(d+1)\exp\left(-\frac{\delta^{2}}{4C}\log(K(d+1))\right)=\exp\left(\log(K(d+1))\left\{ 1-\frac{\delta^{2}}{4C}\right\} \right)
\end{align*}
for any 
\[
0\leq\delta\leq\sqrt{\frac{n}{r_{n}\log(K(d+1))}}\frac{3Cr_{n}\min_{1\leq k\leq K}\sqrt{p_{[k]}}}{2\zeta_{n}}=\frac{3C}{2}\sqrt{\frac{nr_{n}\inf_{1\leq k\leq K}p_{[k]}}{\zeta_{n}^{2}\log(K(d+1))}}.
\]
Let $0<\epsilon<2$ be any number. Since 
\[
\frac{3C}{2}\sqrt{\frac{nr_{n}\inf_{1\leq k\leq K}p_{[k]}}{\zeta_{n}^{2}\log(K(d+1))}}\geq\frac{3C}{2}\sqrt{\frac{n\inf_{1\leq k\leq K}p_{[k]}}{\zeta_{n}^{2}\log(2Kd)}}\to\infty,
\]
there exists an integer $N_{0}$ such that 
\[
\sqrt{4C\left\{ 1-\log\left(\frac{\epsilon}{2}\right)\right\} }<\frac{3C}{2}\sqrt{\frac{nr_{n}\inf_{1\leq k\leq K}p_{[k]}}{\zeta_{n}^{2}\log(K(d+1))}}
\]
for all $n\geq N_{0}$. Let $\delta_{\epsilon}=\sqrt{4C\left\{ 1-\log\left(\frac{\epsilon}{2}\right)\right\} }$,
then we have $1-\frac{\delta_{\epsilon}^{2}}{4C}=\log\left(\frac{\epsilon}{2}\right)<0$
and
\begin{align*}
 & P\left(\sup_{1\leq k\leq K}\frac{\left\Vert \frac{1}{n}\sum_{i=1}^{n}\1(B_{i}=k)\bs{\xi}_{n}^{*}(\bs X_{i})-\e\left[\1(B_{i}=k)\bs{\xi}_{n}^{*}(\bs X_{i})\right]\right\Vert }{\sqrt{p_{[k]}}}\geq\delta_{\epsilon}\sqrt{\frac{r_{n}\log(K(d+1))}{n}}\right)\\
\leq & \exp\left(\log(K(d+1))\left\{ 1-\frac{\delta_{\epsilon}^{2}}{4C}\right\} \right)\leq\exp\left(\log(2)\log\left(\frac{\epsilon}{2}\right)\right)=\epsilon.
\end{align*}
This implies that 
\[
\sup_{1\leq k\leq K}\frac{\left\Vert \frac{1}{n}\sum_{i=1}^{n}\1(B_{i}=k)\bs{\xi}_{n}^{*}(\bs X_{i})-\e\left[\1(B_{i}=k)\bs{\xi}_{n}^{*}(\bs X_{i})\right]\right\Vert }{\sqrt{p_{[k]}}}=O_{P}\left(\sqrt{\frac{r_{n}\log(K(d+1))}{n}}\right)
\]
and thus 
\[
\left\Vert \frac{1}{n}\sum_{i=1}^{n}\1(B_{i}=k)\bs{\xi}_{n}^{*}(\bs X_{i})-\e\left[\1(B_{i}=k)\bs{\xi}_{n}^{*}(\bs X_{i})\right]\right\Vert \leq\sqrt{p_{[k]}}O_{P}\left(\sqrt{\frac{r_{n}\log(2Kd)}{n}}\right),
\]
where the random variable $O_{P}\left(\sqrt{\frac{r_{n}\log(2Kd)}{n}}\right)$
does not depend on $k$.

\textbf{(ii)} Let $\mathcal{C}_{n}$ denote the $\sigma$-algebra
generated by $(A^{(n)},B^{(n)})$ and $\e_{\mathcal{C}_{n}}\left[\cdot\right]:=\e\left[\cdot\mid\mathcal{C}_{n}\right]$,
$P_{\mathcal{C}_{n}}\left(\cdot\right):=P\left(\cdot\mid\mathcal{C}_{n}\right)$.
Then $\bs{\Xi}_{i,[k]}^{*}\bs{\Xi}_{i,[k]}^{*\trans}$, $i=1,\ldots,n$,
are independent conditional on $\mathcal{C}_{n}$. Now, we fix any
value of $(A^{(n)},B^{(n)})$. Let 
\[
\bs Z_{i}=A_{i}\1(B_{i}=k)\widetilde{\bs{\xi}}_{n}^{*}(\bs X_{i})-\e_{\mathcal{C}_{n}}\left[A_{i}\1(B_{i}=k)\widetilde{\bs{\xi}}_{n}^{*}(\bs X_{i})\right]=A_{i}\1(B_{i}=k)\widetilde{\bs{\xi}}_{n}^{*}(\bs X_{i})
\]
and 
\[
\bs S_{n}=\sum_{i=1}^{n}\bs Z_{i}.
\]
Then by Assumption \ref{assu:conditions on =00005Cxi-diverging xi}
it holds that 
\[
\left\Vert \bs Z_{i}\right\Vert \leq2\zeta_{n}
\]
for all $1\leq i\leq n$ and 
\begin{align*}
V & :=\max\left\{ \left\Vert \e_{\mathcal{C}_{n}}\left[\bs S_{n}\bs S_{n}^{\trans}\right]\right\Vert ,\left\Vert \e_{\mathcal{C}_{n}}\left[\bs S_{n}^{\trans}\bs S_{n}\right]\right\Vert \right\} =\max\left\{ \left\Vert \sum_{i=1}^{n}\e_{\mathcal{C}_{n}}\left[\bs Z_{i}\bs Z_{i}^{\trans}\right]\right\Vert ,\left\Vert \sum_{i=1}^{n}\e_{\mathcal{C}_{n}}\left[\bs Z_{i}^{\trans}\bs Z_{i}\right]\right\Vert \right\} \\
 & \leq\sum_{i=1}^{n}\e_{\mathcal{C}_{n}}\left[\left\Vert \bs Z_{i}\right\Vert ^{2}\right]=\sum_{i=1}^{n}\e_{\mathcal{C}_{n}}\left[A_{i}^{2}\1(B_{i}=k)\left\Vert \widetilde{\bs{\xi}}_{n}^{*}(\bs X_{i})\right\Vert ^{2}\right]\leq\sum_{i=1}^{n}\1(B_{i}=k)\e\left[\left\Vert \widetilde{\bs{\xi}}_{n}^{*}(\bs X_{i})\right\Vert ^{2}\mid B_{i}=k\right]\\
 & \leq n_{[k]}\e\left[\left\Vert \bs{\xi}_{n}^{*}(\bs X_{i})\right\Vert ^{2}\mid B_{i}=k\right]=n_{[k]}\tr\left\{ \e\left[\bs{\xi}_{n}^{*}(\bs X_{i})\bs{\xi}_{n}^{*}(\bs X_{i})^{\trans}\mid B_{i}=k\right]\right\} \leq Cr_{n}n_{[k]}
\end{align*}
a.s. It follows from the Bernstain's inequality for matrices (see
\citet[Theorem 1.6.2]{tropp2015Introduction}) that
\[
P_{\mathcal{C}_{n}}\left(\left\Vert \bs S_{n}\right\Vert \geq t\right)\leq2d\exp\left(-\frac{t^{2}/2}{V+2\zeta_{n}t/3}\right)\leq2d\exp\left(-\frac{t^{2}/2}{Cr_{n}n_{[k]}+2\zeta_{n}t/3}\right)
\]
a.s. for all $t\geq0$. Thus, 
\begin{align*}
 & P_{\mathcal{C}_{n}}\left(\frac{\left\Vert \frac{1}{n}\sum_{i=1}^{n}A_{i}\1(B_{i}=k)\widetilde{\bs{\xi}}_{n}^{*}(\bs X_{i})\right\Vert }{\sqrt{n_{[k]}/n}}\geq t\right)\\
\leq & 2d\exp\left(-\frac{nt^{2}/2}{Cr_{n}+2(\sqrt{n_{[k]}/n})^{-1}\zeta_{n}t/3}\right)\leq2d\exp\left(-\frac{nt^{2}}{4Cr_{n}}\right)
\end{align*}
a.s. for all 
\[
0\leq t\leq\frac{3}{2}C\frac{r_{n}}{\zeta_{n}}\sqrt{\frac{n_{[k]}}{n}}.
\]
By the union bound we have 
\begin{align*}
 & P_{\mathcal{C}_{n}}\left(\sup_{1\leq k\leq K}\frac{\left\Vert \frac{1}{n}\sum_{i=1}^{n}A_{i}\1(B_{i}=k)\widetilde{\bs{\xi}}_{n}^{*}(\bs X_{i})\right\Vert }{\sqrt{n_{[k]}/n}}\geq\delta\sqrt{\frac{r_{n}\log(2Kd)}{n}}\right)\\
\leq & 2Kd\exp\left(-\frac{\delta^{2}}{4C}\log(2Kd)\right)\leq\exp\left(\log(2Kd)\left\{ 1-\frac{\delta^{2}}{4C}\right\} \right)
\end{align*}
a.s. for all 
\[
0\leq\delta\leq\frac{3}{2}C\sqrt{\frac{nr_{n}\min_{1\leq k\leq K}(n_{[k]}/n)}{\zeta_{n}^{2}\log(2Kd)}}.
\]
Let $0<\epsilon<2$ be any number. By Lemma \ref{lem:LLN} we have
$\frac{n_{[k]}}{np_{[k]}}=1+o_{P}(1)$, where the random variable
$o_{P}(1)$ does not depend on $k$. As a result, the event $E_{n}:=\left\{ 1/4\leq\inf_{1\leq k\leq K}\frac{n_{[k]}}{np_{[k]}}\right\} $
happens with probability $1-\epsilon$ for all $n\geq N_{1}$, where
$N_{1}$ is a finite integer. Then on the event $E_{n}$ we have 
\begin{align*}
 & \frac{3}{2}C\sqrt{\frac{nr_{n}\min_{1\leq k\leq K}(n_{[k]}/n)}{\zeta_{n}^{2}\log(2Kd)}}\geq\frac{3}{2}C\sqrt{\frac{n\min_{1\leq k\leq K}\left\{ (n_{[k]}/np_{[k]})p_{[k]}\right\} }{\zeta_{n}^{2}\log(2Kd)}}\\
\geq & \frac{3}{2}C\sqrt{\frac{n\inf_{1\leq k\leq K}p_{[k]}}{\zeta_{n}^{2}\log(2Kd)}}\sqrt{\min_{1\leq k\leq K}(n_{[k]}/np_{[k]})}\geq\frac{3}{4}C\sqrt{\frac{n\inf_{1\leq k\leq K}p_{[k]}}{\zeta_{n}^{2}\log(2Kd)}}.
\end{align*}
By Assumption \ref{assu:conditions on =00005Cxi-diverging xi}, we
have $\sqrt{\frac{n\inf_{1\leq k\leq K}p_{[k]}}{\zeta_{n}^{2}\log(2Kd)}}\to\infty$,
thus there exists an integer $N_{0}$ such that on the event $E_{n}$
\[
\sqrt{4C\left\{ 1-\log\left(\frac{\epsilon}{2}\right)\right\} }\leq\frac{3}{4}C\sqrt{\frac{n\inf_{1\leq k\leq K}p_{[k]}}{\zeta_{n}^{2}\log(2Kd)}}\leq\frac{3}{2}C\sqrt{\frac{nr_{n}\min_{1\leq k\leq K}(n_{[k]}/n)}{\zeta_{n}^{2}\log(2Kd)}}
\]
for all $n\geq N_{0}$. Let $\delta_{\epsilon}=\sqrt{4C\left\{ 1-\log\left(\frac{\epsilon}{2}\right)\right\} }$,
then we have $1-\frac{\delta_{\epsilon}^{2}}{4C}=\log\left(\frac{\epsilon}{2}\right)<0$
and
\begin{align*}
 & P_{\mathcal{C}_{n}}\left(\sup_{1\leq k\leq K}\frac{\left\Vert \frac{1}{n}\sum_{i=1}^{n}A_{i}\1(B_{i}=k)\widetilde{\bs{\xi}}_{n}^{*}(\bs X_{i})\right\Vert }{\sqrt{n_{[k]}/n}}\geq\delta_{\epsilon}\sqrt{\frac{r_{n}\log(2Kd)}{n}}\right)\\
\leq & \exp\left(\log(2Kd)\left\{ 1-\frac{\delta_{\epsilon}^{2}}{4C}\right\} \right)\leq\exp\left(\log(2)\log\left(\frac{\epsilon}{2}\right)\right)=\epsilon
\end{align*}
for any fixed value of $(A^{(n)},B^{(n)})$ such that $E_{n}$ happens\footnote{This is feasible since $E_{n}$ is measurable with respect to $\mathcal{C}_{n}$.}.
Then, for all $n\geq\max\{N_{0},N_{1}\}$, we have 
\begin{align*}
 & P\left(\sup_{1\leq k\leq K}\frac{\left\Vert \frac{1}{n}\sum_{i=1}^{n}A_{i}\1(B_{i}=k)\widetilde{\bs{\xi}}_{n}^{*}(\bs X_{i})\right\Vert }{\sqrt{n_{[k]}/n}}\geq\delta_{\epsilon}\sqrt{\frac{r_{n}\log(2Kd)}{n}}\right)\\
= & \e\left[P_{\mathcal{C}_{n}}\left(\sup_{1\leq k\leq K}\frac{\left\Vert \frac{1}{n}\sum_{i=1}^{n}\bs{\Xi}_{i,[k]}^{*}\bs{\Xi}_{i,[k]}^{*\trans}-\mathbf{\Sigma}_{[k]}^{\mathcal{C}_{n}}\right\Vert }{\sqrt{n_{[k]}/n}}\geq\delta_{\epsilon}\sqrt{\frac{r_{n}\log(2Kd)}{n}}\right)\right]\\
\leq & \e\left[\epsilon\1(E_{n}\text{ happens})\right]+\e\left[\1(E_{n}\text{does not happen})\right]\leq\epsilon+\epsilon=2\epsilon.
\end{align*}
This implies that 
\[
\sup_{1\leq k\leq K}\frac{\left\Vert \frac{1}{n}\sum_{i=1}^{n}A_{i}\1(B_{i}=k)\widetilde{\bs{\xi}}_{n}^{*}(\bs X_{i})\right\Vert }{\sqrt{n_{[k]}/n}}=O_{P}\left(\sqrt{\frac{r_{n}\log(2Kd)}{n}}\right),
\]
where the random variable $O_{P}\left(\sqrt{\frac{r_{n}\log(2Kd)}{n}}\right)$
does not depend on $k$. By Lemma \ref{lem:LLN}, we have 
\[
\sqrt{n_{[k]}/n}=\sqrt{p_{[k]}}O_{P}(1),
\]
where the random variable $O_{P}(1)$ does not depend on $k$. As
a result, 
\[
\left\Vert \frac{1}{n}\sum_{i=1}^{n}A_{i}\1(B_{i}=k)\widetilde{\bs{\xi}}_{n}^{*}(\bs X_{i})\right\Vert \leq\sqrt{n_{[k]}/n}O_{P}\left(\sqrt{\frac{r_{n}\log(2Kd)}{n}}\right)=\sqrt{p_{[k]}}O_{P}\left(\sqrt{\frac{r_{n}\log(2Kd)}{n}}\right),
\]
where the random variable $O_{P}\left(\sqrt{\frac{r_{n}\log(2Kd)}{n}}\right)$
does not depend on $k$.

\textbf{(iii)} Let 
\[
\bs Z_{i}=\1(B_{i}=k)\bs{\xi}_{n}^{*}(\bs X_{i})\bs{\xi}_{n}^{*}(\bs X_{i})^{\trans}-\e\left[\1(B_{i}=k)\bs{\xi}_{n}^{*}(\bs X_{i})\bs{\xi}_{n}^{*}(\bs X_{i})^{\trans}\right]
\]
and 
\[
\bs S_{n}=\sum_{i=1}^{n}\bs Z_{i}.
\]
Then by Assumption \ref{assu:conditions on =00005Cxi-diverging xi}
it holds that 
\[
\left\Vert \bs Z_{i}\right\Vert \leq2\zeta_{n}^{2}
\]
for all $1\leq i\leq n$ and 
\begin{align*}
V & :=\max\left\{ \left\Vert \e\left[\bs S_{n}\bs S_{n}^{\trans}\right]\right\Vert ,\left\Vert \e\left[\bs S_{n}^{\trans}\bs S_{n}\right]\right\Vert \right\} \\
 & =\max\left\{ \left\Vert \sum_{i=1}^{n}\e\left[\bs Z_{i}\bs Z_{i}^{\trans}\right]\right\Vert ,\left\Vert \sum_{i=1}^{n}\e\left[\bs Z_{i}^{\trans}\bs Z_{i}\right]\right\Vert \right\} =\left\Vert \sum_{i=1}^{n}\e\left[\bs Z_{i}^{2}\right]\right\Vert \\
 & =\left\Vert \sum_{i=1}^{n}\left\{ \e\left[\1(B_{i}=k)\left\{ \bs{\xi}_{n}^{*}(\bs X_{i})\bs{\xi}_{n}^{*}(\bs X_{i})^{\trans}\right\} ^{2}\right]-\left(\e\left[\1(B_{i}=k)\bs{\xi}_{n}^{*}(\bs X_{i})\bs{\xi}_{n}^{*}(\bs X_{i})^{\trans}\right]\right)^{2}\right\} \right\Vert \\
 & \leq n\left\Vert \e\left[\1(B_{i}=k)\left\{ \bs{\xi}_{n}^{*}(\bs X_{i})\bs{\xi}_{n}^{*}(\bs X_{i})^{\trans}\right\} ^{2}\right]\right\Vert +n\left\Vert \e\left[\1(B_{i}=k)\bs{\xi}_{n}^{*}(\bs X_{i})\bs{\xi}_{n}^{*}(\bs X_{i})^{\trans}\right]\right\Vert ^{2}\\
 & \leq2n\left\Vert \e\left[\1(B_{i}=k)\left\{ \bs{\xi}_{n}^{*}(\bs X_{i})\bs{\xi}_{n}^{*}(\bs X_{i})^{\trans}\right\} ^{2}\right]\right\Vert \\
 & =2np_{[k]}\left\Vert \e\left[\bs{\xi}_{n}^{*}(\bs X_{i})\bs{\xi}_{n}^{*}(\bs X_{i})^{\trans}\bs{\xi}_{n}^{*}(\bs X_{i})\bs{\xi}_{n}^{*}(\bs X_{i})^{\trans}\mid B_{i}=k\right]\right\Vert \\
 & \leq2np_{[k]}\zeta_{n}^{2}\left\Vert \e\left[\bs{\xi}_{n}^{*}(\bs X_{i})\bs{\xi}_{n}^{*}(\bs X_{i})^{\trans}\mid B_{i}=k\right]\right\Vert \leq2Cp_{[k]}n\zeta_{n}^{2}.
\end{align*}
It follows from the Bernstain's inequality for matrices (see \citet[Theorem 1.6.2]{tropp2015Introduction})
that
\[
P\left(\left\Vert \bs S_{n}\right\Vert \geq t\right)\leq2d\exp\left(-\frac{t^{2}/2}{V+2\zeta_{n}^{2}t/3}\right)\leq2d\exp\left(-\frac{t^{2}/2}{2Cp_{[k]}n\zeta_{n}^{2}+2\zeta_{n}^{2}t/3}\right)
\]
for all $t\geq0$. Thus, 
\begin{align*}
 & P\left(\frac{\left\Vert \frac{1}{n}\sum_{i=1}^{n}\1(B_{i}=k)\bs{\xi}_{n}^{*}(\bs X_{i})\bs{\xi}_{n}^{*}(\bs X_{i})^{\trans}-\e\left[\1(B_{i}=k)\bs{\xi}_{n}^{*}(\bs X_{i})\bs{\xi}_{n}^{*}(\bs X_{i})^{\trans}\right]\right\Vert }{\sqrt{p_{[k]}}}\geq t\right)\\
\leq & 2d\exp\left(-\frac{nt^{2}/2}{2C\zeta_{n}^{2}+2\sqrt{p_{[k]}^{-1}}\zeta_{n}^{2}t/3}\right)\leq2d\exp\left(-\frac{nt^{2}}{8C\zeta_{n}^{2}}\right).
\end{align*}
for all $0\leq t\leq3C\sqrt{p_{[k]}}.$ By the union bound we have
{\small
\begin{align*}
 & P\left(\sup_{1\leq k\leq K}\frac{\left\Vert \frac{1}{n}\sum_{i=1}^{n}\1(B_{i}=k)\bs{\xi}_{n}^{*}(\bs X_{i})\bs{\xi}_{n}^{*}(\bs X_{i})^{\trans}-\e\left[\1(B_{i}=k)\bs{\xi}_{n}^{*}(\bs X_{i})\bs{\xi}_{n}^{*}(\bs X_{i})^{\trans}\right]\right\Vert }{\sqrt{p_{[k]}}}\geq\delta\sqrt{\frac{\zeta_{n}^{2}\log(2Kd)}{n}}\right)\\
\leq & 2Kd\exp\left(-\frac{\delta^{2}}{8C}\log(2Kd)\right)\leq\exp\left(\log(2Kd)\left\{ 1-\frac{\delta^{2}}{8C}\right\} \right)
\end{align*}
}for all 
\[
0\leq\delta\leq3C\sqrt{\frac{n\inf_{1\leq k\leq K}p_{[k]}}{\zeta_{n}^{2}\log(2Kd)}}.
\]
Let $0<\epsilon<2$ be any number. Since 
\[
3C\sqrt{\frac{n\inf_{1\leq k\leq K}p_{[k]}}{\zeta_{n}^{2}\log(2Kd)}}\to\infty,
\]
there exists an integer $N_{0}$ such that 
\[
\sqrt{8C\left\{ 1-\log\left(\frac{\epsilon}{2}\right)\right\} }<3C\sqrt{\frac{n\inf_{1\leq k\leq K}p_{[k]}}{\zeta_{n}^{2}\log(2Kd)}}
\]
for all $n\geq N_{0}$. Let $\delta_{\epsilon}=\sqrt{8C\left\{ 1-\log\left(\frac{\epsilon}{2}\right)\right\} }$,
then we have $1-\frac{\delta_{\epsilon}^{2}}{8C}=\log\left(\frac{\epsilon}{2}\right)<0$
and{\small
\begin{align*}
 & P\left(\sup_{1\leq k\leq K}\frac{\left\Vert \frac{1}{n}\sum_{i=1}^{n}\1(B_{i}=k)\bs{\xi}_{n}^{*}(\bs X_{i})\bs{\xi}_{n}^{*}(\bs X_{i})^{\trans}-\e\left[\1(B_{i}=k)\bs{\xi}_{n}^{*}(\bs X_{i})\bs{\xi}_{n}^{*}(\bs X_{i})^{\trans}\right]\right\Vert }{\sqrt{p_{[k]}}}\geq\delta\sqrt{\frac{\zeta_{n}^{2}\log(2Kd)}{n}}\right)\\
\leq & \exp\left(\log(2Kd)\left\{ 1-\frac{\delta_{\epsilon}^{2}}{8C}\right\} \right)\leq\exp\left(\log(2)\log\left(\frac{\epsilon}{2}\right)\right)=\epsilon.
\end{align*}
}This implies that 
\begin{align*}
 & \sup_{1\leq k\leq K}\frac{\left\Vert \frac{1}{n}\sum_{i=1}^{n}\1(B_{i}=k)\bs{\xi}_{n}^{*}(\bs X_{i})\bs{\xi}_{n}^{*}(\bs X_{i})^{\trans}-\e\left[\1(B_{i}=k)\bs{\xi}_{n}^{*}(\bs X_{i})\bs{\xi}_{n}^{*}(\bs X_{i})^{\trans}\right]\right\Vert }{\sqrt{p_{[k]}}}\\
= & \sqrt{\frac{\zeta_{n}^{2}\log(2Kd)}{n}}O_{P}(1)
\end{align*}
and thus 
\begin{align*}
 & \left\Vert \frac{1}{n}\sum_{i=1}^{n}\1(B_{i}=k)\bs{\xi}_{n}^{*}(\bs X_{i})\bs{\xi}_{n}^{*}(\bs X_{i})^{\trans}-\e\left[\1(B_{i}=k)\bs{\xi}_{n}^{*}(\bs X_{i})\bs{\xi}_{n}^{*}(\bs X_{i})^{\trans}\right]\right\Vert \\
\leq & p_{[k]}\sqrt{\frac{\zeta_{n}^{2}\log(2Kd)}{n\inf_{1\leq k\leq K}p_{[k]}}}O_{P}(1)=p_{[k]}o_{P}(1),
\end{align*}
where the random variable $o_{P}(1)$ does not depend on $k$ and
the inequality follows from 
\[
\zeta_{n}^{2}\log(2Kd)/(n\inf_{1\leq k\leq K}p_{[k]})\to0
\]
as $n\to\infty$.

\textbf{(iv)} Let $\mathcal{C}_{n}$ denote the $\sigma$-algebra
generated by $(A^{(n)},B^{(n)})$ and $\e_{\mathcal{C}_{n}}\left[\cdot\right]:=\e\left[\cdot\mid\mathcal{C}_{n}\right]$,
$P_{\mathcal{C}_{n}}\left(\cdot\right):=P\left(\cdot\mid\mathcal{C}_{n}\right)$.
Then $\bs{\Xi}_{i,[k]}^{*}\bs{\Xi}_{i,[k]}^{*\trans}$, $i=1,\ldots,n$,
are independent conditional on $\mathcal{C}_{n}$. Now, we fix any
value of $(A^{(n)},B^{(n)})$. Let 
\[
\bs Z_{i}=\bs{\Xi}_{i,[k]}^{*}\bs{\Xi}_{i,[k]}^{*\trans}-\e_{\mathcal{C}_{n}}\left[\bs{\Xi}_{i,[k]}^{*}\bs{\Xi}_{i,[k]}^{*\trans}\right]
\]
and 
\[
\bs S_{n}=\sum_{i=1}^{n}\bs Z_{i}.
\]
Then by Assumption \ref{assu:conditions on =00005Cxi-diverging xi}
it holds that 
\[
\left\Vert \bs Z_{i}\right\Vert \leq4\zeta_{n}^{2}
\]
for all $1\leq i\leq n$ and 
\begin{align*}
V= & \max\left\{ \left\Vert \e_{\mathcal{C}_{n}}\left[\bs S_{n}\bs S_{n}^{\trans}\right]\right\Vert ,\left\Vert \e_{\mathcal{C}_{n}}\left[\bs S_{n}^{\trans}\bs S_{n}\right]\right\Vert \right\} =\left\Vert \sum_{i=1}^{n}\e_{\mathcal{C}_{n}}\left[\bs Z_{i}^{2}\right]\right\Vert \\
= & \biggl\Vert\sum_{i=1}^{n}\e_{\mathcal{C}_{n}}\left[\1(B_{i}=k)(A_{i}-\pi_{n[k]})^{2}\left\{ \widetilde{\bs{\xi}}_{n}^{*}(\bs X_{i})\widetilde{\bs{\xi}}_{n}^{*}(\bs X_{i})^{\trans}\right\} ^{2}\right]\\
 & -\sum_{i=1}^{n}\left\{ \e_{\mathcal{C}_{n}}\left[(A_{i}-\pi_{n[k]})\1(B_{i}=k)\widetilde{\bs{\xi}}_{n}^{*}(\bs X_{i})\widetilde{\bs{\xi}}_{n}^{*}(\bs X_{i})^{\trans}\right]\right\} ^{2}\biggr\Vert\\
\leq & \sum_{i=1}^{n}\left\Vert \e_{\mathcal{C}_{n}}\left[\1(B_{i}=k)(A_{i}-\pi_{n[k]})^{2}\left\{ \widetilde{\bs{\xi}}_{n}^{*}(\bs X_{i})\widetilde{\bs{\xi}}_{n}^{*}(\bs X_{i})^{\trans}\right\} ^{2}\right]\right\Vert \\
 & +\sum_{i=1}^{n}\left\Vert \e_{\mathcal{C}_{n}}\left[(A_{i}-\pi_{n[k]})\1(B_{i}=k)\widetilde{\bs{\xi}}_{n}^{*}(\bs X_{i})\widetilde{\bs{\xi}}_{n}^{*}(\bs X_{i})^{\trans}\right]\right\Vert ^{2}\\
\leq & 2\sum_{i=1}^{n}\left\Vert \e_{\mathcal{C}_{n}}\left[\1(B_{i}=k)(A_{i}-\pi_{n[k]})^{2}\left\{ \widetilde{\bs{\xi}}_{n}^{*}(\bs X_{i})\widetilde{\bs{\xi}}_{n}^{*}(\bs X_{i})^{\trans}\right\} ^{2}\right]\right\Vert \\
\leq & 2\sum_{i=1}^{n}\1(B_{i}=k)\left\Vert \e\left[\left\{ \widetilde{\bs{\xi}}_{n}^{*}(\bs X_{i})\widetilde{\bs{\xi}}_{n}^{*}(\bs X_{i})^{\trans}\right\} ^{2}\mid B_{i}=k\right]\right\Vert \\
\leq & 8\zeta_{n}^{2}\sum_{i=1}^{n}\1(B_{i}=k)\left\Vert \e\left[\widetilde{\bs{\xi}}_{n}^{*}(\bs X_{i})\widetilde{\bs{\xi}}_{n}^{*}(\bs X_{i})^{\trans}\mid B_{i}=k\right]\right\Vert \\
\leq & 8\zeta_{n}^{2}\sum_{i=1}^{n}\1(B_{i}=k)\left\Vert \e\left[\bs{\xi}_{n}^{*}(\bs X_{i})\bs{\xi}_{n}^{*}(\bs X_{i})^{\trans}\mid B_{i}=k\right]\right\Vert \\
\leq & 8C\zeta_{n}^{2}\sum_{i=1}^{n}\1(B_{i}=k)=8C\zeta_{n}^{2}n_{[k]}
\end{align*}
a.s. It follows from the Bernstain's inequality for matrices (see
\citet[Theorem 1.6.2]{tropp2015Introduction}) that
\[
P_{\mathcal{C}_{n}}\left(\left\Vert \bs S_{n}\right\Vert \geq t\right)\leq2d\exp\left(-\frac{t^{2}/2}{V+4\zeta_{n}^{2}t/3}\right)\leq2d\exp\left(-\frac{t^{2}/2}{8C\zeta_{n}^{2}n_{[k]}+4\zeta_{n}^{2}t/3}\right)
\]
a.s. for all $t\geq0$. Thus, 
\begin{align*}
 & P_{\mathcal{C}_{n}}\left(\frac{\left\Vert \frac{1}{n}\sum_{i=1}^{n}\bs{\Xi}_{i,[k]}^{*}\bs{\Xi}_{i,[k]}^{*\trans}-\mathbf{\Sigma}_{[k]}^{\mathcal{C}_{n}}\right\Vert }{\sqrt{n_{[k]}/n}}\geq t\right)=P_{\mathcal{C}_{n}}\left(\frac{\left\Vert \frac{1}{n}\sum_{i=1}^{n}\left\{ \bs{\Xi}_{i,[k]}^{*}\bs{\Xi}_{i,[k]}^{*\trans}-\e_{\mathcal{C}_{n}}\left[\bs{\Xi}_{i,[k]}^{*}\bs{\Xi}_{i,[k]}^{*\trans}\right]\right\} \right\Vert }{\sqrt{n_{[k]}/n}}\geq t\right)\\
\leq & 2d\exp\left(-\frac{nt^{2}/2}{8C\zeta_{n}^{2}+4(\sqrt{n_{[k]}/n})^{-1}\zeta_{n}^{2}t/3}\right)\leq2d\exp\left(-\frac{nt^{2}}{32C\zeta_{n}^{2}}\right)
\end{align*}
a.s. for all 
\[
0\leq t\leq6C\sqrt{n_{[k]}/n}.
\]
By the union bound we have 
\begin{align*}
 & P_{\mathcal{C}_{n}}\left(\sup_{1\leq k\leq K}\frac{\left\Vert \frac{1}{n}\sum_{i=1}^{n}\bs{\Xi}_{i,[k]}^{*}\bs{\Xi}_{i,[k]}^{*\trans}-\mathbf{\Sigma}_{[k]}^{\mathcal{C}_{n}}\right\Vert }{\sqrt{n_{[k]}/n}}\geq\delta\sqrt{\frac{\zeta_{n}^{2}\log(2Kd)}{n}}\right)\\
\leq & 2Kd\exp\left(-\frac{\delta^{2}}{32C}\log(2Kd)\right)\leq\exp\left(\log(2Kd)\left\{ 1-\frac{\delta^{2}}{32C}\right\} \right)
\end{align*}
a.s. for all 
\[
0\leq\delta\leq6C\sqrt{\frac{n\min_{1\leq k\leq K}(n_{[k]}/n)}{\zeta_{n}^{2}\log(2Kd)}}.
\]
Let $0<\epsilon<2$ be any number. By Lemma \ref{lem:LLN} we have
$\frac{n_{[k]}}{np_{[k]}}=1+o_{P}(1)$, where the random variable
$o_{P}(1)$ does not depend on $k$. As a result, the event $E_{n}:=\left\{ 1/4\leq\inf_{1\leq k\leq K}\frac{n_{[k]}}{np_{[k]}}\right\} $
happens with probability $1-\epsilon$ for all $n\geq N_{1}$, where
$N_{1}$ is a finite integer. Then on the event $E_{n}$ we have 
\begin{align*}
 & 6C\sqrt{\frac{n\min_{1\leq k\leq K}(n_{[k]}/n)}{\zeta_{n}^{2}\log(2Kd)}}\\
= & 6C\sqrt{\frac{n\min_{1\leq k\leq K}\left\{ (n_{[k]}/np_{[k]})p_{[k]}\right\} }{\zeta_{n}^{2}\log(2Kd)}}\geq6C\sqrt{\frac{n\inf_{1\leq k\leq K}p_{[k]}}{\zeta_{n}^{2}\log(2Kd)}}\sqrt{\min_{1\leq k\leq K}(n_{[k]}/np_{[k]})}\\
\geq & 3C\sqrt{\frac{n\inf_{1\leq k\leq K}p_{[k]}}{\zeta_{n}^{2}\log(2Kd)}}.
\end{align*}
By Assumption \ref{assu:conditions on =00005Cxi-diverging xi}, we
have $\sqrt{\frac{n\inf_{1\leq k\leq K}p_{[k]}}{\zeta_{n}^{2}\log(2Kd)}}\to\infty$,
thus there exists an integer $N_{0}$ such that on the event $E_{n}$
\[
\sqrt{32C\left\{ 1-\log\left(\frac{\epsilon}{2}\right)\right\} }\leq3C\sqrt{\frac{n\inf_{1\leq k\leq K}p_{[k]}}{\zeta_{n}^{2}\log(2Kd)}}\leq6C\sqrt{\frac{n\min_{1\leq k\leq K}(n_{[k]}/n)}{\zeta_{n}^{2}\log(2Kd)}}
\]
for all $n\geq N_{0}$. Let $\delta_{\epsilon}=\sqrt{32C\left\{ 1-\log\left(\frac{\epsilon}{2}\right)\right\} }$,
then we have $1-\frac{\delta_{\epsilon}^{2}}{32C}=\log\left(\frac{\epsilon}{2}\right)<0$
and
\begin{align*}
 & P_{\mathcal{C}_{n}}\left(\sup_{1\leq k\leq K}\frac{\left\Vert \frac{1}{n}\sum_{i=1}^{n}\bs{\Xi}_{i,[k]}^{*}\bs{\Xi}_{i,[k]}^{*\trans}-\mathbf{\Sigma}_{[k]}^{\mathcal{C}_{n}}\right\Vert }{\sqrt{n_{[k]}/n}}\geq\delta_{\epsilon}\sqrt{\frac{\zeta_{n}^{2}\log(2Kd)}{n}}\right)\\
\leq & \exp\left(\log(2Kd)\left\{ 1-\frac{\delta_{\epsilon}^{2}}{32C}\right\} \right)\leq\exp\left(\log(2)\log\left(\frac{\epsilon}{2}\right)\right)=\epsilon
\end{align*}
for any fixed value of $(A^{(n)},B^{(n)})$ such that $E_{n}$ happens\footnote{This is feasible since $E_{n}$ is measurable with respect to $\mathcal{C}_{n}$.}.
Then, for all $n\geq\max\{N_{0},N_{1}\}$, we have 
\begin{align*}
 & P\left(\sup_{1\leq k\leq K}\frac{\left\Vert \frac{1}{n}\sum_{i=1}^{n}\bs{\Xi}_{i,[k]}^{*}\bs{\Xi}_{i,[k]}^{*\trans}-\mathbf{\Sigma}_{[k]}^{\mathcal{C}_{n}}\right\Vert }{\sqrt{n_{[k]}/n}}\geq\delta_{\epsilon}\sqrt{\frac{\zeta_{n}^{2}\log(2Kd)}{n}}\right)\\
= & \e\left[P_{\mathcal{C}_{n}}\left(\sup_{1\leq k\leq K}\frac{\left\Vert \frac{1}{n}\sum_{i=1}^{n}\bs{\Xi}_{i,[k]}^{*}\bs{\Xi}_{i,[k]}^{*\trans}-\mathbf{\Sigma}_{[k]}^{\mathcal{C}_{n}}\right\Vert }{\sqrt{n_{[k]}/n}}\geq\delta_{\epsilon}\sqrt{\frac{\zeta_{n}^{2}\log(2Kd)}{n}}\right)\right]\\
\leq & \e\left[\epsilon\1(E_{n}\text{ happens})\right]+\e\left[\1(E_{n}\text{ does not happen})\right]\leq\epsilon+\epsilon=2\epsilon.
\end{align*}
This implies that 
\[
\sup_{1\leq k\leq K}\frac{\left\Vert \frac{1}{n}\sum_{i=1}^{n}\bs{\Xi}_{i,[k]}^{*}\bs{\Xi}_{i,[k]}^{*\trans}-\mathbf{\Sigma}_{[k]}^{\mathcal{C}_{n}}\right\Vert }{\sqrt{n_{[k]}/n}}=O_{P}\left(\sqrt{\frac{\zeta_{n}^{2}\log(2Kd)}{n}}\right),
\]
where the random variable $O_{P}\left(\sqrt{\frac{\zeta_{n}^{2}\log(2Kd)}{n}}\right)$
does not depend on $k$. By Lemma \ref{lem:LLN}, we have 
\[
\sqrt{n_{[k]}/n}=\sqrt{p_{[k]}}O_{P}(1),
\]
where the random variable $O_{P}(1)$ does not depend on $k$. As
a result, 
\begin{align*}
 & \left\Vert \frac{1}{n}\sum_{i=1}^{n}\bs{\Xi}_{i,[k]}^{*}\bs{\Xi}_{i,[k]}^{*\trans}-\mathbf{\Sigma}_{[k]}^{\mathcal{C}_{n}}\right\Vert \leq\sqrt{n_{[k]}/n}O_{P}\left(\sqrt{\frac{\zeta_{n}^{2}\log(2Kd)}{n}}\right)\\
\leq & p_{[k]}\sqrt{\frac{\zeta_{n}^{2}\log(2Kd)}{n\inf_{1\leq k\leq K}p_{[k]}}}O_{P}(1)=p_{[k]}o_{P}(1),
\end{align*}
where the random variable $o_{P}(1)$ does not depend on $k$ and
the last step follows from 
\[
\zeta_{n}^{2}\log(2Kd)/(n\inf_{1\leq k\leq K}p_{[k]})\to0
\]
as $n\to\infty$. The proof is completed.
\end{proof}
\begin{lem}
\label{lem:LLN for triangular array}Let $(\Omega,\mathcal{F},P)$
be a probability space and $\mathcal{C}_{n}$ be a sequence of sub-$\sigma$-algebras
of $\mathcal{F}$. Let $\left\{ \bs X_{n,i}\in\R^{d}:1\leq i\leq n,n\geq1\right\} $
be a triangular array such that for every $n$, $\bs X_{n,1},\ldots,\bs X_{n,n}$
are mutually independent conditional on $\mathcal{C}_{n}$. Let $V_{n}:=\sup_{n\geq1,1\leq i\leq n}\e\left[\left\Vert \bs X_{n,i}\right\Vert ^{q}\mid\mathcal{C}_{n}\right]$
for some $1<q\leq2$. If $\e\left[\bs X_{n,i}\mid\mathcal{C}_{n}\right]=0$
for all $1\leq i\leq n$ and $n\geq1$ , then it holds that 
\[
\left\Vert \frac{1}{n}\sum_{i=1}^{n}\bs X_{n,i}\right\Vert =O_{P}(V_{n}^{1/q}n^{1/q-1}).
\]
In particular, if $V_{n}=O_{P}(1)$, then 
\[
\left\Vert \frac{1}{n}\sum_{i=1}^{n}\bs X_{n,i}\right\Vert =O_{P}(n^{1/q-1})=o_{P}(1).
\]
\end{lem}
\begin{proof}
For any random vector $\bs X$, let $\left\Vert \bs X\right\Vert _{L^{q},\mathcal{C}_{n}}$
denote its $L^{q}$-norm conditional on $\mathcal{C}_{n}$, i.e.,
$\left\Vert \bs X\right\Vert _{L^{q},\mathcal{C}_{n}}:=\left(\e\left[\left\Vert \bs X\right\Vert ^{q}\mid\mathcal{C}_{n}\right]\right)^{1/q}$,
where $\left\Vert \cdot\right\Vert $ stands for the Euclidean norm.
By Theorem 3.1 in \citet{burkholder1988sharp}, we have 
\[
\left\Vert \sum_{i=1}^{n}\bs X_{n,i}\right\Vert _{L^{q},\mathcal{C}_{n}}\leq\frac{1}{q-1}\left\Vert \left(\sum_{i=1}^{n}\left\Vert \bs X_{n,i}\right\Vert ^{2}\right)^{1/2}\right\Vert _{L^{q},\mathcal{C}_{n}}.
\]
Notice that for any non-negative real numbers $a_{1},\ldots,a_{n}$
and $r\geq q>1$, we have $\left(\sum_{i=1}^{n}a_{i}^{q}\right)^{1/q}\geq\left(\sum_{i=1}^{n}a_{i}^{r}\right)^{1/r}$
(Proposition 9.1.5 in \citet{dennis2009matrix}, p.599). Thus, from
the above two inequalities, we have 
\begin{align*}
\left\Vert \frac{1}{n}\sum_{i=1}^{n}\bs X_{n,i}\right\Vert _{L^{q},\mathcal{C}_{n}} & \leq\frac{1}{n}\frac{1}{q-1}\left\Vert \left(\sum_{i=1}^{n}\left\Vert \bs X_{n,i}\right\Vert ^{2}\right)^{1/2}\right\Vert _{L^{q},\mathcal{C}_{n}}\leq\frac{1}{n}\frac{1}{q-1}\left\Vert \left(\sum_{i=1}^{n}\left\Vert \bs X_{n,i}\right\Vert ^{q}\right)^{1/q}\right\Vert _{L^{q},\mathcal{C}_{n}}\\
 & =\frac{1}{n}\frac{1}{q-1}\left(\e\left[\sum_{i=1}^{n}\left\Vert \bs X_{n,i}\right\Vert ^{q}\mid\mathcal{C}_{n}\right]\right)^{1/q}=\frac{1}{n}\frac{1}{q-1}\left(\sum_{i=1}^{n}\left\Vert \bs X_{n,i}\right\Vert _{L^{q},\mathcal{C}_{n}}^{q}\right)^{1/q}\\
 & \leq\frac{1}{n}\frac{1}{q-1}n^{1/q}\sup_{n\geq1,1\leq i\leq n}\left(\e\left[\left\Vert \bs X_{n,i}\right\Vert ^{q}\mid\mathcal{C}_{n}\right]\right)^{1/q}=O_{P}(n^{1/q-1}V_{n}^{1/q}).
\end{align*}
By conditional Markov's inequality, we have, for all $\delta>0$
\begin{align*}
 & P\left(\left\Vert \frac{1}{V_{n}^{1/q}n^{1/q-1}}\frac{1}{n}\sum_{i=1}^{n}\bs X_{n,i}\right\Vert >\delta\mid\mathcal{C}_{n}\right)\\
\leq & \frac{\left\Vert \frac{1}{V_{n}^{1/q}n^{1/q-1}}\frac{1}{n}\sum_{i=1}^{n}\bs X_{n,i}\right\Vert _{L^{q},\mathcal{C}_{n}}^{q}}{\delta^{q}}=\frac{\left\Vert \frac{1}{n}\sum_{i=1}^{n}\bs X_{n,i}\right\Vert _{L^{q},\mathcal{C}_{n}}^{q}}{\left\{ V_{n}^{1/q}n^{1/q-1}\right\} ^{q}\delta^{q}}=\frac{R_{n}}{\delta^{q}},
\end{align*}
where the second step follows from $V_{n}$ is measurable with respect
to $\mathcal{C}_{n}$ and in the last step we let $R_{n}:=\left\{ \left\Vert \frac{1}{n}\sum_{i=1}^{n}\bs X_{n,i}\right\Vert _{L^{q},\mathcal{C}_{n}}/(V_{n}^{1/q}n^{1/q-1})\right\} ^{q}=O_{P}(1)$.
For any $\epsilon>0$, there exists $N,M>1$ such that 
\[
P\left(R_{n}\geq M\right)<\epsilon/2
\]
for all $n>N$. Let $\delta:=(\frac{2M}{\epsilon})^{1/q}$, then we
have 
\begin{align*}
 & P\left(\left\Vert \frac{1}{n^{1/q-1}}\frac{1}{n}\sum_{i=1}^{n}\bs X_{n,i}\right\Vert >(\frac{2M}{\epsilon})^{1/q}\right)=\e\left[P\left(\left\Vert \frac{1}{n^{1/q-1}}\frac{1}{n}\sum_{i=1}^{n}\bs X_{n,i}\right\Vert >(\frac{2M}{\epsilon})^{1/q}\mid\mathcal{C}_{n}\right)\right]\\
\leq & P\left(P\left(\left\Vert \frac{1}{n^{1/q-1}}\frac{1}{n}\sum_{i=1}^{n}\bs X_{n,i}\right\Vert >(\frac{2M}{\epsilon})^{1/q}\mid\mathcal{C}_{n}\right)>\frac{\epsilon}{2}\right)+\frac{\epsilon}{2}\\
\leq & P\left(\frac{R_{n}}{2M/\epsilon}>\frac{\epsilon}{2}\right)+\frac{\epsilon}{2}=P\left(R_{n}\geq M\right)+\frac{\epsilon}{2}<\epsilon
\end{align*}
for all $n>N$. Thus, $\frac{1}{V_{n}^{1/q}n^{1/q-1}}\frac{1}{n}\sum_{i=1}^{n}\bs X_{n,i}=O_{P}(1)$
and the desired result follows.
\end{proof}
\begin{lem}
\label{lem:clt for triangular array}Let $(\Omega,\mathcal{F},P)$
be a probability space and $\mathcal{C}_{n}$ be a sequence of sub-$\sigma$-algebras
of $\mathcal{F}$. Let $\left\{ X_{n,i}\in\R:1\leq i\leq n,n\geq1\right\} $
be a triangular array such that for every $n$, $X_{n,1},\ldots,X_{n,n}$
are mutually independent conditional on $\mathcal{C}_{n}$. Let $S_{n}:=\sum_{i=1}^{n}X_{n,i}$
and $\sigma_{n}^{2}(\mathcal{C}_{n}):=\var\left(S_{n}\mid\mathcal{C}_{n}\right)$.
We assume that
\begin{equation}
P\left(\left\{ \omega:\var(X_{n,i}\mid\mathcal{C}_{n})<\infty,\ \forall1\leq i\leq n\right\} \bigcap\left\{ \omega:\sigma_{n}^{2}(\mathcal{C}_{n})>0\right\} \right)\to1\label{eq: positive variance in clt}
\end{equation}
as $n\to\infty$. If the Lindeberg condition is satisfied: for any
$\delta>0$
\[
\frac{1}{\sigma_{n}^{2}(\mathcal{C}_{n})}\sum_{i=1}^{n}\e\left[\left\{ X_{n,i}-\e\left[X_{n,i}\mid\mathcal{C}_{n}\right]\right\} ^{2}\1\left\{ \left|X_{n,i}-\e\left[X_{n,i}\mid\mathcal{C}_{n}\right]\right|\geq\delta\sigma_{n}(\mathcal{C}_{n})\right\} \mid\mathcal{C}_{n}\right]=o_{P}(1)
\]
as $n\to\infty$, then for every $t\in\R$, (i) 
\[
\e\left[\exp\left\{ \mathrm{i}t\frac{S_{n}-\e\left[S_{n}\mid\mathcal{C}_{n}\right]}{\sigma_{n}(\mathcal{C}_{n})}\right\} \mid\mathcal{C}_{n}\right]\tp\exp\left(-t^{2}/2\right)
\]
as $n\to\infty$ and (ii) 
\[
\e\left[\left|\e\left[\exp\left\{ \mathrm{i}t\frac{S_{n}-\e\left[S_{n}\mid\mathcal{C}_{n}\right]}{\sigma_{n}(\mathcal{C}_{n})}\right\} \mid\mathcal{C}_{n}\right]-\exp\left(-t^{2}/2\right)\right|\right]\to0
\]
as $n\to\infty$.
\end{lem}
\begin{proof}
To obtain the first conclusion, we slightly modify the proof of \citet[Theorem 1]{bulinski2017Conditional}.
Let 
\[
E_{n}:=\left\{ \omega:\var(X_{n,i}\mid\mathcal{C}_{n})<\infty,\ \forall1\leq i\leq n\right\} \bigcap\left\{ \omega:\sigma_{n}^{2}(\mathcal{C}_{n})>0\right\} ,
\]
then $E_{n}\in\mathcal{C}_{n}$ and $P(E_{n})\to1$ as $n\to\infty$.
The only difference between our proof and that of \citet[Theorem 1]{bulinski2017Conditional}
lies in that we restrict our attention to $\omega\in E_{n}$, which
does not ruin the desired result since $P(E_{n})\to1$. We give some
key steps of the proof.

We let 
\[
Z_{n,i}(\omega):=\begin{cases}
\frac{X_{n,i}-\e\left[X_{n,i}\mid\mathcal{C}_{n}\right]}{\sigma_{n}(\mathcal{C}_{n})}, & \omega\in E_{n}\\
0, & \omega\in\Omega\backslash E_{n}
\end{cases},\quad i=1,\ldots n,\quad n\geq1,
\]
and 
\[
T_{n}(\omega):=\sum_{i=1}^{n}Z_{n,i}(\omega)=\begin{cases}
\frac{S_{n}-\e\left[S_{n}\mid\mathcal{C}_{n}\right]}{\sigma_{n}(\mathcal{C}_{n})}, & \omega\in E_{n},\\
0, & \omega\in\Omega\backslash E_{n}.
\end{cases}
\]
By $E_{n}\in\mathcal{C}_{n}$ and Lemma \ref{lem:conditional expectation},
$\e\left[Z_{n,i}\mid\mathcal{C}_{n}\right]=0$ a.s.,
\[
\var\left(Z_{n,i}\mid\mathcal{C}_{n}\right)=\e\left[Z_{n,i}^{2}\mid\mathcal{C}_{n}\right]=\begin{cases}
\frac{\var\left(X_{n,i}\mid\mathcal{C}_{n}\right)}{\sigma_{n}^{2}(\mathcal{C}_{n})}, & \omega\in E_{n},\\
0, & \omega\in\Omega\backslash E_{n},
\end{cases}\quad\text{a.s.}
\]
and 
\[
\sum_{i=1}^{n}\var\left(Z_{n,i}\mid\mathcal{C}_{n}\right)=\sum_{i=1}^{n}\e\left[Z_{n,i}^{2}\mid\mathcal{C}_{n}\right]=\begin{cases}
1, & \omega\in E_{n},\\
0, & \omega\in\Omega\backslash E_{n},
\end{cases}\quad\text{a.s.}
\]
 for $\omega\in E_{n}$. Now for any $\epsilon>0$ we define {\small
\begin{align*}
 & L_{n}^{\mathcal{C}_{n}}(\epsilon):=\sum_{i=1}^{n}\e\left[Z_{n,i}^{2}\1(\left|Z_{n,i}\right|\geq\epsilon)\mid\mathcal{C}_{n}\right]\\
= & \begin{cases}
\frac{1}{\sigma_{n}^{2}(\mathcal{C}_{n})}\sum_{i=1}^{n}\e\left[\left\{ X_{n,i}-\e\left[X_{n,i}\mid\mathcal{C}_{n}\right]\right\} ^{2}\1\left\{ \left|X_{n,i}-\e\left[X_{n,i}\mid\mathcal{C}_{n}\right]\right|\geq\epsilon\sigma_{n}(\mathcal{C}_{n})\right\} \mid\mathcal{C}_{n}\right], & \omega\in E_{n},\\
0, & \omega\in\Omega\backslash E_{n}.
\end{cases}
\end{align*}
}By the Lindeberg condition and $P(E_{n})\to1$ we have 
\[
L_{n}^{\mathcal{C}_{n}}(\epsilon)\tp0
\]
for any $\epsilon>0$. For any $n\geq1$ one has
\begin{align*}
\max_{1\leq i\leq n}\e\left[Z_{n,i}^{2}\mid\mathcal{C}_{n}\right] & =\max_{1\leq i\leq n}\left\{ \e\left[Z_{n,i}^{2}\1(\left|Z_{n,i}\right|<\epsilon)\mid\mathcal{C}_{n}\right]+\e\left[Z_{n,i}^{2}\1(\left|Z_{n,i}\right|\geq\epsilon)\mid\mathcal{C}_{n}\right]\right\} \\
 & \leq\epsilon^{2}+L_{n}^{\mathcal{C}_{n}}(\epsilon)\tp\epsilon^{2}.
\end{align*}
Since $\epsilon>0$ is arbitrary, we have $\e\left[Z_{n,i}^{2}\mid\mathcal{C}_{n}\right]\tp0$
as $n\to\infty$. Similar to the derivation of \citet[Eq.(15)]{bulinski2017Conditional},
we can obtain that, for any $t\in\R$
\[
\max_{1\leq i\leq n}\left|\varphi_{Z_{n,i}}^{\mathcal{C}_{n}}(t)-1\right|\leq\frac{t^{2}}{2}\max_{1\leq i\leq n}\e\left[Z_{n,i}^{2}\mid\mathcal{C}_{n}\right]\tp0,
\]
where $\varphi_{Z_{n,i}}^{\mathcal{C}_{n}}:=\e\left[\exp\left\{ \mathrm{i}tZ_{n,i}\right\} \mid\mathcal{C}_{n}\right]$.
Now, we have shown \citet[Eq.(13)]{bulinski2017Conditional}.

Introduce the events $D_{n}:=\left\{ \max_{1\leq i\leq n}\left|\varphi_{Z_{n,i}}^{\mathcal{C}_{n}}(t)-1\right|\leq1/2\right\} \bigcap E_{n}$,
$n\geq1$. Then $P(D_{n})\to1$ as $n\to\infty$. Besides, we can
show that \citet[Eqs.(17)-(20)]{bulinski2017Conditional} hold in
our case.

Let 
\[
G_{n}^{\mathcal{C}_{n}}(t):=\begin{cases}
\sum_{i=1}^{n}\left(\varphi_{Z_{n,i}}^{\mathcal{C}_{n}}(t)-1\right)+\frac{t^{2}}{2}, & \omega\in D_{n}.\\
0, & \omega\in\Omega\backslash D_{n}.
\end{cases}
\]
For $\omega\in D_{n}$, we have $\sum_{i=1}^{n}\e\left[Z_{n,i}^{2}\mid\mathcal{C}_{n}\right]=1$
and thus
\begin{align*}
G_{n}^{\mathcal{C}_{n}}(t) & =\sum_{i=1}^{n}\e\left[g_{t}(Z_{n,i})\mid\mathcal{C}_{n}\right]\\
 & =\sum_{i=1}^{n}\e\left[g_{t}(Z_{n,i})\1(\left|Z_{n,i}\right|<\epsilon)\mid\mathcal{C}_{n}\right]+\sum_{i=1}^{n}\e\left[g_{t}(Z_{n,i})\1(\left|Z_{n,i}\right|\geq\epsilon)\mid\mathcal{C}_{n}\right],
\end{align*}
where $g_{t}(x):=\exp\{\mathrm{i}tx\}-1-\mathrm{i}tx-(\mathrm{i}tx)^{2}/2$,
$x\in\R$. Applying \citet[Eq.(14)]{bulinski2017Conditional}, for
$\omega\in D_{n}$ we have 
\[
\left|G_{n}^{\mathcal{C}_{n}}(t)\right|\leq\frac{\left|t\right|^{3}\epsilon}{6}+t^{2}L_{n}^{\mathcal{C}_{n}}(\epsilon)
\]
for any $\epsilon>0$. Since $G_{n}^{\mathcal{C}_{n}}(t)=0$ for $\omega\in\Omega\backslash D_{n}$,
we have 
\[
\left|G_{n}^{\mathcal{C}_{n}}(t)\right|\leq\frac{\left|t\right|^{3}\epsilon}{6}+t^{2}L_{n}^{\mathcal{C}_{n}}(\epsilon),\quad\forall\omega\in\Omega.
\]
By $L_{n}^{\mathcal{C}_{n}}(\epsilon)\tp0$, we have $G_{n}^{\mathcal{C}_{n}}(t)\tp0$
as $n\to\infty$ for every $t\in\R$. Following the same steps as
in the proof of \citet[Theorem 1]{bulinski2017Conditional}, we can
obtain that 
\[
\e\left[\exp\left\{ \mathrm{i}tT_{n}\right\} \mid\mathcal{C}_{n}\right]\tp\exp(-t^{2}/2)
\]
as $n\to\infty$. By Lemma \ref{lem:conditional expectation} and
$E_{n}\in\mathcal{C}_{n}$, we have 
\[
\e\left[\exp\left\{ \mathrm{i}tT_{n}\right\} \mid\mathcal{C}_{n}\right]=\begin{cases}
\e\left[\exp\left\{ \mathrm{i}t\frac{S_{n}-\e\left[S_{n}\mid\mathcal{C}_{n}\right]}{\sigma_{n}(\mathcal{C}_{n})}\right\} \mid\mathcal{C}_{n}\right], & \omega\in E_{n}\\
1, & \omega\in\Omega\backslash E_{n}.
\end{cases}
\]
Note that $P(E_{n})\to1$ as $n\to\infty$, the first conclusion follows.

To obtain the second conclusion, we note that $\left|\e\left[\exp\left\{ \mathrm{i}t\frac{S_{n}-\e\left[S_{n}\mid\mathcal{C}_{n}\right]}{\sigma_{n}(\mathcal{C}_{n})}\right\} \mid\mathcal{C}_{n}\right]-\exp\left(-t^{2}/2\right)\right|\leq2$
is bounded. Then, combining $\left|\e\left[\exp\left\{ \mathrm{i}t\frac{S_{n}-\e\left[S_{n}\mid\mathcal{C}_{n}\right]}{\sigma_{n}(\mathcal{C}_{n})}\right\} \mid\mathcal{C}_{n}\right]-\exp\left(-t^{2}/2\right)\right|=o_{P}(1)$
and \citet[Theorem 5.4, chapter 1]{billingsley1968convergence} gives
the desired result.
\end{proof}
\begin{lem}
\label{lem:joint convergence}Let $(\Omega,\mathcal{F},P)$ be a probability
space and $\mathcal{C}_{n}$ be a sequence of sub-$\sigma$-algebras
of $\mathcal{F}$. Suppose that $X_{n}$ is a sequence of random variables
satisfies 
\[
\e\left[\left|\e\left[\exp\left\{ \mathrm{i}tX_{n}\right\} \mid\mathcal{C}_{n}\right]-\exp\left(-t^{2}/2\right)\right|\right]\to0
\]
as $n\to\infty$ for any $t\in\R$, and $Y_{n}$ is a sequence of
$\mathcal{C}_{n}$-measurable random variables satisfying 
\[
\e\left[\exp\left\{ \mathrm{i}sY_{n}\right\} \right]\to\exp\left(-s^{2}/2\right)
\]
as $n\to\infty$ for any $s\in\R$. Then 
\[
(X_{n},Y_{n})^{\trans}\tod N\left(\begin{pmatrix}0\\
0
\end{pmatrix},\begin{pmatrix}1 & 0\\
0 & 1
\end{pmatrix}\right).
\]
\end{lem}
\begin{proof}
Note that we have 
\begin{align*}
 & \left|\e\left[\exp\left\{ \mathrm{i}tX_{n}+\mathrm{i}sY_{n}\right\} \right]-\exp\left(-t^{2}/2-s^{2}/2\right)\right|\\
= & \left|\e\left[\e\left[\exp\left\{ \mathrm{i}tX_{n}+\mathrm{i}sY_{n}\right\} \mid\mathcal{C}_{n}\right]\right]-\exp\left(-t^{2}/2-s^{2}/2\right)\right|\\
= & \left|\e\left[\exp\left\{ \mathrm{i}sY_{n}\right\} \e\left[\exp\left\{ \mathrm{i}tX_{n}\right\} \mid\mathcal{C}_{n}\right]\right]-\exp\left(-t^{2}/2-s^{2}/2\right)\right|\\
\leq & \left|\e\left[\left|\exp\left\{ \mathrm{i}sY_{n}\right\} \right|\left|\e\left[\exp\left\{ \mathrm{i}tX_{n}\right\} \mid\mathcal{C}_{n}\right]-\exp(-t^{2}/2)\right|\right]\right|\\
 & +\left|\exp(-t^{2}/2)\e\left[\exp\left\{ \mathrm{i}sY_{n}\right\} -\exp\left(-s^{2}/2\right)\right]\right|\\
\leq & \e\left[\left|\e\left[\exp\left\{ \mathrm{i}tX_{n}\right\} \mid\mathcal{C}_{n}\right]-\exp(-t^{2}/2)\right|\right]+\exp(-t^{2}/2)\left|\e\left[\exp\left\{ \mathrm{i}sY_{n}\right\} -\exp\left(-s^{2}/2\right)\right]\right|\\
\to & 0
\end{align*}
for any $t\in\R$ and $s\in\R$, where the second equality follows
from the fact that $Y_{n}$ is measurable with respect to $\mathcal{C}_{n}$
and the last step follows from $\e\left[\left|\e\left[\exp\left\{ \mathrm{i}tX_{n}\right\} \mid\mathcal{C}_{n}\right]-\exp\left(-t^{2}/2\right)\right|\right]\to0$
and $\e\left[\exp\left\{ \mathrm{i}sY_{n}\right\} \right]\to\exp\left(-s^{2}/2\right)$.
By Levy's continuity theorem, the conclusion follows.
\end{proof}
\begin{lem}
\label{lem:covergence of combin of normal}Suppose $X_{n}$ and $Y_{n}$
are two sequences of random variables satisfying 
\[
(X_{n},Y_{n})^{\trans}\tod(X,Y)^{\trans}\overset{d}{=}N\left(\begin{pmatrix}0\\
0
\end{pmatrix},\begin{pmatrix}1 & 0\\
0 & 1
\end{pmatrix}\right).
\]
Let $a_{n}$ and $b_{n}$ be two sequences of real numbers satisfying
$a_{n}^{2}+b_{n}^{2}=1$ for all $n\geq1$. Then
\[
a_{n}X_{n}+b_{n}Y_{n}\tod N(0,1).
\]
\end{lem}
\begin{proof}
By Levy's continuity theorem, we have 
\[
\phi_{n}(t,s):=\e\left[\exp\left\{ \mathrm{i}tX_{n}+\mathrm{i}sY_{n}\right\} \right]\to\e\left[\exp\left\{ \mathrm{i}tX+\mathrm{i}sY\right\} \right]=\exp\left(-t^{2}/2-s^{2}/2\right)=:\phi(t,s)
\]
as $n\to\infty$ for any $(t,s)\in\R^{2}$. We show that $\phi_{n}(t,s)$
converges to $\phi(t,s)$ uniformly over any compact set $\mathcal{G}\subset\R^{2}$.
We fix any point $(t_{0},s_{0})\in\R^{2}$ and any number $\ep>0$.
We employ the $3\epsilon$-argument.

\textbf{First $\epsilon$.} Note that $\phi(t,s)$ is continuous,
we have, there exists $\delta_{1}>0$ such that 
\[
\left|\phi(t_{0},s_{0})-\phi(t,s)\right|\leq\epsilon
\]
for any $(t,s)\in\mathcal{B}((t_{0},s_{0});\delta_{1}):=\left\{ (t,s):\left\Vert (t,s)-(t_{0},s_{0})\right\Vert <\delta_{1}\right\} $.

\textbf{Second $\epsilon$.} Since $\phi_{n}(t_{0},s_{0})\to\phi(t_{0},s_{0})$,
there exists $N_{2}\geq1$ such that 
\[
\left|\phi_{n}(t_{0},s_{0})\to\phi(t_{0},s_{0})\right|\leq\epsilon
\]
for any $n\geq N_{2}$.

\textbf{Third $\ep$.} By Prohorov's theorem (\citet[Theorem 2.4]{vaart1998Asymptotic}),
$\{(X_{n},Y_{n}):n\geq1\}$ is uniformly tight. Then there exists
$R>0$ such that 
\[
\sup_{n\geq1}P\left(\left\Vert (X_{n},Y_{n})\right\Vert \geq R\right)<\epsilon/4.
\]
Note that $\exp\{\mathrm{i}t\}$ is continuous at $t=0$, we have,
there exists $\delta_{3}>0$ such that 
\[
\left|\exp\{\mathrm{i}(t-t_{0})x+\mathrm{i}(s-s_{0})y\}-1\right|\leq\epsilon/2
\]
for all $(t,s)\in\mathcal{B}((t_{0},s_{0});\delta_{3})$ and $(x,y)\in\mathcal{B}(\bs 0;R)$.
Now we have 
\begin{align*}
\left|\phi_{n}(t,s)-\phi_{n}(t_{0},s_{0})\right| & \leq\left|\e\left[\exp\left\{ \mathrm{i}tX_{n}+\mathrm{i}sY_{n}\right\} \right]-\e\left[\exp\left\{ \mathrm{i}t_{0}X_{n}+\mathrm{i}s_{0}Y_{n}\right\} \right]\right|\\
 & \leq\left|\e\left[\exp\left\{ \mathrm{i}t_{0}X_{n}+\mathrm{i}s_{0}Y_{n}\right\} \left(\exp\left\{ \mathrm{i}(t-t_{0})X_{n}+\mathrm{i}(s-s_{0})Y_{n}\right\} -1\right)\right]\right|\\
 & \leq\e\left[\left|\exp\left\{ \mathrm{i}(t-t_{0})X_{n}+\mathrm{i}(s-s_{0})Y_{n}\right\} -1\right|\right]\\
 & =\int_{\R^{2}}\left|\exp\left\{ \mathrm{i}(t-t_{0})x+\mathrm{i}(s-s_{0})y\right\} -1\right|dF_{X_{n},Y_{n}}(x,y)\\
 & =\int_{\mathcal{B}(\bs 0;R)}\left|\exp\left\{ \mathrm{i}(t-t_{0})x+\mathrm{i}(s-s_{0})y\right\} -1\right|dF_{X_{n},Y_{n}}(x,y)\\
 & \qquad+\int_{\R^{2}\backslash\mathcal{B}(\bs 0;R)}\left|\exp\left\{ \mathrm{i}(t-t_{0})x+\mathrm{i}(s-s_{0})y\right\} -1\right|dF_{X_{n},Y_{n}}(x,y)\\
 & \leq\epsilon/2+2\times\epsilon/4=\epsilon
\end{align*}
for all $(t,s)\in\mathcal{B}((t_{0},s_{0});\delta_{3})$ and $n\geq1$.

Let $\delta(t_{0},s_{0}):=\min\left\{ \delta_{1},\delta_{3}\right\} >0$
and $N(t_{0},s_{0}):=N_{2}\geq1$, we have 
\begin{align}
\left|\phi_{n}(t,s)-\phi(t,s)\right| & \leq\left|\phi_{n}(t,s)-\phi_{n}(t_{0},s_{0})\right|+\left|\phi_{n}(t_{0},s_{0})-\phi(t_{0},s_{0})\right|+\left|\phi(t_{0},s_{0})-\phi(t,s)\right|\nonumber \\
 & \leq3\epsilon\label{eq:3=00005Cepsilon}
\end{align}
for all $(t,s)\in\mathcal{B}((t_{0},s_{0});\delta(t_{0},s_{0}))$
and $n\geq N(t_{0},s_{0})$.

Let $\mathcal{G}$ be any compact subset of $\R^{2}$. Then $\bigcup_{(t_{0},s_{0})\in G}\mathcal{B}((t_{0},s_{0});\delta(t_{0},s_{0}))$
is an open cover of $\mathcal{G}$. By the Heine--Borel theorem,
there exists a finite constant $G\geq1$ such that $\bigcup_{g=1}^{G}\mathcal{B}((t_{g},s_{g});\delta(t_{g},s_{g}))$
is a finite subcover of $\mathcal{G}$. Let $\delta:=\min_{1\leq g\leq G}\left\{ \delta(t_{g},s_{g})\right\} >0$
and $N:=\max_{1\leq g\leq G}\left\{ N(t_{g},s_{g})\right\} <\infty$,
then it follows from (\ref{eq:3=00005Cepsilon}) that 
\[
\left|\phi_{n}(t,s)-\phi(t,s)\right|\leq3\epsilon
\]
for all $(t,s)\in\mathcal{G}$ and $n\geq N$. Thus, $\phi_{n}(t,s)$
converges to $\phi(t,s)$ uniformly over any compact set $\mathcal{G}\subset\R^{2}$.

For any $u\in\R$, let $\mathcal{G}_{u}:=\left\{ (t,s)\in\R^{2}:t^{2}+s^{2}=u^{2}\right\} $.
Then $\mathcal{G}_{u}$ is a compact subset of $\R^{2}$. Now we have
\begin{align*}
 & \left|\e\left[\exp\left\{ \mathrm{i}u\left\{ a_{n}X_{n}+b_{n}Y_{n}\right\} \right\} \right]-\e\left[\exp\left\{ -u^{2}/2\right\} \right]\right|\\
= & \left|\phi_{n}(ua_{n},ub_{n})-\phi(ua_{n},ub_{n})\right|\leq\sup_{(t,s)\in\mathcal{G}_{u}}\left|\phi_{n}(t,s)-\phi(t,s)\right|\\
\to & 0
\end{align*}
as $n\to\infty$. Since $u\in\R$ is arbitrary, the result follows
from Levy's continuity theorem.
\end{proof}
\begin{lem}
\label{lem:conditional expectation}Let $(\Omega,\mathcal{F},P)$
be a probability space and $\mathcal{C}_{n}$ be a sequence of sub-$\sigma$-algebras
of $\mathcal{F}$. Let $\left\{ X_{n}\in\R,n\geq1\right\} $ and $\left\{ Y_{n}\in\R,n\geq1\right\} $
be two sequences of random variables. Let $E_{n}$ be a subset of
$\Omega$ and 
\[
Z_{n}(\omega):=\begin{cases}
X_{n}(\omega), & \omega\in E_{n},\\
Y_{n}(\omega), & \omega\in\Omega\backslash E_{n}.
\end{cases}
\]
If $E_{n}\in\mathcal{C}_{n}$, then $\e\left[Z_{n}\mid\mathcal{C}_{n}\right]=R_{n}$
a.s., where 
\[
R_{n}(\omega):=\begin{cases}
\e\left[X_{n}\mid\mathcal{C}_{n}\right](\omega), & \omega\in E_{n},\\
\e\left[Y_{n}\mid\mathcal{C}_{n}\right](\omega), & \omega\in\Omega\backslash E_{n}.
\end{cases}
\]
\end{lem}
\begin{proof}
For any $H\in\mathcal{C}_{n}$, we have 
\begin{align*}
\int_{H}R_{n}dP & =\int_{H\cap E_{n}}R_{n}dP+\int_{H\cap\left(\Omega\backslash E_{n}\right)}R_{n}dP\\
 & =\int_{H\cap E_{n}}\e\left[X_{n}\mid\mathcal{C}_{n}\right]dP+\int_{H\cap\left(\Omega\backslash E_{n}\right)}\e\left[Y_{n}\mid\mathcal{C}_{n}\right]dP\\
 & =\int_{H\cap E_{n}}X_{n}dP+\int_{H\cap\left(\Omega\backslash E_{n}\right)}Y_{n}dP=\int_{H\cap E_{n}}Z_{n}dP+\int_{H\cap\left(\Omega\backslash E_{n}\right)}Z_{n}dP\\
 & =\int_{H}Z_{n}dP,
\end{align*}
where the third inequality follows from $E_{n}\in\mathcal{C}_{n}$
and the definitions of $\e\left[X_{n}\mid\mathcal{C}_{n}\right]$
and $\e\left[Y_{n}\mid\mathcal{C}_{n}\right]$. By the definition
of the conditional expectation, the desired result follows.
\end{proof}

\section{Proofs for the main theorems\label{sec:Proofs-for-the}}

\subsection{Notation}

We give the expressions for the asymptotic variances:
\[
\varsigma_{\widetilde{Y}}^{2}:=\sum_{k=1}^{K}\frac{p_{[k]}}{1-\pi_{[k]}}\var\left(\widetilde{Y}_{i}(0)\mid B_{i}=k\right)+\sum_{k=1}^{K}\frac{p_{[k]}}{\pi_{[k]}}\var\left(\widetilde{Y}_{i}(1)\mid B_{i}=k\right)
\]
\[
\varsigma_{H}^{2}:=\sum_{k=1}^{K}p_{[k]}\left\{ \e\left[Y_{i}(1)-Y_{i}(0)\mid B_{i}=k\right]\right\} ^{2}-\left\{ \sum_{k=1}^{K}p_{[k]}\e\left[Y_{i}(1)-Y_{i}(0)\mid B_{i}=k\right]\right\} ^{2},
\]
\[
\mathbf{\Sigma}_{[k]\widetilde{\bs{\xi}}_{n}^{*}\widetilde{Y}(a)}:=\e\left[\widetilde{\bs{\xi}}_{n}^{*}(\bs X_{i})\widetilde{Y}_{i}(a)\mid B_{i}=k\right],\ 1\leq k\leq K,a\in\{0,1\},
\]
\[
\mathbf{\Sigma}_{[k]\widetilde{\bs{\xi}}_{n}^{*}\widetilde{\bs{\xi}}_{n}^{*}}:=\e\left[\widetilde{\bs{\xi}}_{n}^{*}(\bs X_{i})\widetilde{\bs{\xi}}_{n}^{*}(\bs X_{i})^{\trans}\mid B_{i}=k\right],\ 1\leq k\leq K,
\]
\begin{align*}
\varsigma_{\widetilde{Y}\mid\widetilde{\bs{\xi}}_{n}^{*}}^{2} & :=\sum_{k=1}^{K}\frac{p_{[k]}}{\pi_{[k]}(1-\pi_{[k]})}\left\{ (1-\pi_{[k]})\mathbf{\Sigma}_{[k]\widetilde{\bs{\xi}}_{n}^{*}\widetilde{Y}(1)}+\pi_{[k]}\mathbf{\Sigma}_{[k]\widetilde{\bs{\xi}}_{n}^{*}\widetilde{Y}(0)}\right\} ^{\trans}\\
 & \qquad\times\mathbf{\Sigma}_{[k]\widetilde{\bs{\xi}}_{n}^{*}\widetilde{\bs{\xi}}_{n}^{*}}^{+}\left\{ (1-\pi_{[k]})\mathbf{\Sigma}_{[k]\widetilde{\bs{\xi}}_{n}^{*}\widetilde{Y}(1)}+\pi_{[k]}\mathbf{\Sigma}_{[k]\widetilde{\bs{\xi}}_{n}^{*}\widetilde{Y}(0)}\right\} ,
\end{align*}
and
\[
\varsigma_{\widetilde{Y}\mid\widetilde{h}^{*}}^{2}:=\sum_{k=1}^{K}\frac{p_{[k]}}{\pi_{[k]}(1-\pi_{[k]})}\e\left[\left\{ (1-\pi_{[k]})\widetilde{h}_{1[k]}^{*}(\bs X_{i})+\pi_{[k]}\widetilde{h}_{0[k]}^{*}(\bs X_{i})\right\} ^{2}\mid B_{i}=k\right].
\]
The estimates of the asymptotic variances are defined as:
\begin{align}
\widehat{\varsigma}_{\widetilde{Y}}^{2} & :=\sum_{k=1}^{K}\frac{n_{[k]}}{n_{[k]}-\widehat{r}_{[k]}-1}\frac{p_{n[k]}}{1-\pi_{n[k]}}\frac{1}{n_{0[k]}}\sum_{i\in[k]}(1-A_{i})(Y_{i}-\overline{Y}_{0[k]})^{2}\nonumber \\
 & \quad+\sum_{k=1}^{K}\frac{n_{[k]}}{n_{[k]}-\widehat{r}_{[k]}-1}\frac{p_{n[k]}}{\pi_{n[k]}}\frac{1}{n_{1[k]}}\sum_{i\in[k]}A_{i}(Y_{i}-\overline{Y}_{1[k]})^{2},\label{eq:hat varsigma Y}
\end{align}
\begin{equation}
\widehat{\varsigma}_{H}^{2}:=\sum_{k=1}^{K}\frac{n_{[k]}}{n_{[k]}-\widehat{r}_{[k]}-1}p_{n[k]}\left\{ \overline{Y}_{1[k]}-\overline{Y}_{0[k]}-(\overline{Y}_{1}-\overline{Y}_{0})\right\} ^{2},\label{eq:hat varsigma H}
\end{equation}
\[
\widehat{\mathbf{\Gamma}}_{[k]}:=\frac{1}{n}\sum_{i\in[k]}\left\{ \frac{(1-\pi_{n[k]})A_{i}}{\pi_{n[k]}}(\bs{\xi}_{n}(\bs X_{i})-\overline{\bs{\xi}}_{n[k]})(Y_{i}-\overline{Y}_{1[k]})+\frac{\pi_{n[k]}(1-A_{i})}{1-\pi_{n[k]}}(\bs{\xi}_{n}(\bs X_{i})-\overline{\bs{\xi}}_{n[k]})(Y_{i}-\overline{Y}_{0[k]})\right\} ,
\]
\[
\widehat{\mathbf{\Sigma}}_{[k]}:=\frac{1}{n}\sum_{i\in[k]}(A_{i}-\pi_{n[k]})^{2}(\bs{\xi}_{n}(\bs X_{i})-\overline{\bs{\xi}}_{n[k]})(\bs{\xi}_{n}(\bs X_{i})-\overline{\bs{\xi}}_{n[k]})^{\trans},\ 1\leq k\leq K,
\]
and 
\begin{equation}
\widehat{\varsigma}_{\widetilde{Y}\mid\widetilde{\bs{\xi}}_{n}^{*}}^{2}:=\sum_{k=1}^{K}\frac{n_{[k]}}{n_{[k]}-\widehat{r}_{[k]}-1}\widehat{\mathbf{\Gamma}}_{[k]}^{\trans}\widehat{\mathbf{\Sigma}}_{[k]}^{+}\widehat{\mathbf{\Gamma}}_{[k]}.\label{eq:hat varsigma Y|xi}
\end{equation}
Note that we adjust the degree of freedom in every stratum by a factor
of $n_{[k]}/(n_{[k]}-\widehat{r}_{n,[k]}-1)$ in (\ref{eq:hat varsigma Y})-(\ref{eq:hat varsigma Y|xi}),
where $\widehat{r}_{[k]}$ is a consistent estimator of $r_{[k]}=\mathrm{rank}\left\{ \e\left[\bs{\xi}_{n}^{*}(\bs X_{i})\bs{\xi}_{n}^{*}(\bs X_{i})^{\trans}\mid B_{i}=k\right]\right\} $,
i.e., $\sup_{1\leq k\leq K}\left\{ \left|\widehat{r}_{[k]}-r_{[k]}\right|/r_{[k]}\right\} =o_{P}(1)$
as $n\to\infty$. If $\e\left[\bs{\xi}_{n}^{*}(\bs X_{i})\bs{\xi}_{n}^{*}(\bs X_{i})^{\trans}\mid B_{i}=k\right]$
is non-singular, then we can simply take $\widehat{r}_{[k]}=d$.

Table~\ref{tab:The-Notation-that} lists the notation that will be
used in this section. Let $\mathcal{C}_{n}$ denote the $\sigma$-algebra
generated by $(A^{(n)},B^{(n)})$ and $\e_{\mathcal{C}_{n}}\left[\cdot\right]:=\e\left[\cdot\mid\mathcal{C}_{n}\right]$.
\begin{center}
\begin{table}[!tbh]
\begin{centering}
\caption{\label{tab:The-Notation-that}The list of notation used in Section
\ref{sec:Proofs-for-the}.}
\par\end{centering}
\resizebox{\textwidth}{!}{%
\begin{centering}
\begin{tabular}{llll}
\toprule 
Notation & Expression & Notation & Expression\tabularnewline
\midrule
\midrule 
$\widetilde{Y}_{i}(a)$ & $Y_{i}(a)-\e\left[Y_{i}(a)\mid B_{i}\right]$ & $\widetilde{\bs{\xi}}_{n}^{*}(\bs X_{i})$ & $\bs{\xi}_{n}^{*}(\bs X_{i})-\e\left[\bs{\xi}_{n}^{*}(\bs X_{i})\mid B_{i}\right]$\tabularnewline
\midrule 
$\bs{\Xi}_{i,[k]}$ & $\left\{ A_{i}-\pi_{n[k]}\right\} \1(B_{i}=k)\left\{ \bs{\xi}_{n}(\bs X_{i})-\overline{\bs{\xi}}_{n[k]}\right\} $ & $\bs{\Xi}_{i}$ & $\left(\bs{\Xi}_{i,[1]}^{\trans},\ldots,\bs{\Xi}_{i,[K]}^{\trans}\right)^{\trans}$\tabularnewline
\midrule 
$\bs{\Xi}_{i,[k]}^{*}$ & $\left\{ A_{i}-\pi_{n[k]}\right\} \1(B_{i}=k)\widetilde{\bs{\xi}}_{n}^{*}(\bs X_{i})$ & $\bs{\Xi}_{i}^{*}$ & $\left(\bs{\Xi}_{i,[1]}^{*\trans},\ldots,\bs{\Xi}_{i,[K]}^{*\trans}\right)^{\trans}$\tabularnewline
\midrule 
$\widetilde{Y}_{i,[k]}$ & $\1(B_{i}=k)\{\frac{A_{i}}{\pi_{n[k]}}(Y_{i}-\overline{Y}_{1[k]})-\frac{1-A_{i}}{1-\pi_{n[k]}}(Y_{i}-\overline{Y}_{0[k]})\}$ & $\widetilde{Y}_{i}$ & $\sum_{k=1}^{K}\widetilde{Y}_{i,[k]}$\tabularnewline
\midrule 
$\widetilde{Y}_{i,[k]}^{*}$ & $\1(B_{i}=k)\{\frac{A_{i}}{\pi_{n[k]}}\widetilde{Y}_{i}(1)-\frac{1-A_{i}}{1-\pi_{n[k]}}\widetilde{Y}_{i}(0)\}$ & $\widetilde{Y}_{i}^{*}$ & $\sum_{k=1}^{K}\widetilde{Y}_{i,[k]}^{*}$\tabularnewline
\midrule 
$\mathbf{\Sigma}_{[k]}$ & $\pi_{[k]}\left(1-\pi_{[k]}\right)p_{[k]}\mathbf{\Sigma}_{[k]\widetilde{\bs{\xi}}_{n}^{*}\widetilde{\bs{\xi}}_{n}^{*}}$ & $\mathbf{\Sigma}$ & $\diag\left\{ \mathbf{\Sigma}_{[k]}:k=1,\ldots,K\right\} $\tabularnewline
\midrule 
$\mathbf{\Sigma}_{[k]}^{\mathcal{C}_{n}}$ & $\frac{1}{n}\sum_{i=1}^{n}\left\{ A_{i}-\pi_{n[k]}\right\} ^{2}\1(B_{i}=k)\mathbf{\Sigma}_{[k]\widetilde{\bs{\xi}}_{n}^{*}\widetilde{\bs{\xi}}_{n}^{*}}$ & $\mathbf{\Sigma}_{\bs{\Xi}\bs{\Xi}\epsilon[k]}^{\mathcal{C}_{n}}$ & $\frac{1}{n}\sum_{i=1}^{n}\e_{\mathcal{C}_{n}}\left[\bs{\Xi}_{i,[k]}^{*}\bs{\Xi}_{i,[k]}^{*\trans}\epsilon_{i,[k]}^{*}\right]$\tabularnewline
\midrule 
$\bs{\beta}_{[k]}^{*}$ & $\mathbf{\Sigma}_{[k]\widetilde{\bs{\xi}}_{n}^{*}\widetilde{\bs{\xi}}_{n}^{*}}^{+}\left\{ \frac{1}{\pi_{[k]}}\mathbf{\Sigma}_{[k]\widetilde{\bs{\xi}}_{n}^{*}\widetilde{Y}(1)}+\frac{1}{1-\pi_{[k]}}\mathbf{\Sigma}_{[k]\widetilde{\bs{\xi}}_{n}^{*}\widetilde{Y}(0)}\right\} $ & $\bs{\beta}^{*}$ & $\left(\bs{\beta}_{[1]}^{*\trans},\ldots,\bs{\beta}_{[K]}^{*\trans}\right)^{\trans}$\tabularnewline
\midrule 
$\bs{\beta}_{[k],\mathcal{C}_{n}}$ & $\left\{ \mathbf{\Sigma}_{[k]}^{\mathcal{C}_{n}}\right\} ^{+}\left\{ \frac{1}{n}\sum_{i=1}^{n}\e_{\mathcal{C}_{n}}\left[\widetilde{Y}_{i,[k]}^{*}\bs{\Xi}_{i,[k]}^{*}\right]\right\} $ & $\bs{\beta}_{\mathcal{C}_{n}}$ & $\left(\bs{\beta}_{[1],\mathcal{C}_{n}}^{\trans},\ldots,\bs{\beta}_{[K],\mathcal{C}_{n}}^{\trans}\right)^{\trans}$\tabularnewline
\midrule 
$\widehat{\bs{\beta}}_{[k]}$ & $\left\{ \frac{1}{n}\sum_{i=1}^{n}\bs{\Xi}_{i,[k]}\bs{\Xi}_{i,[k]}^{\trans}\right\} ^{+}\frac{1}{n}\sum_{i=1}^{n}\widetilde{Y}_{i,[k]}\bs{\Xi}_{i,[k]}$ & $\widehat{\bs{\beta}}$ & $\left(\widehat{\bs{\beta}}_{[1]}^{\trans},\ldots,\widehat{\bs{\beta}}_{[K]}^{\trans}\right)^{\trans}$\tabularnewline
\midrule 
$\epsilon_{i,[k]}$ & $\widetilde{Y}_{i,[k]}-\bs{\beta}_{[k],\mathcal{C}_{n}}^{\trans}\bs{\Xi}_{i,[k]}$ & $\epsilon_{i}$ & $\left(\epsilon_{i,[1]},\ldots,\epsilon_{i,[K]}\right)^{\trans}$\tabularnewline
\midrule 
$\epsilon_{i,[k]}^{*}$ & $\widetilde{Y}_{i,[k]}^{*}-\bs{\beta}_{[k],\mathcal{C}_{n}}^{\trans}\bs{\Xi}_{i,[k]}^{*}$ & $\epsilon_{i}^{*}$ & $\left(\epsilon_{i,[1]}^{*},\ldots,\epsilon_{i,[K]}^{*}\right)^{\trans}$\tabularnewline
\bottomrule
\end{tabular}
\par\end{centering}
}
\end{table}
\par\end{center}

\subsection{Proof of Theorem \ref{thm:asymptotic properties of tau_cal}}

\textbf{Derive the expressions for the calibration weights.} We first
derive the expressions for the calibration weights $\widehat{w}_{i}$'s.
For all $i\in\left\{ 1,\ldots,n\right\} $, let 
\[
\bs{\Xi}_{i}:=\begin{pmatrix}\left\{ A_{i}-\pi_{n[1]}\right\} \1(B_{i}=1)\left\{ \bs{\xi}_{n}(\bs X_{i})-\overline{\bs{\xi}}_{n[1]}\right\} \\
\left\{ A_{i}-\pi_{n[2]}\right\} \1(B_{i}=2)\left\{ \bs{\xi}_{n}(\bs X_{i})-\overline{\bs{\xi}}_{n[2]}\right\} \\
\vdots\\
\left\{ A_{i}-\pi_{n[K]}\right\} \1(B_{i}=K)\left\{ \bs{\xi}_{n}(\bs X_{i})-\overline{\bs{\xi}}_{n[K]}\right\} 
\end{pmatrix}\in\R^{Kd}.
\]
Let $D(\bs w)=\sum_{i=1}^{n}D(w_{i},1)$ . Then the calibration problem
can be rewritten as 
\begin{equation}
\min_{w_{i}:1\leq i\leq n}D(\bs w)\ \text{ s.t. }\ \left(\bs{\Xi}_{1},\ldots,\bs{\Xi}_{n}\right)\bs w=0.\label{eq:cal step rewritten}
\end{equation}
Note that the conjugate function of $D(\bs w)$ is 
\[
D^{*}(\bs z)=\sum_{i=1}^{n}\left\{ z_{i}\cdot(D^{\prime})^{-1}(z_{i})-D\left\{ (D^{\prime})^{-1}(z_{i})\right\} \right\} =\sum_{i=1}^{n}-\rho(-z_{i}),\quad\forall\bs z=(z_{1},\ldots,z_{n})^{\trans},
\]
where $\rho(v):=D\left\{ (D^{\prime})^{-1}(-v)\right\} +v\cdot(D^{\prime})^{-1}(-v)$
for any $v\in\R$. By \citet{tseng1991Relaxation} (see also \citet{boyd2004Convex}),
the dual problem of (\ref{eq:cal step rewritten}) is 
\begin{align*}
\max_{\bs{\lambda}\in\R^{Kd}}\left\{ -D^{*}((\bs{\Xi}_{1},\ldots,\bs{\Xi}_{n})^{\trans}\bs{\lambda})\right\}  & =\max_{\bs{\lambda}\in\R^{Kd}}\sum_{i=1}^{n}\rho(-\bs{\lambda}^{\trans}\bs{\Xi}_{i})=\max_{\bs{\lambda}\in\R^{Kd}}\frac{1}{n}\sum_{i=1}^{n}\rho(\bs{\lambda}^{\trans}\bs{\Xi}_{i}).
\end{align*}
Let $\widehat{\bs{\lambda}}:=\arg\max_{\bs{\lambda}\in\R^{Kd}}\frac{1}{n}\sum_{i=1}^{n}\rho(\bs{\lambda}^{\trans}\bs{\Xi}_{i})$.
In the following, considering the special case that $D(v)=(v-1)^{2}/2$,
we have $\rho(v)=-v^{2}/2+v$, $\rho^{\prime}(v)=-v+1$ and $\rho^{\prime\prime}(v)=-1$.
Then 
\begin{equation}
\widehat{\bs{\lambda}}=\underset{\bs{\lambda}\in\R^{Kd}}{\arg\max}\frac{1}{n}\sum_{i=1}^{n}\left\{ -\frac{1}{2}\bs{\lambda}^{\trans}\bs{\Xi}_{i}\bs{\Xi}_{i}^{\trans}\bs{\lambda}+\bs{\lambda}^{\trans}\bs{\Xi}_{i}\right\} =\underset{\bs{\lambda}\in\R^{Kd}}{\arg\max}\left\{ -\frac{1}{2}\bs{\lambda}^{\trans}\frac{1}{n}\sum_{i=1}^{n}\bs{\Xi}_{i}\bs{\Xi}_{i}^{\trans}\bs{\lambda}+\bs{\lambda}^{\trans}\frac{1}{n}\sum_{i=1}^{n}\bs{\Xi}_{i}\right\} .\label{eq:solve lambda}
\end{equation}
We claim that 
\begin{equation}
\frac{1}{n}\sum_{i=1}^{n}\bs{\Xi}_{i}\in\text{range}\left(\frac{1}{n}\sum_{i=1}^{n}\bs{\Xi}_{i}\bs{\Xi}_{i}^{\trans}\right),\label{eq:mean in variance}
\end{equation}
otherwise the optimization problem (\ref{eq:solve lambda}) is unbounded,
which contradicts to 
\[
-\frac{1}{2}\bs{\lambda}^{\trans}\frac{1}{n}\sum_{i=1}^{n}\bs{\Xi}_{i}\bs{\Xi}_{i}^{\trans}\bs{\lambda}+\bs{\lambda}^{\trans}\frac{1}{n}\sum_{i=1}^{n}\bs{\Xi}_{i}=\frac{1}{n}\sum_{i=1}^{n}\left\{ -\frac{1}{2}\left(\bs{\lambda}^{\trans}\bs{\Xi}_{i}\right)^{2}+\bs{\lambda}^{\trans}\bs{\Xi}_{i}\right\} \leq\frac{1}{2}.
\]
By the first order condition of (\ref{eq:solve lambda}) we have 
\[
\frac{1}{n}\sum_{i=1}^{n}\bs{\Xi}_{i}\bs{\Xi}_{i}^{\trans}\bs{\lambda}=\frac{1}{n}\sum_{i=1}^{n}\bs{\Xi}_{i}.
\]
It follows from (\ref{eq:mean in variance}) that this equation has
at least one solution. We always take the solution as

\begin{equation}
\widehat{\bs{\lambda}}=\left\{ \frac{1}{n}\sum_{i=1}^{n}\bs{\Xi}_{i}\bs{\Xi}_{i}^{\trans}\right\} ^{+}\frac{1}{n}\sum_{i=1}^{n}\bs{\Xi}_{i},\label{eq:def of lambdahat}
\end{equation}
where $\left\{ \frac{1}{n}\sum_{i=1}^{n}\bs{\Xi}_{i}\bs{\Xi}_{i}^{\trans}\right\} ^{+}$
is the Moore-Penrose inverse of $\frac{1}{n}\sum_{i=1}^{n}\bs{\Xi}_{i}\bs{\Xi}_{i}^{\trans}$.

Note that 
\begin{equation}
\sum_{k=1}^{K}p_{n[k]}\left(\overline{Y}_{1[k]}-\overline{Y}_{0[k]}\right)=\frac{1}{n}\sum_{i=1}^{n}\sum_{k=1}^{K}\left\{ \frac{A_{i}}{\pi_{n[k]}}-\frac{1-A_{i}}{1-\pi_{n[k]}}\right\} \1(B_{i}=k)\cdot Y_{i}\label{eq:modification of tau_hat1}
\end{equation}
and 
\begin{equation}
\frac{1}{n}\sum_{i=1}^{n}\sum_{k=1}^{K}\left\{ \frac{A_{i}}{\pi_{n[k]}}\left(Y_{i}-\overline{Y}_{1[k]}\right)-\frac{1-A_{i}}{1-\pi_{n[k]}}\left(Y_{i}-\overline{Y}_{0[k]}\right)\right\} \1(B_{i}=k)=0.\label{eq:modification of tau_hat2}
\end{equation}
We have

\begin{align}
 & \widehat{\tau}_{\mathrm{cal}}-\tau\nonumber \\
= & \frac{1}{n}\sum_{i=1}^{n}\sum_{k=1}^{K}\left\{ \frac{A_{i}}{\pi_{n[k]}}-\frac{1-A_{i}}{1-\pi_{n[k]}}\right\} \1(B_{i}=k)\cdot Y_{i}-\tau\nonumber \\
 & \quad+\frac{1}{n}\sum_{i=1}^{n}(\widehat{w}_{i}-1)\sum_{k=1}^{K}\left\{ \frac{A_{i}}{\pi_{n[k]}}\left(Y_{i}-\overline{Y}_{1[k]}\right)-\frac{1-A_{i}}{1-\pi_{n[k]}}\left(Y_{i}-\overline{Y}_{0[k]}\right)\right\} \1(B_{i}=k)\nonumber \\
= & \frac{1}{n}\sum_{i=1}^{n}\sum_{k=1}^{K}\left\{ \frac{A_{i}}{\pi_{n[k]}}-\frac{1-A_{i}}{1-\pi_{n[k]}}\right\} \1(B_{i}=k)\cdot Y_{i}-\tau-\frac{1}{n}\sum_{i=1}^{n}\widehat{\bs{\lambda}}^{\trans}\bs{\Xi}_{i}\sum_{k=1}^{K}\widetilde{Y}_{i,[k]}\nonumber \\
= & \frac{1}{n}\sum_{i=1}^{n}\sum_{k=1}^{K}\left\{ \frac{A_{i}}{\pi_{n[k]}}-\frac{1-A_{i}}{1-\pi_{n[k]}}\right\} \1(B_{i}=k)\cdot Y_{i}-\tau-\frac{1}{n}\sum_{i=1}^{n}\sum_{k=1}^{K}\bs{\Xi}_{i}^{\trans}\widetilde{Y}_{i,[k]}\left\{ \frac{1}{n}\sum_{i=1}^{n}\bs{\Xi}_{i}\bs{\Xi}_{i}^{\trans}\right\} ^{+}\frac{1}{n}\sum_{i=1}^{n}\bs{\Xi}_{i}\nonumber \\
= & \sum_{k=1}^{K}\frac{1}{\pi_{n[k]}}\frac{1}{n}\sum_{i=1}^{n}A_{i}\1(B_{i}=k)\widetilde{Y}_{i}(1)-\sum_{k=1}^{K}\frac{1}{1-\pi_{n[k]}}\frac{1}{n}\sum_{i=1}^{n}\left(1-A_{i}\right)\1(B_{i}=k)\widetilde{Y}_{i}(0)\nonumber \\
 & \quad+\sum_{k=1}^{K}\e\left[Y_{i}(1)\mid B_{i}=k\right]\frac{1}{n}\sum_{i=1}^{n}\left(\1(B_{i}=k)-p_{[k]}\right)-\sum_{k=1}^{K}\e\left[Y_{i}(0)\mid B_{i}=k\right]\frac{1}{n}\sum_{i=1}^{n}\left(\1(B_{i}=k)-p_{[k]}\right)\nonumber \\
 & \quad-\sum_{k=1}^{K}\frac{1}{n}\sum_{i=1}^{n}\bs{\Xi}_{i,[k]}^{\trans}\widetilde{Y}_{i,[k]}\left\{ \frac{1}{n}\sum_{i=1}^{n}\bs{\Xi}_{i,[k]}\bs{\Xi}_{i,[k]}^{\trans}\right\} ^{+}\frac{1}{n}\sum_{i=1}^{n}\bs{\Xi}_{i,[k]}\nonumber \\
= & \underbrace{\frac{1}{n}\sum_{i=1}^{n}\sum_{k=1}^{K}\widetilde{Y}_{i,[k]}^{*}}_{R_{1}}+\underbrace{\left\{ -\frac{1}{n}\sum_{i=1}^{n}\sum_{k=1}^{K}\widehat{\bs{\beta}}_{[k]}^{\trans}\bs{\Xi}_{i,[k]}\right\} }_{R_{2}}+\underbrace{\left\{ \sum_{k=1}^{K}\e\left[Y_{i}(1)\mid B_{i}=k\right]\cdot\frac{1}{n}\sum_{i=1}^{n}\left(\1(B_{i}=k)-p_{[k]}\right)\right\} }_{R_{3}}\nonumber \\
 & \quad+\underbrace{\left\{ -\sum_{k=1}^{K}\e\left[Y_{i}(0)\mid B_{i}=k\right]\cdot\frac{1}{n}\sum_{i=1}^{n}\left(\1(B_{i}=k)-p_{[k]}\right)\right\} }_{R_{4}},\label{eq:main estimator decomposition}
\end{align}
where the first equality uses (\ref{eq:modification of tau_hat1})
and (\ref{eq:modification of tau_hat2}), the third one follows from
$\widehat{\bs{\lambda}}=\left\{ \frac{1}{n}\sum_{i=1}^{n}\bs{\Xi}_{i}\bs{\Xi}_{i}^{\trans}\right\} ^{+}\frac{1}{n}\sum_{i=1}^{n}\bs{\Xi}_{i}$,
and the fourth one follows from 
\[
\sum_{k=1}^{K}\bs{\Xi}_{i}^{\trans}\widetilde{Y}_{i,[k]}=\left(\bs{\Xi}_{i,[1]}^{\trans}\sum_{k=1}^{K}\widetilde{Y}_{i,[k]},\ldots,\bs{\Xi}_{i,[K]}^{\trans}\sum_{k=1}^{K}\widetilde{Y}_{i,[k]}\right)=\left(\bs{\Xi}_{i,[1]}^{\trans}\widetilde{Y}_{i,[1]},\ldots,\bs{\Xi}_{i,[K]}^{\trans}\widetilde{Y}_{i,[K]}\right)
\]
 and $\frac{1}{n}\sum_{i=1}^{n}A_{i}\1(B_{i}=k)=\pi_{n[k]}\frac{1}{n}\sum_{i=1}^{n}\1(B_{i}=k)$.

The last two terms ($R_{3}$ and $R_{4}$) in (\ref{eq:main estimator decomposition})
can be handled using the classical central limit theorem. We focus
on the first two terms ($R_{1}$ and $R_{2}$) in (\ref{eq:main estimator decomposition}).
To begin with, we show the consistency of $\widehat{\bs{\beta}}_{[k]}$.
It suffices to show the consistency results of $\frac{1}{n}\sum_{i=1}^{n}\bs{\Xi}_{i,[k]}\bs{\Xi}_{i,[k]}^{\trans}$
and $\frac{1}{n}\sum_{i=1}^{n}\bs{\Xi}_{i,[k]}^{\trans}\widetilde{Y}_{i,[k]}$
for every $k=1,\ldots,K$.

\textbf{The probability limit of $\left\{ \frac{1}{n}\sum_{i=1}^{n}\bs{\Xi}_{i,[k]}\bs{\Xi}_{i,[k]}^{\trans}\right\} ^{+}$.}
We have {\footnotesize
\begin{align}
 & \frac{1}{n}\sum_{i=1}^{n}\bs{\Xi}_{i,[k]}\bs{\Xi}_{i,[k]}^{\trans}\nonumber \\
= & \frac{1}{n}\sum_{i=1}^{n}A_{i}\1(B_{i}=k)\left\{ \bs{\xi}_{n}^{*}(\bs X_{i})-\overline{\bs{\xi}}_{n[k]}^{*}\right\} \left\{ \bs{\xi}_{n}^{*}(\bs X_{i})-\overline{\bs{\xi}}_{n[k]}^{*}\right\} ^{\trans}-\nonumber \\
 & \frac{2}{n}\sum_{i=1}^{n}\pi_{n[k]}A_{i}\1(B_{i}=k)\left\{ \bs{\xi}_{n}^{*}(\bs X_{i})-\overline{\bs{\xi}}_{n[k]}^{*}\right\} \left\{ \bs{\xi}_{n}^{*}(\bs X_{i})-\overline{\bs{\xi}}_{n[k]}^{*}\right\} ^{\trans}+\nonumber \\
 & \frac{1}{n}\sum_{i=1}^{n}\pi_{n[k]}^{2}\1(B_{i}=k)\left\{ \bs{\xi}_{n}^{*}(\bs X_{i})-\overline{\bs{\xi}}_{n[k]}^{*}\right\} \left\{ \bs{\xi}_{n}^{*}(\bs X_{i})-\overline{\bs{\xi}}_{n[k]}^{*}\right\} ^{\trans}+\nonumber \\
 & \underbrace{\frac{1}{n}\sum_{i=1}^{n}(A_{i}-\pi_{n[k]})^{2}\1(B_{i}=k)\left[\left(\bs{\xi}_{n}(\bs X_{i})-\overline{\bs{\xi}}_{n[k]}\right)\left(\bs{\xi}_{n}(\bs X_{i})-\overline{\bs{\xi}}_{n[k]}\right)^{\trans}-\left(\bs{\xi}_{n}^{*}(\bs X_{i})-\overline{\bs{\xi}}_{n[k]}^{*}\right)\left(\bs{\xi}_{n}^{*}(\bs X_{i})-\overline{\bs{\xi}}_{n[k]}^{*}\right)^{\trans}\right]}_{R_{5}},\label{eq:analyze empirical var}
\end{align}
}where $\overline{\bs{\xi}}_{n[k]}^{*}:=\frac{1}{n_{[k]}}\sum_{i=1}^{n}\1(B_{i}=k)\bs{\xi}_{n}^{*}(\bs X_{i})$
is the stratum-specific sample mean for $\bs{\xi}_{n}^{*}(\bs X_{i})$.
For every $1\leq k\leq K$ and $\bs{\alpha}\in\R^{d}$, by the Cauchy--Schwarz
inequality we have 
\begin{align*}
\bs{\alpha}^{\trans}R_{5}\bs{\alpha} & =\frac{1}{n}\sum_{i=1}^{n}\left(A_{i}-\pi_{n[k]}\right)^{2}\1(B_{i}=k)\left[\left(\bs{\alpha}^{\trans}\bs{\xi}_{n}(\bs X_{i})-\bs{\alpha}^{\trans}\overline{\bs{\xi}}_{n[k]}\right)^{2}-\left(\bs{\alpha}^{\trans}\bs{\xi}_{n}^{*}(\bs X_{i})-\bs{\alpha}^{\trans}\overline{\bs{\xi}}_{n[k]}^{*}\right)^{2}\right]\\
 & \leq\sqrt{\frac{1}{n}\sum_{i=1}^{n}\1(B_{i}=k)\left\{ \bs{\alpha}^{\trans}\bs{\xi}_{n}(\bs X_{i})-\bs{\alpha}^{\trans}\overline{\bs{\xi}}_{n[k]}-\left(\bs{\alpha}^{\trans}\bs{\xi}_{n}^{*}(\bs X_{i})-\bs{\alpha}^{\trans}\overline{\bs{\xi}}_{n[k]}^{*}\right)\right\} ^{2}}\times\\
 & \qquad\sqrt{\frac{1}{n}\sum_{i=1}^{n}\1(B_{i}=k)\left\{ \bs{\alpha}^{\trans}\bs{\xi}_{n}(\bs X_{i})-\bs{\alpha}^{\trans}\overline{\bs{\xi}}_{n[k]}+\left(\bs{\alpha}^{\trans}\bs{\xi}_{n}^{*}(\bs X_{i})-\bs{\alpha}^{\trans}\overline{\bs{\xi}}_{n[k]}^{*}\right)\right\} ^{2}}\\
 & \leq\left\Vert \bs{\alpha}\right\Vert ^{2}\sqrt{\underbrace{\frac{1}{n}\sum_{i=1}^{n}\1(B_{i}=k)\left\Vert \bs{\xi}_{n}(\bs X_{i})-\overline{\bs{\xi}}_{n[k]}-\left(\bs{\xi}_{n}^{*}(\bs X_{i})-\overline{\bs{\xi}}_{n[k]}^{*}\right)\right\Vert ^{2}}_{R_{6}}}\times\\
 & \qquad\sqrt{\underbrace{\frac{1}{n}\sum_{i=1}^{n}\1(B_{i}=k)\left\Vert \bs{\xi}_{n}(\bs X_{i})-\overline{\bs{\xi}}_{n[k]}+\left(\bs{\xi}_{n}^{*}(\bs X_{i})-\overline{\bs{\xi}}_{n[k]}^{*}\right)\right\Vert ^{2}}_{R_{7}}},
\end{align*}
which leads to 
\[
\left\Vert R_{5}\right\Vert \leq\sqrt{R_{6}}\times\sqrt{R_{7}}.
\]
We control $R_{6}$ and $R_{7}$ separately. By Assumption \ref{assu:conditions on =00005Cxi},
we have 
\begin{align}
\left\Vert \overline{\bs{\xi}}_{n[k]}-\overline{\bs{\xi}}_{n[k]}^{*}\right\Vert  & =\left\Vert \frac{1}{n_{[k]}}\sum_{i=1}^{n}\1(B_{i}=k)\left\{ \bs{\xi}_{n}(\bs X_{i})-\bs{\xi}_{n}^{*}(\bs X_{i})\right\} \right\Vert \nonumber \\
 & \leq\sqrt{\frac{1}{n_{[k]}}\sum_{i=1}^{n}\1(B_{i}=k)\left\Vert \bs{\xi}_{n}(\bs X_{i})-\bs{\xi}_{n}^{*}(\bs X_{i})\right\Vert ^{2}}=o_{P}(1)\label{eq:sample group mean-pop group mean}
\end{align}
for all $k=1,\ldots K$. Then it follows from Assumption \ref{assu:conditions on =00005Cxi},
weak law of large numbers and (\ref{eq:sample group mean-pop group mean})
that 
\begin{align}
R_{6} & \leq2\times\frac{1}{n}\sum_{i=1}^{n}\1(B_{i}=k)\left\Vert \bs{\xi}_{n}(\bs X_{i})-\bs{\xi}_{n}^{*}(\bs X_{i})\right\Vert ^{2}+2\times\frac{1}{n}\sum_{i=1}^{n}\1(B_{i}=k)\left\Vert \overline{\bs{\xi}}_{n[k]}-\overline{\bs{\xi}}_{n[k]}^{*}\right\Vert ^{2}=o_{P}(1).\label{eq:R6 in main theorem}
\end{align}
By Lemma \ref{lem:LLN for triangular array} and Assumption \ref{assu:conditions on =00005Cxi},
we have 
\[
\frac{1}{n}\sum_{i=1}^{n}\left\{ \1(B_{i}=k)\left\Vert \bs{\xi}_{n}^{*}(\bs X_{i})\right\Vert ^{2}-\e\left[\1(B_{i}=k)\left\Vert \bs{\xi}_{n}^{*}(\bs X_{i})\right\Vert ^{2}\right]\right\} =o_{P}(1)
\]
and 
\begin{align}
 & \overline{\bs{\xi}}_{n[k]}^{*}-\e\left[\bs{\xi}_{n}^{*}(\bs X_{i})\mid B_{i}=k\right]\nonumber \\
= & \frac{n}{n_{[k]}}\frac{1}{n}\sum_{i=1}^{n}\left\{ \1(B_{i}=k)\bs{\xi}_{n}^{*}(\bs X_{i})-\e\left[\1(B_{i}=k)\bs{\xi}_{n}^{*}(\bs X_{i})\right]\right\} +o_{P}(1)=o_{P}(1).\label{eq:main theorem xibar-true mean}
\end{align}
As a result, we have
\begin{align*}
R_{7} & \leq2\times R_{6}+8\times\frac{1}{n}\sum_{i=1}^{n}\1(B_{i}=k)\left\Vert \bs{\xi}_{n}^{*}(\bs X_{i})-\overline{\bs{\xi}}_{n[k]}^{*}\right\Vert ^{2}\\
 & \leq o_{P}(1)+\frac{16}{n}\sum_{i=1}^{n}\1(B_{i}=k)\left\Vert \bs{\xi}_{n}^{*}(\bs X_{i})\right\Vert ^{2}+\frac{16}{n}\sum_{i=1}^{n}\1(B_{i}=k)\left\Vert \overline{\bs{\xi}}_{n[k]}^{*}\right\Vert ^{2}\\
 & =o_{P}(1)+\left\{ O(1)+o_{P}(1)\right\} +O_{P}(1)\left\{ O(1)+o_{P}(1)\right\} ^{2}=O_{P}(1).
\end{align*}
 Now, we can obtain that $\left\Vert R_{5}\right\Vert \leq\sqrt{R_{6}}\times\sqrt{R_{7}}=o_{P}(1)$
and thus
\begin{align*}
 & \frac{1}{n}\sum_{i=1}^{n}\bs{\Xi}_{i,[k]}\bs{\Xi}_{i,[k]}^{\trans}\\
= & \frac{1}{n}\sum_{i=1}^{n}A_{i}\1(B_{i}=k)\left\{ \bs{\xi}_{n}^{*}(\bs X_{i})-\overline{\bs{\xi}}_{n[k]}^{*}\right\} \left\{ \bs{\xi}_{n}^{*}(\bs X_{i})-\overline{\bs{\xi}}_{n[k]}^{*}\right\} ^{\trans}\\
 & -\frac{2}{n}\sum_{i=1}^{n}\pi_{n[k]}A_{i}\1(B_{i}=k)\left\{ \bs{\xi}_{n}^{*}(\bs X_{i})-\overline{\bs{\xi}}_{n[k]}^{*}\right\} \left\{ \bs{\xi}_{n}^{*}(\bs X_{i})-\overline{\bs{\xi}}_{n[k]}^{*}\right\} ^{\trans}\\
 & +\frac{1}{n}\sum_{i=1}^{n}\pi_{n[k]}^{2}\1(B_{i}=k)\left\{ \bs{\xi}_{n}^{*}(\bs X_{i})-\overline{\bs{\xi}}_{n[k]}^{*}\right\} \left\{ \bs{\xi}_{n}^{*}(\bs X_{i})-\overline{\bs{\xi}}_{n[k]}^{*}\right\} ^{\trans}+o_{P}(1).
\end{align*}
Let $\mathcal{C}_{n}$ denote the $\sigma$-algebra generated by $(A^{(n)},B^{(n)})$.
Then by Assumptions \ref{assu:independent sampling} and \ref{assu:treatment assignment}
we have $(Y_{i}(1),Y_{i}(0),\bs X_{i}^{\trans})^{\trans},$ $i=1,\ldots,n$,
are independent conditional on $\mathcal{C}_{n}$.For the sake of
simplicity we let $\e_{\mathcal{C}_{n}}\left[\cdot\right]:=\e\left[\cdot\mid\mathcal{C}_{n}\right]$.

Note that 
\[
\e_{\mathcal{C}_{n}}\left[A_{i}\1(B_{i}=k)\bs{\xi}_{n}^{*}(\bs X_{i})\bs{\xi}_{n}^{*}(\bs X_{i})^{\trans}\right]=A_{i}\1(B_{i}=k)\e\left[\bs{\xi}_{n}^{*}(\bs X_{i})\bs{\xi}_{n}^{*}(\bs X_{i})^{\trans}\mid B_{i}=k\right]
\]
and 
\begin{align*}
 & \e_{\mathcal{C}_{n}}\left[\left\{ \left\Vert A_{i}\1(B_{i}=k)\bs{\xi}_{n}^{*}(\bs X_{i})\bs{\xi}_{n}^{*}(\bs X_{i})^{\trans}-\e_{\mathcal{C}}\left[A_{i}\1(B_{i}=k)\bs{\xi}_{n}^{*}(\bs X_{i})\bs{\xi}_{n}^{*}(\bs X_{i})^{\trans}\right]\right\Vert _{F}\right\} ^{1+\epsilon/2}\right]\\
\leq & \e_{\mathcal{C}_{n}}\left[\left\{ \left\Vert A_{i}\1(B_{i}=k)\bs{\xi}_{n}^{*}(\bs X_{i})\bs{\xi}_{n}^{*}(\bs X_{i})^{\trans}\right\Vert _{F}\right\} ^{1+\epsilon/2}\right]\\
\leq & \e\left[\left\{ \sqrt{\tr\left\{ \bs{\xi}_{n}^{*}(\bs X_{i})\bs{\xi}_{n}^{*}(\bs X_{i})^{\trans}\bs{\xi}_{n}^{*}(\bs X_{i})\bs{\xi}_{n}^{*}(\bs X_{i})^{\trans}\right\} }\right\} ^{1+\epsilon/2}\mid B_{i}=k\right]\\
= & \e\left[\left\Vert \bs{\xi}_{n}^{*}(\bs X_{i})\right\Vert ^{2+\epsilon}\mid B_{i}=k\right]=O(1),
\end{align*}
where the last step follows from Assumption \ref{assu:conditions on =00005Cxi}.
By Lemma \ref{lem:LLN for triangular array}, we have\footnote{View $A_{i}\1(B_{i}=k)\bs{\xi}_{n}^{*}(\bs X_{i})\bs{\xi}_{n}^{*}(\bs X_{i})^{\trans}$
as a vector and use the fact that 
\[
\left\Vert \vec\left\{ A_{i}\1(B_{i}=k)\bs{\xi}_{n}^{*}(\bs X_{i})\bs{\xi}_{n}^{*}(\bs X_{i})^{\trans}\right\} \right\Vert =\left\Vert A_{i}\1(B_{i}=k)\bs{\xi}_{n}^{*}(\bs X_{i})\bs{\xi}_{n}^{*}(\bs X_{i})^{\trans}\right\Vert _{F}.
\]
} 
\[
\frac{1}{n}\sum_{i=1}^{n}A_{i}\1(B_{i}=k)\bs{\xi}_{n}^{*}(\bs X_{i})\bs{\xi}_{n}^{*}(\bs X_{i})^{\trans}=\frac{1}{n}\sum_{i=1}^{n}A_{i}\1(B_{i}=k)\e\left[\bs{\xi}_{n}^{*}(\bs X_{i})\bs{\xi}_{n}^{*}(\bs X_{i})^{\trans}\mid B_{i}=k\right]+o_{P}(1).
\]
Similarly, we can also obtain that 
\[
\frac{1}{n}\sum_{i=1}^{n}A_{i}\1(B_{i}=k)\bs{\xi}_{n}^{*}(\bs X_{i})=\frac{1}{n}\sum_{i=1}^{n}A_{i}\1(B_{i}=k)\e\left[\bs{\xi}_{n}^{*}(\bs X_{i})\mid B_{i}=k\right]+o_{P}(1),
\]
\[
\frac{1}{n}\sum_{i=1}^{n}\1(B_{i}=k)\bs{\xi}_{n}^{*}(\bs X_{i})\bs{\xi}_{n}^{*}(\bs X_{i})^{\trans}=\frac{1}{n}\sum_{i=1}^{n}\1(B_{i}=k)\e\left[\bs{\xi}_{n}^{*}(\bs X_{i})\bs{\xi}_{n}^{*}(\bs X_{i})^{\trans}\mid B_{i}=k\right]+o_{P}(1),
\]
\[
\frac{1}{n}\sum_{i=1}^{n}\1(B_{i}=k)\bs{\xi}_{n}^{*}(\bs X_{i})=\frac{1}{n}\sum_{i=1}^{n}\1(B_{i}=k)\e\left[\bs{\xi}_{n}^{*}(\bs X_{i})\mid B_{i}=k\right]+o_{P}(1),
\]
and 
\begin{equation}
\overline{\bs{\xi}}_{n[k]}^{*}=\frac{n}{n_{[k]}}\frac{1}{n}\sum_{i=1}^{n}\1(B_{i}=k)\bs{\xi}_{n}^{*}(\bs X_{i})=\e\left[\bs{\xi}_{n}^{*}(\bs X_{i})\mid B_{i}=k\right]+o_{P}(1).\label{eq:lln for =00005Coverline=00007B=00005Cbs=00007B=00005Cxi=00007D=00007D_=00007Bn=00005Bk=00005D=00007D^=00007B*=00007D}
\end{equation}
Since $\pi_{n[k]}=O_{P}(1)$ and $\overline{\bs{\xi}}_{n[k]}^{*}=O_{P}(1)$,
we can derive that 
\begin{align*}
 & \frac{1}{n}\sum_{i=1}^{n}\bs{\Xi}_{i,[k]}\bs{\Xi}_{i,[k]}^{\trans}\\
= & o_{P}(1)+\frac{1}{n}\sum_{i=1}^{n}\left\{ A_{i}-2\pi_{n[k]}A_{i}+\pi_{n[k]}^{2}\right\} \1(B_{i}=k)\times\\
 & \quad\left\{ \e\left[\bs{\xi}_{n}^{*}(\bs X_{i})\bs{\xi}_{n}^{*}(\bs X_{i})^{\trans}\mid B_{i}=k\right]-\e\left[\bs{\xi}_{n}^{*}(\bs X_{i})\mid B_{i}=k\right]\e\left[\bs{\xi}_{n}^{*}(\bs X_{i})^{\trans}\mid B_{i}=k\right]\right\} \\
= & \frac{1}{n}\sum_{i=1}^{n}\e_{\mathcal{C}_{n}}\left[\bs{\Xi}_{i,[k]}^{*}\bs{\Xi}_{i,[k]}^{*\trans}\right]+o_{P}(1)\\
= & \frac{1}{n}\sum_{i=1}^{n}\left\{ A_{i}-\pi_{n[k]}\right\} ^{2}\1(B_{i}=k)\mathbf{\Sigma}_{[k]\widetilde{\bs{\xi}}_{n}^{*}\widetilde{\bs{\xi}}_{n}^{*}}+o_{P}(1)=\mathbf{\Sigma}_{[k]}^{\mathcal{C}_{n}}+o_{P}(1).
\end{align*}
By Lemma \ref{lem:LLN}, we have $\frac{1}{n}\sum_{i=1}^{n}A_{i}\1(B_{i}=k)=\pi_{[k]}p_{[k]}+o_{P}(1)$.
Then it follows from 
\[
\left\Vert \e\left[\mathbf{\Sigma}_{[k]\widetilde{\bs{\xi}}_{n}^{*}\widetilde{\bs{\xi}}_{n}^{*}}\mid B_{i}=k\right]\right\Vert \leq\e\left[\left\Vert \bs{\xi}_{n}^{*}(\bs X_{i})\right\Vert ^{2}\mid B_{i}=k\right]<\infty
\]
that
\[
\mathbf{\Sigma}_{[k]}^{\mathcal{C}_{n}}=\frac{1}{n}\sum_{i=1}^{n}\left\{ A_{i}-\pi_{n[k]}\right\} ^{2}\1(B_{i}=k)\mathbf{\Sigma}_{[k]\widetilde{\bs{\xi}}_{n}^{*}\widetilde{\bs{\xi}}_{n}^{*}}=\mathbf{\Sigma}_{[k]}+o_{P}(1).
\]
As a result, we have 
\begin{equation}
\frac{1}{n}\sum_{i=1}^{n}\bs{\Xi}_{i,[k]}\bs{\Xi}_{i,[k]}^{\trans}=\mathbf{\Sigma}_{[k]}+o_{P}(1).\label{eq:convergence of var each k}
\end{equation}
By Assumption \ref{assu:conditions on =00005Cxi} and the derivation
of the Moore-Penrose inverse, we have
\[
\left\Vert \mathbf{\Sigma}_{[k]\widetilde{\bs{\xi}}_{n}^{*}\widetilde{\bs{\xi}}_{n}^{*}}^{+}\right\Vert \leq c^{-1}
\]
for all $k=1,\ldots,K$ and $n\geq1$. Then we have 
\[
\left\Vert \mathbf{\Sigma}_{[k]}^{+}\right\Vert =\left\Vert \pi_{[k]}^{-1}(1-\pi_{[k]})^{-1}p_{[k]}^{-1}\mathbf{\Sigma}_{[k]\widetilde{\bs{\xi}}_{n}^{*}\widetilde{\bs{\xi}}_{n}^{*}}^{+}\right\Vert \leq\pi_{[k]}^{-1}(1-\pi_{[k]})^{-1}p_{[k]}^{-1}c^{-1}.
\]
By Theorem 3.3 in \citet{stewart1977Perturbation} and Assumption
\ref{assu:conditions on =00005Cxi}, we have 
\begin{align}
\left\Vert \left\{ \frac{1}{n}\sum_{i=1}^{n}\bs{\Xi}_{i,[k]}\bs{\Xi}_{i,[k]}^{\trans}\right\} ^{+}-\mathbf{\Sigma}_{[k]}^{+}\right\Vert  & \leq3\max\left\{ \left\Vert \left\{ \frac{1}{n}\sum_{i=1}^{n}\bs{\Xi}_{i,[k]}\bs{\Xi}_{i,[k]}^{\trans}\right\} ^{+}\right\Vert ^{2},\left\Vert \mathbf{\Sigma}_{[k]}^{+}\right\Vert ^{2}\right\} \left\Vert \frac{1}{n}\sum_{i=1}^{n}\bs{\Xi}_{i,[k]}\bs{\Xi}_{i,[k]}^{\trans}-\mathbf{\Sigma}_{[k]}^{+}\right\Vert \nonumber \\
 & =O_{P}(1)\times o_{P}(1)=o_{P}(1).\label{eq:consistency of inverse of var}
\end{align}

\textbf{The probability limit of $\frac{1}{n}\sum_{i=1}^{n}\bs{\Xi}_{i,[k]}^{\trans}\widetilde{Y}_{i,[k]}$.}
Note that for every $k=1,\ldots,K$ we have the following decomposition:
\begin{align}
 & \frac{1}{n}\sum_{i=1}^{n}\frac{1-\pi_{n[k]}}{\pi_{n[k]}}\1(B_{i}=k)A_{i}(\bs{\xi}_{n}(\bs X_{i})-\overline{\bs{\xi}}_{n[k]})(Y_{i}-\overline{Y}_{1[k]})\nonumber \\
= & \frac{1}{n}\sum_{i=1}^{n}\frac{1-\pi_{n[k]}}{\pi_{n[k]}}\1(B_{i}=k)A_{i}(\bs{\xi}_{n}(\bs X_{i})-\overline{\bs{\xi}}_{n[k]})Y_{i}(1)\nonumber \\
 & -\overline{Y}_{1[k]}\times\frac{1}{n}\sum_{i=1}^{n}\frac{1-\pi_{n[k]}}{\pi_{n[k]}}\1(B_{i}=k)A_{i}(\bs{\xi}_{n}(\bs X_{i})-\overline{\bs{\xi}}_{n[k]})\nonumber \\
= & \frac{1}{n}\sum_{i=1}^{n}\frac{1-\pi_{n[k]}}{\pi_{n[k]}}\1(B_{i}=k)A_{i}(\bs{\xi}_{n}^{*}(\bs X_{i})-\overline{\bs{\xi}}_{n[k]}^{*})(Y_{i}(1)-\overline{Y}_{1[k]})\nonumber \\
 & +\underbrace{\frac{1}{n}\sum_{i=1}^{n}\frac{1-\pi_{n[k]}}{\pi_{n[k]}}\1(B_{i}=k)A_{i}\left\{ \bs{\xi}_{n}(\bs X_{i})-\overline{\bs{\xi}}_{n[k]}-\left(\bs{\xi}_{n}^{*}(\bs X_{i})-\overline{\bs{\xi}}_{n[k]}^{*}\right)\right\} Y_{i}}_{R_{8}}\nonumber \\
 & -\underbrace{\overline{Y}_{1[k]}\times\frac{1}{n}\sum_{i=1}^{n}\frac{1-\pi_{n[k]}}{\pi_{n[k]}}\1(B_{i}=k)A_{i}\left\{ \bs{\xi}_{n}(\bs X_{i})-\overline{\bs{\xi}}_{n[k]}-(\bs{\xi}_{n}^{*}(\bs X_{i})-\overline{\bs{\xi}}_{n[k]}^{*})\right\} }_{R_{9}}.\label{eq:analyze empirical cov}
\end{align}
We control $R_{8}$ and $R_{9}$ separately. We have 
\begin{align*}
\left|R_{8}\right| & \leq\frac{1-\pi_{n[k]}}{\pi_{n[k]}}\sqrt{\frac{1}{n}\sum_{i=1}^{n}Y_{i}^{2}}\times\sqrt{\frac{1}{n}\sum_{i=1}^{n}\1(B_{i}=k)\left\Vert \bs{\xi}_{n}(\bs X_{i})-\overline{\bs{\xi}}_{n[k]}-\left(\bs{\xi}_{n}^{*}(\bs X_{i})-\overline{\bs{\xi}}_{n[k]}^{*}\right)\right\Vert ^{2}}\\
 & =O_{P}(1)\times\sqrt{R_{6}}=o_{P}(1)
\end{align*}
for every $k=1,\ldots,K$, where the first step follows from the Cauchy--Schwarz
inequality, the second step follows from Assumption \ref{assu:independent sampling}
and the last step follows from $R_{6}=o_{P}(1)$. Similarly, by noting
that $\overline{Y}_{1[k]}=\frac{1}{n_{1[k]}}\sum_{i=1}^{n}\1(B_{i}=k)A_{i}Y_{i}(1)=O_{P}(1)$
we have 
\[
\left|R_{9}\right|\leq O_{P}(1)\times\sqrt{R_{6}}=o_{P}(1).
\]
Combining the results for $R_{8}$ and $R_{9}$ together, we have
\begin{align}
 & \frac{1}{n}\sum_{i=1}^{n}\frac{1-\pi_{n[k]}}{\pi_{n[k]}}\1(B_{i}=k)A_{i}(\bs{\xi}_{n}(\bs X_{i})-\overline{\bs{\xi}}_{n[k]})(Y_{i}-\overline{Y}_{1[k]})\nonumber \\
= & \frac{1}{n}\sum_{i=1}^{n}\frac{1-\pi_{n[k]}}{\pi_{n[k]}}\1(B_{i}=k)A_{i}(\bs{\xi}_{n}^{*}(\bs X_{i})-\overline{\bs{\xi}}_{n[k]}^{*})(Y_{i}(1)-\overline{Y}_{1[k]})+o_{P}(1)\label{eq:mid step for cov limit}
\end{align}
Note that we have
\[
\e_{\mathcal{C}_{n}}\left[\1(B_{i}=k)A_{i}\bs{\xi}_{n}^{*}(\bs X_{i})Y_{i}(1)\right]=\1(B_{i}=k)A_{i}\e\left[\bs{\xi}_{n}^{*}(\bs X_{i})Y_{i}(1)\mid B_{i}=k\right]
\]
and 
\begin{align*}
 & \e_{\mathcal{C}_{n}}\left[\left\Vert \1(B_{i}=k)A_{i}\bs{\xi}_{n}^{*}(\bs X_{i})Y_{i}(1)\right\Vert ^{1+\widetilde{\epsilon}}\right]\leq\e\left[\left\Vert \bs{\xi}_{n}^{*}(\bs X_{i})Y_{i}(1)\right\Vert ^{1+\widetilde{\epsilon}}\mid B_{i}=k\right]\\
\leq & \left\{ \e\left[\left\Vert \bs{\xi}_{n}^{*}(\bs X_{i})\right\Vert ^{(1+\widetilde{\epsilon})\times\frac{2}{1-\widetilde{\epsilon}}}\mid B_{i}=k\right]\right\} ^{\frac{1-\widetilde{\epsilon}}{2}}\left\{ \e\left[\left|Y_{i}(1)\right|^{(1+\widetilde{\epsilon})\times\frac{2}{1+\widetilde{\epsilon}}}\mid B_{i}=k\right]\right\} ^{\frac{1+\widetilde{\epsilon}}{2}}\\
= & \left\{ \e\left[\left\Vert \bs{\xi}_{n}^{*}(\bs X_{i})\right\Vert ^{2+\epsilon}\mid B_{i}=k\right]\right\} ^{\frac{1}{2+\epsilon/2}}\left\{ \e\left[\left|Y_{i}(1)\right|^{2}\mid B_{i}=k\right]\right\} ^{\frac{1+\epsilon/2}{2+\epsilon/2}}\\
= & O(1),
\end{align*}
where the second inequality follows from the Hölder inequality, the
first equality follows from $\widetilde{\epsilon}:=\frac{\epsilon/4}{1+\epsilon/4}$
and the last step follows from Assumptions \ref{assu:independent sampling}
and \ref{assu:conditions on =00005Cxi}. By Lemma \ref{lem:LLN for triangular array}
we have 
\[
\frac{1}{n}\sum_{i=1}^{n}\1(B_{i}=k)A_{i}\bs{\xi}_{n}^{*}(\bs X_{i})Y_{i}(1)=\frac{1}{n}\sum_{i=1}^{n}\1(B_{i}=k)A_{i}\e\left[\bs{\xi}_{n}^{*}(\bs X_{i})Y_{i}(1)\mid B_{i}=k\right]+o_{P}(1).
\]
This, combined with $\frac{1}{n}\sum_{i=1}^{n}\1(B_{i}=k)A_{i}=\pi_{[k]}p_{[k]}+o_{P}(1)$
gives that 
\[
\frac{1}{n}\sum_{i=1}^{n}\1(B_{i}=k)A_{i}\bs{\xi}_{n}^{*}(\bs X_{i})Y_{i}(1)=\pi_{[k]}p_{[k]}\e\left[\bs{\xi}_{n}^{*}(\bs X_{i})Y_{i}(1)\mid B_{i}=k\right]+o_{P}(1).
\]
Similarly we have 
\[
\frac{1}{n}\sum_{i=1}^{n}\1(B_{i}=k)A_{i}Y_{i}(1)=\pi_{[k]}p_{[k]}\e\left[Y_{i}(1)\mid B_{i}=k\right]+o_{P}(1),
\]
\[
\frac{1}{n}\sum_{i=1}^{n}\1(B_{i}=k)A_{i}\bs{\xi}_{n}^{*}(\bs X_{i})=\pi_{[k]}p_{[k]}\e\left[\bs{\xi}_{n}^{*}(\bs X_{i})\mid B_{i}=k\right]+o_{P}(1)
\]
and 
\[
\overline{Y}_{1[k]}=\frac{1}{n_{1[k]}}\sum_{i=1}^{n}\1(B_{i}=k)A_{i}Y_{i}(1)=\e\left[Y_{i}(1)\mid B_{i}=k\right]+o_{P}(1).
\]
Recalling (\ref{eq:lln for =00005Coverline=00007B=00005Cbs=00007B=00005Cxi=00007D=00007D_=00007Bn=00005Bk=00005D=00007D^=00007B*=00007D})
and (\ref{eq:mid step for cov limit}), we can obtain that 
\begin{align*}
 & \frac{1}{n}\sum_{i=1}^{n}\frac{1-\pi_{n[k]}}{\pi_{n[k]}}\1(B_{i}=k)A_{i}(\bs{\xi}_{n}(\bs X_{i})-\overline{\bs{\xi}}_{n[k]})(Y_{i}-\overline{Y}_{1[k]})\\
= & \frac{1-\pi_{[k]}}{\pi_{[k]}}\pi_{[k]}p_{[k]}\e\left[\left\{ \bs{\xi}_{n}^{*}(\bs X_{i})-\e\left[\bs{\xi}_{n}^{*}(\bs X_{i})\mid B_{i}=k\right]\right\} \left\{ Y_{i}(1)-\e\left[Y_{i}(1)\mid B_{i}=k\right]\right\} \mid B_{i}=k\right]+o_{P}(1)\\
= & (1-\pi_{[k]})p_{[k]}\e\left[\widetilde{\bs{\xi}}_{n}^{*}(\bs X_{i})\widetilde{Y}_{i}(1)\mid B_{i}=k\right]+o_{P}(1)=(1-\pi_{[k]})p_{[k]}\mathbf{\Sigma}_{[k]\widetilde{\bs{\xi}}_{n}^{*}\widetilde{Y}(1)}+o_{P}(1).
\end{align*}
By a similar argument, we can also obtain that 
\[
\frac{1}{n}\sum_{i=1}^{n}\frac{\pi_{n[k]}}{1-\pi_{n[k]}}\1(B_{i}=k)(1-A_{i})(\bs{\xi}_{n}(\bs X_{i})-\overline{\bs{\xi}}_{n[k]})(Y_{i}-\overline{Y}_{0[k]})=\pi_{[k]}p_{[k]}\mathbf{\Sigma}_{[k]\widetilde{\bs{\xi}}_{n}^{*}\widetilde{Y}(0)}+o_{P}(1).
\]
Then we have 
\begin{align}
 & \frac{1}{n}\sum_{i=1}^{n}\bs{\Xi}_{i,[k]}^{\trans}\widetilde{Y}_{i,[k]}\nonumber \\
= & \frac{1}{n}\sum_{i=1}^{n}\frac{1-\pi_{n[k]}}{\pi_{n[k]}}\1(B_{i}=k)A_{i}(\bs{\xi}_{n}(\bs X_{i})-\overline{\bs{\xi}}_{n[k]})(Y_{i}-\overline{Y}_{1[k]})\nonumber \\
 & +\frac{1}{n}\sum_{i=1}^{n}\frac{\pi_{n[k]}}{1-\pi_{n[k]}}\1(B_{i}=k)(1-A_{i})(\bs{\xi}_{n}(\bs X_{i})-\overline{\bs{\xi}}_{n[k]})(Y_{i}-\overline{Y}_{0[k]})\nonumber \\
= & (1-\pi_{[k]})p_{[k]}\mathbf{\Sigma}_{[k]\widetilde{\bs{\xi}}_{n}^{*}\widetilde{Y}(1)}+\pi_{[k]}p_{[k]}\mathbf{\Sigma}_{[k]\widetilde{\bs{\xi}}_{n}^{*}\widetilde{Y}(0)}+o_{P}(1).\label{eq:consistency of cov}
\end{align}

Combining the results for $\left\{ \frac{1}{n}\sum_{i=1}^{n}\bs{\Xi}_{i,[k]}\bs{\Xi}_{i,[k]}^{\trans}\right\} ^{+}$
and $\frac{1}{n}\sum_{i=1}^{n}\bs{\Xi}_{i,[k]}^{\trans}\widetilde{Y}_{i,[k]}$
gives that 
\begin{align}
\widehat{\bs{\beta}}_{[k]} & =\left\{ \frac{1}{n}\sum_{i=1}^{n}\bs{\Xi}_{i,[k]}\bs{\Xi}_{i,[k]}^{\trans}\right\} ^{+}\frac{1}{n}\sum_{i=1}^{n}\bs{\Xi}_{i,[k]}^{\trans}\widetilde{Y}_{i,[k]}\nonumber \\
 & =\left\{ \mathbf{\Sigma}_{[k]}^{+}+o_{P}(1)\right\} \left\{ (1-\pi_{[k]})p_{[k]}\mathbf{\Sigma}_{[k]\widetilde{\bs{\xi}}_{n}^{*}\widetilde{Y}(1)}+\pi_{[k]}p_{[k]}\mathbf{\Sigma}_{[k]\widetilde{\bs{\xi}}_{n}^{*}\widetilde{Y}(0)}+o_{P}(1)\right\} \nonumber \\
 & =\bs{\beta}_{[k]}^{*}+o_{P}(1).\label{eq:converge of beta_hat}
\end{align}
Besides, by Assumptions \ref{assu:independent sampling}-\ref{assu:conditions on =00005Cxi}
and the Cauchy--Schwarz inequality, it holds that $\left\Vert \mathbf{\Sigma}_{[k]\widetilde{\bs{\xi}}_{n}^{*}\widetilde{Y}(a)}\right\Vert =O(1)$
for all $a=0,1$. Then we have 
\begin{equation}
\left\Vert \bs{\beta}_{[k]}^{*}\right\Vert \leq\left\Vert \mathbf{\Sigma}_{[k]}^{+}\right\Vert \left\Vert (1-\pi_{[k]})p_{[k]}\mathbf{\Sigma}_{[k]\widetilde{\bs{\xi}}_{n}^{*}\widetilde{Y}(1)}+\pi_{[k]}p_{[k]}\mathbf{\Sigma}_{[k]\widetilde{\bs{\xi}}_{n}^{*}\widetilde{Y}(0)}\right\Vert =O(1).\label{eq:beta* is bounded}
\end{equation}

\textbf{The analysis of $\frac{1}{n}\sum_{i=1}^{n}\bs{\Xi}_{i,[k]}$.}
Let $\bs{\xi}_{[k]}^{*}:=\e\left[\bs{\xi}_{n}^{*}(\bs X_{i})\mid B_{i}=k\right]$.
By $\frac{1}{n}\sum_{i=1}^{n}A_{i}\1(B_{i}=k)=\pi_{n[k]}\frac{1}{n}\sum_{i=1}^{n}\1(B_{i}=k)$
and $\1(B_{i}=k)\widetilde{\bs{\xi}}_{n}^{*}(\bs X_{i})=\1(B_{i}=k)\left\{ \bs{\xi}_{n}^{*}(\bs X_{i})-\bs{\xi}_{[k]}^{*}\right\} $,
we have 
\begin{align}
 & \frac{1}{n}\sum_{i=1}^{n}\bs{\Xi}_{i,[k]}\nonumber \\
= & \frac{1}{n}\sum_{i=1}^{n}\left\{ A_{i}-\pi_{n[k]}\right\} \1(B_{i}=k)\left\{ \bs{\xi}_{n}(\bs X_{i})-\overline{\bs{\xi}}_{n[k]}\right\} \nonumber \\
= & \frac{1}{n}\sum_{i=1}^{n}\left\{ A_{i}-\pi_{n[k]}\right\} \1(B_{i}=k)\left\{ \bs{\xi}_{n}(\bs X_{i})-\bs{\xi}_{[k]}^{*}\right\} \nonumber \\
= & \frac{1}{n}\sum_{i=1}^{n}\left\{ A_{i}-\pi_{n[k]}\right\} \1(B_{i}=k)\widetilde{\bs{\xi}}_{n}^{*}(\bs X_{i})+\frac{1}{n}\sum_{i=1}^{n}\left\{ A_{i}-\pi_{n[k]}\right\} \1(B_{i}=k)\left\{ \bs{\xi}_{n}(\bs X_{i})-\bs{\xi}_{n}^{*}(\bs X_{i})\right\} \nonumber \\
= & \frac{1}{n}\sum_{i=1}^{n}\bs{\Xi}_{i,[k]}^{*}+\frac{1}{n}\sum_{i=1}^{n}\left\{ A_{i}-\pi_{n[k]}\right\} \1(B_{i}=k)\left\{ \bs{\xi}_{n}(\bs X_{i})-\bs{\xi}_{n}^{*}(\bs X_{i})\right\} \label{eq:sample ave of xi,k}
\end{align}
By Assumption \ref{assu:conditions on =00005Cxi}, we have 
\begin{align*}
 & \frac{1}{n}\sum_{i=1}^{n}\left\{ A_{i}-\pi_{n[k]}\right\} \1(B_{i}=k)\left\{ \bs{\xi}_{n}(\bs X_{i})-\bs{\xi}_{n}^{*}(\bs X_{i})\right\} \\
= & \frac{n_{[k]}}{n}\pi_{n[k]}(1-\pi_{n[k]})\times\\
 & \quad\left\{ \frac{1}{n_{1[k]}}\sum_{i=1}^{n}A_{i}\1(B_{i}=k)\left\{ \bs{\xi}_{n}(\bs X_{i})-\bs{\xi}_{n}^{*}(\bs X_{i})\right\} -\frac{1}{n_{0[k]}}\sum_{i=1}^{n}\left(1-A_{i}\right)\1(B_{i}=k)\left\{ \bs{\xi}_{n}(\bs X_{i})-\bs{\xi}_{n}^{*}(\bs X_{i})\right\} \right\} \\
= & o_{P}(n^{-1/2}).
\end{align*}
Thus we have 
\begin{align*}
\frac{1}{n}\sum_{i=1}^{n}\bs{\Xi}_{i,[k]} & =\frac{1}{n}\sum_{i=1}^{n}\bs{\Xi}_{i,[k]}^{*}+o_{P}(n^{-1/2}).
\end{align*}
Note that $\bs{\Xi}_{i,[k]}^{*}$, $i=1,\ldots,n$, are independent
conditional on $\mathcal{C}_{n}$,
\begin{align*}
\e_{\mathcal{C}_{n}}\left[\bs{\Xi}_{i,[k]}^{*}\right] & =\e_{\mathcal{C}_{n}}\left[\left\{ A_{i}-\pi_{n[k]}\right\} \1(B_{i}=k)\widetilde{\bs{\xi}}_{n}^{*}(\bs X_{i})\right]\\
 & =\left\{ A_{i}-\pi_{n[k]}\right\} \1(B_{i}=k)\e\left[\widetilde{\bs{\xi}}_{n}^{*}(\bs X_{i})\mid B_{i}=k\right]=0
\end{align*}
and 
\[
\e_{\mathcal{C}_{n}}\left[\left\Vert \bs{\Xi}_{i,[k]}^{*}\right\Vert ^{2}\right]=\e_{\mathcal{C}_{n}}\left[\left\Vert \left\{ A_{i}-\pi_{n[k]}\right\} \1(B_{i}=k)\widetilde{\bs{\xi}}_{n}^{*}(\bs X_{i})\right\Vert ^{2}\right]\leq\e\left[\left\Vert \widetilde{\bs{\xi}}_{n}^{*}(\bs X_{i})\right\Vert ^{2}\mid B_{i}=k\right]=O(1),
\]
where the last step follows from Assumption \ref{assu:conditions on =00005Cxi}.
Applying Lemma \ref{lem:LLN for triangular array} gives that 
\begin{align*}
\frac{1}{n}\sum_{i=1}^{n}\bs{\Xi}_{i,[k]}^{*} & =O_{P}(n^{-1/2}).
\end{align*}
As a result, we can conclude that 
\begin{equation}
\frac{1}{n}\sum_{i=1}^{n}\bs{\Xi}_{i,[k]}=\frac{1}{n}\sum_{i=1}^{n}\bs{\Xi}_{i,[k]}^{*}+o_{P}(n^{-1/2})=O_{P}(n^{-1/2}).\label{eq: mean of xi is root-n}
\end{equation}
This, combined with (\ref{eq:converge of beta_hat}) and (\ref{eq:beta* is bounded})
yields that
\begin{align*}
\widehat{\bs{\beta}}_{[k]}^{\trans}\frac{1}{n}\sum_{i=1}^{n}\bs{\Xi}_{i,[k]} & =\left\{ \bs{\beta}_{[k]}^{*}+o_{P}(1)\right\} ^{\trans}\left\{ \frac{1}{n}\sum_{i=1}^{n}\bs{\Xi}_{i,[k]}^{*}+o_{P}(n^{-1/2})\right\} =\bs{\beta}_{[k]}^{*\trans}\frac{1}{n}\sum_{i=1}^{n}\bs{\Xi}_{i,[k]}^{*}+o_{P}(n^{-1/2})
\end{align*}

As a result, the first two terms ($R_{1}$ and $R_{2}$) in (\ref{eq:main estimator decomposition})
satisfy
\begin{align}
R_{1}+R_{2} & =\frac{1}{n}\sum_{i=1}^{n}\left\{ \sum_{k=1}^{K}\left(\widetilde{Y}_{i,[k]}^{*}-\bs{\beta}_{[k]}^{*\trans}\bs{\Xi}_{i,[k]}^{*}\right)\right\} +o_{P}(n^{-1/2})\label{eq:R_1+R_2}
\end{align}
Note that $\sum_{k=1}^{K}\left(\widetilde{Y}_{i,[k]}^{*}-\bs{\beta}_{[k]}^{*\trans}\bs{\Xi}_{i,[k]}^{*}\right)$
are independent conditional on $\mathcal{C}_{n}$, we intend to apply
Lemma~\ref{lem:clt for triangular array}. By direct calculation
we have 
\begin{align*}
 & \e_{\mathcal{C}_{n}}\left[\sum_{k=1}^{K}\left(\widetilde{Y}_{i,[k]}^{*}-\bs{\beta}_{[k]}^{*\trans}\bs{\Xi}_{i,[k]}^{*}\right)\right]\\
= & \sum_{k=1}^{K}\1(B_{i}=k)\left\{ \frac{A_{i}}{\pi_{n[k]}}\e\left[\widetilde{Y}_{i}(1)\mid B_{i}=k\right]-\frac{1-A_{i}}{1-\pi_{n[k]}}\e\left[\widetilde{Y}_{i}(0)\mid B_{i}=k\right]\right\} \\
 & \qquad-\sum_{k=1}^{K}\1(B_{i}=k)(A_{i}-\pi_{n[k]})\bs{\beta}_{[k]}^{*\trans}\e\left[\widetilde{\bs{\xi}}_{n}^{*}(\bs X_{i})\mid B_{i}=k\right]\\
= & 0,
\end{align*}
\begin{align}
 & \var\left(\sum_{k=1}^{K}\left(\widetilde{Y}_{i,[k]}^{*}-\bs{\beta}_{[k]}^{*\trans}\bs{\Xi}_{i,[k]}^{*}\right)\mid\mathcal{C}_{n}\right)=\sum_{k=1}^{K}\e\left[\left(\widetilde{Y}_{i,[k]}^{*}-\bs{\beta}_{[k]}^{*\trans}\bs{\Xi}_{i,[k]}^{*}\right)^{2}\mid\mathcal{C}_{n}\right]\nonumber \\
= & \sum_{k=1}^{K}\1(B_{i}=k)\left\{ \frac{A_{i}}{\pi_{n[k]}^{2}}\e\left[\widetilde{Y}_{i}^{2}(1)\mid B_{i}=k\right]+\frac{1-A_{i}}{(1-\pi_{n[k]})^{2}}\e\left[\widetilde{Y}_{i}^{2}(0)\mid B_{i}=k\right]\right\} \nonumber \\
 & -2\sum_{k=1}^{K}\1(B_{i}=k)A_{i}\frac{1-\pi_{n[k]}}{\pi_{n[k]}}\e\left[\widetilde{Y}_{i}(1)\bs{\beta}_{[k]}^{*\trans}\widetilde{\bs{\xi}}_{n}^{*}(\bs X_{i})\mid B_{i}=k\right]\nonumber \\
 & -2\sum_{k=1}^{K}\1(B_{i}=k)(1-A_{i})\frac{\pi_{n[k]}}{1-\pi_{n[k]}}\e\left[\widetilde{Y}_{i}(0)\bs{\beta}_{[k]}^{*\trans}\widetilde{\bs{\xi}}_{n}^{*}(\bs X_{i})\mid B_{i}=k\right]\nonumber \\
 & +\sum_{k=1}^{K}\1(B_{i}=k)(A_{i}-\pi_{n[k]})^{2}\e\left[(\bs{\beta}_{[k]}^{*\trans}\widetilde{\bs{\xi}}_{n}^{*}(\bs X_{i}))^{2}\mid B_{i}=k\right]\nonumber \\
= & \sum_{k=1}^{K}\1(B_{i}=k)\left\{ \frac{A_{i}}{\pi_{n[k]}^{2}}\e\left[\widetilde{Y}_{i}^{2}(1)\mid B_{i}=k\right]+\frac{1-A_{i}}{(1-\pi_{n[k]})^{2}}\e\left[\widetilde{Y}_{i}^{2}(0)\mid B_{i}=k\right]\right\} \nonumber \\
 & -2\sum_{k=1}^{K}\1(B_{i}=k)\bs{\beta}_{[k]}^{*\trans}\left\{ A_{i}\frac{1-\pi_{n[k]}}{\pi_{n[k]}}\mathbf{\Sigma}_{[k]\widetilde{\bs{\xi}}_{n}^{*}\widetilde{Y}(1)}+(1-A_{i})\frac{\pi_{n[k]}}{1-\pi_{n[k]}}\mathbf{\Sigma}_{[k]\widetilde{\bs{\xi}}_{n}^{*}\widetilde{Y}(0)}\right\} \nonumber \\
 & +\sum_{k=1}^{K}\1(B_{i}=k)(A_{i}-\pi_{n[k]})^{2}\bs{\beta}_{[k]}^{*\trans}\mathbf{\Sigma}_{[k]\widetilde{\bs{\xi}}_{n}^{*}\widetilde{\bs{\xi}}_{n}^{*}}\bs{\beta}_{[k]}^{*}\label{eq:var(sum delta)}
\end{align}
and 
\begin{align*}
\sigma_{n}^{2}(\mathcal{C}_{n}) & :=\var\left(\sum_{i=1}^{n}\sum_{k=1}^{K}\left(\widetilde{Y}_{i,[k]}^{*}-\bs{\beta}_{[k]}^{*\trans}\bs{\Xi}_{i,[k]}^{*}\right)\mid\mathcal{C}_{n}\right)=\sum_{i=1}^{n}\var\left(\sum_{k=1}^{K}\left(\widetilde{Y}_{i,[k]}^{*}-\bs{\beta}_{[k]}^{*\trans}\bs{\Xi}_{i,[k]}^{*}\right)\mid\mathcal{C}_{n}\right)\\
 & =\sum_{i=1}^{n}\sum_{k=1}^{K}\1(B_{i}=k)\left\{ \frac{A_{i}}{\pi_{n[k]}^{2}}\e\left[\widetilde{Y}_{i}^{2}(1)\mid B_{i}=k\right]+\frac{1-A_{i}}{(1-\pi_{n[k]})^{2}}\e\left[\widetilde{Y}_{i}^{2}(0)\mid B_{i}=k\right]\right\} \\
 & \quad-2\sum_{i=1}^{n}\sum_{k=1}^{K}\1(B_{i}=k)\bs{\beta}_{[k]}^{*\trans}\left\{ A_{i}\frac{1-\pi_{n[k]}}{\pi_{n[k]}}\mathbf{\Sigma}_{[k]\widetilde{\bs{\xi}}_{n}^{*}\widetilde{Y}(1)}+(1-A_{i})\frac{\pi_{n[k]}}{1-\pi_{n[k]}}\mathbf{\Sigma}_{[k]\widetilde{\bs{\xi}}_{n}^{*}\widetilde{Y}(0)}\right\} \\
 & \quad+\sum_{i=1}^{n}\sum_{k=1}^{K}\1(B_{i}=k)(A_{i}-\pi_{n[k]})^{2}\bs{\beta}_{[k]}^{*\trans}\mathbf{\Sigma}_{[k]\widetilde{\bs{\xi}}_{n}^{*}\widetilde{\bs{\xi}}_{n}^{*}}\bs{\beta}_{[k]}^{*}.
\end{align*}
By Lemma \ref{lem:LLN} and $\pi_{n[k]}\tp\pi_{[k]}$, we have 
\begin{align}
\frac{\sigma_{n}^{2}(\mathcal{C}_{n})}{n} & =\sum_{k=1}^{K}\left\{ \frac{p_{[k]}}{\pi_{[k]}}\e\left[\widetilde{Y}_{i}^{2}(1)\mid B_{i}=k\right]+\frac{p_{[k]}}{1-\pi_{[k]}}\e\left[\widetilde{Y}_{i}^{2}(0)\mid B_{i}=k\right]\right\} \nonumber \\
 & \quad-2\sum_{k=1}^{K}p_{[k]}\bs{\beta}_{[k]}^{*\trans}\left\{ (1-\pi_{[k]})\mathbf{\Sigma}_{[k]\widetilde{\bs{\xi}}_{n}^{*}\widetilde{Y}(1)}+\pi_{[k]}\mathbf{\Sigma}_{[k]\widetilde{\bs{\xi}}_{n}^{*}\widetilde{Y}(0)}\right\} \nonumber \\
 & \quad+\sum_{k=1}^{K}p_{[k]}\pi_{[k]}(1-\pi_{[k]})\bs{\beta}_{[k]}^{*\trans}\mathbf{\Sigma}_{[k]\widetilde{\bs{\xi}}_{n}^{*}\widetilde{\bs{\xi}}_{n}^{*}}\bs{\beta}_{[k]}^{*}+o_{P}(1)\nonumber \\
 & =\varsigma_{\widetilde{Y}}^{2}-\varsigma_{\widetilde{Y}\mid\widetilde{\bs{\xi}}_{n}^{*}}^{2}+o_{P}(1)\label{eq:limit of sigma_n}\\
 & =O_{P}(1),\nonumber 
\end{align}
where the second equality follows from the definition of $\bs{\beta}_{[k]}^{*}$.
According to (\ref{eq:var(sum delta)}), (\ref{eq:limit of sigma_n})
and Assumption \ref{assu:conditions on =00005Cxi}, Condition (\ref{eq: positive variance in clt})
in Lemma \ref{lem:clt for triangular array} is satisfied.

It remains to verify the Lindeberg condition. By Loève’s $c_{r}$
inequality (\citealp[Theorem 9.32]{davidson2021Stochastic}) and Assumptions
\ref{assu:independent sampling} and \ref{assu:conditions on =00005Cxi},
we have 
\begin{align*}
 & \e\left[\left|\sum_{k=1}^{K}\left(\widetilde{Y}_{i,[k]}^{*}-\bs{\beta}_{[k]}^{*\trans}\bs{\Xi}_{i,[k]}^{*}\right)\right|^{2+\epsilon}\mid\mathcal{C}_{n}\right]\\
\leq & \frac{(3K)^{1+\epsilon}}{\min_{1\leq k\leq K}\left\{ \pi_{n[k]}^{2+\epsilon},(1-\pi_{n[k]})^{2+\epsilon}\right\} }\times\\
 & \quad\left\{ 2\max_{a\in\{0,1\}}\e\left[\left|\widetilde{Y}_{i}(a)\right|^{2+\epsilon}\mid B_{i}=k\right]+\left\Vert \bs{\beta}_{[k]}^{*}\right\Vert ^{2+\epsilon}\e\left[\left\Vert \widetilde{\bs{\xi}}_{n}^{*}(\bs X_{i})\right\Vert ^{2+\epsilon}\mid B_{i}=k\right]\right\} \\
= & O_{P}(1).
\end{align*}
Then, the Lindeberg condition in Lemma \ref{lem:clt for triangular array}
follows from 
\begin{align*}
 & \frac{1}{\sigma_{n}^{2}(\mathcal{C}_{n})}\sum_{i=1}^{n}\e\left[\left\{ \sum_{k=1}^{K}\left(\widetilde{Y}_{i,[k]}^{*}-\bs{\beta}_{[k]}^{*\trans}\bs{\Xi}_{i,[k]}^{*}\right)\right\} ^{2}\1\left\{ \left|\sum_{k=1}^{K}\left(\widetilde{Y}_{i,[k]}^{*}-\bs{\beta}_{[k]}^{*\trans}\bs{\Xi}_{i,[k]}^{*}\right)\right|\geq\delta\sigma_{n}(\mathcal{C}_{n})\right\} \mid\mathcal{C}_{n}\right]\\
\leq & \frac{1}{\sigma_{n}^{2}(\mathcal{C}_{n})}\sum_{i=1}^{n}\frac{\e\left[\left|\sum_{k=1}^{K}\left(\widetilde{Y}_{i,[k]}^{*}-\bs{\beta}_{[k]}^{*\trans}\bs{\Xi}_{i,[k]}^{*}\right)\right|^{2+\epsilon}\mid\mathcal{C}_{n}\right]}{\left\{ \delta\sigma_{n}(\mathcal{C}_{n})\right\} ^{\epsilon}}=\frac{1}{\varsigma_{\widetilde{Y}}^{2}-\varsigma_{\widetilde{Y}\mid\widetilde{\bs{\xi}}_{n}^{*}}^{2}+o_{P}(1)}\frac{O_{P}(1)}{O_{P}(n^{\epsilon/2})}\frac{1}{\delta^{\ep}}=o_{P}(1)
\end{align*}
for any $\delta>0$, where the second inequality follows from (\ref{eq:limit of sigma_n}).
Applying Lemma~\ref{lem:clt for triangular array} we have 
\[
\e\left[\left|\e\left[\exp\left\{ \mathrm{i}t\frac{\sum_{i=1}^{n}\sum_{k=1}^{K}\left(\widetilde{Y}_{i,[k]}^{*}-\bs{\beta}_{[k]}^{*\trans}\bs{\Xi}_{i,[k]}^{*}\right)}{\sigma_{n}(\mathcal{C}_{n})}\right\} \mid\mathcal{C}_{n}\right]-\exp\left(-t^{2}/2\right)\right|\right]\to0
\]
as $n\to\infty$ for all $t\in\R$.

Recalling the definitions of $R_{3}$ and $R_{4}$ in (\ref{eq:main estimator decomposition}),
by the classical central limit theorem we have 
\[
\e\left[\left|\e\left[\exp\left\{ \mathrm{i}s\frac{\sqrt{n}(R_{3}+R_{4})}{\varsigma_{H}}\right\} \mid\mathcal{C}_{n}\right]-\exp\left(-s^{2}/2\right)\right|\right]\to0
\]
for any $s\in\R$. Note that $R_{3}+R_{4}$ is measurable with respect
to $\mathcal{C}_{n}$, by Lemma \ref{lem:joint convergence} we have
\[
\left(\frac{\sum_{i=1}^{n}\sum_{k=1}^{K}\left(\widetilde{Y}_{i,[k]}^{*}-\bs{\beta}_{[k]}^{*\trans}\bs{\Xi}_{i,[k]}^{*}\right)}{\sigma_{n}(\mathcal{C}_{n})},\frac{\sqrt{n}(R_{3}+R_{4})}{\varsigma_{H}}\right)^{\trans}\tod N\left(\begin{pmatrix}0\\
0
\end{pmatrix},\begin{pmatrix}1 & 0\\
0 & 1
\end{pmatrix}\right).
\]
Combining (\ref{eq:main estimator decomposition}) and (\ref{eq:R_1+R_2}),
we have 
\begin{align*}
 & \frac{\sqrt{n}\left(\widehat{\tau}_{\mathrm{cal}}-\tau\right)}{\sqrt{\varsigma_{H}^{2}+\varsigma_{\widetilde{Y}}^{2}-\varsigma_{\widetilde{Y}\mid\widetilde{\bs{\xi}}_{n}^{*}}^{2}}}=\frac{\sqrt{n}\left(R_{1}+R_{2}\right)}{\sqrt{\varsigma_{H}^{2}+\varsigma_{\widetilde{Y}}^{2}-\varsigma_{\widetilde{Y}\mid\widetilde{\bs{\xi}}_{n}^{*}}^{2}}}+\frac{\sqrt{n}\left(R_{3}+R_{4}\right)}{\sqrt{\varsigma_{H}^{2}+\varsigma_{\widetilde{Y}}^{2}-\varsigma_{\widetilde{Y}\mid\widetilde{\bs{\xi}}_{n}^{*}}^{2}}}\\
= & \frac{\sigma_{n}(\mathcal{C}_{n})}{\sqrt{n}\sqrt{\varsigma_{H}^{2}+\varsigma_{\widetilde{Y}}^{2}-\varsigma_{\widetilde{Y}\mid\widetilde{\bs{\xi}}_{n}^{*}}^{2}}}\frac{\sum_{i=1}^{n}\sum_{k=1}^{K}\left(\widetilde{Y}_{i,[k]}^{*}-\bs{\beta}_{[k]}^{*\trans}\bs{\Xi}_{i,[k]}^{*}\right)}{\sigma_{n}(\mathcal{C}_{n})}+\frac{\varsigma_{H}}{\sqrt{\varsigma_{H}^{2}+\varsigma_{\widetilde{Y}}^{2}-\varsigma_{\widetilde{Y}\mid\widetilde{\bs{\xi}}_{n}^{*}}^{2}}}\frac{\sqrt{n}(R_{3}+R_{4})}{\varsigma_{H}}+o_{P}(1)\\
= & \frac{\sqrt{\varsigma_{\widetilde{Y}}^{2}-\varsigma_{\widetilde{Y}\mid\widetilde{\bs{\xi}}_{n}^{*}}^{2}}+o_{P}(1)}{\sqrt{\varsigma_{H}^{2}+\varsigma_{\widetilde{Y}}^{2}-\varsigma_{\widetilde{Y}\mid\widetilde{\bs{\xi}}_{n}^{*}}^{2}}}\frac{\sum_{i=1}^{n}\sum_{k=1}^{K}\left(\widetilde{Y}_{i,[k]}^{*}-\bs{\beta}_{[k]}^{*\trans}\bs{\Xi}_{i,[k]}^{*}\right)}{\sigma_{n}(\mathcal{C}_{n})}+\frac{\varsigma_{H}}{\sqrt{\varsigma_{H}^{2}+\varsigma_{\widetilde{Y}}^{2}-\varsigma_{\widetilde{Y}\mid\widetilde{\bs{\xi}}_{n}^{*}}^{2}}}\frac{\sqrt{n}(R_{3}+R_{4})}{\varsigma_{H}}+o_{P}(1)\\
= & \frac{\sqrt{\varsigma_{\widetilde{Y}}^{2}-\varsigma_{\widetilde{Y}\mid\widetilde{\bs{\xi}}_{n}^{*}}^{2}}}{\sqrt{\varsigma_{H}^{2}+\varsigma_{\widetilde{Y}}^{2}-\varsigma_{\widetilde{Y}\mid\widetilde{\bs{\xi}}_{n}^{*}}^{2}}}\frac{\sum_{i=1}^{n}\sum_{k=1}^{K}\left(\widetilde{Y}_{i,[k]}^{*}-\bs{\beta}_{[k]}^{*\trans}\bs{\Xi}_{i,[k]}^{*}\right)}{\sigma_{n}(\mathcal{C}_{n})}+\frac{\varsigma_{H}}{\sqrt{\varsigma_{H}^{2}+\varsigma_{\widetilde{Y}}^{2}-\varsigma_{\widetilde{Y}\mid\widetilde{\bs{\xi}}_{n}^{*}}^{2}}}\frac{\sqrt{n}(R_{3}+R_{4})}{\varsigma_{H}}+o_{P}(1)\\
\tod & N(0,1),
\end{align*}
where the third equality follows from (\ref{eq:limit of sigma_n}),
the fourth one follows from $\sum_{i=1}^{n}\sum_{k=1}^{K}(\widetilde{Y}_{i,[k]}^{*}-\bs{\beta}_{[k]}^{*\trans}\bs{\Xi}_{i,[k]}^{*})/\sigma_{n}(\mathcal{C}_{n})$
is uniformly tight and the last step follows from Lemma \ref{lem:covergence of combin of normal}
and Slutsky's theorem. The asymptotic normality of $\widehat{\tau}_{\mathrm{cal}}$
is now established.

\textbf{Consistency of the variance estimator.} The consistency of
$\widehat{\varsigma}_{\widetilde{Y}}^{2}$ and $\widehat{\varsigma}_{H}^{2}$
can be established in a manner similar to that in \citet{bugni2018Inference,ma2022Regression},
so we omit the details here. It remains to show the consistency of
$\widehat{\varsigma}_{\widetilde{Y}\mid\widetilde{\bs{\xi}}_{n}^{*}}^{2}$.

From (\ref{eq:consistency of inverse of var}), we have 
\[
\widehat{\mathbf{\Sigma}}_{[k]}^{+}=\pi_{[k]}^{-1}(1-\pi_{[k]})^{-1}p_{[k]}^{-1}\mathbf{\Sigma}_{[k]\widetilde{\bs{\xi}}_{n}^{*}\widetilde{\bs{\xi}}_{n}^{*}}^{+}+o_{P}(1),\ \forall1\leq k\leq K.
\]
Similarly, from (\ref{eq:consistency of cov}), we have 
\[
\widehat{\mathbf{\Gamma}}_{[k]}=(1-\pi_{[k]})p_{[k]}\mathbf{\Sigma}_{[k]\widetilde{\bs{\xi}}_{n}^{*}\widetilde{Y}(1)}+\pi_{[k]}p_{[k]}\mathbf{\Sigma}_{[k]\widetilde{\bs{\xi}}_{n}^{*}\widetilde{Y}(0)}+o_{P}(1),\ \forall1\leq k\leq K.
\]
This, combined with $\left|n_{[k]}/(n_{[k]}-\widehat{r}_{[k]}-1)-1\right|=o_{P}(1)$
gives that 
\[
\widehat{\varsigma}_{\widetilde{Y}\mid\widetilde{\bs{\xi}}_{n}^{*}}^{2}=\sum_{k=1}^{K}\frac{n_{[k]}}{n_{[k]}-\widehat{r}_{[k]}-1}\widehat{\mathbf{\Gamma}}_{[k]}^{\trans}\widehat{\mathbf{\Sigma}}_{[k]}^{+}\widehat{\mathbf{\Gamma}}_{[k]}=\varsigma_{\widetilde{Y}\mid\widetilde{\bs{\xi}}_{n}^{*}}^{2}+o_{P}(1).
\]
The proof is completed.

\textbf{Establishing the semiparametric efficiency.} If for every
$k=1,\ldots,K$, there exists a non-stochastic vector $\bs{\alpha}_{k}$
such that 
\[
\left\{ \sqrt{\frac{1-\pi_{[k]}}{\pi_{[k]}}}\widetilde{h}_{1[k]}^{*}(\bs X)+\sqrt{\frac{\pi_{[k]}}{1-\pi_{[k]}}}\widetilde{h}_{0[k]}^{*}(\bs X)\right\} \1(B=k)=\bs{\alpha}_{k}^{\trans}\left\{ \bs{\xi}_{n}^{*}(\bs X)-\e\left[\bs{\xi}_{n}^{*}(\bs X)\mid B=k\right]\right\} \1(B=k),
\]
then 
\begin{align*}
 & \varsigma_{\widetilde{Y}\mid\widetilde{\bs{\xi}}_{n}^{*}}^{2}\\
= & \sum_{k=1}^{K}p_{[k]}\bs{\alpha}_{k}^{\trans}\e\left[\widetilde{\bs{\xi}}_{n}^{*}(\bs X_{i})\widetilde{\bs{\xi}}_{n}^{*}(\bs X_{i})^{\trans}\mid B_{i}=k\right]\left\{ \var\left(\widetilde{\bs{\xi}}_{n}^{*}(\bs X_{i})\mid B_{i}=k\right)\right\} ^{+}\e\left[\widetilde{\bs{\xi}}_{n}^{*}(\bs X_{i})\widetilde{\bs{\xi}}_{n}^{*}(\bs X_{i})^{\trans}\mid B_{i}=k\right]\bs{\alpha}_{k}\\
= & \sum_{k=1}^{K}p_{[k]}\bs{\alpha}_{k}^{\trans}\e\left[\widetilde{\bs{\xi}}_{n}^{*}(\bs X_{i})\widetilde{\bs{\xi}}_{n}^{*}(\bs X_{i})^{\trans}\1(B_{i}=k)\mid B_{i}=k\right]\bs{\alpha}_{k}\\
= & \sum_{k=1}^{K}p_{[k]}\e\left[\left\{ \sqrt{\frac{1-\pi_{[k]}}{\pi_{[k]}}}\widetilde{h}_{1[k]}^{*}(\bs X)+\sqrt{\frac{\pi_{[k]}}{1-\pi_{[k]}}}\widetilde{h}_{0[k]}^{*}(\bs X)\right\} ^{2}\mid B=k\right]\\
= & \varsigma_{\widetilde{Y}\mid\widetilde{h}^{*}}^{2}
\end{align*}
and thus
\[
\sqrt{n}\left(\widehat{\tau}_{\mathrm{cal}}-\tau\right)\tod\varsigma_{H}^{2}+\varsigma_{\widetilde{Y}}^{2}-\varsigma_{\widetilde{Y}\mid\widetilde{h}^{*}}^{2}.
\]
The semiparametric efficient bound in \citet[Theorem 3.1]{rafi2023Efficient}
is given by 
\begin{align*}
V^{*} & =\e\left[\frac{\var\left(Y(1)\mid\bs X,B\right)}{\pi_{[B]}}+\frac{\var\left(Y(0)\mid\bs X,B\right)}{1-\pi_{[B]}}+\left\{ h_{1[B]}^{*}(\bs X)-h_{0[B]}^{*}(\bs X)-\tau\right\} ^{2}\right].
\end{align*}
By direct calculation, we have 
\begin{align*}
 & \e\left[\frac{\var\left(Y(1)\mid\bs X,B\right)}{\pi_{[B]}}+\frac{\var\left(Y(0)\mid\bs X,B\right)}{1-\pi_{[B]}}\right]\\
= & \sum_{k=1}^{K}p_{[k]}\e\left[\frac{\var\left(Y(1)\mid\bs X,B=k\right)}{\pi_{[k]}}+\frac{\var\left(Y(0)\mid\bs X,B=k\right)}{1-\pi_{[k]}}\mid B=k\right]\\
= & \sum_{k=1}^{K}\frac{p_{[k]}}{\pi_{[k]}}\left\{ \var\left(Y(1)\mid B=k\right)-\var\left(\e\left[Y(1)\mid\bs X,B=k\right]\right)\right\} \\
 & +\sum_{k=1}^{K}\frac{p_{[k]}}{1-\pi_{[k]}}\left\{ \var\left(Y(0)\mid B=k\right)-\var\left(\e\left[Y(0)\mid\bs X,B=k\right]\right)\right\} \\
= & \varsigma_{\widetilde{Y}}^{2}-\sum_{k=1}^{K}\frac{p_{[k]}}{\pi_{[k]}}\var\left(\widetilde{h}_{1[k]}^{*}(\bs X)\right)-\sum_{k=1}^{K}\frac{p_{[k]}}{1-\pi_{[k]}}\var\left(\widetilde{h}_{0[k]}^{*}(\bs X)\right)
\end{align*}
and 
\begin{align*}
 & \e\left[\left\{ h_{1[B]}^{*}(\bs X)-h_{0[B]}^{*}(\bs X)-\tau\right\} ^{2}\right]\\
= & \sum_{k=1}^{K}p_{[k]}\e\left[\left\{ h_{1[k]}^{*}(\bs X)-h_{0[k]}^{*}(\bs X)-\tau\right\} ^{2}\mid B=k\right]\\
= & \sum_{k=1}^{K}p_{[k]}\e\left[\left\{ \widetilde{h}_{1[k]}^{*}(\bs X)-\widetilde{h}_{0[k]}^{*}(\bs X)+\e\left[h_{1[k]}^{*}(\bs X)\mid B=k\right]-\e\left[h_{0[k]}^{*}(\bs X)\mid B=k\right]-\tau\right\} ^{2}\mid B=k\right]\\
= & \sum_{k=1}^{K}p_{[k]}\e\left[\left\{ \widetilde{h}_{1[k]}^{*}(\bs X)-\widetilde{h}_{0[k]}^{*}(\bs X)\right\} ^{2}\mid B=k\right]\\
 & +\sum_{k=1}^{K}p_{[k]}\left\{ \e\left[h_{1[k]}^{*}(\bs X)\mid B=k\right]-\e\left[h_{0[k]}^{*}(\bs X)\mid B=k\right]-\tau\right\} ^{2}\\
= & \sum_{k=1}^{K}p_{[k]}\left\{ \var\left(\widetilde{h}_{1[k]}^{*}(\bs X)\right)+\var\left(\widetilde{h}_{0[k]}^{*}(\bs X)\right)-2\e\left[\widetilde{h}_{1[k]}^{*}(\bs X)\widetilde{h}_{0[k]}^{*}(\bs X)\mid B=k\right]\right\} \\
 & +\sum_{k=1}^{K}p_{[k]}\left\{ \e\left[Y(1)\mid B=k\right]-\e\left[Y(0)\mid B=k\right]-\tau\right\} ^{2}\\
= & \sum_{k=1}^{K}p_{[k]}\left\{ \var\left(\widetilde{h}_{1[k]}^{*}(\bs X)\right)+\var\left(\widetilde{h}_{0[k]}^{*}(\bs X)\right)-2\e\left[\widetilde{h}_{1[k]}^{*}(\bs X)\widetilde{h}_{0[k]}^{*}(\bs X)\mid B=k\right]\right\} +\varsigma_{H}^{2}.
\end{align*}
Thus, we have 
\begin{align*}
V^{*} & =\e\left[\frac{\var\left(Y(1)\mid\bs X,B\right)}{\pi_{[B]}}+\frac{\var\left(Y(0)\mid\bs X,B\right)}{1-\pi_{[B]}}+\left\{ h_{1[B]}^{*}(\bs X)-h_{0[B]}^{*}(\bs X)-\tau\right\} ^{2}\right].\\
 & =\varsigma_{\widetilde{Y}}^{2}-\sum_{k=1}^{K}\frac{p_{[k]}}{\pi_{[k]}}\var\left(\widetilde{h}_{1[k]}^{*}(\bs X)\right)-\sum_{k=1}^{K}\frac{p_{[k]}}{1-\pi_{[k]}}\var\left(\widetilde{h}_{0[k]}^{*}(\bs X)\right)\\
 & \qquad+\sum_{k=1}^{K}p_{[k]}\left\{ \var\left(\widetilde{h}_{1[k]}^{*}(\bs X)\right)+\var\left(\widetilde{h}_{0[k]}^{*}(\bs X)\right)-2\e\left[\widetilde{h}_{1[k]}^{*}(\bs X)\widetilde{h}_{0[k]}^{*}(\bs X)\mid B=k\right]\right\} +\varsigma_{H}^{2}\\
 & =\varsigma_{\widetilde{Y}}^{2}+\varsigma_{H}^{2}-\varsigma_{\widetilde{Y}\mid\widetilde{h}^{*}}^{2}.
\end{align*}
The proof is completed.$\hfill\qedsymbol$

\textbf{Proof of Theorem \ref{thm:efficiency comparison thm}.} By
Theorem \ref{thm:asymptotic properties of tau_cal} we have 
\[
\frac{\sqrt{n}\left(\widehat{\tau}_{\mathrm{cal}}(\bs{\xi}_{n})-\tau\right)}{\sqrt{\varsigma_{H}^{2}+\varsigma_{\widetilde{Y}}^{2}-\varsigma_{\widetilde{Y}\mid\widetilde{\bs{\xi}}_{n}^{*}}^{2}}}\tod N(0,1).
\]
Note that Assumption \ref{assu:conditions on =00005Cxi} still holds
if we replace $\bs{\xi}_{n}$ by $\mathbf{\Lambda}\bs{\xi}_{n}$.
It follows from Theorem \ref{thm:asymptotic properties of tau_cal}
that 
\[
\frac{\sqrt{n}\left(\widehat{\tau}_{\mathrm{cal}}(\mathbf{\Lambda}\bs{\xi}_{n})-\tau\right)}{\sqrt{\varsigma_{H}^{2}+\varsigma_{\widetilde{Y}}^{2}-\varsigma_{\widetilde{Y}\mid\mathbf{\Lambda}\widetilde{\bs{\xi}}_{n}^{*}}^{2}}}\tod N(0,1).
\]
Now, it remains to show that $\varsigma_{\widetilde{Y}\mid\widetilde{\bs{\xi}}_{n}^{*}}^{2}\geq\varsigma_{\widetilde{Y}\mid\mathbf{\Lambda}\widetilde{\bs{\xi}}_{n}^{*}}^{2}$.
Note that we have 
\[
\mathbf{\Sigma}_{[k]\mathbf{\Lambda}\widetilde{\bs{\xi}}_{n}^{*}\widetilde{Y}(a)}=\mathbf{\Lambda}\mathbf{\Sigma}_{[k]\widetilde{\bs{\xi}}_{n}^{*}\widetilde{Y}(a)}\text{ and }\mathbf{\Sigma}_{[k]\mathbf{\Lambda}\widetilde{\bs{\xi}}_{n}^{*}\mathbf{\Lambda}\widetilde{\bs{\xi}}_{n}^{*}}^{+}=\left\{ \mathbf{\Lambda}\mathbf{\Sigma}_{[k]\widetilde{\bs{\xi}}_{n}^{*}\widetilde{\bs{\xi}}_{n}^{*}}\mathbf{\Lambda}^{\trans}\right\} ^{+}.
\]
As a result, 
\begin{align*}
 & \varsigma_{\widetilde{Y}\mid\widetilde{\bs{\xi}}_{n}^{*}}^{2}-\varsigma_{\widetilde{Y}\mid\mathbf{\Lambda}\widetilde{\bs{\xi}}_{n}^{*}}^{2}\\
= & \sum_{k=1}^{K}\frac{p_{[k]}}{\pi_{[k]}(1-\pi_{[k]})}\left\{ (1-\pi_{[k]})\mathbf{\Sigma}_{[k]\widetilde{\bs{\xi}}_{n}^{*}\widetilde{Y}(1)}+\pi_{[k]}\mathbf{\Sigma}_{[k]\widetilde{\bs{\xi}}_{n}^{*}\widetilde{Y}(0)}\right\} ^{\trans}\times\\
 & \qquad\left\{ \mathbf{\Sigma}_{[k]\widetilde{\bs{\xi}}_{n}^{*}\widetilde{\bs{\xi}}_{n}^{*}}^{+}-\mathbf{\Lambda}^{\trans}\left(\mathbf{\Lambda}\mathbf{\Sigma}_{[k]\widetilde{\bs{\xi}}_{n}^{*}\widetilde{\bs{\xi}}_{n}^{*}}\mathbf{\Lambda}^{\trans}\right)^{+}\mathbf{\Lambda}\right\} \times\\
 & \qquad\left\{ (1-\pi_{[k]})\mathbf{\Sigma}_{[k]\widetilde{\bs{\xi}}_{n}^{*}\widetilde{Y}(1)}+\pi_{[k]}\mathbf{\Sigma}_{[k]\widetilde{\bs{\xi}}_{n}^{*}\widetilde{Y}(0)}\right\} .
\end{align*}

We claim that 
\begin{equation}
\mathbf{\Sigma}_{[k]\widetilde{\bs{\xi}}_{n}^{*}\widetilde{Y}(a)}\in\text{range}(\mathbf{\Sigma}_{[k]\widetilde{\bs{\xi}}_{n}^{*}\widetilde{\bs{\xi}}_{n}^{*}}),\ \forall k=1,\ldots,K\text{ and }a\in\{0,1\}.\label{eq: cov in range=00007Bvar=00007D}
\end{equation}
Since $\R^{d}=\text{range}(\mathbf{\Sigma}_{[k]\widetilde{\bs{\xi}}_{n}^{*}\widetilde{\bs{\xi}}_{n}^{*}})\oplus\text{kernel}(\mathbf{\Sigma}_{[k]\widetilde{\bs{\xi}}_{n}^{*}\widetilde{\bs{\xi}}_{n}^{*}})$,
it suffices to show that 
\[
\bs v^{\trans}\mathbf{\Sigma}_{[k]\widetilde{\bs{\xi}}_{n}^{*}\widetilde{Y}(a)}=0,\ \forall\bs v\in\text{kernel}(\mathbf{\Sigma}_{[k]\widetilde{\bs{\xi}}_{n}^{*}\widetilde{\bs{\xi}}_{n}^{*}}).
\]
For any $\bs v\in\text{kernel}(\mathbf{\Sigma}_{[k]\widetilde{\bs{\xi}}_{n}^{*}\widetilde{\bs{\xi}}_{n}^{*}})$,
we have 
\[
\mathbf{\Sigma}_{[k]\widetilde{\bs{\xi}}_{n}^{*}\widetilde{\bs{\xi}}_{n}^{*}}\bs v=0\Longrightarrow\e\left[\bs v^{\trans}\widetilde{\bs{\xi}}_{n}^{*}(\bs X_{i})\widetilde{\bs{\xi}}_{n}^{*}(\bs X_{i})^{\trans}\bs v\mid B_{i}=k\right]=0\Longrightarrow\e\left[\left(\bs v^{\trans}\widetilde{\bs{\xi}}_{n}^{*}(\bs X_{i})\right)^{2}\1(B_{i}=k)\right]=0,
\]
which further implies that $\bs v^{\trans}\widetilde{\bs{\xi}}_{n}^{*}(\bs X_{i})\1(B_{i}=k)=0$
a.s. As a result, it holds that 
\[
\bs v^{\trans}\mathbf{\Sigma}_{[k]\widetilde{\bs{\xi}}_{n}^{*}\widetilde{Y}(a)}=\e\left[\bs v^{\trans}\widetilde{\bs{\xi}}_{n}^{*}(\bs X_{i})\widetilde{Y}_{i}(a)\mid B_{i}=k\right]=\frac{1}{p_{[k]}}\e\left[\bs v^{\trans}\widetilde{\bs{\xi}}_{n}^{*}(\bs X_{i})\1(B_{i}=k)\widetilde{Y}_{i}(a)\right]=0.
\]
Thus, (\ref{eq: cov in range=00007Bvar=00007D}) holds.

Now, to show $\varsigma_{\widetilde{Y}\mid\widetilde{\bs{\xi}}_{n}^{*}}^{2}\geq\varsigma_{\widetilde{Y}\mid\mathbf{\Lambda}\widetilde{\bs{\xi}}_{n}^{*}}^{2}$,
it suffices to show that $\bs x^{\trans}\left\{ \mathbf{\Sigma}_{[k]\widetilde{\bs{\xi}}_{n}^{*}\widetilde{\bs{\xi}}_{n}^{*}}^{+}-\mathbf{\Lambda}^{\trans}\left(\mathbf{\Lambda}\mathbf{\Sigma}_{[k]\widetilde{\bs{\xi}}_{n}^{*}\widetilde{\bs{\xi}}_{n}^{*}}\mathbf{\Lambda}^{\trans}\right)^{+}\mathbf{\Lambda}\right\} \bs x\geq0$
for all $\bs x\in\text{range}(\mathbf{\Sigma}_{[k]\widetilde{\bs{\xi}}_{n}^{*}\widetilde{\bs{\xi}}_{n}^{*}})$.
Since $\mathbf{\Sigma}_{[k]\widetilde{\bs{\xi}}_{n}^{*}\widetilde{\bs{\xi}}_{n}^{*}}$
is positive semi-definite, there exists a positive semi-definite matrix
$\mathbf{\Sigma}_{[k]\widetilde{\bs{\xi}}_{n}^{*}\widetilde{\bs{\xi}}_{n}^{*}}^{1/2}$
\[
\mathbf{\Sigma}_{[k]\widetilde{\bs{\xi}}_{n}^{*}\widetilde{\bs{\xi}}_{n}^{*}}=\mathbf{\Sigma}_{[k]\widetilde{\bs{\xi}}_{n}^{*}\widetilde{\bs{\xi}}_{n}^{*}}^{1/2}\mathbf{\Sigma}_{[k]\widetilde{\bs{\xi}}_{n}^{*}\widetilde{\bs{\xi}}_{n}^{*}}^{1/2}.
\]
For any $\bs x\in\text{range}(\mathbf{\Sigma}_{[k]\widetilde{\bs{\xi}}_{n}^{*}\widetilde{\bs{\xi}}_{n}^{*}})$,
there exists a vector $\bs w$ such that $\bs x=\mathbf{\Sigma}_{[k]\widetilde{\bs{\xi}}_{n}^{*}\widetilde{\bs{\xi}}_{n}^{*}}^{1/2}\mathbf{\Sigma}_{[k]\widetilde{\bs{\xi}}_{n}^{*}\widetilde{\bs{\xi}}_{n}^{*}}^{1/2}\bs w$.
Thus, 
\[
\bs x^{\trans}\mathbf{\Sigma}_{[k]\widetilde{\bs{\xi}}_{n}^{*}\widetilde{\bs{\xi}}_{n}^{*}}^{+}\bs x=\bs w^{\trans}\mathbf{\Sigma}_{[k]\widetilde{\bs{\xi}}_{n}^{*}\widetilde{\bs{\xi}}_{n}^{*}}^{1/2}\mathbf{\Sigma}_{[k]\widetilde{\bs{\xi}}_{n}^{*}\widetilde{\bs{\xi}}_{n}^{*}}^{1/2}\mathbf{\Sigma}_{[k]\widetilde{\bs{\xi}}_{n}^{*}\widetilde{\bs{\xi}}_{n}^{*}}^{+}\mathbf{\Sigma}_{[k]\widetilde{\bs{\xi}}_{n}^{*}\widetilde{\bs{\xi}}_{n}^{*}}^{1/2}\mathbf{\Sigma}_{[k]\widetilde{\bs{\xi}}_{n}^{*}\widetilde{\bs{\xi}}_{n}^{*}}^{1/2}\bs w.
\]
By singular value decomposition, we have $\mathbf{\Sigma}_{[k]\widetilde{\bs{\xi}}_{n}^{*}\widetilde{\bs{\xi}}_{n}^{*}}^{+}=\left(\mathbf{\Sigma}_{[k]\widetilde{\bs{\xi}}_{n}^{*}\widetilde{\bs{\xi}}_{n}^{*}}^{1/2}\right)^{+}\left(\mathbf{\Sigma}_{[k]\widetilde{\bs{\xi}}_{n}^{*}\widetilde{\bs{\xi}}_{n}^{*}}^{1/2}\right)^{+}$,
which leads to 
\begin{align*}
\bs x^{\trans}\mathbf{\Sigma}_{[k]\widetilde{\bs{\xi}}_{n}^{*}\widetilde{\bs{\xi}}_{n}^{*}}^{+}\bs x & =\bs w^{\trans}\mathbf{\Sigma}_{[k]\widetilde{\bs{\xi}}_{n}^{*}\widetilde{\bs{\xi}}_{n}^{*}}^{1/2}\mathbf{\Sigma}_{[k]\widetilde{\bs{\xi}}_{n}^{*}\widetilde{\bs{\xi}}_{n}^{*}}^{1/2}\left(\mathbf{\Sigma}_{[k]\widetilde{\bs{\xi}}_{n}^{*}\widetilde{\bs{\xi}}_{n}^{*}}^{1/2}\right)^{+}\left(\mathbf{\Sigma}_{[k]\widetilde{\bs{\xi}}_{n}^{*}\widetilde{\bs{\xi}}_{n}^{*}}^{1/2}\right)^{+}\mathbf{\Sigma}_{[k]\widetilde{\bs{\xi}}_{n}^{*}\widetilde{\bs{\xi}}_{n}^{*}}^{1/2}\mathbf{\Sigma}_{[k]\widetilde{\bs{\xi}}_{n}^{*}\widetilde{\bs{\xi}}_{n}^{*}}^{1/2}\bs w\\
 & =\bs w^{\trans}\mathbf{\Sigma}_{[k]\widetilde{\bs{\xi}}_{n}^{*}\widetilde{\bs{\xi}}_{n}^{*}}^{1/2}\mathbf{P}^{2}\mathbf{\Sigma}_{[k]\widetilde{\bs{\xi}}_{n}^{*}\widetilde{\bs{\xi}}_{n}^{*}}^{1/2}\bs w
\end{align*}
where $\mathbf{P}:=\mathbf{\Sigma}_{[k]\widetilde{\bs{\xi}}_{n}^{*}\widetilde{\bs{\xi}}_{n}^{*}}^{1/2}\left(\mathbf{\Sigma}_{[k]\widetilde{\bs{\xi}}_{n}^{*}\widetilde{\bs{\xi}}_{n}^{*}}^{1/2}\right)^{+}$
is the projection matrix onto $\text{range}(\mathbf{\Sigma}_{[k]\widetilde{\bs{\xi}}_{n}^{*}\widetilde{\bs{\xi}}_{n}^{*}}^{1/2})$
and we have used the fact that $\mathbf{P}=\mathbf{P}^{\trans}$. Let
$\bs u=\mathbf{\Sigma}_{[k]\widetilde{\bs{\xi}}_{n}^{*}\widetilde{\bs{\xi}}_{n}^{*}}^{1/2}\bs w\in\text{range}(\mathbf{\Sigma}_{[k]\widetilde{\bs{\xi}}_{n}^{*}\widetilde{\bs{\xi}}_{n}^{*}}^{1/2})$,
then
\begin{equation}
\bs x^{\trans}\mathbf{\Sigma}_{[k]\widetilde{\bs{\xi}}_{n}^{*}\widetilde{\bs{\xi}}_{n}^{*}}^{+}\bs x=\bs u^{\trans}\mathbf{P}^{2}\bs u=\bs u^{\trans}\bs u.\label{eq:x Sigma x}
\end{equation}
On the other hand we have 
\begin{align}
\bs x^{\trans}\mathbf{\Lambda}^{\trans}\left(\mathbf{\Lambda}\mathbf{\Sigma}_{[k]\widetilde{\bs{\xi}}_{n}^{*}\widetilde{\bs{\xi}}_{n}^{*}}\mathbf{\Lambda}^{\trans}\right)^{+}\mathbf{\Lambda}\bs x & =\bs u^{\trans}\mathbf{\Sigma}_{[k]\widetilde{\bs{\xi}}_{n}^{*}\widetilde{\bs{\xi}}_{n}^{*}}^{1/2}\mathbf{\Lambda}^{\trans}\left(\mathbf{\Lambda}\mathbf{\Sigma}_{[k]\widetilde{\bs{\xi}}_{n}^{*}\widetilde{\bs{\xi}}_{n}^{*}}^{1/2}\mathbf{\Sigma}_{[k]\widetilde{\bs{\xi}}_{n}^{*}\widetilde{\bs{\xi}}_{n}^{*}}^{1/2}\mathbf{\Lambda}^{\trans}\right)^{+}\mathbf{\Lambda}\mathbf{\Sigma}_{[k]\widetilde{\bs{\xi}}_{n}^{*}\widetilde{\bs{\xi}}_{n}^{*}}^{1/2}\bs u\nonumber \\
 & =\bs u^{\trans}\mathbf{Q}^{\trans}\left(\mathbf{Q}\mathbf{Q}^{\trans}\right)^{+}\mathbf{Q}\bs u,\label{eq:xLambda()^+x}
\end{align}
where $\mathbf{Q}:=\mathbf{\Lambda}\mathbf{\Sigma}_{[k]\widetilde{\bs{\xi}}_{n}^{*}\widetilde{\bs{\xi}}_{n}^{*}}^{1/2}$.
Note that $\left\{ \mathbf{Q}^{\trans}\left(\mathbf{Q}\mathbf{Q}^{\trans}\right)^{+}\mathbf{Q}\right\} ^{2}=\mathbf{Q}^{\trans}\left(\mathbf{Q}\mathbf{Q}^{\trans}\right)^{+}\mathbf{Q}$
and $\mathbf{Q}^{\trans}\left(\mathbf{Q}\mathbf{Q}^{\trans}\right)^{+}\mathbf{Q}$
is symmetric. We have $\mathbf{Q}^{\trans}\left(\mathbf{Q}\mathbf{Q}^{\trans}\right)^{+}\mathbf{Q}$
is idempotent and thus $\left\Vert \mathbf{Q}^{\trans}\left(\mathbf{Q}\mathbf{Q}^{\trans}\right)^{+}\mathbf{Q}\right\Vert \leq1$.
Combining (\ref{eq:x Sigma x}) and (\ref{eq:xLambda()^+x}) gives
that 
\begin{align*}
\bs x^{\trans}\mathbf{\Sigma}_{[k]\widetilde{\bs{\xi}}_{n}^{*}\widetilde{\bs{\xi}}_{n}^{*}}^{+}\bs x-\bs x^{\trans}\mathbf{\Lambda}^{\trans}\left(\mathbf{\Lambda}\mathbf{\Sigma}_{[k]\widetilde{\bs{\xi}}_{n}^{*}\widetilde{\bs{\xi}}_{n}^{*}}\mathbf{\Lambda}^{\trans}\right)^{+}\mathbf{\Lambda}\bs x & =\bs u^{\trans}\bs u-\bs u^{\trans}\mathbf{Q}^{\trans}\left(\mathbf{Q}\mathbf{Q}^{\trans}\right)^{+}\mathbf{Q}\bs u\\
 & \geq\left\Vert \bs u\right\Vert ^{2}-\left\Vert \bs u\right\Vert ^{2}=0.
\end{align*}
Since $\bs x\in\text{range}(\mathbf{\Sigma}_{[k]\widetilde{\bs{\xi}}_{n}^{*}\widetilde{\bs{\xi}}_{n}^{*}})$
is arbitrary, the proof is completed.$\hfill\qedsymbol$

\subsection{Proof of Theorem \ref{thm:asymptotic properties of tau_cal-diverging xi}}

We maintain the notation conventions in the proof of Theorem \ref{thm:asymptotic properties of tau_cal}.
Recall (\ref{eq:main estimator decomposition}), we have 
\begin{equation}
\widehat{\tau}_{\mathrm{cal}}-\tau=R_{1}+R_{2}+R_{3}+R_{4}.\label{eq:decomposition of tau_hat-tau*}
\end{equation}
We first handle $R_{1}+R_{2}$.

\textbf{Handling $R_{1}+R_{2}$.} By the definition of $\epsilon_{i,[k]}$
we have 
\begin{align}
R_{1}+R_{2} & =\sum_{k=1}^{K}\frac{1}{n}\sum_{i=1}^{n}\left\{ \widetilde{Y}_{i,[k]}^{*}-\widehat{\bs{\beta}}_{[k]}^{\trans}\bs{\Xi}_{i,[k]}\right\} \nonumber \\
 & =\sum_{k=1}^{K}\frac{1}{n}\sum_{i=1}^{n}\left\{ \widetilde{Y}_{i,[k]}^{*}-\left\{ \frac{1}{n}\sum_{i=1}^{n}\bs{\Xi}_{i,[k]}^{\trans}\widetilde{Y}_{i,[k]}\right\} \left\{ \frac{1}{n}\sum_{i=1}^{n}\bs{\Xi}_{i,[k]}\bs{\Xi}_{i,[k]}^{\trans}\right\} ^{+}\bs{\Xi}_{i,[k]}\right\} \nonumber \\
 & =\sum_{k=1}^{K}\frac{1}{n}\sum_{i=1}^{n}\left\{ \widetilde{Y}_{i,[k]}^{*}-\left\{ \frac{1}{n}\sum_{i=1}^{n}\bs{\Xi}_{i,[k]}^{\trans}\left(\bs{\beta}_{[k],\mathcal{C}_{n}}^{\trans}\bs{\Xi}_{i,[k]}+\epsilon_{i,[k]}\right)\right\} \left\{ \frac{1}{n}\sum_{i=1}^{n}\bs{\Xi}_{i,[k]}\bs{\Xi}_{i,[k]}^{\trans}\right\} ^{+}\bs{\Xi}_{i,[k]}\right\} \nonumber \\
 & =\sum_{k=1}^{K}\frac{1}{n}\sum_{i=1}^{n}\left\{ \widetilde{Y}_{i,[k]}^{*}-\left\{ \bs{\beta}_{[k],\mathcal{C}_{n}}^{\trans}\frac{1}{n}\sum_{i=1}^{n}\bs{\Xi}_{i,[k]}\bs{\Xi}_{i,[k]}^{\trans}+\frac{1}{n}\sum_{i=1}^{n}\bs{\Xi}_{i,[k]}^{\trans}\epsilon_{i,[k]}\right\} \left\{ \frac{1}{n}\sum_{i=1}^{n}\bs{\Xi}_{i,[k]}\bs{\Xi}_{i,[k]}^{\trans}\right\} ^{+}\bs{\Xi}_{i,[k]}\right\} \nonumber \\
 & =\sum_{k=1}^{K}\frac{1}{n}\sum_{i=1}^{n}\widetilde{Y}_{i,[k]}^{*}-\sum_{k=1}^{K}\bs{\beta}_{[k],\mathcal{C}_{n}}^{\trans}\frac{1}{n}\sum_{i=1}^{n}\bs{\Xi}_{i,[k]}\bs{\Xi}_{i,[k]}^{\trans}\left\{ \frac{1}{n}\sum_{i=1}^{n}\bs{\Xi}_{i,[k]}\bs{\Xi}_{i,[k]}^{\trans}\right\} ^{+}\frac{1}{n}\sum_{i=1}^{n}\bs{\Xi}_{i,[k]}\nonumber \\
 & \qquad-\sum_{k=1}^{K}\frac{1}{n}\sum_{i=1}^{n}\bs{\Xi}_{i,[k]}^{\trans}\epsilon_{i,[k]}\left\{ \frac{1}{n}\sum_{i=1}^{n}\bs{\Xi}_{i,[k]}\bs{\Xi}_{i,[k]}^{\trans}\right\} ^{+}\frac{1}{n}\sum_{i=1}^{n}\bs{\Xi}_{i,[k]}\nonumber \\
 & =\underbrace{\sum_{k=1}^{K}\frac{1}{n}\sum_{i=1}^{n}\left\{ \widetilde{Y}_{i,[k]}^{*}-\bs{\beta}_{[k],\mathcal{C}_{n}}^{\trans}\bs{\Xi}_{i,[k]}\right\} }_{Q_{2}}+\underbrace{-\sum_{k=1}^{K}\frac{1}{n}\sum_{i=1}^{n}\bs{\Xi}_{i,[k]}^{\trans}\epsilon_{i,[k]}\left\{ \frac{1}{n}\sum_{i=1}^{n}\bs{\Xi}_{i,[k]}\bs{\Xi}_{i,[k]}^{\trans}\right\} ^{+}\frac{1}{n}\sum_{i=1}^{n}\bs{\Xi}_{i,[k]}}_{Q_{1}},\label{eq:Q_2+Q_1}
\end{align}
where the last step follows from (\ref{eq:mean in variance}) and
\[
\frac{1}{n}\sum_{i=1}^{n}\bs{\Xi}_{i}\bs{\Xi}_{i}^{\trans}\left\{ \frac{1}{n}\sum_{i=1}^{n}\bs{\Xi}_{i}\bs{\Xi}_{i}^{\trans}\right\} ^{+}\frac{1}{n}\sum_{i=1}^{n}\bs{\Xi}_{i}=\frac{1}{n}\sum_{i=1}^{n}\bs{\Xi}_{i}.
\]
We bound $Q_{1}$and $Q_{2}$, respectively.

To bound $Q_{1}$, we derive the convergence rates of $\left\{ \frac{1}{n}\sum_{i=1}^{n}\bs{\Xi}_{i,[k]}\bs{\Xi}_{i,[k]}^{\trans}\right\} ^{+}$,
$\frac{1}{n}\sum_{i=1}^{n}\bs{\Xi}_{i,[k]}^{\trans}\epsilon_{i,[k]}$
and $\frac{1}{n}\sum_{i=1}^{n}\bs{\Xi}_{i,[k]}$ uniformly over $k=1,\ldots,K$.
In this proof, the random variables hidden by $o_{P}(\cdot)$, $O_{P}(\cdot)$,
$O(1)$ and $o(1)$ do not depend on $k$. For any random vector $\bs Z_{n}$
and a sequence of sub-$\sigma$-algebras $\mathcal{C}_{n}$, $\bs Z_{n}=O_{L^{2},\mathcal{C}_{n}}\left(a_{n}\right)$
stands for $\e\left[\left\Vert \bs Z_{n}/a_{n}\right\Vert _{2}\mid\mathcal{C}_{n}\right]=O_{P}(1)$
as $n\to\infty$, where the random variable $O_{P}(1)$ does not depend
on $k$ and $a_{n}$ is a sequence of non-negative numbers.

\textbf{The analysis of $\left\{ \frac{1}{n}\sum_{i=1}^{n}\bs{\Xi}_{i,[k]}\bs{\Xi}_{i,[k]}^{\trans}\right\} ^{+}$.}
Recall that {\footnotesize
\begin{align}
 & \frac{1}{n}\sum_{i=1}^{n}\bs{\Xi}_{i,[k]}\bs{\Xi}_{i,[k]}^{\trans}\nonumber \\
= & \frac{1}{n}\sum_{i=1}^{n}A_{i}\1(B_{i}=k)\left\{ \bs{\xi}_{n}^{*}(\bs X_{i})-\overline{\bs{\xi}}_{n[k]}^{*}\right\} \left\{ \bs{\xi}_{n}^{*}(\bs X_{i})-\overline{\bs{\xi}}_{n[k]}^{*}\right\} ^{\trans}-\nonumber \\
 & \frac{2}{n}\sum_{i=1}^{n}\pi_{n[k]}A_{i}\1(B_{i}=k)\left\{ \bs{\xi}_{n}^{*}(\bs X_{i})-\overline{\bs{\xi}}_{n[k]}^{*}\right\} \left\{ \bs{\xi}_{n}^{*}(\bs X_{i})-\overline{\bs{\xi}}_{n[k]}^{*}\right\} ^{\trans}+\nonumber \\
 & \frac{1}{n}\sum_{i=1}^{n}\pi_{n[k]}^{2}\1(B_{i}=k)\left\{ \bs{\xi}_{n}^{*}(\bs X_{i})-\overline{\bs{\xi}}_{n[k]}^{*}\right\} \left\{ \bs{\xi}_{n}^{*}(\bs X_{i})-\overline{\bs{\xi}}_{n[k]}^{*}\right\} ^{\trans}+\nonumber \\
 & \underbrace{\frac{1}{n}\sum_{i=1}^{n}\left(A_{i}-\pi_{n[k]}\right)^{2}\1(B_{i}=k)\left[\left(\bs{\xi}_{n}(\bs X_{i})-\overline{\bs{\xi}}_{n[k]}\right)\left(\bs{\xi}_{n}(\bs X_{i})-\overline{\bs{\xi}}_{n[k]}\right)^{\trans}-\left(\bs{\xi}_{n}^{*}(\bs X_{i})-\overline{\bs{\xi}}_{n[k]}^{*}\right)\left(\bs{\xi}_{n}^{*}(\bs X_{i})-\overline{\bs{\xi}}_{n[k]}^{*}\right)^{\trans}\right]}_{R_{5}},\label{eq: sample variance decomposition}
\end{align}
}where $\overline{\bs{\xi}}_{n[k]}^{*}:=\frac{1}{n_{[k]}}\sum_{i=1}^{n}\1(B_{i}=k)\bs{\xi}_{n}^{*}(\bs X_{i})$
is the stratum-specific sample mean for $\bs{\xi}_{n}^{*}(\bs X_{i})$.
Let $S^{d-1}:=\{\bs{\alpha}\in\R^{d}:\left\Vert \bs{\alpha}\right\Vert =1\}$
denote the set of all unit vectors in $\R^{d}$. For every $1\leq k\leq K$
and $\bs{\alpha}\in\R^{d}$, by the Cauchy--Schwarz inequality we
have 
\begin{align*}
 & \sup_{\bs{\alpha}\in S^{d-1}}\left|\bs{\alpha}^{\trans}R_{5}\bs{\alpha}\right|\\
= & \sup_{\bs{\alpha}\in S^{d-1}}\left|\frac{1}{n}\sum_{i=1}^{n}\left(A_{i}-\pi_{n[k]}\right)^{2}\1(B_{i}=k)\left[\left(\bs{\alpha}^{\trans}\bs{\xi}_{n}(\bs X_{i})-\bs{\alpha}^{\trans}\overline{\bs{\xi}}_{n[k]}\right)^{2}-\left(\bs{\alpha}^{\trans}\bs{\xi}_{n}^{*}(\bs X_{i})-\bs{\alpha}^{\trans}\overline{\bs{\xi}}_{n[k]}^{*}\right)^{2}\right]\right|\\
\leq & \sup_{\bs{\alpha}\in S^{d-1}}\sqrt{\frac{1}{n}\sum_{i=1}^{n}\1(B_{i}=k)\left\{ \bs{\alpha}^{\trans}\bs{\xi}_{n}(\bs X_{i})-\bs{\alpha}^{\trans}\overline{\bs{\xi}}_{n[k]}-\left(\bs{\alpha}^{\trans}\bs{\xi}_{n}^{*}(\bs X_{i})-\bs{\alpha}^{\trans}\overline{\bs{\xi}}_{n[k]}^{*}\right)\right\} ^{2}}\times\\
 & \qquad\sup_{\bs{\alpha}\in S^{d-1}}\sqrt{\frac{1}{n}\sum_{i=1}^{n}\1(B_{i}=k)\left\{ \bs{\alpha}^{\trans}\bs{\xi}_{n}(\bs X_{i})-\bs{\alpha}^{\trans}\overline{\bs{\xi}}_{n[k]}+\left(\bs{\alpha}^{\trans}\bs{\xi}_{n}^{*}(\bs X_{i})-\bs{\alpha}^{\trans}\overline{\bs{\xi}}_{n[k]}^{*}\right)\right\} ^{2}}\\
\leq & \sqrt{\underbrace{\frac{1}{n}\sum_{i=1}^{n}\1(B_{i}=k)\left\Vert \bs{\xi}_{n}(\bs X_{i})-\overline{\bs{\xi}}_{n[k]}-\left(\bs{\xi}_{n}^{*}(\bs X_{i})-\overline{\bs{\xi}}_{n[k]}^{*}\right)\right\Vert ^{2}}_{R_{6}}}\times\\
 & \qquad\sqrt{\underbrace{\sup_{\bs{\alpha}\in S^{d-1}}\frac{1}{n}\sum_{i=1}^{n}\1(B_{i}=k)\left\{ \bs{\alpha}^{\trans}\bs{\xi}_{n}(\bs X_{i})-\bs{\alpha}^{\trans}\overline{\bs{\xi}}_{n[k]}+\left(\bs{\alpha}^{\trans}\bs{\xi}_{n}^{*}(\bs X_{i})-\bs{\alpha}^{\trans}\overline{\bs{\xi}}_{n[k]}^{*}\right)\right\} ^{2}}_{R_{7}}},
\end{align*}
which leads to 
\[
\left\Vert R_{5}\right\Vert \leq\sqrt{R_{6}}\times\sqrt{R_{7}}.
\]
We control $R_{6}$ and $R_{7}$ separately. By Assumption \ref{assu:conditions on =00005Cxi-diverging xi},
we have 
\begin{align}
\left\Vert \overline{\bs{\xi}}_{n[k]}-\overline{\bs{\xi}}_{n[k]}^{*}\right\Vert  & =\left\Vert \frac{1}{n_{[k]}}\sum_{i=1}^{n}\1(B_{i}=k)\left\{ \bs{\xi}_{n}(\bs X_{i})-\bs{\xi}_{n}^{*}(\bs X_{i})\right\} \right\Vert \nonumber \\
 & \leq\sqrt{\frac{1}{n_{[k]}}\sum_{i=1}^{n}\1(B_{i}=k)\left\Vert \bs{\xi}_{n}(\bs X_{i})-\bs{\xi}_{n}^{*}(\bs X_{i})\right\Vert ^{2}}=O_{P}(n^{-1/4})\label{eq:sample group mean-pop group mean-1}
\end{align}
uniformly over $k=1,\ldots K$. Then we have 
\begin{align}
R_{6} & \leq2\times\frac{1}{n}\sum_{i=1}^{n}\1(B_{i}=k)\left\Vert \bs{\xi}_{n}(\bs X_{i})-\bs{\xi}_{n}^{*}(\bs X_{i})\right\Vert ^{2}+2\times\frac{1}{n}\sum_{i=1}^{n}\1(B_{i}=k)\left\Vert \overline{\bs{\xi}}_{n[k]}-\overline{\bs{\xi}}_{n[k]}^{*}\right\Vert ^{2}\nonumber \\
 & \leq\frac{n_{[k]}}{n}\frac{2}{n_{[k]}}\sum_{i=1}^{n}\1(B_{i}=k)\left\Vert \bs{\xi}_{n}(\bs X_{i})-\bs{\xi}_{n}^{*}(\bs X_{i})\right\Vert ^{2}+\frac{n_{[k]}}{n}\times2\left\Vert \overline{\bs{\xi}}_{n[k]}-\overline{\bs{\xi}}_{n[k]}^{*}\right\Vert ^{2}\nonumber \\
 & =\frac{n_{[k]}}{n}O_{P}(n^{-1/2})=p_{[k]}n^{-1/2}O_{P}(1),\label{eq:R6 diverging xi}
\end{align}
where the last step follows from Lemma \ref{lem:LLN}. For $R_{7}$,
we note that
\begin{align*}
R_{7} & \leq2\sup_{\bs{\alpha}\in S^{d-1}}\frac{1}{n}\sum_{i=1}^{n}\1(B_{i}=k)\left\{ \bs{\alpha}^{\trans}\bs{\xi}_{n}(\bs X_{i})-\bs{\alpha}^{\trans}\overline{\bs{\xi}}_{n[k]}-\left(\bs{\alpha}^{\trans}\bs{\xi}_{n}^{*}(\bs X_{i})-\bs{\alpha}^{\trans}\overline{\bs{\xi}}_{n[k]}^{*}\right)\right\} ^{2}\\
 & \qquad+2\sup_{\bs{\alpha}\in S^{d-1}}\frac{1}{n}\sum_{i=1}^{n}\1(B_{i}=k)\left\{ 2\left(\bs{\alpha}^{\trans}\bs{\xi}_{n}^{*}(\bs X_{i})-\bs{\alpha}^{\trans}\overline{\bs{\xi}}_{n[k]}^{*}\right)\right\} ^{2}\\
 & \leq2R_{6}+16\sup_{\bs{\alpha}\in S^{d-1}}\frac{1}{n}\sum_{i=1}^{n}\1(B_{i}=k)\left\{ \bs{\alpha}^{\trans}\overline{\bs{\xi}}_{n[k]}^{*}\right\} ^{2}+16\sup_{\bs{\alpha}\in S^{d-1}}\frac{1}{n}\sum_{i=1}^{n}\1(B_{i}=k)\left\{ \bs{\alpha}^{\trans}\bs{\xi}_{n}^{*}(\bs X_{i})\right\} ^{2}\\
 & \leq2R_{6}+16\underbrace{\frac{1}{n}\sum_{i=1}^{n}\1(B_{i}=k)\left\Vert \overline{\bs{\xi}}_{n[k]}^{*}\right\Vert ^{2}}_{R_{7,1}}+16\underbrace{\left\Vert \frac{1}{n}\sum_{i=1}^{n}\1(B_{i}=k)\bs{\xi}_{n}^{*}(\bs X_{i})\bs{\xi}_{n}^{*}(\bs X_{i})^{\trans}\right\Vert }_{R_{7,2}}.
\end{align*}
We deal with $R_{7,1}$. By Lemma \ref{lem:uniform LLN} we have 
\begin{align}
\left\Vert \frac{n_{[k]}}{n}\overline{\bs{\xi}}_{n[k]}^{*}-\e\left[\1(B_{i}=k)\bs{\xi}_{n}^{*}(\bs X_{i})\right]\right\Vert  & =\left\Vert \frac{1}{n}\sum_{i=1}^{n}\left\{ \1(B_{i}=k)\bs{\xi}_{n}^{*}(\bs X_{i})-\e\left[\1(B_{i}=k)\bs{\xi}_{n}^{*}(\bs X_{i})\right]\right\} \right\Vert \nonumber \\
 & =\sqrt{\frac{p_{[k]}r_{n}\log(2Kd)}{n}}O_{P}(1).\label{eq:lln for =00005Coverline=00007B=00005Cbs=00007B=00005Cxi=00007D=00007D_=00007Bn=00005Bk=00005D=00007D^=00007B*=00007D-1}
\end{align}
On the other hand, by Assumption \ref{assu:conditions on =00005Cxi-diverging xi}
we have 
\begin{align}
 & \left\Vert \e\left[\1(B_{i}=k)\bs{\xi}_{n}^{*}(\bs X_{i})\right]\right\Vert \nonumber \\
= & p_{[k]}\sup_{\bs{\alpha}\in S^{d-1}}\left|\e\left[\bs{\alpha}^{\trans}\bs{\xi}_{n}^{*}(\bs X_{i})\mid B_{i}=k\right]\right|\leq p_{[k]}\sup_{\bs{\alpha}\in S^{d-1}}\sqrt{\e\left[\left\{ \bs{\alpha}^{\trans}\bs{\xi}_{n}^{*}(\bs X_{i})\right\} ^{2}\mid B_{i}=k\right]}\nonumber \\
= & p_{[k]}\sup_{\bs{\alpha}\in S^{d-1}}\sqrt{p_{[k]}\bs{\alpha}^{\trans}\e\left[\bs{\xi}_{n}^{*}(\bs X_{i})\bs{\xi}_{n}^{*}(\bs X_{i})^{\trans}\mid B_{i}=k\right]\bs{\alpha}}\leq p_{[k]}\sqrt{C}=p_{[k]}O(1).\label{eq:E xin* is bounded}
\end{align}
As a result, by Lemma \ref{lem:LLN} we have 
\begin{align*}
\left\Vert \overline{\bs{\xi}}_{n[k]}^{*}\right\Vert  & =\frac{n}{n_{[k]}}\left\{ p_{[k]}O(1)+\sqrt{\frac{p_{[k]}r_{n}\log(2Kd)}{n}}O_{P}(1)\right\} =O_{P}(1)+\sqrt{\log(2Kd)r_{n}/(np_{[k]})}O_{P}(1)=O_{P}(1),
\end{align*}
where the last step follows from $\log(2Kd)r_{n}/(n\inf_{1\leq k\leq K}p_{[k]})\to0$
as $n\to\infty$. Recall the definition of $R_{7,1}$, then it follows
from Lemma \ref{lem:LLN} that 
\begin{align*}
R_{7,1} & =\frac{1}{n}\sum_{i=1}^{n}\1(B_{i}=k)\left\Vert \overline{\bs{\xi}}_{n[k]}^{*}\right\Vert ^{2}=\left\{ p_{[k]}+p_{[k]}o_{P}(1)\right\} \left\{ O_{P}(1)\right\} ^{2}=p_{[k]}O_{P}(1).
\end{align*}
By Lemma \ref{lem:uniform LLN} we have 
\[
\left\Vert \frac{1}{n}\sum_{i=1}^{n}\1(B_{i}=k)\bs{\xi}_{n}^{*}(\bs X_{i})\bs{\xi}_{n}^{*}(\bs X_{i})^{\trans}-\e\left[\1(B_{i}=k)\bs{\xi}_{n}^{*}(\bs X_{i})\bs{\xi}_{n}^{*}(\bs X_{i})^{\trans}\right]\right\Vert \leq p_{[k]}o_{P}(1).
\]
Thus,
\begin{align*}
R_{7,2} & \leq\left\Vert \e\left[\1(B_{i}=k)\bs{\xi}_{n}^{*}(\bs X_{i})\bs{\xi}_{n}^{*}(\bs X_{i})^{\trans}\right]\right\Vert +p_{[k]}o_{P}(1)\\
 & =p_{[k]}O(1)+p_{[k]}o_{P}(1)=p_{[k]}O_{P}(1).
\end{align*}
Combining (\ref{eq:R6 diverging xi}) and the results for $R_{7,1}$
and $R_{7,2}$ gives that 
\begin{align*}
R_{7} & \leq2R_{6}+16R_{7,1}+16R_{7,2}=p_{[k]}n^{-1/2}O_{P}(1)+p_{[k]}O_{P}(1)=p_{[k]}O_{P}(1).
\end{align*}
As a result, we have 
\begin{align}
\left\Vert R_{5}\right\Vert \leq\sqrt{R_{6}}\times\sqrt{R_{7}}= & p_{[k]}n^{-1/4}O_{P}(1).\label{eq:result for R5 diverging xi}
\end{align}
Recalling (\ref{eq: sample variance decomposition}), we have 
\begin{align}
 & \frac{1}{n}\sum_{i=1}^{n}\bs{\Xi}_{i,[k]}\bs{\Xi}_{i,[k]}^{\trans}\nonumber \\
= & \frac{1}{n}\sum_{i=1}^{n}A_{i}\1(B_{i}=k)\left\{ \bs{\xi}_{n}^{*}(\bs X_{i})-\overline{\bs{\xi}}_{n[k]}^{*}\right\} \left\{ \bs{\xi}_{n}^{*}(\bs X_{i})-\overline{\bs{\xi}}_{n[k]}^{*}\right\} ^{\trans}\nonumber \\
 & -\frac{2}{n}\sum_{i=1}^{n}\pi_{n[k]}A_{i}\1(B_{i}=k)\left\{ \bs{\xi}_{n}^{*}(\bs X_{i})-\overline{\bs{\xi}}_{n[k]}^{*}\right\} \left\{ \bs{\xi}_{n}^{*}(\bs X_{i})-\overline{\bs{\xi}}_{n[k]}^{*}\right\} ^{\trans}\nonumber \\
 & +\frac{1}{n}\sum_{i=1}^{n}\pi_{n[k]}^{2}\1(B_{i}=k)\left\{ \bs{\xi}_{n}^{*}(\bs X_{i})-\overline{\bs{\xi}}_{n[k]}^{*}\right\} \left\{ \bs{\xi}_{n}^{*}(\bs X_{i})-\overline{\bs{\xi}}_{n[k]}^{*}\right\} ^{\trans}+p_{[k]}n^{-1/4}O_{P}(1)\nonumber \\
= & \frac{1}{n}\sum_{i=1}^{n}\left(A_{i}-\pi_{n[k]}\right)^{2}\1(B_{i}=k)\left\{ \widetilde{\bs{\xi}}_{n}^{*}(\bs X_{i})-\overline{\widetilde{\bs{\xi}}}_{n[k]}^{*}\right\} \left\{ \widetilde{\bs{\xi}}_{n}^{*}(\bs X_{i})-\overline{\widetilde{\bs{\xi}}}_{n[k]}^{*}\right\} ^{\trans}+p_{[k]}n^{-1/4}O_{P}(1)\nonumber \\
= & \frac{1}{n}\sum_{i=1}^{n}\bs{\Xi}_{i,[k]}^{*}\bs{\Xi}_{i,[k]}^{*\trans}-2\times\overline{\widetilde{\bs{\xi}}}_{n[k]}^{*}\times\frac{1}{n}\sum_{i=1}^{n}\left(A_{i}-\pi_{n[k]}\right)^{2}\1(B_{i}=k)\widetilde{\bs{\xi}}_{n}^{*\trans}(\bs X_{i})\nonumber \\
 & +\overline{\widetilde{\bs{\xi}}}_{n[k]}^{*}\times\overline{\widetilde{\bs{\xi}}}_{n[k]}^{*\trans}\times\frac{1}{n}\sum_{i=1}^{n}\left(A_{i}-\pi_{n[k]}\right)^{2}\1(B_{i}=k)+p_{[k]}n^{-1/4}O_{P}(1)\label{eq:decomposition of sample variance}
\end{align}
where $\overline{\widetilde{\bs{\xi}}}_{n[k]}^{*}:=\frac{1}{n_{[k]}}\sum_{i=1}^{n}\1(B_{i}=k)\widetilde{\bs{\xi}}_{n}^{*}(\bs X_{i})$.
It follows from (\ref{eq:lln for =00005Coverline=00007B=00005Cbs=00007B=00005Cxi=00007D=00007D_=00007Bn=00005Bk=00005D=00007D^=00007B*=00007D-1}),
\begin{align*}
 & \left\Vert \e\left[\bs{\xi}_{n}^{*}(\bs X_{i})\mid B_{i}=k\right]\right\Vert =\sup_{\bs{\alpha}\in S^{d-1}}\left|\e\left[\bs{\alpha}^{\trans}\bs{\xi}_{n}^{*}(\bs X_{i})\mid B_{i}=k\right]\right|\\
\leq & \sup_{\bs{\alpha}\in S^{d-1}}\sqrt{\e\left[\left\{ \bs{\alpha}^{\trans}\bs{\xi}_{n}^{*}(\bs X_{i})\right\} ^{2}\mid B_{i}=k\right]}\leq C
\end{align*}
and $\frac{n_{[k]}}{n}=p_{[k]}+\sqrt{p_{[k]}}O_{P}(\sqrt{\log(2K)/n})$
that 
\begin{align}
 & \left\Vert \frac{1}{n}\sum_{i=1}^{n}\1(B_{i}=k)\widetilde{\bs{\xi}}_{n}^{*}(\bs X_{i})\right\Vert =\left\Vert \frac{1}{n}\sum_{i=1}^{n}\1(B_{i}=k)\bs{\xi}_{n}^{*}(\bs X_{i})-\frac{n_{[k]}}{n}\e\left[\bs{\xi}_{n}^{*}(\bs X_{i})\mid B_{i}=k\right]\right\Vert \nonumber \\
\leq & \left\Vert \frac{1}{n}\sum_{i=1}^{n}\1(B_{i}=k)\bs{\xi}_{n}^{*}(\bs X_{i})-p_{[k]}\e\left[\bs{\xi}_{n}^{*}(\bs X_{i})\mid B_{i}=k\right]\right\Vert +\left|\frac{n_{[k]}}{n}-p_{[k]}\right|\left\Vert \e\left[\bs{\xi}_{n}^{*}(\bs X_{i})\mid B_{i}=k\right]\right\Vert \nonumber \\
= & \sqrt{\frac{p_{[k]}r_{n}\log(2Kd)}{n}}O_{P}(1)+\sqrt{\frac{p_{[k]}\log(2K)}{n}}O_{P}(1).\label{eq:xi tilde * mean}
\end{align}
Thus, it follows from Lemma \ref{lem:LLN} that 
\begin{align}
\left\Vert \overline{\widetilde{\bs{\xi}}}_{n[k]}^{*}\right\Vert  & =\frac{n}{n_{[k]}}\left\Vert \frac{1}{n}\sum_{i=1}^{n}\1(B_{i}=k)\widetilde{\bs{\xi}}_{n}^{*}(\bs X_{i})\right\Vert =p_{[k]}^{-1}(1+o_{P}(1))\left\Vert \frac{1}{n}\sum_{i=1}^{n}\1(B_{i}=k)\widetilde{\bs{\xi}}_{n}^{*}(\bs X_{i})\right\Vert \nonumber \\
 & =\sqrt{r_{n}\log(2Kd)/(np_{[k]})}O_{P}(1)+\sqrt{\log(2K)/(np_{[k]})}O_{P}(1).\label{eq:rate for bar tilde xi}
\end{align}
In addition, by Lemma \ref{lem:uniform LLN} we have 
\[
\left\Vert \frac{1}{n}\sum_{i=1}^{n}A_{i}\1(B_{i}=k)\widetilde{\bs{\xi}}_{n}^{*}(\bs X_{i})\right\Vert \leq\sqrt{\frac{p_{[k]}r_{n}\log(2Kd)}{n}}O_{P}(1).
\]
Then it follows from Assumption \ref{assu:treatment assignment} and
(\ref{eq:xi tilde * mean}) that 
\begin{align*}
 & \left\Vert \frac{1}{n}\sum_{i=1}^{n}\left(A_{i}-\pi_{n[k]}\right)^{2}\1(B_{i}=k)\widetilde{\bs{\xi}}_{n}^{*}(\bs X_{i})\right\Vert \\
= & \left\Vert \frac{1}{n}\sum_{i=1}^{n}A_{i}\1(B_{i}=k)\widetilde{\bs{\xi}}_{n}^{*}(\bs X_{i})-2\pi_{n[k]}\frac{1}{n}\sum_{i=1}^{n}A_{i}\1(B_{i}=k)\widetilde{\bs{\xi}}_{n}^{*}(\bs X_{i})+\pi_{n[k]}^{2}\frac{1}{n}\sum_{i=1}^{n}\1(B_{i}=k)\widetilde{\bs{\xi}}_{n}^{*}(\bs X_{i})\right\Vert \\
= & \sqrt{\frac{p_{[k]}r_{n}\log(2Kd)}{n}}O_{P}(1)+\sqrt{\frac{p_{[k]}r_{n}\log(2Kd)}{n}}O_{P}(1)+\sqrt{\frac{p_{[k]}\log(2K)}{n}}O_{P}(1)\\
= & \sqrt{\frac{p_{[k]}r_{n}\log(2Kd)}{n}}O_{P}(1)+\sqrt{\frac{p_{[k]}\log(2K)}{n}}O_{P}(1).
\end{align*}
This, combined with (\ref{eq:decomposition of sample variance}),
(\ref{eq:rate for bar tilde xi}), (\ref{eq:result for R5 diverging xi})
and 
\[
\frac{1}{n}\sum_{i=1}^{n}\left(A_{i}-\pi_{n[k]}\right)^{2}\1(B_{i}=k)\leq\frac{n_{[k]}}{n}=p_{[k]}\left(1+o_{P}(1)\right)=p_{[k]}O_{P}(1)
\]
implies that 
\begin{align}
 & \left\Vert \frac{1}{n}\sum_{i=1}^{n}\bs{\Xi}_{i,[k]}\bs{\Xi}_{i,[k]}^{\trans}-\frac{1}{n}\sum_{i=1}^{n}\bs{\Xi}_{i,[k]}^{*}\bs{\Xi}_{i,[k]}^{*\trans}\right\Vert \nonumber \\
= & \left\{ \sqrt{\frac{p_{[k]}r_{n}\log(2Kd)}{n}}+\sqrt{\frac{p_{[k]}\log(2K)}{n}}\right\} \left\{ \sqrt{\frac{r_{n}\log(2Kd)}{np_{[k]}}}+\sqrt{\frac{\log(2K)}{np_{[k]}}}\right\} O_{P}(1)\nonumber \\
 & +\left\{ \frac{r_{n}\log(2Kd)}{np_{[k]}}+\frac{\log(2K)}{np_{[k]}}\right\} p_{[k]}O_{P}(1)+p_{[k]}n^{-1/4}O_{P}(1)\nonumber \\
= & \left\{ \frac{r_{n}\log(2Kd)}{n}+\frac{\log(2K)}{n}\right\} O_{P}(1)+p_{[k]}n^{-1/4}O_{P}(1).\label{eq:xixi=00003Dxixi*+r/n}
\end{align}

By Lemma \ref{lem:uniform LLN} we have 
\begin{align}
\left\Vert \frac{1}{n}\sum_{i=1}^{n}\bs{\Xi}_{i,[k]}^{*}\bs{\Xi}_{i,[k]}^{*\trans}-\mathbf{\Sigma}_{[k]}^{\mathcal{C}_{n}}\right\Vert  & =p_{[k]}o_{P}(1).\label{eq:conditional lln for var}
\end{align}
Besides, it follows from $\left\Vert \e\left[\widetilde{\bs{\xi}}_{n}^{*}(\bs X_{i})\widetilde{\bs{\xi}}_{n}^{*}(\bs X_{i})^{\trans}\mid B_{i}=k\right]\right\Vert \leq C<\infty$
and Lemma \ref{lem:LLN} that 
\begin{equation}
\left\Vert \mathbf{\Sigma}_{[k]}^{\mathcal{C}_{n}}-\mathbf{\Sigma}_{[k]}\right\Vert =p_{[k]}o_{P}(1).\label{eq:converge of Sigma_=00005Bk=00005D^C_n}
\end{equation}
Combining (\ref{eq:xixi=00003Dxixi*+r/n}), (\ref{eq:conditional lln for var})
and (\ref{eq:converge of Sigma_=00005Bk=00005D^C_n}) gives that 
\begin{align}
 & \left\Vert \frac{1}{n}\sum_{i=1}^{n}\bs{\Xi}_{i,[k]}\bs{\Xi}_{i,[k]}^{\trans}-\mathbf{\Sigma}_{[k]}\right\Vert \nonumber \\
= & \left\{ \frac{r_{n}\log(2Kd)}{n}+\frac{\log(2K)}{n}\right\} O_{P}(1)+p_{[k]}n^{-1/4}O_{P}(1)+p_{[k]}o_{P}(1).\label{eq:convergence of var-1}
\end{align}
By Assumption \ref{assu:conditions on =00005Cxi-diverging xi} and
the derivation of the Moore-Penrose inverse, we have
\[
\left\Vert \mathbf{\Sigma}_{[k]\widetilde{\bs{\xi}}_{n}^{*}\widetilde{\bs{\xi}}_{n}^{*}}^{+}\right\Vert \leq c^{-1}
\]
uniformly over $k=1,\ldots,K$ and $n\geq1$. Then by Assumption \ref{assu:treatment assignment}
we have 
\begin{equation}
\left\Vert \mathbf{\Sigma}_{[k]}^{+}\right\Vert =\left\Vert \pi_{[k]}^{-1}(1-\pi_{[k]})^{-1}p_{[k]}^{-1}\mathbf{\Sigma}_{[k]\widetilde{\bs{\xi}}_{n}^{*}\widetilde{\bs{\xi}}_{n}^{*}}^{+}\right\Vert \leq p_{[k]}^{-1}O(1).\label{eq: Sigma+=00005Bk=00005D is bounded}
\end{equation}
It follows from Theorem 3.3 in \citet{stewart1977Perturbation}, Assumption
\ref{assu:conditions on =00005Cxi-diverging xi} and Lemma \ref{lem:LLN}
that
\begin{align}
 & \left\Vert \left\{ \frac{1}{n}\sum_{i=1}^{n}\bs{\Xi}_{i,[k]}\bs{\Xi}_{i,[k]}^{\trans}\right\} ^{+}-\mathbf{\Sigma}_{[k]}^{+}\right\Vert \nonumber \\
\leq & 3\max\left\{ \left\Vert \left\{ \frac{1}{n}\sum_{i=1}^{n}\bs{\Xi}_{i,[k]}\bs{\Xi}_{i,[k]}^{\trans}\right\} ^{+}\right\Vert ^{2},\left\Vert \mathbf{\Sigma}_{[k]}^{+}\right\Vert ^{2}\right\} \left\Vert \frac{1}{n}\sum_{i=1}^{n}\bs{\Xi}_{i,[k]}\bs{\Xi}_{i,[k]}^{\trans}-\mathbf{\Sigma}_{[k]}\right\Vert \nonumber \\
= & \max\left\{ p_{[k]}^{-2}O(1),\left(\frac{n}{n_{[k]}}\right)^{2}O_{P}(1)\right\} \times\left\Vert \frac{1}{n}\sum_{i=1}^{n}\bs{\Xi}_{i,[k]}\bs{\Xi}_{i,[k]}^{\trans}-\mathbf{\Sigma}_{[k]}\right\Vert \nonumber \\
= & p_{[k]}^{-1}\left\{ \left\{ \frac{r_{n}\log(2Kd)}{np_{[k]}}+\frac{\log(2K)}{np_{[k]}}\right\} O_{P}(1)+n^{-1/4}O_{P}(1)+o_{P}(1)\right\} \\
= & p_{[k]}^{-1}o_{P}(1).\label{eq:consistency of inverse of var-1}
\end{align}
where the last step follows from and $r_{n}\log(2Kd)/(n\inf_{1\leq k\leq K}p_{[k]})\to0$
as $n\to\infty$.

\textbf{The analysis of $\frac{1}{n}\sum_{i=1}^{n}\bs{\Xi}_{i,[k]}\epsilon_{i,[k]}$.}
We have 
\[
\frac{1}{n}\sum_{i=1}^{n}\bs{\Xi}_{i,[k]}\epsilon_{i,[k]}=\frac{1}{n}\sum_{i=1}^{n}\bs{\Xi}_{i,[k]}\widetilde{Y}_{i,[k]}-\frac{1}{n}\sum_{i=1}^{n}\bs{\Xi}_{i,[k]}\bs{\Xi}_{i,[k]}^{\trans}\bs{\beta}_{[k],\mathcal{C}_{n}}.
\]
We handle $\frac{1}{n}\sum_{i=1}^{n}\bs{\Xi}_{i,[k]}\widetilde{Y}_{i,[k]}$
and $\frac{1}{n}\sum_{i=1}^{n}\bs{\Xi}_{i,[k]}\bs{\Xi}_{i,[k]}^{\trans}\bs{\beta}_{[k],\mathcal{C}_{n}}$
one by one.

\textbf{The analysis of $\frac{1}{n}\sum_{i=1}^{n}\bs{\Xi}_{i,[k]}\widetilde{Y}_{i,[k]}$.}
We have the following decomposition:
\begin{align*}
 & \frac{1}{n}\sum_{i=1}^{n}\frac{1-\pi_{n[k]}}{\pi_{n[k]}}\1(B_{i}=k)A_{i}(\bs{\xi}_{n}(\bs X_{i})-\overline{\bs{\xi}}_{n[k]})(Y_{i}-\overline{Y}_{1[k]})\\
= & \frac{1}{n}\sum_{i=1}^{n}\frac{1-\pi_{n[k]}}{\pi_{n[k]}}\1(B_{i}=k)A_{i}(\bs{\xi}_{n}(\bs X_{i})-\overline{\bs{\xi}}_{n[k]})Y_{i}(1)\\
 & -\overline{Y}_{1[k]}\times\frac{1}{n}\sum_{i=1}^{n}\frac{1-\pi_{n[k]}}{\pi_{n[k]}}\1(B_{i}=k)A_{i}(\bs{\xi}_{n}(\bs X_{i})-\overline{\bs{\xi}}_{n[k]})\\
= & \frac{1}{n}\sum_{i=1}^{n}\frac{1-\pi_{n[k]}}{\pi_{n[k]}}\1(B_{i}=k)A_{i}(\bs{\xi}_{n}^{*}(\bs X_{i})-\overline{\bs{\xi}}_{n[k]}^{*})(Y_{i}(1)-\overline{Y}_{1[k]})\\
 & +\underbrace{\frac{1}{n}\sum_{i=1}^{n}\frac{1-\pi_{n[k]}}{\pi_{n[k]}}\1(B_{i}=k)A_{i}\left\{ \bs{\xi}_{n}(\bs X_{i})-\overline{\bs{\xi}}_{n[k]}-\left(\bs{\xi}_{n}^{*}(\bs X_{i})-\overline{\bs{\xi}}_{n[k]}^{*}\right)\right\} Y_{i}}_{R_{8}}\\
 & -\underbrace{\overline{Y}_{1[k]}\times\frac{1}{n}\sum_{i=1}^{n}\frac{1-\pi_{n[k]}}{\pi_{n[k]}}\1(B_{i}=k)A_{i}\left\{ \bs{\xi}_{n}(\bs X_{i})-\overline{\bs{\xi}}_{n[k]}-(\bs{\xi}_{n}^{*}(\bs X_{i})-\overline{\bs{\xi}}_{n[k]}^{*})\right\} }_{R_{9}}.
\end{align*}
We control $R_{8}$ and $R_{9}$ separately. Note that by Assumption
\ref{assu:independent sampling} and Lemma \ref{lem:LLN} we have
\[
\e_{\mathcal{C}_{n}}\left[\frac{1}{n}\sum_{i=1}^{n}\1(B_{i}=k)Y_{i}^{2}\right]=\frac{1}{n}\sum_{i=1}^{n}\1(B_{i}=k)\e\left[Y_{i}^{2}\mid B_{i}=k\right]=\frac{n_{[k]}}{n}O(1)=p_{[k]}O_{P}(1).
\]
It follows from Cauchy-Schwarz inequality, Assumption \ref{assu:independent sampling}
and (\ref{eq:R6 diverging xi}) that 
\begin{align*}
 & \left|R_{8}\right|\\
\leq & \frac{1-\pi_{n[k]}}{\pi_{n[k]}}\sqrt{\frac{1}{n}\sum_{i=1}^{n}\1(B_{i}=k)Y_{i}^{2}}\times\sqrt{\frac{1}{n}\sum_{i=1}^{n}\1(B_{i}=k)\left\Vert \bs{\xi}_{n}(\bs X_{i})-\overline{\bs{\xi}}_{n[k]}-\left(\bs{\xi}_{n}^{*}(\bs X_{i})-\overline{\bs{\xi}}_{n[k]}^{*}\right)\right\Vert ^{2}}\\
= & O_{P}(1)\sqrt{p_{[k]}}O_{L^{2},\mathcal{C}_{n}}(1)\times\sqrt{R_{6}}=O_{P}(1)\sqrt{p_{[k]}}O_{L^{2},\mathcal{C}_{n}}(1)\times\sqrt{p_{[k]}}n^{-1/4}O_{P}(1)\\
= & p_{[k]}n^{-1/4}O_{P}(1)O_{L^{2},\mathcal{C}_{n}}(1).
\end{align*}
Similarly, by noting that 
\begin{align*}
\left|\overline{Y}_{1[k]}\right| & =\frac{n}{n_{[k]}}\frac{n_{[k]}}{n_{1[k]}}\left|\frac{1}{n}\sum_{i=1}^{n}\1(B_{i}=k)A_{i}Y_{i}(1)\right|\leq\frac{n}{n_{[k]}}\frac{n_{[k]}}{n_{1[k]}}\sqrt{\frac{n_{[k]}}{n}}\sqrt{p_{[k]}}O_{L^{2},\mathcal{C}_{n}}(1)\\
 & =\sqrt{p_{[k]}^{-1}(1+o_{P}(1))}\left\{ \frac{1}{\pi_{[k]}}+o_{P}(1)\right\} \sqrt{p_{[k]}}O_{L^{2},\mathcal{C}_{n}}(1)=O_{P}(1)O_{L^{2},\mathcal{C}_{n}}(1)
\end{align*}
we have 
\[
\left|R_{9}\right|\leq O_{P}(1)O_{L^{2},\mathcal{C}_{n}}(1)\times\sqrt{\frac{n_{[k]}}{n}}\times\sqrt{R_{6}}=p_{[k]}n^{-1/4}O_{P}(1)O_{L^{2},\mathcal{C}_{n}}(1).
\]
Combining the results for $R_{8}$ and $R_{9}$, we have
\begin{align}
 & \frac{1}{n}\sum_{i=1}^{n}\frac{1-\pi_{n[k]}}{\pi_{n[k]}}\1(B_{i}=k)A_{i}(\bs{\xi}_{n}(\bs X_{i})-\overline{\bs{\xi}}_{n[k]})(Y_{i}-\overline{Y}_{1[k]})\nonumber \\
= & \frac{1-\pi_{n[k]}}{\pi_{n[k]}}\times\frac{1}{n}\sum_{i=1}^{n}\1(B_{i}=k)A_{i}(\bs{\xi}_{n}^{*}(\bs X_{i})-\overline{\bs{\xi}}_{n[k]}^{*})(Y_{i}(1)-\overline{Y}_{1[k]})+p_{[k]}n^{-1/4}O_{P}(1)O_{L^{2},\mathcal{C}_{n}}(1)\nonumber \\
= & \frac{1-\pi_{n[k]}}{\pi_{n[k]}}\times\frac{1}{n}\sum_{i=1}^{n}\1(B_{i}=k)A_{i}(\widetilde{\bs{\xi}}_{n}^{*}(\bs X_{i})-\overline{\widetilde{\bs{\xi}}}_{n[k]}^{*})(\widetilde{Y}_{i}(1)-\overline{\widetilde{Y}}_{1[k]})+p_{[k]}n^{-1/4}O_{P}(1)O_{L^{2},\mathcal{C}_{n}}(1)\nonumber \\
= & \frac{1-\pi_{n[k]}}{\pi_{n[k]}}\times\frac{1}{n}\sum_{i=1}^{n}\1(B_{i}=k)A_{i}\widetilde{\bs{\xi}}_{n}^{*}(\bs X_{i})\widetilde{Y}_{i}(1)-\overline{\widetilde{Y}}_{1[k]}\times\frac{1-\pi_{n[k]}}{\pi_{n[k]}}\frac{1}{n}\sum_{i=1}^{n}\1(B_{i}=k)A_{i}\widetilde{\bs{\xi}}_{n}^{*}(\bs X_{i})\nonumber \\
 & -\overline{\widetilde{\bs{\xi}}}_{n[k]}^{*}\frac{1-\pi_{n[k]}}{\pi_{n[k]}}\frac{1}{n}\sum_{i=1}^{n}\1(B_{i}=k)A_{i}\widetilde{Y}_{i}(1)+\overline{\widetilde{\bs{\xi}}}_{n[k]}^{*}\overline{\widetilde{Y}}_{1[k]}\frac{1-\pi_{n[k]}}{\pi_{n[k]}}\frac{1}{n}\sum_{i=1}^{n}\1(B_{i}=k)\nonumber \\
 & +p_{[k]}n^{-1/4}O_{P}(1)O_{L^{2},\mathcal{C}_{n}}(1),\label{eq:mid step for cov limit-1}
\end{align}
where $\overline{\widetilde{\bs{\xi}}}_{n[k]}^{*}:=\frac{1}{n_{[k]}}\sum_{i=1}^{n}\1(B_{i}=k)\widetilde{\bs{\xi}}_{n}^{*}(\bs X_{i})=\overline{\bs{\xi}}_{n[k]}^{*}-\e\left[\bs{\xi}_{n}^{*}(\bs X_{i})\mid B_{i}=k\right]$
and $\overline{\widetilde{Y}}_{1[k]}:=\frac{1}{n_{1[k]}}\sum_{i=1}^{n}\1(B_{i}=k)A_{i}\widetilde{Y}_{i}(1)=\overline{Y}_{1[k]}-\e\left[Y_{i}(1)\mid B_{i}=k\right]$.
Note that 
\[
\e_{\mathcal{C}_{n}}\left[\1(B_{i}=k)A_{i}\widetilde{Y}_{i}(1)\right]=\1(B_{i}=k)A_{i}\e\left[\widetilde{Y}_{i}(1)\mid B_{i}=k\right]=0
\]
and 
\begin{align*}
\e_{\mathcal{C}_{n}}\left[\left\{ \frac{1}{n}\sum_{i=1}^{n}\1(B_{i}=k)A_{i}\widetilde{Y}_{i}(1)\right\} ^{2}\right] & =\frac{1}{n^{2}}\sum_{i=1}^{n}\e_{\mathcal{C}_{n}}\left[\1(B_{i}=k)A_{i}\widetilde{Y}_{i}(1)^{2}\right]\\
 & \leq\frac{1}{n}\frac{n_{[k]}}{n}\e\left[\widetilde{Y}_{i}(1)^{2}\mid B_{i}=k\right]=\frac{1}{n}p_{[k]}O_{P}(1)
\end{align*}
We have 
\[
\frac{1}{n}\sum_{i=1}^{n}\1(B_{i}=k)A_{i}\widetilde{Y}_{i}(1)=\sqrt{p_{[k]}/n}O_{L^{2},\mathcal{C}_{n}}(1)
\]
and thus
\begin{align*}
\overline{\widetilde{Y}}_{1[k]} & =\frac{n}{n_{1[k]}}\frac{1}{n}\sum_{i=1}^{n}\1(B_{i}=k)A_{i}\widetilde{Y}_{i}(1)=\frac{n}{n_{1[k]}}\sqrt{p_{[k]}/n}O_{L^{2},\mathcal{C}_{n}}(1)=\frac{n}{n_{[k]}}\frac{n_{[k]}}{n_{1[k]}}\sqrt{p_{[k]}/n}O_{L^{2},\mathcal{C}_{n}}(1)\\
 & =p_{[k]}^{-1}(1+o_{P}(1))O_{P}(1)\sqrt{p_{[k]}/n}O_{L^{2},\mathcal{C}_{n}}(1)=\sqrt{1/(np_{[k]})}O_{P}(1)O_{L^{2},\mathcal{C}_{n}}(1),
\end{align*}
where the fourth step follows from Assumption \ref{assu:treatment assignment}.
By Lemma \ref{lem:uniform LLN} we have 
\[
\frac{1}{n}\sum_{i=1}^{n}\1(B_{i}=k)A_{i}\widetilde{\bs{\xi}}_{n}^{*}(\bs X_{i})=\sqrt{\frac{r_{n}p_{[k]}\log(2Kd)}{n}}O_{P}(1).
\]
Combining the above results with (\ref{eq:rate for bar tilde xi}),
(\ref{eq:mid step for cov limit-1}), Assumption \ref{assu:treatment assignment}
and Lemma \ref{lem:LLN}, we have {\footnotesize
\begin{align*}
 & \left\Vert \frac{1}{n}\sum_{i=1}^{n}\frac{1-\pi_{n[k]}}{\pi_{n[k]}}\1(B_{i}=k)A_{i}(\bs{\xi}_{n}(\bs X_{i})-\overline{\bs{\xi}}_{n[k]})(Y_{i}-\overline{Y}_{1[k]})-\frac{1-\pi_{n[k]}}{\pi_{n[k]}}\times\frac{1}{n}\sum_{i=1}^{n}\1(B_{i}=k)A_{i}\widetilde{\bs{\xi}}_{n}^{*}(\bs X_{i})\widetilde{Y}_{i}(1)\right\Vert \\
\leq & \frac{\sqrt{r_{n}\log(2Kd)}}{n}O_{P}(1)O_{L^{2},\mathcal{C}_{n}}(1)+p_{[k]}n^{-1/4}O_{P}(1)O_{L^{2},\mathcal{C}_{n}}(1).
\end{align*}
}Similarly, we also have {\footnotesize
\begin{align*}
 & \frac{\pi_{n[k]}}{1-\pi_{n[k]}}\left|\frac{1}{n}\sum_{i=1}^{n}\1(B_{i}=k)(1-A_{i})(\bs{\xi}_{n}(\bs X_{i})-\overline{\bs{\xi}}_{n[k]})(Y_{i}-\overline{Y}_{0[k]})-\frac{1}{n}\sum_{i=1}^{n}\1(B_{i}=k)(1-A_{i})\widetilde{\bs{\xi}}_{n}^{*}(\bs X_{i})\widetilde{Y}_{i}(0)\right|\\
\leq & \frac{\sqrt{r_{n}\log(2Kd)}}{n}O_{P}(1)O_{L^{2},\mathcal{C}_{n}}(1)+p_{[k]}n^{-1/4}O_{P}(1)O_{L^{2},\mathcal{C}_{n}}(1).
\end{align*}
}Note that 
\begin{align*}
 & \frac{1}{n}\sum_{i=1}^{n}\bs{\Xi}_{i,[k]}\widetilde{Y}_{i,[k]}\\
= & \frac{1}{n}\sum_{i=1}^{n}\frac{1-\pi_{n[k]}}{\pi_{n[k]}}\1(B_{i}=k)A_{i}(\bs{\xi}_{n}(\bs X_{i})-\overline{\bs{\xi}}_{n[k]})(Y_{i}-\overline{Y}_{1[k]})\\
 & \qquad+\frac{1}{n}\sum_{i=1}^{n}\frac{\pi_{n[k]}}{1-\pi_{n[k]}}\1(B_{i}=k)(1-A_{i})(\bs{\xi}_{n}(\bs X_{i})-\overline{\bs{\xi}}_{n[k]})(Y_{i}-\overline{Y}_{0[k]}).
\end{align*}
We have 
\begin{align}
 & \frac{1}{n}\sum_{i=1}^{n}\bs{\Xi}_{i,[k]}\widetilde{Y}_{i,[k]}\nonumber \\
= & \frac{1-\pi_{n[k]}}{\pi_{n[k]}}\times\frac{1}{n}\sum_{i=1}^{n}\1(B_{i}=k)A_{i}\widetilde{\bs{\xi}}_{n}^{*}(\bs X_{i})\widetilde{Y}_{i}(1)\nonumber \\
 & \qquad+\frac{\pi_{n[k]}}{1-\pi_{n[k]}}\times\frac{1}{n}\sum_{i=1}^{n}\1(B_{i}=k)(1-A_{i})\widetilde{\bs{\xi}}_{n}^{*}(\bs X_{i})\widetilde{Y}_{i}(0)\nonumber \\
 & \qquad+\frac{\sqrt{r_{n}\log(2Kd)}}{n}O_{P}(1)O_{L^{2}}(1)+p_{[k]}n^{-1/4}O_{P}(1)O_{L^{2}}(1)\nonumber \\
= & \frac{1}{n}\sum_{i=1}^{n}\bs{\Xi}_{i,[k]}^{*}\widetilde{Y}_{i,[k]}^{*}+\frac{\sqrt{r_{n}\log(2Kd)}}{n}O_{P}(1)O_{L^{2},\mathcal{C}_{n}}(1)+p_{[k]}n^{-1/4}O_{P}(1)O_{L^{2},\mathcal{C}_{n}}(1).\label{eq:xiY=00003Dxi*Y*+error}
\end{align}

\textbf{The analysis of $\frac{1}{n}\sum_{i=1}^{n}\bs{\Xi}_{i,[k]}\bs{\Xi}_{i,[k]}^{\trans}\bs{\beta}_{[k],\mathcal{C}_{n}}$.}
From the definition of $\bs{\beta}_{[k],\mathcal{C}_{n}}$ we have
\[
\bs{\beta}_{[k],\mathcal{C}_{n}}=\left\{ \mathbf{\Sigma}_{[k]}^{\mathcal{C}_{n}}\right\} ^{+}\left\{ \frac{1}{n}\sum_{i=1}^{n}\e_{\mathcal{C}_{n}}\left[\widetilde{Y}_{i,[k]}^{*}\bs{\Xi}_{i,[k]}^{*}\right]\right\} .
\]
Note that 
\begin{align*}
 & \frac{1}{n}\sum_{i=1}^{n}\e_{\mathcal{C}_{n}}\left[\widetilde{Y}_{i,[k]}^{*}\bs{\Xi}_{i,[k]}^{*}\right]\\
= & \frac{1-\pi_{n[k]}}{\pi_{n[k]}}\times\frac{1}{n}\sum_{i=1}^{n}\1(B_{i}=k)A_{i}\mathbf{\Sigma}_{[k]\widetilde{\bs{\xi}}_{n}^{*}\widetilde{Y}(1)}+\frac{\pi_{n[k]}}{1-\pi_{n[k]}}\times\frac{1}{n}\sum_{i=1}^{n}\1(B_{i}=k)(1-A_{i})\mathbf{\Sigma}_{[k]\widetilde{\bs{\xi}}_{n}^{*}\widetilde{Y}(0)}.
\end{align*}
For any $a\in\{0,1\}$, by Assumption \ref{assu:conditions on =00005Cxi-diverging xi}
we have
\begin{align}
\left\Vert \mathbf{\Sigma}_{[k]\widetilde{\bs{\xi}}_{n}^{*}\widetilde{Y}(a)}\right\Vert  & =\left\Vert \e\left[\widetilde{\bs{\xi}}(\bs X_{i})\widetilde{Y}_{i}(a)\mid B_{i}=k\right]\right\Vert =\sup_{\bs{\alpha}\in S^{d-1}}\left|\e\left[\bs{\alpha}^{\trans}\widetilde{\bs{\xi}}(\bs X_{i})\widetilde{Y}_{i}(a)\mid B_{i}=k\right]\right|\nonumber \\
 & \leq\sqrt{\e\left[\left\{ \bs{\alpha}^{\trans}\widetilde{\bs{\xi}}(\bs X_{i})\right\} ^{2}\mid B_{i}=k\right]\e\left[\widetilde{Y}_{i}(a)^{2}\mid B_{i}=k\right]}\nonumber \\
 & \leq\sqrt{\e\left[\widetilde{Y}_{i}(a)^{2}\mid B_{i}=k\right]}\times\sqrt{\bs{\alpha}^{\trans}\e\left[\widetilde{\bs{\xi}}(\bs X_{i})\widetilde{\bs{\xi}}(\bs X_{i})^{\trans}\mid B_{i}=k\right]\bs{\alpha}}=O(1).\label{eq:SigmaxiY is bounded}
\end{align}
Then, applying Lemma \ref{lem:LLN} we have 
\begin{equation}
\left\Vert \frac{1}{n}\sum_{i=1}^{n}\e_{\mathcal{C}_{n}}\left[\widetilde{Y}_{i,[k]}^{*}\bs{\Xi}_{i,[k]}^{*}\right]-p_{[k]}(1-\pi_{[k]})\mathbf{\Sigma}_{[k]\widetilde{\bs{\xi}}_{n}^{*}\widetilde{Y}(1)}-p_{[k]}\pi_{[k]}\mathbf{\Sigma}_{[k]\widetilde{\bs{\xi}}_{n}^{*}\widetilde{Y}(0)}\right\Vert =p_{[k]}o_{P}(1).\label{eq:consistency of conditional Y*XI*}
\end{equation}
Furthermore, by Lemma \ref{lem:LLN} and (\ref{eq: Sigma+=00005Bk=00005D is bounded})
we have 
\begin{align*}
\left\{ \mathbf{\Sigma}_{[k]}^{\mathcal{C}_{n}}\right\} ^{+} & =\left\{ \frac{1}{n}\sum_{i=1}^{n}\left\{ A_{i}-\pi_{n[k]}\right\} ^{2}\1(B_{i}=k)\right\} ^{-1}\mathbf{\Sigma}_{[k]\widetilde{\bs{\xi}}_{n}^{*}\widetilde{\bs{\xi}}_{n}^{*}}^{+}\\
 & =\left\{ p_{[k]}\pi_{[k]}(1-\pi_{[k]})+p_{[k]}o_{P}(1)\right\} ^{-1}\mathbf{\Sigma}_{[k]\widetilde{\bs{\xi}}_{n}^{*}\widetilde{\bs{\xi}}_{n}^{*}}^{+}\\
 & =p_{[k]}^{-1}\left\{ \pi_{[k]}^{-1}(1-\pi_{[k]}^{-1})+o_{P}(1)\right\} \mathbf{\Sigma}_{[k]\widetilde{\bs{\xi}}_{n}^{*}\widetilde{\bs{\xi}}_{n}^{*}}^{+}\\
 & =\mathbf{\Sigma}_{[k]}^{+}+p_{[k]}^{-1}\mathbf{\Sigma}_{[k]\widetilde{\bs{\xi}}_{n}^{*}\widetilde{\bs{\xi}}_{n}^{*}}^{+}o_{P}(1)=\mathbf{\Sigma}_{[k]}^{+}+p_{[k]}^{-1}o_{P}(1).
\end{align*}
This, combined with (\ref{eq:SigmaxiY is bounded}) and (\ref{eq:consistency of conditional Y*XI*})
yields that 
\begin{align}
\bs{\beta}_{[k],\mathcal{C}_{n}} & =\left\{ \left\{ \mathbf{\Sigma}_{[k]}^{\mathcal{C}_{n}}\right\} ^{+}+p_{[k]}^{-1}o_{P}(1)\right\} \left\{ p_{[k]}(1-\pi_{[k]})\mathbf{\Sigma}_{[k]\widetilde{\bs{\xi}}_{n}^{*}\widetilde{Y}(1)}+p_{[k]}\pi_{[k]}\mathbf{\Sigma}_{[k]\widetilde{\bs{\xi}}_{n}^{*}\widetilde{Y}(0)}+p_{[k]}o_{P}(1)\right\} \nonumber \\
 & =\bs{\beta}_{[k]}^{*}+o_{P}(1)\label{eq:betak C_n is consistent}\\
 & =O_{P}(1)\label{eq:betak C_n is bounded}
\end{align}
Then it follows from (\ref{eq:xixi=00003Dxixi*+r/n}) that 
\begin{align*}
 & \left\Vert \frac{1}{n}\sum_{i=1}^{n}\bs{\Xi}_{i,[k]}\bs{\Xi}_{i,[k]}^{\trans}\bs{\beta}_{[k],\mathcal{C}_{n}}-\frac{1}{n}\sum_{i=1}^{n}\bs{\Xi}_{i,[k]}^{*}\bs{\Xi}_{i,[k]}^{*\trans}\bs{\beta}_{[k],\mathcal{C}_{n}}\right\Vert \\
= & \left\{ \frac{r_{n}\log(2Kd)}{n}+\frac{\log(2K)}{n}\right\} O_{P}(1)+p_{[k]}n^{-1/4}O_{P}(1).
\end{align*}

Recall (\ref{eq:xiY=00003Dxi*Y*+error}), then we obtain that 
\begin{equation}
\left\Vert \frac{1}{n}\sum_{i=1}^{n}\bs{\Xi}_{i,[k]}\epsilon_{i,[k]}-\frac{1}{n}\sum_{i=1}^{n}\bs{\Xi}_{i,[k]}^{*}\epsilon_{i,[k]}^{*}\right\Vert =\frac{r_{n}\log(2Kd)}{n}O_{P}(1)O_{L^{2},\mathcal{C}_{n}}(1)+p_{[k]}n^{-1/4}O_{P}(1)O_{L^{2},\mathcal{C}_{n}}(1).\label{eq:xiepsilon-xi*epsilon*=00003Dsmall}
\end{equation}
Note that 
\begin{align*}
 & \e_{\mathcal{C}_{n}}\left[\left\Vert \bs{\Xi}_{i,[k]}^{*}\epsilon_{i,[k]}^{*}-\e_{\mathcal{C}_{n}}\left[\bs{\Xi}_{i,[k]}^{*}\epsilon_{i,[k]}^{*}\right]\right\Vert ^{2}\right]\leq\e_{\mathcal{C}_{n}}\left[\left\Vert \bs{\Xi}_{i,[k]}^{*}\epsilon_{i,[k]}^{*}\right\Vert ^{2}\right]\\
\leq & 2\e_{\mathcal{C}_{n}}\left[\left\Vert \bs{\Xi}_{i,[k]}^{*}\widetilde{Y}_{i,[k]}^{*}\right\Vert ^{2}\right]+2\e_{\mathcal{C}_{n}}\left[\left\Vert \bs{\Xi}_{i,[k]}^{*}\bs{\Xi}_{i,[k]}^{*\trans}\bs{\beta}_{[k],\mathcal{C}_{n}}\right\Vert ^{2}\right]\\
\leq & 4\max\left\{ \frac{1-\pi_{n[k]}}{\pi_{n[k]}},\frac{\pi_{n[k]}}{1-\pi_{n[k]}}\right\} ^{2}\1(B_{i}=k)\max_{a=0,1}\e\left[\left\Vert \widetilde{\bs{\xi}}_{n}^{*}(\bs X_{i})\widetilde{Y}_{i}(a)\right\Vert ^{2}\mid B_{i}=k\right]\\
 & +2\e_{\mathcal{C}_{n}}\left[\left\Vert \widetilde{\bs{\xi}}_{n}^{*}(\bs X_{i})\widetilde{\bs{\xi}}_{n}^{*}(\bs X_{i})^{\trans}\bs{\beta}_{[k],\mathcal{C}_{n}}\right\Vert ^{2}\right]\1(B_{i}=k)\\
\leq & \zeta_{n}^{2}O_{P}(1)\max_{a=0,1}\e\left[\left\Vert \widetilde{Y}_{i}(a)\right\Vert ^{2}\mid B_{i}=k\right]\1(B_{i}=k)\\
 & +\zeta_{n}^{2}O_{P}(1)\bs{\beta}_{[k],\mathcal{C}_{n}}^{\trans}\e\left[\widetilde{\bs{\xi}}_{n}^{*}(\bs X_{i})\widetilde{\bs{\xi}}_{n}^{*}(\bs X_{i})^{\trans}\mid B_{i}=k\right]\bs{\beta}_{[k],\mathcal{C}_{n}}\1(B_{i}=k)\\
= & \1(B_{i}=k)\zeta_{n}^{2}O_{P}(1),
\end{align*}
where $O_{P}(1)$ depends on neither $k$ nor $i$ and the last step
follows from Assumptions \ref{assu:independent sampling} and \ref{assu:conditions on =00005Cxi-diverging xi}
and (\ref{eq:betak C_n is bounded}). Thus, it follows from Lemma
\ref{lem:LLN} that 
\begin{align*}
 & \e_{\mathcal{C}_{n}}\left[\left\Vert \frac{1}{n}\sum_{i=1}^{n}\left\{ \bs{\Xi}_{i,[k]}^{*}\epsilon_{i,[k]}^{*}-\e_{\mathcal{C}_{n}}\left[\bs{\Xi}_{i,[k]}^{*}\epsilon_{i,[k]}^{*}\right]\right\} \right\Vert ^{2}\right]\\
= & \frac{1}{n^{2}}\sum_{i=1}^{n}\e_{\mathcal{C}_{n}}\left[\left\Vert \bs{\Xi}_{i,[k]}^{*}\epsilon_{i,[k]}^{*}-\e_{\mathcal{C}_{n}}\left[\bs{\Xi}_{i,[k]}^{*}\epsilon_{i,[k]}^{*}\right]\right\Vert ^{2}\right]\\
\leq & \frac{1}{n^{2}}\sum_{i=1}^{n}\1(B_{i}=k)\zeta_{n}^{2}O_{P}(1)=\frac{1}{n}\zeta_{n}^{2}O_{P}(1)p_{[k]}.
\end{align*}
As a result, 
\begin{equation}
\left\Vert \frac{1}{n}\sum_{i=1}^{n}\left\{ \bs{\Xi}_{i,[k]}^{*}\epsilon_{i,[k]}^{*}-\e_{\mathcal{C}_{n}}\left[\bs{\Xi}_{i,[k]}^{*}\epsilon_{i,[k]}^{*}\right]\right\} \right\Vert =\sqrt{\zeta_{n}^{2}p_{[k]}/n}O_{P}(1)O_{L^{2},\mathcal{C}_{n}}(1).\label{eq:lln for xi*ep*}
\end{equation}

Since 
\begin{align*}
 & \frac{1}{n}\sum_{i=1}^{n}\e_{\mathcal{C}_{n}}\left[\widetilde{Y}_{i,[k]}^{*}\bs{\Xi}_{i,[k]}^{*}\right]\\
= & \frac{1-\pi_{n[k]}}{\pi_{n[k]}}\times\frac{1}{n}\sum_{i=1}^{n}\1(B_{i}=k)A_{i}\mathbf{\Sigma}_{[k]\widetilde{\bs{\xi}}_{n}^{*}\widetilde{Y}(1)}+\frac{\pi_{n[k]}}{1-\pi_{n[k]}}\times\frac{1}{n}\sum_{i=1}^{n}\1(B_{i}=k)(1-A_{i})\mathbf{\Sigma}_{[k]\widetilde{\bs{\xi}}_{n}^{*}\widetilde{Y}(0)}
\end{align*}
and 
\[
\frac{1}{n}\sum_{i=1}^{n}\e_{\mathcal{C}_{n}}\left[\bs{\Xi}_{i,[k]}^{*}\bs{\Xi}_{i,[k]}^{*\trans}\right]=\frac{1}{n}\sum_{i=1}^{n}(A_{i}-\pi_{n[k]})^{2}\1(B_{i}=k)\mathbf{\Sigma}_{[k]\widetilde{\bs{\xi}}_{n}^{*}\widetilde{\bs{\xi}}_{n}^{*}}^{+},
\]
by (\ref{eq: cov in range=00007Bvar=00007D}) we have 
\[
\frac{1}{n}\sum_{i=1}^{n}\e_{\mathcal{C}_{n}}\left[\widetilde{Y}_{i,[k]}^{*}\bs{\Xi}_{i,[k]}^{*}\right]\in\text{range}\left(\frac{1}{n}\sum_{i=1}^{n}\e_{\mathcal{C}_{n}}\left[\bs{\Xi}_{i,[k]}^{*}\bs{\Xi}_{i,[k]}^{*\trans}\right]\right).
\]
Using the property of the Moore-Penrose inverse, we have 
\[
\frac{1}{n}\sum_{i=1}^{n}\e_{\mathcal{C}_{n}}\left[\bs{\Xi}_{i,[k]}^{*}\bs{\Xi}_{i,[k]}^{*\trans}\right]\left\{ \frac{1}{n}\sum_{i=1}^{n}\e_{\mathcal{C}_{n}}\left[\bs{\Xi}_{i,[k]}^{*}\bs{\Xi}_{i,[k]}^{*\trans}\right]\right\} ^{+}\left\{ \frac{1}{n}\sum_{i=1}^{n}\e_{\mathcal{C}_{n}}\left[\widetilde{Y}_{i,[k]}^{*}\bs{\Xi}_{i,[k]}^{*}\right]\right\} =\frac{1}{n}\sum_{i=1}^{n}\e_{\mathcal{C}_{n}}\left[\widetilde{Y}_{i,[k]}^{*}\bs{\Xi}_{i,[k]}^{*}\right].
\]
As a result, we have 
\begin{align}
 & \frac{1}{n}\sum_{i=1}^{n}\e_{\mathcal{C}_{n}}\left[\bs{\Xi}_{i,[k]}^{*}\epsilon_{i,[k]}^{*}\right]=\frac{1}{n}\sum_{i=1}^{n}\e_{\mathcal{C}_{n}}\left[\bs{\Xi}_{i,[k]}^{*}\widetilde{Y}_{i,[k]}^{*}\right]-\frac{1}{n}\sum_{i=1}^{n}\e_{\mathcal{C}_{n}}\left[\bs{\Xi}_{i,[k]}^{*}\bs{\Xi}_{i,[k]}^{*\trans}\right]\bs{\beta}_{[k],\mathcal{C}_{n}}\nonumber \\
= & \frac{1}{n}\sum_{i=1}^{n}\e_{\mathcal{C}_{n}}\left[\bs{\Xi}_{i,[k]}^{*}\widetilde{Y}_{i,[k]}^{*}\right]-\frac{1}{n}\sum_{i=1}^{n}\e_{\mathcal{C}_{n}}\left[\bs{\Xi}_{i,[k]}^{*}\bs{\Xi}_{i,[k]}^{*\trans}\right]\left\{ \frac{1}{n}\sum_{i=1}^{n}\e_{\mathcal{C}_{n}}\left[\bs{\Xi}_{i,[k]}^{*}\bs{\Xi}_{i,[k]}^{*\trans}\right]\right\} ^{+}\left\{ \frac{1}{n}\sum_{i=1}^{n}\e_{\mathcal{C}_{n}}\left[\widetilde{Y}_{i,[k]}^{*}\bs{\Xi}_{i,[k]}^{*}\right]\right\} \nonumber \\
= & \frac{1}{n}\sum_{i=1}^{n}\e_{\mathcal{C}_{n}}\left[\bs{\Xi}_{i,[k]}^{*}\widetilde{Y}_{i,[k]}^{*}\right]-\frac{1}{n}\sum_{i=1}^{n}\e_{\mathcal{C}_{n}}\left[\bs{\Xi}_{i,[k]}^{*}\widetilde{Y}_{i,[k]}^{*}\right]=0.\label{eq:condi exp of Xiepsilon}
\end{align}
This, combined with (\ref{eq:lln for xi*ep*}) gives that 
\[
\left\Vert \frac{1}{n}\sum_{i=1}^{n}\bs{\Xi}_{i,[k]}^{*}\epsilon_{i,[k]}^{*}\right\Vert =\sqrt{\zeta_{n}^{2}p_{[k]}/n}O_{P}(1)O_{L^{2},\mathcal{C}_{n}}(1).
\]
Recall (\ref{eq:xiepsilon-xi*epsilon*=00003Dsmall}), we have 
\begin{equation}
\left\Vert \frac{1}{n}\sum_{i=1}^{n}\bs{\Xi}_{i,[k]}\epsilon_{i,[k]}\right\Vert =\left\{ \frac{r_{n}\log(2Kd)}{n}+p_{[k]}n^{-1/4}+\sqrt{\frac{\zeta_{n}^{2}p_{[k]}}{n}}\right\} O_{P}(1)O_{L^{2},\mathcal{C}_{n}}(1).\label{eq:conver rate for =00005Cfrac=00007B1=00007D=00007Bn=00007D=00005Csum_=00007Bi=00003D1=00007D^=00007Bn=00007D=00005Cbs=00007B=00005CXi=00007D_=00007Bi,=00005Bk=00005D=00007D=00005Cepsilon_=00007Bi,=00005Bk=00005D=00007D}
\end{equation}

\textbf{The analysis of $\frac{1}{n}\sum_{i=1}^{n}\bs{\Xi}_{i,[k]}$.}
For every $k=1,\ldots,K$, we let $\bs{\xi}_{[k]}^{*}:=\e\left[\bs{\xi}_{n}^{*}(\bs X_{i})\mid B_{i}=k\right]$.
By $\frac{1}{n}\sum_{i=1}^{n}A_{i}\1(B_{i}=k)=\pi_{n[k]}\frac{1}{n}\sum_{i=1}^{n}\1(B_{i}=k)$
and $\1(B_{i}=k)\widetilde{\bs{\xi}}_{n}^{*}(\bs X_{i})=\1(B_{i}=k)\left\{ \bs{\xi}_{n}^{*}(\bs X_{i})-\bs{\xi}_{[k]}^{*}\right\} $,
we have 
\begin{align*}
\frac{1}{n}\sum_{i=1}^{n}\bs{\Xi}_{i,[k]} & =\frac{1}{n}\sum_{i=1}^{n}\left\{ A_{i}-\pi_{n[k]}\right\} \1(B_{i}=k)\left\{ \bs{\xi}_{n}(\bs X_{i})-\overline{\bs{\xi}}_{n[k]}\right\} \\
 & =\frac{1}{n}\sum_{i=1}^{n}\left\{ A_{i}-\pi_{n[k]}\right\} \1(B_{i}=k)\left\{ \bs{\xi}_{n}(\bs X_{i})-\bs{\xi}_{[k]}^{*}\right\} \\
 & =\frac{1}{n}\sum_{i=1}^{n}\bs{\Xi}_{i,[k]}^{*}+\frac{1}{n}\sum_{i=1}^{n}\left\{ A_{i}-\pi_{n[k]}\right\} \1(B_{i}=k)\left\{ \bs{\xi}_{n}(\bs X_{i})-\bs{\xi}_{n}^{*}(\bs X_{i})\right\} .
\end{align*}
By Assumption \ref{assu:conditions on =00005Cxi-diverging xi} and
Lemma \ref{lem:LLN}, we have 
\begin{align*}
 & \left\Vert \frac{1}{n}\sum_{i=1}^{n}\left\{ A_{i}-\pi_{n[k]}\right\} \1(B_{i}=k)\left\{ \bs{\xi}_{n}(\bs X_{i})-\bs{\xi}_{n}^{*}(\bs X_{i})\right\} \right\Vert \\
= & \frac{n_{[k]}}{n}\pi_{n[k]}(1-\pi_{n[k]})\times\\
 & \quad\left\Vert \frac{1}{n_{1[k]}}\sum_{i=1}^{n}A_{i}\1(B_{i}=k)\left\{ \bs{\xi}_{n}(\bs X_{i})-\bs{\xi}_{n}^{*}(\bs X_{i})\right\} -\frac{1}{n_{0[k]}}\sum_{i=1}^{n}\left(1-A_{i}\right)\1(B_{i}=k)\left\{ \bs{\xi}_{n}(\bs X_{i})-\bs{\xi}_{n}^{*}(\bs X_{i})\right\} \right\Vert \\
= & p_{[k]}O_{P}(1)\times o_{P}(n^{-1/2})=p_{[k]}n^{-1/2}o_{P}(1),
\end{align*}
which leads to 
\begin{equation}
\left\Vert \frac{1}{n}\sum_{i=1}^{n}\bs{\Xi}_{i,[k]}-\frac{1}{n}\sum_{i=1}^{n}\bs{\Xi}_{i,[k]}^{*}\right\Vert =p_{[k]}n^{-1/2}o_{P}(1).\label{eq: xi -xi*}
\end{equation}
Note that 
\[
\e_{\mathcal{C}_{n}}\left[\bs{\Xi}_{i,[k]}^{*}\right]=(A_{i}-\pi_{n[k]})\1(B_{i}=k)\e\left[\widetilde{\bs{\xi}}_{n}^{*}(\bs X_{i})\mid B_{i}=k\right]
\]
and
\begin{align*}
\e_{\mathcal{C}_{n}}\left[\left\Vert \frac{1}{n}\sum_{i=1}^{n}\bs{\Xi}_{i,[k]}^{*}\right\Vert ^{2}\right] & =\frac{1}{n^{2}}\sum_{i=1}^{n}(A_{i}-\pi_{n[k]})^{2}\1(B_{i}=k)\e\left[\left\Vert \widetilde{\bs{\xi}}_{n}^{*}(\bs X_{i})\right\Vert ^{2}\mid B_{i}=k\right]\\
 & \leq\frac{1}{n}\frac{n_{[k]}}{n}\tr\left\{ \e\left[\bs{\xi}_{n}^{*}(\bs X_{i})\bs{\xi}_{n}^{*}(\bs X_{i})^{\trans}\mid B_{i}=k\right]\right\} =p_{[k]}\frac{r_{n}}{n}O_{P}(1).
\end{align*}
We have 
\[
\left\Vert \frac{1}{n}\sum_{i=1}^{n}\bs{\Xi}_{i,[k]}^{*}\right\Vert =\sqrt{\frac{r_{n}p_{[k]}}{n}}O_{P}(1)O_{L^{2},\mathcal{C}_{n}}(1)
\]
 Thus, we can conclude that 
\begin{equation}
\left\Vert \frac{1}{n}\sum_{i=1}^{n}\bs{\Xi}_{i,[k]}\right\Vert \leq\left\Vert \frac{1}{n}\sum_{i=1}^{n}\bs{\Xi}_{i,[k]}^{*}\right\Vert +p_{[k]}n^{-1/2}o_{P}(1)=\sqrt{\frac{r_{n}p_{[k]}}{n}}O_{P}(1)O_{L^{2},\mathcal{C}_{n}}(1),\label{eq:conver rate for =00005Cfrac=00007B1=00007D=00007Bn=00007D=00005Csum_=00007Bi=00003D1=00007D^=00007Bn=00007D=00005Cbs=00007B=00005CXi=00007D_=00007Bi,=00005Bk=00005D=00007D^=00007B*=00007D}
\end{equation}
where the last step follows from $p_{[k]}\leq\sqrt{r_{n}p_{[k]}}$.

Combining (\ref{eq: Sigma+=00005Bk=00005D is bounded}), (\ref{eq:consistency of inverse of var-1}),
(\ref{eq:conver rate for =00005Cfrac=00007B1=00007D=00007Bn=00007D=00005Csum_=00007Bi=00003D1=00007D^=00007Bn=00007D=00005Cbs=00007B=00005CXi=00007D_=00007Bi,=00005Bk=00005D=00007D=00005Cepsilon_=00007Bi,=00005Bk=00005D=00007D})
and (\ref{eq:conver rate for =00005Cfrac=00007B1=00007D=00007Bn=00007D=00005Csum_=00007Bi=00003D1=00007D^=00007Bn=00007D=00005Cbs=00007B=00005CXi=00007D_=00007Bi,=00005Bk=00005D=00007D^=00007B*=00007D}),
we have
\begin{align}
\left|Q_{1}\right| & \leq\sum_{k=1}^{K}\left\Vert \frac{1}{n}\sum_{i=1}^{n}\bs{\Xi}_{i,[k]}^{\trans}\epsilon_{i,[k]}\right\Vert \left\Vert \left\{ \frac{1}{n}\sum_{i=1}^{n}\bs{\Xi}_{i,[k]}\bs{\Xi}_{i,[k]}^{\trans}\right\} ^{+}\right\Vert \left\Vert \frac{1}{n}\sum_{i=1}^{n}\bs{\Xi}_{i,[k]}\right\Vert \nonumber \\
 & \leq\sum_{k=1}^{K}\left\{ \frac{r_{n}\log(2Kd)}{n}+p_{[k]}n^{-1/4}+\sqrt{\frac{\zeta_{n}^{2}p_{[k]}}{n}}\right\} O_{P}(1)O_{L^{2},\mathcal{C}_{n}}(1)\nonumber \\
 & \quad\times\left\{ p_{[k]}^{-1}O(1)+p_{[k]}^{-1}o_{P}(1)\right\} \times\left\{ \sqrt{\frac{r_{n}p_{[k]}}{n}}O_{P}(1)O_{L^{2},\mathcal{C}_{n}}(1)\right\} \nonumber \\
 & \leq O_{P}(1)\sqrt{\frac{r_{n}^{3}\log^{2}(2Kd)}{n^{3}\inf_{1\leq k\leq K}p_{[k]}}}\sum_{k=1}^{K}O_{L^{2},\mathcal{C}_{n}}(1)O_{L^{2},\mathcal{C}_{n}}(1)\nonumber \\
 & \quad+O_{P}(1)\left(\frac{r_{n}^{2}}{n^{3}}\right)^{1/4}\sum_{k=1}^{K}\sqrt{p_{[k]}}O_{L^{2},\mathcal{C}_{n}}(1)O_{L^{2},\mathcal{C}_{n}}(1)\nonumber \\
 & \quad+O_{P}(1)\frac{\sqrt{\zeta_{n}^{2}r_{n}}}{n}\sum_{k=1}^{K}O_{L^{2},\mathcal{C}_{n}}(1)O_{L^{2},\mathcal{C}_{n}}(1)\nonumber \\
 & =O_{P}(n^{-1/2})\left\{ \sqrt{\frac{K^{2}r_{n}^{2}}{n}\frac{r_{n}\log^{2}(2Kd)}{n\inf_{1\leq k\leq K}p_{[k]}}}+\left(\frac{K^{2}r_{n}^{2}}{n}\right)^{1/4}+\sqrt{\frac{K^{2}\zeta_{n}^{2}r_{n}}{n}}\right\} \nonumber \\
 & =o_{P}(n^{-1/2})\label{eq:Q1}
\end{align}
The first equality in the above display follows because by \citet[Lemma 6.1]{chernozhukov2018Double}
and 
\begin{align*}
 & \e_{\mathcal{C}_{n}}\left[\left|\sum_{k=1}^{K}O_{L^{2},\mathcal{C}_{n}}(1)O_{L^{2},\mathcal{C}_{n}}(1)\right|\right]\leq\sum_{k=1}^{K}\e_{\mathcal{C}_{n}}\left[\left|O_{L^{2},\mathcal{C}_{n}}(1)O_{L^{2},\mathcal{C}_{n}}(1)\right|\right]\\
\leq & \sum_{k=1}^{K}\sqrt{\e_{\mathcal{C}_{n}}\left[O_{L^{2},\mathcal{C}_{n}}(1)^{2}\right]}\sqrt{\e_{\mathcal{C}_{n}}\left[O_{L^{2},\mathcal{C}_{n}}(1)^{2}\right]}\leq\sum_{k=1}^{K}O_{P}(1)\leq KO_{P}(1),
\end{align*}
\begin{align*}
 & \e_{\mathcal{C}_{n}}\left[\sqrt{p_{[k]}}\left|\sum_{k=1}^{K}O_{L^{2},\mathcal{C}_{n}}(1)O_{L^{2},\mathcal{C}_{n}}(1)\right|\right]\leq\sum_{k=1}^{K}\sqrt{p_{[k]}}\e_{\mathcal{C}_{n}}\left[\left|O_{L^{2},\mathcal{C}_{n}}(1)O_{L^{2},\mathcal{C}_{n}}(1)\right|\right]\\
\leq & \sum_{k=1}^{K}\sqrt{p_{[k]}}\sqrt{\e_{\mathcal{C}_{n}}\left[O_{L^{2},\mathcal{C}_{n}}(1)^{2}\right]}\sqrt{\e_{\mathcal{C}_{n}}\left[O_{L^{2},\mathcal{C}_{n}}(1)^{2}\right]}\leq\sum_{k=1}^{K}\sqrt{p_{[k]}}O_{P}(1)=O_{P}(1)\sum_{k=1}^{K}\sqrt{p_{[k]}}\\
\leq & O_{P}(1)\sqrt{K}\sqrt{\sum_{k=1}^{K}p_{[k]}}=\sqrt{K}O_{P}(1),
\end{align*}
it holds that 
\[
\left|\sum_{k=1}^{K}O_{L^{2},\mathcal{C}_{n}}(1)O_{L^{2},\mathcal{C}_{n}}(1)\right|=O_{P}(K)\text{ and }\left|\sum_{k=1}^{K}\sqrt{p_{[k]}}O_{L^{2},\mathcal{C}_{n}}(1)O_{L^{2},\mathcal{C}_{n}}(1)\right|=O_{P}(\sqrt{K}).
\]

By (\ref{eq:betak C_n is consistent}) and (\ref{eq: xi -xi*}), we
have 
\begin{align}
 & Q_{2}=\sum_{k=1}^{K}\frac{1}{n}\sum_{i=1}^{n}\left\{ \widetilde{Y}_{i,[k]}^{*}-\bs{\beta}_{[k],\mathcal{C}_{n}}^{\trans}\bs{\Xi}_{i,[k]}\right\} \nonumber \\
= & \sum_{k=1}^{K}\frac{1}{n}\sum_{i=1}^{n}\left\{ \widetilde{Y}_{i,[k]}^{*}-\bs{\beta}_{[k],\mathcal{C}_{n}}^{\trans}\bs{\Xi}_{i,[k]}^{*}\right\} +\sum_{k=1}^{K}\bs{\beta}_{[k],\mathcal{C}_{n}}^{\trans}\frac{1}{n}\sum_{i=1}^{n}\left(\bs{\Xi}_{i,[k]}^{*}-\bs{\Xi}_{i,[k]}\right)\nonumber \\
= & \frac{1}{n}\sum_{i=1}^{n}\sum_{k=1}^{K}\epsilon_{i,[k]}^{*}+n^{-1/2}o_{P}(1)\sum_{k=1}^{K}p_{[k]}=\frac{1}{n}\sum_{i=1}^{n}\sum_{k=1}^{K}\epsilon_{i,[k]}^{*}+o_{P}(n^{-1/2}),\label{eq:Q2}
\end{align}
Note that $\sum_{k=1}^{K}\epsilon_{i,[k]}^{*}$ are independent conditional
on $\mathcal{C}_{n}$, we intend to apply Lemma \ref{lem:clt for triangular array}.
By direct calculation we have 
\begin{align*}
\e_{\mathcal{C}_{n}}\left[\sum_{k=1}^{K}\epsilon_{i,[k]}^{*}\right] & =\sum_{k=1}^{K}\1(B_{i}=k)\left\{ \frac{A_{i}}{\pi_{n[k]}}\e\left[\widetilde{Y}_{i}(1)\mid B_{i}=k\right]-\frac{1-A_{i}}{1-\pi_{n[k]}}\e\left[\widetilde{Y}_{i}(0)\mid B_{i}=k\right]\right\} \\
 & \qquad-\sum_{k=1}^{K}\1(B_{i}=k)(A_{i}-\pi_{n[k]})\bs{\beta}_{[k],\mathcal{C}_{n}}^{\trans}\e\left[\widetilde{\bs{\xi}}_{n}^{*}(\bs X_{i})\mid B_{i}=k\right]\\
 & =0,
\end{align*}
\begin{align}
 & \var\left(\sum_{k=1}^{K}\epsilon_{i,[k]}^{*}\mid\mathcal{C}_{n}\right)=\e\left[\left(\sum_{k=1}^{K}\epsilon_{i,[k]}^{*}\right)^{2}\mid\mathcal{C}_{n}\right]=\sum_{k=1}^{K}\e\left[\epsilon_{i,[k]}^{*2}\mid\mathcal{C}_{n}\right]\nonumber \\
= & \sum_{k=1}^{K}\1(B_{i}=k)\left\{ \frac{A_{i}}{\pi_{n[k]}^{2}}\e\left[\widetilde{Y}_{i}^{2}(1)\mid B_{i}=k\right]+\frac{1-A_{i}}{(1-\pi_{n[k]})^{2}}\e\left[\widetilde{Y}_{i}^{2}(0)\mid B_{i}=k\right]\right\} \nonumber \\
 & -2\sum_{k=1}^{K}\1(B_{i}=k)A_{i}\frac{1-\pi_{n[k]}}{\pi_{n[k]}}\bs{\beta}_{[k],\mathcal{C}_{n}}^{\trans}\e\left[\widetilde{Y}_{i}(1)\widetilde{\bs{\xi}}_{n}^{*}(\bs X_{i})\mid B_{i}=k\right]\nonumber \\
 & -2\sum_{k=1}^{K}\1(B_{i}=k)(1-A_{i})\frac{\pi_{n[k]}}{1-\pi_{n[k]}}\bs{\beta}_{[k],\mathcal{C}_{n}}^{\trans}\e\left[\widetilde{Y}_{i}(0)\widetilde{\bs{\xi}}_{n}^{*}(\bs X_{i})\mid B_{i}=k\right]\nonumber \\
 & +\sum_{k=1}^{K}\1(B_{i}=k)(A_{i}-\pi_{n[k]})^{2}\bs{\beta}_{[k],\mathcal{C}_{n}}^{\trans}\e\left[\widetilde{\bs{\xi}}_{n}^{*}(\bs X_{i})\widetilde{\bs{\xi}}_{n}^{*}(\bs X_{i})^{\trans}\mid B_{i}=k\right]\bs{\beta}_{[k],\mathcal{C}_{n}}\nonumber \\
= & \sum_{k=1}^{K}\1(B_{i}=k)\left\{ \frac{A_{i}}{\pi_{n[k]}^{2}}\e\left[\widetilde{Y}_{i}^{2}(1)\mid B_{i}=k\right]+\frac{1-A_{i}}{(1-\pi_{n[k]})^{2}}\e\left[\widetilde{Y}_{i}^{2}(0)\mid B_{i}=k\right]\right\} \nonumber \\
 & -2\sum_{k=1}^{K}\1(B_{i}=k)\bs{\beta}_{[k],\mathcal{C}_{n}}^{\trans}\left\{ A_{i}\frac{1-\pi_{n[k]}}{\pi_{n[k]}}\mathbf{\Sigma}_{[k]\widetilde{\bs{\xi}}_{n}^{*}\widetilde{Y}(1)}+(1-A_{i})\frac{\pi_{n[k]}}{1-\pi_{n[k]}}\mathbf{\Sigma}_{[k]\widetilde{\bs{\xi}}_{n}^{*}\widetilde{Y}(0)}\right\} \nonumber \\
 & +\sum_{k=1}^{K}\1(B_{i}=k)(A_{i}-\pi_{n[k]})^{2}\bs{\beta}_{[k],\mathcal{C}_{n}}^{\trans}\mathbf{\Sigma}_{[k]\widetilde{\bs{\xi}}_{n}^{*}\widetilde{\bs{\xi}}_{n}^{*}}\bs{\beta}_{[k],\mathcal{C}_{n}}\label{eq:var(sum delta)-1}
\end{align}
and 
\begin{align*}
\sigma_{n}^{2}(\mathcal{C}_{n}) & :=\var\left(\sum_{i=1}^{n}\sum_{k=1}^{K}\epsilon_{i,[k]}^{*}\mid\mathcal{C}_{n}\right)=\sum_{i=1}^{n}\var\left(\sum_{k=1}^{K}\epsilon_{i,[k]}^{*}\mid\mathcal{C}_{n}\right)\\
 & =\sum_{i=1}^{n}\sum_{k=1}^{K}\1(B_{i}=k)\left\{ \frac{A_{i}}{\pi_{n[k]}^{2}}\e\left[\widetilde{Y}_{i}^{2}(1)\mid B_{i}=k\right]+\frac{1-A_{i}}{(1-\pi_{n[k]})^{2}}\e\left[\widetilde{Y}_{i}^{2}(0)\mid B_{i}=k\right]\right\} \\
 & \quad-2\sum_{i=1}^{n}\sum_{k=1}^{K}\1(B_{i}=k)\bs{\beta}_{[k],\mathcal{C}_{n}}^{\trans}\left\{ A_{i}\frac{1-\pi_{n[k]}}{\pi_{n[k]}}\mathbf{\Sigma}_{[k]\widetilde{\bs{\xi}}_{n}^{*}\widetilde{Y}(1)}+(1-A_{i})\frac{\pi_{n[k]}}{1-\pi_{n[k]}}\mathbf{\Sigma}_{[k]\widetilde{\bs{\xi}}_{n}^{*}\widetilde{Y}(0)}\right\} \\
 & \quad+\sum_{i=1}^{n}\sum_{k=1}^{K}\1(B_{i}=k)(A_{i}-\pi_{n[k]})^{2}\bs{\beta}_{[k],\mathcal{C}_{n}}^{\trans}\mathbf{\Sigma}_{[k]\widetilde{\bs{\xi}}_{n}^{*}\widetilde{\bs{\xi}}_{n}^{*}}\bs{\beta}_{[k],\mathcal{C}_{n}}.
\end{align*}
By Lemma \ref{lem:LLN}, $\left|\pi_{n[k]}-\pi_{[k]}\right|=o_{P}(1)$,
$\left\Vert \mathbf{\Sigma}_{[k]\widetilde{\bs{\xi}}_{n}^{*}\widetilde{\bs{\xi}}_{n}^{*}}\right\Vert =O(1)$,
$\left\Vert \bs{\beta}_{[k],\mathcal{C}_{n}}\right\Vert =O_{P}(1)$,
$\left\Vert \mathbf{\Sigma}_{[k]\widetilde{\bs{\xi}}_{n}^{*}\widetilde{Y}(a)}\right\Vert =O(1)$
and (\ref{eq:betak C_n is consistent}) we have 
\begin{align}
\frac{\sigma_{n}^{2}(\mathcal{C}_{n})}{n} & =\sum_{k=1}^{K}\left\{ \frac{p_{[k]}}{\pi_{[k]}}\e\left[\widetilde{Y}_{i}^{2}(1)\mid B_{i}=k\right]+\frac{p_{[k]}}{1-\pi_{[k]}}\e\left[\widetilde{Y}_{i}^{2}(0)\mid B_{i}=k\right]\right\} \nonumber \\
 & \quad-2\sum_{k=1}^{K}p_{[k]}\bs{\beta}_{[k]}^{*\trans}\left\{ (1-\pi_{[k]})\mathbf{\Sigma}_{[k]\widetilde{\bs{\xi}}_{n}^{*}\widetilde{Y}(1)}+\pi_{[k]}\mathbf{\Sigma}_{[k]\widetilde{\bs{\xi}}_{n}^{*}\widetilde{Y}(0)}\right\} \nonumber \\
 & \quad+\sum_{k=1}^{K}p_{[k]}\pi_{[k]}(1-\pi_{[k]})\bs{\beta}_{[k]}^{*\trans}\mathbf{\Sigma}_{[k]\widetilde{\bs{\xi}}_{n}^{*}\widetilde{\bs{\xi}}_{n}^{*}}\bs{\beta}_{[k]}^{*}+o_{P}(1)\nonumber \\
 & =\varsigma_{\widetilde{Y}}^{2}-\varsigma_{\widetilde{Y}\mid\widetilde{\bs{\xi}}_{n}^{*}}^{2}+o_{P}(1)\label{eq:limit of sigma_n-1}\\
 & =O_{P}(1).\nonumber 
\end{align}
According to (\ref{eq:var(sum delta)-1}), (\ref{eq:limit of sigma_n-1})
and Assumption \ref{assu:conditions on =00005Cxi-diverging xi}, Condition
(\ref{eq: positive variance in clt}) in Lemma \ref{lem:clt for triangular array}
is satisfied.

It remains to verify the Lindeberg condition. Using Assumptions \ref{assu:independent sampling}
and \ref{assu:conditions on =00005Cxi-diverging xi} and (\ref{eq:betak C_n is bounded}),
we have 
\begin{align*}
 & \e\left[\left|\sum_{k=1}^{K}\epsilon_{i,[k]}^{*}\right|^{2+\epsilon}\mid\mathcal{C}_{n}\right]\\
\leq & 3^{1+\epsilon}\e\left[\left|\sum_{k=1}^{K}\1(B_{i}=k)\frac{A_{i}}{\pi_{n[k]}}\right|^{2+\epsilon}\left|\widetilde{Y}_{i}(1)\right|^{2+\epsilon}+\left|\sum_{k=1}^{K}\1(B_{i}=k)\frac{1-A_{i}}{1-\pi_{n[k]}}\right|^{2+\epsilon}\left|\widetilde{Y}_{i}(0)\right|^{2+\epsilon}\mid\mathcal{C}_{n}\right]\\
 & +3^{1+\epsilon}\e\left[\left|\sum_{k=1}^{K}\1(B_{i}=k)(A_{i}-\pi_{n[k]})\right|^{2+\epsilon}\left|\bs{\beta}_{[k],\mathcal{C}_{n}}^{\trans}\widetilde{\bs{\xi}}_{n}^{*}(\bs X_{i})\right|^{2+\epsilon}\mid\mathcal{C}_{n}\right]\\
\leq & \frac{2\times3^{1+\epsilon}}{\min_{1\leq k\leq K}\left\{ \pi_{n[k]}^{2+\epsilon},(1-\pi_{n[k]})^{2+\epsilon}\right\} }\left\{ \sup_{a\in\{0,1\},1\leq k\leq K}\e\left[\left|\widetilde{Y}_{i}(a)\right|^{2+\epsilon}\mid B_{i}=k\right]+\e\left[\left|\bs{\beta}_{[k],\mathcal{C}_{n}}^{\trans}\widetilde{\bs{\xi}}_{n}^{*}(\bs X_{i})\right|^{2+\epsilon}\mid\mathcal{C}_{n}\right]\right\} \\
\leq & \frac{2\times3^{1+\epsilon}}{\min_{1\leq k\leq K}\left\{ \pi_{n[k]}^{2+\epsilon},(1-\pi_{n[k]})^{2+\epsilon}\right\} }\left\{ O(1)+\zeta_{n}^{\epsilon}\bs{\beta}_{[k],\mathcal{C}_{n}}^{\trans}\sup_{1\leq k\leq K}\e\left[\widetilde{\bs{\xi}}_{n}^{*}(\bs X_{i})\widetilde{\bs{\xi}}_{n}^{*}(\bs X_{i})^{\trans}\mid B_{i}=k\right]\bs{\beta}_{[k],\mathcal{C}_{n}}\right\} \\
= & O_{P}(\zeta_{n}^{\epsilon}),
\end{align*}
where the first step follows from Loève’s $c_{r}$ inequality (\citealp[Theorem 9.32]{davidson2021Stochastic}),
the second step follows from 
\[
\max\left\{ \sum_{k=1}^{K}\1(B_{i}=k)\frac{A_{i}}{\pi_{n[k]}},\sum_{k=1}^{K}\1(B_{i}=k)\frac{1-A_{i}}{1-\pi_{n[k]}}\right\} \leq\frac{1}{\min_{1\leq k\leq K}\left\{ \pi_{n[k]},(1-\pi_{n[k]})\right\} }
\]
and 
\[
\e\left[\left|\widetilde{Y}_{i}(a)\right|^{2+\epsilon}\mid\mathcal{C}_{n}\right]=\e\left[\left|\widetilde{Y}_{i}(a)\right|^{2+\epsilon}\mid B_{i}\right]\leq2^{2+\epsilon}\sup_{a\in\{0,1\},1\leq k\leq K}\e\left[\left|Y_{i}(a)\right|^{2+\epsilon}\mid B_{i}=k\right],
\]
and the third step follows from 
\[
\e\left[\widetilde{\bs{\xi}}_{n}^{*}(\bs X_{i})\widetilde{\bs{\xi}}_{n}^{*}(\bs X_{i})^{\trans}\mid\mathcal{C}_{n}\right]=\e\left[\widetilde{\bs{\xi}}_{n}^{*}(\bs X_{i})\widetilde{\bs{\xi}}_{n}^{*}(\bs X_{i})^{\trans}\mid B_{i}=k\right]\leq\sup_{1\leq k\leq K}\e\left[\widetilde{\bs{\xi}}_{n}^{*}(\bs X_{i})\widetilde{\bs{\xi}}_{n}^{*}(\bs X_{i})^{\trans}\mid B_{i}=k\right].
\]
Then, the Lindeberg condition in Lemma \ref{lem:clt for triangular array}
follows from 
\begin{align*}
 & \frac{1}{\sigma_{n}^{2}(\mathcal{C}_{n})}\sum_{i=1}^{n}\e\left[\left\{ \sum_{k=1}^{K}\epsilon_{i,[k]}^{*}\right\} ^{2}\1\left\{ \left|\sum_{k=1}^{K}\epsilon_{i,[k]}^{*}\right|\geq\delta\sigma_{n}(\mathcal{C}_{n})\right\} \mid\mathcal{C}_{n}\right]\\
\leq & \frac{1}{\sigma_{n}^{2}(\mathcal{C}_{n})}\sum_{i=1}^{n}\frac{\e\left[\left|\sum_{k=1}^{K}\epsilon_{i,[k]}^{*}\right|^{2+\epsilon}\mid\mathcal{C}_{n}\right]}{\left\{ \delta\sigma_{n}(\mathcal{C}_{n})\right\} ^{\epsilon}}=\frac{1}{\varsigma_{\widetilde{Y}}^{2}-\varsigma_{\widetilde{Y}\mid\widetilde{\bs{\xi}}_{n}^{*}}^{2}+o_{P}(1)}\frac{O_{P}(\zeta_{n}^{\epsilon})}{O_{P}(n^{\epsilon/2})}\frac{1}{\delta^{\ep}}=o_{P}(1)
\end{align*}
for any $\delta>0$, where the first step follows from (\ref{eq:limit of sigma_n-1})
and the last step follows from $K^{2}\zeta_{n}^{2}r_{n}/n\to0$ as
$n\to\infty$. Applying Lemma \ref{lem:clt for triangular array}
we have 
\[
\e\left[\left|\e\left[\exp\left\{ \mathrm{i}t\frac{\sum_{i=1}^{n}\sum_{k=1}^{K}\epsilon_{i,[k]}^{*}}{\sigma_{n}(\mathcal{C}_{n})}\right\} \mid\mathcal{C}_{n}\right]-\exp\left(-t^{2}/2\right)\right|\right]\to0
\]
as $n\to\infty$ for all $t\in\R$.

Recalling $R_{3}$ and $R_{4}$ defined in (\ref{eq:main estimator decomposition}),
by the classical central limit theorem we have 
\[
\e\left[\left|\e\left[\exp\left\{ \mathrm{i}s\frac{\sqrt{n}(R_{3}+R_{4})}{\varsigma_{H}}\right\} \mid\mathcal{C}_{n}\right]-\exp\left(-s^{2}/2\right)\right|\right]\to0
\]
for any $s\in\R$. Note that $R_{3}+R_{4}$ is measurable with respect
to $\mathcal{C}_{n}$, by Lemma \ref{lem:joint convergence} we have
\[
\left(\frac{\sum_{i=1}^{n}\sum_{k=1}^{K}\epsilon_{i,[k]}^{*}}{\sigma_{n}(\mathcal{C}_{n})},\frac{\sqrt{n}(R_{3}+R_{4})}{\varsigma_{H}}\right)^{\trans}\tod N\left(\begin{pmatrix}0\\
0
\end{pmatrix},\begin{pmatrix}1 & 0\\
0 & 1
\end{pmatrix}\right).
\]

Combining (\ref{eq:decomposition of tau_hat-tau*}), (\ref{eq:Q_2+Q_1}),
(\ref{eq:Q1}) and (\ref{eq:Q2}), we have 
\begin{align*}
 & \frac{\sqrt{n}\left(\widehat{\tau}_{\mathrm{cal}}-\tau\right)}{\sqrt{\varsigma_{H}^{2}+\varsigma_{\widetilde{Y}}^{2}-\varsigma_{\widetilde{Y}\mid\widetilde{\bs{\xi}}_{n}^{*}}^{2}}}\\
= & \frac{R_{1}+R_{2}+R_{3}+R_{4}}{\sqrt{\varsigma_{H}^{2}+\varsigma_{\widetilde{Y}}^{2}-\varsigma_{\widetilde{Y}\mid\widetilde{\bs{\xi}}_{n}^{*}}^{2}}}=\frac{1}{\sqrt{n}}\frac{\sum_{i=1}^{n}\sum_{k=1}^{K}\epsilon_{i,[k]}^{*}}{\sqrt{\varsigma_{H}^{2}+\varsigma_{\widetilde{Y}}^{2}-\varsigma_{\widetilde{Y}\mid\widetilde{\bs{\xi}}_{n}^{*}}^{2}}}+\frac{\sqrt{n}(R_{3}+R_{4})}{\sqrt{\varsigma_{H}^{2}+\varsigma_{\widetilde{Y}}^{2}-\varsigma_{\widetilde{Y}\mid\widetilde{\bs{\xi}}_{n}^{*}}^{2}}}+o_{P}(1)\\
= & \frac{\sigma_{n}(\mathcal{C}_{n})}{\sqrt{n}\sqrt{\varsigma_{H}^{2}+\varsigma_{\widetilde{Y}}^{2}-\varsigma_{\widetilde{Y}\mid\widetilde{\bs{\xi}}_{n}^{*}}^{2}}}\frac{\sum_{i=1}^{n}\sum_{k=1}^{K}\epsilon_{i,[k]}^{*}}{\sigma_{n}(\mathcal{C}_{n})}+\frac{\varsigma_{H}}{\sqrt{\varsigma_{H}^{2}+\varsigma_{\widetilde{Y}}^{2}-\varsigma_{\widetilde{Y}\mid\widetilde{\bs{\xi}}_{n}^{*}}^{2}}}\frac{\sqrt{n}(R_{3}+R_{4})}{\varsigma_{H}}+o_{P}(1)\\
= & \frac{\sqrt{\varsigma_{\widetilde{Y}}^{2}-\varsigma_{\widetilde{Y}\mid\widetilde{\bs{\xi}}_{n}^{*}}^{2}}+o_{P}(1)}{\sqrt{\varsigma_{H}^{2}+\varsigma_{\widetilde{Y}}^{2}-\varsigma_{\widetilde{Y}\mid\widetilde{\bs{\xi}}_{n}^{*}}^{2}}}\frac{\sum_{i=1}^{n}\sum_{k=1}^{K}\epsilon_{i,[k]}^{*}}{\sigma_{n}(\mathcal{C}_{n})}+\frac{\varsigma_{H}}{\sqrt{\varsigma_{H}^{2}+\varsigma_{\widetilde{Y}}^{2}-\varsigma_{\widetilde{Y}\mid\widetilde{\bs{\xi}}_{n}^{*}}^{2}}}\frac{\sqrt{n}(R_{3}+R_{4})}{\varsigma_{H}}+o_{P}(1)\\
= & \frac{\sqrt{\varsigma_{\widetilde{Y}}^{2}-\varsigma_{\widetilde{Y}\mid\widetilde{\bs{\xi}}_{n}^{*}}^{2}}}{\sqrt{\varsigma_{H}^{2}+\varsigma_{\widetilde{Y}}^{2}-\varsigma_{\widetilde{Y}\mid\widetilde{\bs{\xi}}_{n}^{*}}^{2}}}\frac{\sum_{i=1}^{n}\sum_{k=1}^{K}\epsilon_{i,[k]}^{*}}{\sigma_{n}(\mathcal{C}_{n})}+\frac{\varsigma_{H}}{\sqrt{\varsigma_{H}^{2}+\varsigma_{\widetilde{Y}}^{2}-\varsigma_{\widetilde{Y}\mid\widetilde{\bs{\xi}}_{n}^{*}}^{2}}}\frac{\sqrt{n}(R_{3}+R_{4})}{\varsigma_{H}}+o_{P}(1)\\
\tod & N(0,1),
\end{align*}
where the third equality follows from (\ref{eq:limit of sigma_n-1}),
the fourth one follows from $\sum_{i=1}^{n}\sum_{k=1}^{K}\epsilon_{i,[k]}^{*}/\sigma_{n}(\mathcal{C}_{n})$
is uniformly tight and the last step follows from Lemma \ref{lem:covergence of combin of normal}
and Slutsky's theorem. The asymptotic normality of $\widehat{\tau}_{\mathrm{cal}}$
is now established.

\textbf{Consistency of the variance estimator.} We only show the consistency
of $\widehat{\varsigma}_{\widetilde{Y}\mid\widetilde{\bs{\xi}}_{n}^{*}}^{2}$
as the consistency of $\widehat{\varsigma}_{\widetilde{Y}}^{2}$ can
be established in a similar manner.

By (\ref{eq:consistency of inverse of var-1}), we have 
\[
\left\Vert \widehat{\mathbf{\Sigma}}_{[k]}^{+}-\pi_{[k]}^{-1}(1-\pi_{[k]})^{-1}p_{[k]}^{-1}\mathbf{\Sigma}_{[k]\widetilde{\bs{\xi}}_{n}^{*}\widetilde{\bs{\xi}}_{n}^{*}}^{+}\right\Vert =p_{[k]}^{-1}o_{P}(1),\ \forall1\leq k\leq K.
\]
Note that 
\begin{align*}
 & \e_{\mathcal{C}_{n}}\left[\left\Vert \bs{\Xi}_{i,[k]}^{*}\widetilde{Y}_{i,[k]}^{*}-\e_{\mathcal{C}_{n}}\left[\bs{\Xi}_{i,[k]}^{*}\widetilde{Y}_{i,[k]}^{*}\right]\right\Vert ^{2}\right]\leq\e_{\mathcal{C}_{n}}\left[\left\Vert \bs{\Xi}_{i,[k]}^{*}\widetilde{Y}_{i,[k]}^{*}\right\Vert ^{2}\right]\\
\leq & 2\max\left\{ \frac{1-\pi_{n[k]}}{\pi_{n[k]}},\frac{\pi_{n[k]}}{1-\pi_{n[k]}}\right\} ^{2}\1(B_{i}=k)\max_{a=0,1}\e\left[\left\Vert \widetilde{\bs{\xi}}_{n}^{*}(\bs X_{i})\widetilde{Y}_{i}(a)\right\Vert ^{2}\mid B_{i}=k\right]\\
\leq & O_{P}(\zeta_{n}^{2})\max_{a=0,1}\e\left[\left\Vert \widetilde{Y}_{i}(a)\right\Vert ^{2}\mid B_{i}=k\right]\1(B_{i}=k)=O_{P}(\zeta_{n}^{2})\1(B_{i}=k),
\end{align*}
where the last step follows from Assumptions \ref{assu:independent sampling}
and \ref{assu:conditions on =00005Cxi-diverging xi} and the random
variable $O_{P}(\zeta_{n}^{2})$ does not depend on $k$. Thus, by
Lemma \ref{lem:LLN} we have 
\begin{align*}
 & \e_{\mathcal{C}_{n}}\left[\left\Vert \frac{1}{n}\sum_{i=1}^{n}\left\{ \bs{\Xi}_{i,[k]}^{*}\widetilde{Y}_{i,[k]}^{*}-\e_{\mathcal{C}_{n}}\left[\bs{\Xi}_{i,[k]}^{*}\widetilde{Y}_{i,[k]}^{*}\right]\right\} \right\Vert ^{2}\right]\\
\leq & \frac{1}{n^{2}}\sum_{i=1}^{n}\e_{\mathcal{C}_{n}}\left[\left\Vert \bs{\Xi}_{i,[k]}^{*}\widetilde{Y}_{i,[k]}^{*}-\e_{\mathcal{C}_{n}}\left[\bs{\Xi}_{i,[k]}^{*}\widetilde{Y}_{i,[k]}^{*}\right]\right\Vert ^{2}\right]\\
\leq & \frac{1}{n}\frac{n_{[k]}}{n}O_{P}(\zeta_{n}^{2})=\frac{\zeta_{n}^{2}}{n}p_{[k]}O_{P}(1)
\end{align*}
and 
\[
\left\Vert \frac{1}{n}\sum_{i=1}^{n}\left\{ \bs{\Xi}_{i,[k]}^{*}\widetilde{Y}_{i,[k]}^{*}-\e_{\mathcal{C}_{n}}\left[\bs{\Xi}_{i,[k]}^{*}\widetilde{Y}_{i,[k]}^{*}\right]\right\} \right\Vert =\sqrt{\frac{\zeta_{n}^{2}p_{[k]}}{n}}O_{L^{2},\mathcal{C}_{n}}(1)O_{P}(1).
\]
This, combined with (\ref{eq:xiY=00003Dxi*Y*+error}) and (\ref{eq:consistency of conditional Y*XI*})
together yields that 
\begin{align*}
 & \left\Vert \widehat{\mathbf{\Gamma}}_{[k]}-(1-\pi_{[k]})p_{[k]}\mathbf{\Sigma}_{[k]\widetilde{\bs{\xi}}_{n}^{*}\widetilde{Y}(1)}-\pi_{[k]}p_{[k]}\mathbf{\Sigma}_{[k]\widetilde{\bs{\xi}}_{n}^{*}\widetilde{Y}(0)}\right\Vert \\
= & \left\{ \sqrt{\frac{\zeta_{n}^{2}p_{[k]}}{n}}+\frac{\sqrt{r_{n}\log(2Kd)}}{n}+p_{[k]}n^{-1/4}\right\} O_{P}(1)O_{L^{2},\mathcal{C}_{n}}(1)+p_{[k]}o_{P}(1)\\
= & p_{[k]}o_{P}(1)O_{L^{2},\mathcal{C}_{n}}(1).
\end{align*}
Then it follows from $\left\Vert \mathbf{\Sigma}_{[k]\widetilde{\bs{\xi}}_{n}^{*}\widetilde{Y}(a)}\right\Vert =O(1)$,
$\left\Vert \mathbf{\Sigma}_{[k]\widetilde{\bs{\xi}}_{n}^{*}\widetilde{\bs{\xi}}_{n}^{*}}^{+}\right\Vert =O(1)$,
\begin{align*}
\sup_{1\leq k\leq K}\left|\frac{\widehat{r}_{[k]}}{np_{[k]}}\right| & \le\sup_{1\leq k\leq K}\left|\frac{r_{[k]}}{np_{[k]}}\right|+\sup_{1\leq k\leq K}\left|\frac{r_{[k]}}{np_{[k]}}\right|\left|\frac{\widehat{r}_{[k]}-r_{[k]}}{r_{[k]}}\right|\\
 & \leq\frac{r_{n}}{n\min_{1\leq k\le K}p_{[k]}}+\frac{r_{n}}{n\min_{1\leq k\le K}p_{[k]}}\sup_{1\leq k\leq K}\left|\frac{\widehat{r}_{[k]}-r_{[k]}}{r_{[k]}}\right|\\
 & =o_{P}(1)
\end{align*}
 and Assumption \ref{assu:treatment assignment} that 
\begin{align*}
\widehat{\varsigma}_{\widetilde{Y}\mid\widetilde{\bs{\xi}}_{n}^{*}}^{2} & =\sum_{k=1}^{K}\{1+o_{P}(1)\}\left\{ (1-\pi_{[k]})p_{[k]}\mathbf{\Sigma}_{[k]\widetilde{\bs{\xi}}_{n}^{*}\widetilde{Y}(1)}+\pi_{[k]}p_{[k]}\mathbf{\Sigma}_{[k]\widetilde{\bs{\xi}}_{n}^{*}\widetilde{Y}(0)}+p_{[k]}o_{P}(1)O_{L^{2},\mathcal{C}_{n}}(1)\right\} \\
 & \qquad\times\left\{ \pi_{[k]}^{-1}(1-\pi_{[k]})^{-1}p_{[k]}^{-1}\mathbf{\Sigma}_{[k]\widetilde{\bs{\xi}}_{n}^{*}\widetilde{\bs{\xi}}_{n}^{*}}^{+}+p_{[k]}^{-1}o_{P}(1)\right\} \\
 & \qquad\times\left\{ (1-\pi_{[k]})p_{[k]}\mathbf{\Sigma}_{[k]\widetilde{\bs{\xi}}_{n}^{*}\widetilde{Y}(1)}+\pi_{[k]}p_{[k]}\mathbf{\Sigma}_{[k]\widetilde{\bs{\xi}}_{n}^{*}\widetilde{Y}(0)}+p_{[k]}o_{P}(1)O_{L^{2},\mathcal{C}_{n}}(1)\right\} \\
 & =\varsigma_{\widetilde{Y}\mid\widetilde{\bs{\xi}}_{n}^{*}}^{2}+o_{P}(1)\sum_{k=1}^{K}p_{[k]}O_{L^{2},\mathcal{C}_{n}}(1)O_{L^{2},\mathcal{C}_{n}}(1)+o_{P}(1)\sum_{k=1}^{K}p_{[k]}O_{L^{2},\mathcal{C}_{n}}(1)\\
 & =\varsigma_{\widetilde{Y}\mid\widetilde{\bs{\xi}}_{n}^{*}}^{2}+o_{P}(1)\sum_{k=1}^{K}p_{[k]}O_{L^{2},\mathcal{C}_{n}}(1)O_{L^{2},\mathcal{C}_{n}}(1)\\
 & =\varsigma_{\widetilde{Y}\mid\widetilde{\bs{\xi}}_{n}^{*}}^{2}+o_{P}(1),
\end{align*}
where the last step holds is because by \citet[Lemma 6.1]{chernozhukov2018Double}
and 
\begin{align*}
 & \e_{\mathcal{C}_{n}}\left[\left|\sum_{k=1}^{K}p_{[k]}O_{L^{2},\mathcal{C}_{n}}(1)O_{L^{2},\mathcal{C}_{n}}(1)\right|\right]\leq\sum_{k=1}^{K}p_{[k]}\e_{\mathcal{C}_{n}}\left[\left|O_{L^{2},\mathcal{C}_{n}}(1)O_{L^{2},\mathcal{C}_{n}}(1)\right|\right]\\
\leq & \sum_{k=1}^{K}p_{[k]}\sqrt{\e_{\mathcal{C}_{n}}\left[O_{L^{2},\mathcal{C}_{n}}(1)^{2}\right]}\sqrt{\e_{\mathcal{C}_{n}}\left[O_{L^{2},\mathcal{C}_{n}}(1)^{2}\right]}\leq\sum_{k=1}^{K}p_{[k]}O_{P}(1)=O_{P}(1),
\end{align*}
it holds that 
\[
\left|\sum_{k=1}^{K}p_{[k]}O_{L^{2},\mathcal{C}_{n}}(1)O_{L^{2},\mathcal{C}_{n}}(1)\right|=O_{P}(1).
\]
The proof is completed.

The establishment of the semiparametric efficiency bound is identical
to that of Theorem \ref{thm:asymptotic properties of tau_cal}; thus
we omit it.$\hfill\qedsymbol$

\subsection{Proof of Theorem \ref{thm:asymptotic properties of tau_cal-general D(v)}\label{subsec:Proof-of-Theorem}}

\textbf{Derive the expressions for the calibration weights.} We first
derive the expressions for the calibration weights $\widehat{w}_{i}$'s.
For all $i\in\left\{ 1,\ldots,n\right\} $, let 
\[
\bs{\Xi}_{i}:=\begin{pmatrix}\left\{ A_{i}-\pi_{n[1]}\right\} \1(B_{i}=1)\left\{ \bs{\xi}_{n}(\bs X_{i})-\overline{\bs{\xi}}_{n[1]}\right\} \\
\left\{ A_{i}-\pi_{n[2]}\right\} \1(B_{i}=2)\left\{ \bs{\xi}_{n}(\bs X_{i})-\overline{\bs{\xi}}_{n[2]}\right\} \\
\vdots\\
\left\{ A_{i}-\pi_{n[K]}\right\} \1(B_{i}=K)\left\{ \bs{\xi}_{n}(\bs X_{i})-\overline{\bs{\xi}}_{n[K]}\right\} 
\end{pmatrix}\in\R^{Kd}.
\]
Let $D(\bs w)=\sum_{i=1}^{n}D(w_{i},1)$ . Then the calibration problem
can be rewritten as 
\begin{equation}
\min_{w_{i}:1\leq i\leq n}D(\bs w)\ \text{ s.t. }\ \left(\bs{\Xi}_{1},\ldots,\bs{\Xi}_{n}\right)\bs w=0.\label{eq:cal step rewritten-1}
\end{equation}
Note that the conjugate function of $D(\bs w)$ is 
\[
D^{*}(\bs z)=\sum_{i=1}^{n}\left\{ z_{i}\cdot(D^{\prime})^{-1}(z_{i})-D\left\{ (D^{\prime})^{-1}(z_{i})\right\} \right\} =\sum_{i=1}^{n}-\rho(-z_{i}),\quad\forall\bs z=(z_{1},\ldots,z_{n})^{\trans},
\]
where $\rho(v):=D\left\{ (D^{\prime})^{-1}(-v)\right\} +v\cdot(D^{\prime})^{-1}(-v)$
for any $v\in\R$. By \citet{tseng1991Relaxation} (see also \citet{boyd2004Convex}),
the dual problem of (\ref{eq:cal step rewritten-1}) is 
\begin{align*}
\max_{\bs{\lambda}\in\R^{Kd}}\left\{ -D^{*}((\bs{\Xi}_{1},\ldots,\bs{\Xi}_{n})^{\trans}\bs{\lambda})\right\}  & =\max_{\bs{\lambda}\in\R^{Kd}}\sum_{i=1}^{n}\rho(-\bs{\lambda}^{\trans}\bs{\Xi}_{i})=\max_{\bs{\lambda}\in\R^{Kd}}\frac{1}{n}\sum_{i=1}^{n}\rho(\bs{\lambda}^{\trans}\bs{\Xi}_{i}).
\end{align*}
Let $\widehat{\bs{\lambda}}:=\arg\max_{\bs{\lambda}\in\R^{Kd}}\frac{1}{n}\sum_{i=1}^{n}\rho(\bs{\lambda}^{\trans}\bs{\Xi}_{i})$.
Then, by the first order condition, the solution of the minimization
problem (\ref{eq:cal step rewritten}) is given by 
\begin{equation}
\widehat{w}_{i}=\rho^{\prime}(\widehat{\bs{\lambda}}^{\trans}\bs{\Xi}_{i})\text{ for all }i\in\{1,\ldots,n\},\label{eq:expression for wi}
\end{equation}
where $\widehat{\bs{\lambda}}$ satisfies 
\begin{equation}
\frac{1}{n}\sum_{i=1}^{n}\rho^{\prime}(\widehat{\bs{\lambda}}^{\trans}\bs{\Xi}_{i})\bs{\Xi}_{i}=0\label{eq:property of lambda_hat}
\end{equation}
and $\rho^{\prime}(\cdot)$ is the derivative of $\rho(\cdot)$. We
first derive the convergence rate of $\widehat{\bs{\lambda}}$.

\textbf{The convergence rate of $\widehat{\bs{\lambda}}$.} For any
$\epsilon>0$, let $\Upsilon(\epsilon):=\{\bs{\lambda}\in\R^{Kd}:\left\Vert \bs{\lambda}\right\Vert \leq C(\epsilon)n^{-1/2}\}$
and $\partial\Upsilon:=\{\bs{\lambda}\in\R^{Kd}:\left\Vert \bs{\lambda}\right\Vert =C(\epsilon)n^{-1/2}\}$,
where the constant $C(\epsilon)$ will be determined later. Let $\widehat{G}(\bs{\lambda}):=\frac{1}{n}\sum_{i=1}^{n}\rho(\bs{\lambda}^{\trans}\bs{\Xi}_{i})$,
then $\frac{\partial\widehat{G}(\bs{\lambda})}{\partial\bs{\lambda}}=\frac{1}{n}\sum_{i=1}^{n}\rho^{\prime}(\bs{\lambda}^{\trans}\bs{\Xi}_{i})\bs{\Xi}_{i}$
and 
\[
\frac{\partial\widehat{G}(\bs 0)}{\partial\bs{\lambda}}=\frac{1}{n}\sum_{i=1}^{n}\bs{\Xi}_{i}=O_{P}(n^{-1/2}),
\]
where the last step follows from (\ref{eq: mean of xi is root-n}).
Then, there exists a constant $M(\epsilon)>0$ such that 
\begin{equation}
\left\Vert \frac{\partial\widehat{G}(\bs 0)}{\partial\bs{\lambda}}\right\Vert \leq M(\epsilon)\times n^{-1/2}\label{eq: bound for G derivative}
\end{equation}
with probability at least $1-\epsilon$ for sufficiently large $n$.

Note that 
\begin{align*}
\e\left[\max_{1\leq i\leq n}\left\Vert \bs{\Xi}_{i}\right\Vert ^{2+\epsilon}\right] & \leq\sum_{i=1}^{n}\e\left[\left\Vert \bs{\Xi}_{i}\right\Vert ^{2+\epsilon}\right]=\sum_{i=1}^{n}\e\left[\left\{ \sum_{k=1}^{K}\left(A_{i}-\pi_{n[k]}\right)^{2}\1(B_{i}=k)\left\Vert \bs{\xi}_{n}(\bs X_{i})-\overline{\bs{\xi}}_{n[k]}\right\Vert ^{2}\right\} ^{\frac{2+\epsilon}{2}}\right]\\
 & \leq\sum_{i=1}^{n}\e\left[\left\{ \sum_{k=1}^{K}\left\{ 2\left\Vert \bs{\xi}_{n}(\bs X_{i})\right\Vert ^{2}+2\left\Vert \overline{\bs{\xi}}_{n[k]}\right\Vert ^{2}\right\} \right\} ^{\frac{2+\epsilon}{2}}\right]\\
 & \leq\sum_{i=1}^{n}(2K)^{\frac{2+\epsilon}{2}}2^{\frac{2+\epsilon}{2}}\e\left[\left\Vert \bs{\xi}_{n}(\bs X_{i})\right\Vert ^{2+\epsilon}+\max_{1\leq k\leq K}\left\Vert \overline{\bs{\xi}}_{n[k]}\right\Vert ^{2+\epsilon}\right]\\
 & \leq\sum_{i=1}^{n}(2K)^{\frac{2+\epsilon}{2}}2^{\frac{2+\epsilon}{2}}\e\left[\left\Vert \bs{\xi}_{n}(\bs X_{i})\right\Vert ^{2+\epsilon}\right]+\sum_{i=1}^{n}(2K)^{\frac{2+\epsilon}{2}}2^{\frac{2+\epsilon}{2}}\e\left[\sum_{k=1}^{K}\frac{1}{n_{[k]}}\sum_{i\in[k]}\left\Vert \bs{\xi}_{n}(\bs X_{i})\right\Vert ^{2+\epsilon}\right]\\
 & =O(n)
\end{align*}
where the third inequality uses the Power-Mean inequality, the fourth
one follows from
\[
\left\Vert \overline{\bs{\xi}}_{n[k]}\right\Vert ^{2+\epsilon}=\left\Vert \frac{1}{n_{[k]}}\sum_{i\in[k]}\bs{\xi}_{n}(\bs X_{i})\right\Vert ^{2+\epsilon}\leq\left\{ \frac{1}{n_{[k]}}\sum_{i\in[k]}\left\Vert \bs{\xi}_{n}(\bs X_{i})\right\Vert \right\} ^{2+\epsilon}\leq\frac{1}{n_{[k]}}\sum_{i\in[k]}\left\Vert \bs{\xi}_{n}(\bs X_{i})\right\Vert ^{2+\epsilon},
\]
and the last step follows from $\sup_{n\geq1}\sup_{1\leq i\leq n}\e\left[\left\Vert \bs{\xi}_{n}(\bs X_{i})\right\Vert ^{2+\epsilon}\right]<\infty$.
Then, by Markov's inequality we have 
\begin{equation}
\max_{1\leq i\leq n}\left\Vert \bs{\Xi}_{i}\right\Vert =O_{P}(n^{\frac{1}{2+\epsilon}}).\label{eq:maximal of Xi}
\end{equation}
As a result, by the Cauchy-Schwarz inequality we have 
\[
\max_{1\leq i\leq n}\sup_{\bs{\lambda}\in\Upsilon(\epsilon)}\left|\bs{\lambda}^{\trans}\bs{\Xi}_{i}\right|\leq O_{P}(n^{\frac{1}{2+\epsilon}})\times C(\epsilon)\times n^{-1/2}=o_{P}(1).
\]
Then it follows from $\rho^{\prime\prime}(v)$ is continuous at $v=0$
and $\rho^{\prime\prime}(0)<0$ that 
\begin{equation}
2\rho^{\prime\prime}(0)\leq\rho^{\prime\prime}(\bs{\lambda}^{\trans}\bs{\Xi}_{i})\leq\rho^{\prime\prime}(0)/2\ \ \forall\bs{\lambda}\in\Upsilon(\epsilon)\text{ and }i=1,\ldots,n\label{eq: range of rho''}
\end{equation}
with probability $1-o(1)$. From (\ref{eq:convergence of var each k})
and the definition of $\mathbf{\Sigma}$ we have 
\[
\frac{1}{n}\sum_{i=1}^{n}\bs{\Xi}_{i}\bs{\Xi}_{i}^{\trans}=\mathbf{\Sigma}+o_{P}(1).
\]
It follows from Weyl's inequality, Assumption \ref{assu:conditions on =00005Cxi}
and the assumption that $\mathbf{\Sigma}$ is non-singular we have
\begin{equation}
\lambda_{\text{min}}\left\{ \frac{1}{n}\sum_{i=1}^{n}\bs{\Xi}_{i}\bs{\Xi}_{i}^{\trans}\right\} \geq c/2\label{eq:minimal eigen value of sample var}
\end{equation}
with probability $1-o(1)$ and 
\[
\left\Vert \left\{ \frac{1}{n}\sum_{i=1}^{n}\bs{\Xi}_{i}\bs{\Xi}_{i}^{\trans}\right\} ^{+}-\mathbf{\Sigma}^{+}\right\Vert =o_{P}(1).
\]

By Taylor's expansion, for any $\bs{\lambda}\in\partial\Upsilon(\epsilon)$,
there exists $\widetilde{\bs{\lambda}}_{1}\in\Upsilon(\epsilon)$
such that 
\begin{align}
 & \widehat{G}(\bs{\lambda})-\widehat{G}(\bs 0)=\bs{\lambda}^{\trans}\frac{\partial\widehat{G}(\bs 0)}{\partial\bs{\lambda}}+\frac{1}{2}\bs{\lambda}^{\trans}\frac{\partial^{2}\widehat{G}(\widetilde{\bs{\lambda}}_{1})}{\partial\bs{\lambda}\partial\bs{\lambda}^{\trans}}\bs{\lambda}\nonumber \\
\leq & C(\epsilon)n^{-1/2}\left\Vert \frac{\partial\widehat{G}(\bs 0)}{\partial\bs{\lambda}}\right\Vert +\frac{1}{2}\bs{\lambda}^{\trans}\frac{1}{n}\sum_{i=1}^{n}\rho^{\prime\prime}(\widetilde{\bs{\lambda}}_{1}^{\trans}\bs{\Xi}_{i})\bs{\Xi}_{i}\bs{\Xi}_{i}^{\trans}\bs{\lambda}\nonumber \\
\leq & C(\epsilon)\times M(\epsilon)n^{-1}+\frac{\rho^{\prime\prime}(0)}{4}\bs{\lambda}^{\trans}\frac{1}{n}\sum_{i=1}^{n}\bs{\Xi}_{i}\bs{\Xi}_{i}^{\trans}\bs{\lambda}\nonumber \\
\leq & C(\epsilon)\times M(\epsilon)n^{-1}+\frac{\rho^{\prime\prime}(0)c}{8}\left(C(\epsilon)n^{-1/2}\right)^{2}=C(\epsilon)n^{-1}\left(M(\epsilon)+\frac{\rho^{\prime\prime}(0)c}{8}C(\epsilon)\right)\label{eq:taylor of G}
\end{align}
with probability $1-o(1)-\epsilon$, where the second inequality follows
from (\ref{eq: bound for G derivative}) and (\ref{eq: range of rho''})
and the last one follows from (\ref{eq:minimal eigen value of sample var}).
If we choose $C(\epsilon)=-16M(\epsilon)/(\rho^{\prime\prime}(0)c)$,
then 
\[
\widehat{G}(\bs{\lambda})-\widehat{G}(\bs 0)<0
\]
for any $\bs{\lambda}\in\partial\Upsilon(\epsilon)$. Since $\widehat{G}(\bs{\lambda})$
is continuous, there exists a local maximum of $\widehat{G}(\bs{\lambda})$
in the interior of $\Upsilon(\epsilon)$. Note that $\widehat{G}(\bs{\lambda})$
is also strictly concave with a unique global maximum point $\widehat{\bs{\lambda}}$,
we conclude that $\widehat{\bs{\lambda}}\in\Upsilon(\epsilon)\backslash\partial\Upsilon(\epsilon)$
and thus $\left\Vert \widehat{\bs{\lambda}}\right\Vert \leq C(\epsilon)n^{-1/2}$.
Since (\ref{eq:taylor of G}) holds with probability $1-o(1)-\epsilon$,
we have 
\[
\left\Vert \widehat{\bs{\lambda}}\right\Vert =O_{P}(n^{-1/2}).
\]

\textbf{Asymptotic expansion of $\widehat{\bs{\lambda}}$.} Recall
(\ref{eq:property of lambda_hat}). By Taylor's expansion, we have
\begin{equation}
\bs 0=\frac{1}{n}\sum_{i=1}^{n}\rho^{\prime}(\widehat{\bs{\lambda}}^{\trans}\bs{\Xi}_{i})\bs{\Xi}_{i}=\frac{1}{n}\sum_{i=1}^{n}\bs{\Xi}_{i}+\frac{1}{n}\sum_{i=1}^{n}\rho^{\prime\prime}(\widetilde{\bs{\lambda}}^{\trans}\bs{\Xi}_{i})\bs{\Xi}_{i}\bs{\Xi}_{i}^{\trans}\widehat{\bs{\lambda}}\label{eq:second order tylor of lambda}
\end{equation}
where $\widetilde{\bs{\lambda}}$ satisfies $\left\Vert \widetilde{\bs{\lambda}}\right\Vert \leq\left\Vert \widehat{\bs{\lambda}}\right\Vert $.
By $\left\Vert \widetilde{\bs{\lambda}}\right\Vert \leq\left\Vert \widehat{\bs{\lambda}}\right\Vert =O_{P}(n^{-1/2})$,
(\ref{eq:maximal of Xi}) and the mean value theorem, we have 
\begin{align*}
 & \sup_{1\leq i\leq n}\left|\rho^{\prime\prime}(\widetilde{\bs{\lambda}}^{\trans}\bs{\Xi}_{i})-\rho^{\prime\prime}(0)\right|\\
\leq & \sup_{1\leq i\leq n}\sup_{\bs{\lambda}:\left\Vert \bs{\lambda}\right\Vert \leq\left\Vert \widehat{\bs{\lambda}}\right\Vert }\left|\rho^{\prime\prime\prime}(\bs{\lambda}^{\trans}\bs{\Xi}_{i})\right|\sup_{1\leq i\leq n}\left|\widetilde{\bs{\lambda}}^{\trans}\bs{\Xi}_{i}\right|\\
\leq & \sup_{1\leq i\leq n}\sup_{\bs{\lambda}:\left\Vert \bs{\lambda}\right\Vert \leq\left\Vert \widehat{\bs{\lambda}}\right\Vert }\left|\rho^{\prime\prime\prime}(\bs{\lambda}^{\trans}\bs{\Xi}_{i})\right|\left\Vert \widehat{\bs{\lambda}}\right\Vert \sup_{1\leq i\leq n}\left\Vert \bs{\Xi}_{i}\right\Vert \\
= & \sup_{1\leq i\leq n}\sup_{\bs{\lambda}:\left\Vert \bs{\lambda}\right\Vert \leq\left\Vert \widehat{\bs{\lambda}}\right\Vert }\left|\rho^{\prime\prime\prime}(\bs{\lambda}^{\trans}\bs{\Xi}_{i})\right|o_{P}(1).
\end{align*}
Besides, we have 
\begin{align*}
\sup_{1\leq i\leq n}\sup_{\bs{\lambda}:\left\Vert \bs{\lambda}\right\Vert \leq\left\Vert \widehat{\bs{\lambda}}\right\Vert }\left|\rho^{\prime\prime\prime}(\bs{\lambda}^{\trans}\bs{\Xi}_{i})\right| & \leq\left|\rho^{\prime\prime\prime}(0)\right|+\sup_{1\leq i\leq n}\sup_{\bs{\lambda}:\left\Vert \bs{\lambda}\right\Vert \leq\left\Vert \widehat{\bs{\lambda}}\right\Vert }\left|\rho^{\prime\prime\prime}(\bs{\lambda}^{\trans}\bs{\Xi}_{i})-\rho^{\prime\prime\prime}(0)\right|\\
 & \leq\left|\rho^{\prime\prime\prime}(0)\right|+C_{\rho}\sup_{1\leq i\leq n}\sup_{\bs{\lambda}:\left\Vert \bs{\lambda}\right\Vert \leq\left\Vert \widehat{\bs{\lambda}}\right\Vert }\left|\bs{\lambda}^{\trans}\bs{\Xi}_{i}\right|=O_{P}(1).
\end{align*}
As a result, we have 
\begin{equation}
\sup_{1\leq i\leq n}\left|\rho^{\prime\prime}(\widetilde{\bs{\lambda}}^{\trans}\bs{\Xi}_{i})-\rho^{\prime\prime}(0)\right|=o_{P}(1)\label{eq:second derivatice of rho consistent}
\end{equation}
and 
\begin{align*}
 & \left|\frac{1}{n}\sum_{i=1}^{n}\rho^{\prime\prime}(\widetilde{\bs{\lambda}}^{\trans}\bs{\Xi}_{i})\bs{\Xi}_{i}\bs{\Xi}_{i}^{\trans}\widehat{\bs{\lambda}}-\frac{1}{n}\sum_{i=1}^{n}\rho^{\prime\prime}(0)\bs{\Xi}_{i}\bs{\Xi}_{i}^{\trans}\widehat{\bs{\lambda}}\right|\leq o_{P}(1)\left\Vert \frac{1}{n}\sum_{i=1}^{n}\bs{\Xi}_{i}\bs{\Xi}_{i}^{\trans}\right\Vert \left\Vert \widehat{\bs{\lambda}}\right\Vert \\
= & o_{P}(1)\left\Vert \mathbf{\Sigma}_{[k]}+o_{P}(1)\right\Vert \left\Vert \widehat{\bs{\lambda}}\right\Vert =o_{P}(n^{-1/2}).
\end{align*}
Then it follows from (\ref{eq:second order tylor of lambda}) that
\[
\widehat{\bs{\lambda}}=-\rho^{\prime\prime}(0)^{-1}\left(\frac{1}{n}\sum_{i=1}^{n}\bs{\Xi}_{i}\bs{\Xi}_{i}^{\trans}\right)^{-1}\frac{1}{n}\sum_{i=1}^{n}\bs{\Xi}_{i}+o_{P}(n^{-1/2}).
\]
By (\ref{eq:modification of tau_hat1})-(\ref{eq:modification of tau_hat2})
we have 
\begin{align*}
\widehat{\tau}_{\mathrm{cal}}-\tau & =\frac{1}{n}\sum_{i=1}^{n}\sum_{k=1}^{K}\left\{ \frac{A_{i}}{\pi_{n[k]}}-\frac{1-A_{i}}{1-\pi_{n[k]}}\right\} \1(B_{i}=k)\cdot Y_{i}-\tau\\
 & \qquad+\frac{1}{n}\sum_{i=1}^{n}(\widehat{w}_{i}-1)\sum_{k=1}^{K}\left\{ \frac{A_{i}}{\pi_{n[k]}}\left(Y_{i}-\overline{Y}_{1[k]}\right)-\frac{1-A_{i}}{1-\pi_{n[k]}}\left(Y_{i}-\overline{Y}_{0[k]}\right)\right\} \1(B_{i}=k)\\
 & =\frac{1}{n}\sum_{i=1}^{n}\sum_{k=1}^{K}\left\{ \frac{A_{i}}{\pi_{n[k]}}-\frac{1-A_{i}}{1-\pi_{n[k]}}\right\} \1(B_{i}=k)\cdot Y_{i}-\tau+\frac{1}{n}\sum_{i=1}^{n}(\rho^{\prime}(\widehat{\bs{\lambda}}^{\trans}\bs{\Xi}_{i})-1)\widetilde{Y}_{i}\\
 & =\frac{1}{n}\sum_{i=1}^{n}\sum_{k=1}^{K}\widetilde{Y}_{i}^{*}+\frac{1}{n}\sum_{i=1}^{n}(\rho^{\prime}(\widehat{\bs{\lambda}}^{\trans}\bs{\Xi}_{i})-1)\widetilde{Y}_{i}\\
 & \qquad+\sum_{k=1}^{K}\left\{ \e\left[Y_{i}(1)\mid B_{i}=k\right]-\e\left[Y_{i}(0)\mid B_{i}=k\right]\right\} \times\frac{1}{n}\sum_{i=1}^{n}\left(\1(B_{i}=k)-p_{[k]}\right).
\end{align*}
Recall (\ref{eq:main estimator decomposition}) in the proof of Theorem
\ref{thm:asymptotic properties of tau_cal}. If we can show that $\frac{1}{n}\sum_{i=1}^{n}(\rho^{\prime}(\widehat{\bs{\lambda}}^{\trans}\bs{\Xi}_{i})-1)\widetilde{Y}_{i}=-\frac{1}{n}\sum_{i=1}^{n}\sum_{k=1}^{K}\widehat{\bs{\beta}}_{[k]}^{\trans}\bs{\Xi}_{i,[k]}+o_{P}(n^{-1/2})$,
then the conclusion of Theorem \ref{thm:asymptotic properties of tau_cal}
holds. By the mean value theorem we have 
\begin{equation}
\frac{1}{n}\sum_{i=1}^{n}(\rho^{\prime}(\widehat{\bs{\lambda}}^{\trans}\bs{\Xi}_{i})-1)\widetilde{Y}_{i}=\widehat{\bs{\lambda}}^{\trans}\frac{1}{n}\sum_{i=1}^{n}\rho^{\prime\prime}(\bs{\lambda}_{1}^{\trans}\bs{\Xi}_{i})\widetilde{Y}_{i}\bs{\Xi}_{i}\label{eq:taylor of remainder term}
\end{equation}
where $\bs{\lambda}_{1}$ satisfies $\left\Vert \bs{\lambda}_{1}\right\Vert \leq\left\Vert \widehat{\bs{\lambda}}\right\Vert $.
Similar to the derivation of (\ref{eq:second derivatice of rho consistent}),
we can show that 
\[
\sup_{1\leq i\leq n}\left|\rho^{\prime\prime}(\bs{\lambda}_{1}^{\trans}\bs{\Xi}_{i})-\rho^{\prime\prime}(0)\right|=o_{P}(1).
\]
As a result, 
\[
\left|\frac{1}{n}\sum_{i=1}^{n}\rho^{\prime\prime}(\bs{\lambda}_{1}^{\trans}\bs{\Xi}_{i})\widetilde{Y}_{i}\bs{\Xi}_{i}-\frac{1}{n}\sum_{i=1}^{n}\rho^{\prime\prime}(0)\widetilde{Y}_{i}\bs{\Xi}_{i}\right|\leq\sup_{1\leq i\leq n}\left|\rho^{\prime\prime}(\bs{\lambda}_{1}^{\trans}\bs{\Xi}_{i})-\rho^{\prime\prime}(0)\right|\frac{1}{n}\sum_{i=1}^{n}\left\Vert \widetilde{Y}_{i}\bs{\Xi}_{i}\right\Vert .
\]
By the Cauchy-Schwarz inequality, we have 
\[
\frac{1}{n}\sum_{i=1}^{n}\left\Vert \widetilde{Y}_{i}\bs{\Xi}_{i}\right\Vert \leq\sqrt{\frac{1}{n}\sum_{i=1}^{n}\widetilde{Y}_{i}^{2}}\sqrt{\frac{1}{n}\sum_{i=1}^{n}\left\Vert \bs{\Xi}_{i}\right\Vert ^{2}}=\sqrt{\frac{1}{n}\sum_{i=1}^{n}\widetilde{Y}_{i}^{2}}\sqrt{\sum_{k=1}^{K}\frac{1}{n}\sum_{i=1}^{n}\left\Vert \bs{\Xi}_{i,[k]}\right\Vert ^{2}}.
\]
First, we have 
\begin{align*}
 & \frac{1}{n}\sum_{i=1}^{n}\left\Vert \bs{\Xi}_{i,[k]}-\bs{\Xi}_{i,[k]}^{*}\right\Vert ^{2}=\frac{1}{n}\sum_{i=1}^{n}(A_{i}-\pi_{n[k]})^{2}\1(B_{i}=k)\left\Vert \bs{\xi}_{n}(\bs X_{i})-\overline{\bs{\xi}}_{n[k]}-\widetilde{\bs{\xi}}_{n}^{*}(\bs X_{i})\right\Vert ^{2}\\
\leq & \frac{1}{n}\sum_{i=1}^{n}\1(B_{i}=k)\left\Vert \bs{\xi}_{n}(\bs X_{i})-\overline{\bs{\xi}}_{n[k]}-\widetilde{\bs{\xi}}_{n}^{*}(\bs X_{i})\right\Vert ^{2}\\
\leq & \frac{2}{n}\sum_{i=1}^{n}\1(B_{i}=k)\left\Vert \bs{\xi}_{n}(\bs X_{i})-\overline{\bs{\xi}}_{n[k]}-\bs{\xi}_{n}^{*}(\bs X_{i})-\overline{\bs{\xi}}_{n[k]}^{*}\right\Vert ^{2}+2\left\Vert \overline{\bs{\xi}}_{n[k]}^{*}-\e\left[\bs{\xi}_{n}^{*}(\bs X_{i})\mid B_{i}=k\right]\right\Vert ^{2}\\
\leq & 2R_{6}+o_{P}(1)=o_{P}(1),
\end{align*}
where the last two inequalities follows from (\ref{eq:R6 in main theorem})
and (\ref{eq:main theorem xibar-true mean}). Second, we have 
\[
\frac{1}{n}\sum_{i=1}^{n}\left\Vert \bs{\Xi}_{i,[k]}^{*}\right\Vert ^{2}\leq\frac{1}{n}\sum_{i=1}^{n}\1(B_{i}=k)\left\Vert \widetilde{\bs{\xi}}_{n}^{*}(\bs X_{i})\right\Vert ^{2}=O_{P}(1)
\]
by Markov's inequality and Assumption \ref{assu:conditions on =00005Cxi}.
Thus, we have 
\[
\frac{1}{n}\sum_{i=1}^{n}\left\Vert \bs{\Xi}_{i,[k]}\right\Vert ^{2}\leq2\frac{1}{n}\sum_{i=1}^{n}\left\Vert \bs{\Xi}_{i,[k]}-\bs{\Xi}_{i,[k]}^{*}\right\Vert ^{2}+2\frac{1}{n}\sum_{i=1}^{n}\left\Vert \bs{\Xi}_{i,[k]}^{*}\right\Vert ^{2}=O_{P}(1).
\]
Note that we have 
\[
\frac{1}{n}\sum_{i=1}^{n}\widetilde{Y}_{i}^{2}\leq K\frac{1}{n}\sum_{i=1}^{n}\sum_{k=1}^{K}\widetilde{Y}_{i,[k]}^{2}\leq O(1)\sum_{k=1}^{K}\left\{ \frac{1}{n}\sum_{i=1}^{n}Y_{i}^{2}+\overline{Y}_{1[k]}^{2}+\overline{Y}_{0[k]}^{2}\right\} =O_{P}(1).
\]
We can conclude that 
\[
\frac{1}{n}\sum_{i=1}^{n}\left\Vert \widetilde{Y}_{i}\bs{\Xi}_{i}\right\Vert \leq\sqrt{\frac{1}{n}\sum_{i=1}^{n}\widetilde{Y}_{i}^{2}}\sqrt{\sum_{k=1}^{K}\frac{1}{n}\sum_{i=1}^{n}\left\Vert \bs{\Xi}_{i,[k]}\right\Vert ^{2}}=O_{P}(1)
\]
and thus
\begin{equation}
\left|\frac{1}{n}\sum_{i=1}^{n}\rho^{\prime\prime}(\bs{\lambda}_{1}^{\trans}\bs{\Xi}_{i})\widetilde{Y}_{i}\bs{\Xi}_{i}-\frac{1}{n}\sum_{i=1}^{n}\rho^{\prime\prime}(0)\widetilde{Y}_{i}\bs{\Xi}_{i}\right|\leq\sup_{1\leq i\leq n}\left|\rho^{\prime\prime}(\bs{\lambda}_{1}^{\trans}\bs{\Xi}_{i})-\rho^{\prime\prime}(0)\right|\frac{1}{n}\sum_{i=1}^{n}\left\Vert \widetilde{Y}_{i}\bs{\Xi}_{i}\right\Vert =o_{P}(1).\label{eq:remainder taylor convergens}
\end{equation}
Recall (\ref{eq:taylor of remainder term}), then we have 
\begin{align*}
 & \frac{1}{n}\sum_{i=1}^{n}(\rho^{\prime}(\widehat{\bs{\lambda}}^{\trans}\bs{\Xi}_{i})-1)\widetilde{Y}_{i}=\widehat{\bs{\lambda}}^{\trans}\frac{1}{n}\sum_{i=1}^{n}\rho^{\prime\prime}(0)\widetilde{Y}_{i}\bs{\Xi}_{i}+\widehat{\bs{\lambda}}^{\trans}\left\{ \frac{1}{n}\sum_{i=1}^{n}\rho^{\prime\prime}(\bs{\lambda}_{1}^{\trans}\bs{\Xi}_{i})\widetilde{Y}_{i}\bs{\Xi}_{i}-\frac{1}{n}\sum_{i=1}^{n}\rho^{\prime\prime}(0)\widetilde{Y}_{i}\bs{\Xi}_{i}\right\} \\
= & \widehat{\bs{\lambda}}^{\trans}\frac{1}{n}\sum_{i=1}^{n}\rho^{\prime\prime}(0)\widetilde{Y}_{i}\bs{\Xi}_{i}+o_{P}(n^{-1/2})\\
= & -\frac{1}{n}\sum_{i=1}^{n}\widetilde{Y}_{i}\bs{\Xi}_{i}^{\trans}\left(\frac{1}{n}\sum_{i=1}^{n}\bs{\Xi}_{i}\bs{\Xi}_{i}^{\trans}\right)^{-1}\frac{1}{n}\sum_{i=1}^{n}\bs{\Xi}_{i}+o_{P}(n^{-1/2})\frac{1}{n}\sum_{i=1}^{n}\rho^{\prime\prime}(0)\widetilde{Y}_{i}\bs{\Xi}_{i}+o_{P}(n^{-1/2})\\
= & -\widehat{\bs{\beta}}^{\trans}\frac{1}{n}\sum_{i=1}^{n}\bs{\Xi}_{i}+o_{P}(n^{-1/2})=-\frac{1}{n}\sum_{i=1}^{n}\sum_{k=1}^{K}\widehat{\bs{\beta}}_{[k]}^{\trans}\bs{\Xi}_{i,[k]}+o_{P}(n^{-1/2}),
\end{align*}
where the second step follows from $\left\Vert \widehat{\bs{\lambda}}\right\Vert =O_{P}(n^{-1/2})$
and (\ref{eq:remainder taylor convergens}), the third step follows
from $\widehat{\bs{\lambda}}=-\rho^{\prime\prime}(0)^{-1}\left(\frac{1}{n}\sum_{i=1}^{n}\bs{\Xi}_{i}\bs{\Xi}_{i}^{\trans}\right)^{-1}\frac{1}{n}\sum_{i=1}^{n}\bs{\Xi}_{i}+o_{P}(n^{-1/2})$
and the fourth step follows from $\left\Vert \frac{1}{n}\sum_{i=1}^{n}\rho^{\prime\prime}(0)\widetilde{Y}_{i}\bs{\Xi}_{i}\right\Vert \leq\left|\rho^{\prime\prime}(0)\right|\frac{1}{n}\sum_{i=1}^{n}\left\Vert \widetilde{Y}_{i}\bs{\Xi}_{i}\right\Vert =O_{P}(1)$.
Therefore, the conclusion of Theorem \ref{thm:asymptotic properties of tau_cal}
holds.

\textbf{Now, we proceed to prove the second result.} Recall (\ref{eq:property of lambda_hat}).
By Taylor's expansion we have

\begin{align*}
\bs 0 & =\frac{1}{n}\sum_{i=1}^{n}\rho^{\prime}(\widehat{\bs{\lambda}}^{\trans}\bs{\Xi}_{i})\bs{\Xi}_{i}=\frac{1}{n}\sum_{i=1}^{n}\bs{\Xi}_{i}+\rho^{\prime\prime}(0)\frac{1}{n}\sum_{i=1}^{n}\bs{\Xi}_{i}\bs{\Xi}_{i}^{\trans}\widehat{\bs{\lambda}}\\
 & \qquad+\frac{1}{2}\frac{1}{n}\sum_{i=1}^{n}\rho^{\prime\prime\prime}(\widetilde{\bs{\lambda}}^{\trans}\bs{\Xi}_{i})\bs{\Xi}_{i}\left(\bs{\Xi}_{i}^{\trans}\widehat{\bs{\lambda}}\right)^{2},
\end{align*}
where $\widetilde{\bs{\lambda}}$ satisfies $\left\Vert \widetilde{\bs{\lambda}}\right\Vert \leq\left\Vert \widehat{\bs{\lambda}}\right\Vert $.
Then 
\begin{align*}
\widehat{\bs{\lambda}} & =-\left\{ \rho^{\prime\prime}(0)\frac{1}{n}\sum_{i=1}^{n}\bs{\Xi}_{i}\bs{\Xi}_{i}^{\trans}\right\} ^{-1}\left\{ \frac{1}{n}\sum_{i=1}^{n}\bs{\Xi}_{i}+\frac{1}{2}\frac{1}{n}\sum_{i=1}^{n}\rho^{\prime\prime\prime}(\widetilde{\bs{\lambda}}^{\trans}\bs{\Xi}_{i})\widehat{\bs{\lambda}}^{\trans}\bs{\Xi}_{i}\bs{\Xi}_{i}\bs{\Xi}_{i}^{\trans}\widehat{\bs{\lambda}}\right\} \\
 & =-\left\{ \rho^{\prime\prime}(0)\frac{1}{n}\sum_{i=1}^{n}\bs{\Xi}_{i}\bs{\Xi}_{i}^{\trans}\right\} ^{-1}\frac{1}{n}\sum_{i=1}^{n}\bs{\Xi}_{i}-\left\{ \rho^{\prime\prime}(0)\frac{1}{n}\sum_{i=1}^{n}\bs{\Xi}_{i}\bs{\Xi}_{i}^{\trans}\right\} ^{-1}\frac{1}{2}\frac{1}{n}\sum_{i=1}^{n}\rho^{\prime\prime\prime}(\widetilde{\bs{\lambda}}^{\trans}\bs{\Xi}_{i})\widehat{\bs{\lambda}}^{\trans}\bs{\Xi}_{i}\bs{\Xi}_{i}\bs{\Xi}_{i}^{\trans}\widehat{\bs{\lambda}}.
\end{align*}
By $\left\Vert \widetilde{\bs{\lambda}}\right\Vert \leq\left\Vert \widehat{\bs{\lambda}}\right\Vert =O_{P}(n^{-1/2})$,
(\ref{eq:maximal of Xi}) and Assumption \ref{assu:conditions on D(v)},
we have 
\[
\sup_{1\leq i\leq n}\left|\rho^{\prime\prime\prime}(\widetilde{\bs{\lambda}}^{\trans}\bs{\Xi}_{i})-\rho^{\prime\prime\prime}(0)\right|\leq C_{\rho}\sup_{1\leq i\leq n}\left|\widetilde{\bs{\lambda}}^{\trans}\bs{\Xi}_{i}\right|\leq C_{\rho}\left\Vert \widehat{\bs{\lambda}}\right\Vert \sup_{1\leq i\leq n}\left\Vert \bs{\Xi}_{i}\right\Vert =o_{P}(1).
\]
As a result, 
\begin{align*}
 & \left\Vert \frac{1}{n}\sum_{i=1}^{n}\rho^{\prime\prime\prime}(\widetilde{\bs{\lambda}}^{\trans}\bs{\Xi}_{i})\widehat{\bs{\lambda}}^{\trans}\bs{\Xi}_{i}\bs{\Xi}_{i}\bs{\Xi}_{i}^{\trans}\widehat{\bs{\lambda}}-\frac{1}{n}\sum_{i=1}^{n}\rho^{\prime\prime\prime}(0)\widehat{\bs{\lambda}}^{\trans}\bs{\Xi}_{i}\bs{\Xi}_{i}\bs{\Xi}_{i}^{\trans}\widehat{\bs{\lambda}}\right\Vert \\
\leq & \frac{1}{n}\sum_{i=1}^{n}\left|\rho^{\prime\prime\prime}(\widetilde{\bs{\lambda}}^{\trans}\bs{\Xi}_{i})-\rho^{\prime\prime\prime}(0)\right|\left\Vert \bs{\Xi}_{i}\right\Vert ^{3}\left\Vert \widehat{\bs{\lambda}}\right\Vert ^{2}\leq o_{P}(n^{-1})\cdot\frac{1}{n}\sum_{i=1}^{n}\left\Vert \bs{\Xi}_{i}\right\Vert ^{3}=o_{P}(n^{-1}),
\end{align*}
where the last step uses (\ref{eq:emprical 4th moment bounded}) and
the fact that 
\begin{align*}
\frac{1}{n}\sum_{i=1}^{n}\left\Vert \bs{\Xi}_{i}\right\Vert ^{3} & \leq\left\{ \frac{1}{n}\sum_{i=1}^{n}\left\Vert \bs{\Xi}_{i}\right\Vert ^{4}\right\} ^{3/4}\leq\left\{ \frac{1}{n}\sum_{i=1}^{n}\left\{ \sum_{k=1}^{K}\left\Vert \bs{\Xi}_{i,[k]}\right\Vert ^{2}\right\} ^{2}\right\} ^{3/4}\\
 & \leq\left\{ K\sum_{k=1}^{K}\frac{1}{n}\sum_{i=1}^{n}\left\Vert \bs{\Xi}_{i,[k]}\right\Vert ^{4}\right\} ^{3/4}=O_{P}(1).
\end{align*}
Then 
\begin{align}
\widehat{\bs{\lambda}} & =\rho^{\prime\prime}(0)^{-1}\widehat{\bs{\lambda}}_{lin}-\left\{ \rho^{\prime\prime}(0)\frac{1}{n}\sum_{i=1}^{n}\bs{\Xi}_{i}\bs{\Xi}_{i}^{\trans}\right\} ^{-1}\frac{1}{2}\frac{1}{n}\sum_{i=1}^{n}\rho^{\prime\prime\prime}(0)\widehat{\bs{\lambda}}^{\trans}\bs{\Xi}_{i}\bs{\Xi}_{i}\bs{\Xi}_{i}^{\trans}\widehat{\bs{\lambda}}+o_{P}(n^{-1}),\label{eq:lambda hat expression 1}
\end{align}
where we let 
\[
\widehat{\bs{\lambda}}_{lin}:=-\left\{ \frac{1}{n}\sum_{i=1}^{n}\bs{\Xi}_{i}\bs{\Xi}_{i}^{\trans}\right\} ^{-1}\frac{1}{n}\sum_{i=1}^{n}\bs{\Xi}_{i}.
\]
By (\ref{eq:minimal eigen value of sample var}), $\frac{1}{n}\sum_{i=1}^{n}\left\Vert \bs{\Xi}_{i}\right\Vert ^{3}=O_{P}(1)$
and $\left\Vert \widehat{\bs{\lambda}}\right\Vert =O_{P}(n^{-1/2})$,
we have 
\[
\left\{ \rho^{\prime\prime}(0)\frac{1}{n}\sum_{i=1}^{n}\bs{\Xi}_{i}\bs{\Xi}_{i}^{\trans}\right\} ^{-1}\frac{1}{2}\frac{1}{n}\sum_{i=1}^{n}\rho^{\prime\prime\prime}(0)\widehat{\bs{\lambda}}^{\trans}\bs{\Xi}_{i}\bs{\Xi}_{i}\bs{\Xi}_{i}^{\trans}\widehat{\bs{\lambda}}=O_{P}(n^{-1})
\]
and thus $\widehat{\bs{\lambda}}=\rho^{\prime\prime}(0)^{-1}\widehat{\bs{\lambda}}_{lin}+O_{P}(n^{-1})$.
Plugging this into (\ref{eq:lambda hat expression 1}) and using $\frac{1}{n}\sum_{i=1}^{n}\left\Vert \bs{\Xi}_{i}\right\Vert ^{3}=O_{P}(1)$
and $\left\Vert \widehat{\bs{\lambda}}\right\Vert =O_{P}(n^{-1/2})$
we have 
\begin{equation}
\widehat{\bs{\lambda}}=\rho^{\prime\prime}(0)^{-1}\widehat{\bs{\lambda}}_{lin}-\left\{ \rho^{\prime\prime}(0)\frac{1}{n}\sum_{i=1}^{n}\bs{\Xi}_{i}\bs{\Xi}_{i}^{\trans}\right\} ^{-1}\frac{1}{2n}\sum_{i=1}^{n}\frac{\rho^{\prime\prime\prime}(0)}{\rho^{\prime\prime}(0)^{3}}\widehat{\bs{\lambda}}_{lin}^{\trans}\bs{\Xi}_{i}\bs{\Xi}_{i}\bs{\Xi}_{i}^{\trans}\widehat{\bs{\lambda}}_{lin}+o_{P}(n^{-1}).\label{eq:lambda hat expression final}
\end{equation}

By (\ref{eq:modification of tau_hat1})-(\ref{eq:modification of tau_hat2})
we have 
\begin{align}
\widehat{\tau}_{\mathrm{cal}}-\tau & =\frac{1}{n}\sum_{i=1}^{n}\sum_{k=1}^{K}\left\{ \frac{A_{i}}{\pi_{n[k]}}-\frac{1-A_{i}}{1-\pi_{n[k]}}\right\} \1(B_{i}=k)\cdot Y_{i}-\tau\nonumber \\
 & \qquad+\frac{1}{n}\sum_{i=1}^{n}(\widehat{w}_{i}-1)\sum_{k=1}^{K}\left\{ \frac{A_{i}}{\pi_{n[k]}}\left(Y_{i}-\overline{Y}_{1[k]}\right)-\frac{1-A_{i}}{1-\pi_{n[k]}}\left(Y_{i}-\overline{Y}_{0[k]}\right)\right\} \1(B_{i}=k)\nonumber \\
 & =\frac{1}{n}\sum_{i=1}^{n}\sum_{k=1}^{K}\left\{ \frac{A_{i}}{\pi_{n[k]}}-\frac{1-A_{i}}{1-\pi_{n[k]}}\right\} \1(B_{i}=k)\cdot Y_{i}-\tau+\frac{1}{n}\sum_{i=1}^{n}(\rho^{\prime}(\widehat{\bs{\lambda}}^{\trans}\bs{\Xi}_{i})-1)\widetilde{Y}_{i}\nonumber \\
 & =\frac{1}{n}\sum_{i=1}^{n}\widetilde{Y}_{i}^{*}+\frac{1}{n}\sum_{i=1}^{n}(\rho^{\prime}(\widehat{\bs{\lambda}}^{\trans}\bs{\Xi}_{i})-1)\widetilde{Y}_{i}\nonumber \\
 & \qquad+\sum_{k=1}^{K}\left\{ \e\left[Y_{i}(1)\mid B_{i}=k\right]-\e\left[Y_{i}(0)\mid B_{i}=k\right]\right\} \times\frac{1}{n}\sum_{i=1}^{n}\left(\1(B_{i}=k)-p_{[k]}\right).\label{eq:decomposition of tau-tau*}
\end{align}

\textbf{Asymptotic expansion for $\frac{1}{n}\sum_{i=1}^{n}\widetilde{Y}_{i}^{*}+\frac{1}{n}\sum_{i=1}^{n}(\rho^{\prime}(\widehat{\bs{\lambda}}^{\trans}\bs{\Xi}_{i})-1)\widetilde{Y}_{i}$.}
We first deal with $\frac{1}{n}\sum_{i=1}^{n}(\rho^{\prime}(\widehat{\bs{\lambda}}^{\trans}\bs{\Xi}_{i})-1)\widetilde{Y}_{i}$.
By Taylor's expansion we have 
\begin{align*}
\frac{1}{n}\sum_{i=1}^{n}(\rho^{\prime}(\widehat{\bs{\lambda}}^{\trans}\bs{\Xi}_{i})-1)\widetilde{Y}_{i} & =\frac{1}{n}\sum_{i=1}^{n}\left\{ \rho^{\prime\prime}(0)\widehat{\bs{\lambda}}^{\trans}\bs{\Xi}_{i}+\frac{1}{2}\rho^{\prime\prime\prime}(\widetilde{\bs{\lambda}}^{\trans}\bs{\Xi}_{i})\widehat{\bs{\lambda}}^{\trans}\bs{\Xi}_{i}\bs{\Xi}_{i}^{\trans}\widehat{\bs{\lambda}}\right\} \widetilde{Y}_{i},
\end{align*}
where $\widetilde{\bs{\lambda}}$ satisfies $\left\Vert \widetilde{\bs{\lambda}}\right\Vert \leq\left\Vert \widehat{\bs{\lambda}}\right\Vert $.
By $\left\Vert \widetilde{\bs{\lambda}}\right\Vert \leq\left\Vert \widehat{\bs{\lambda}}\right\Vert =O_{P}(n^{-1/2})$,
(\ref{eq:maximal of Xi}) and Assumption \ref{assu:conditions on D(v)},
we have 
\[
\sup_{1\leq i\leq n}\left|\rho^{\prime\prime\prime}(\widetilde{\bs{\lambda}}^{\trans}\bs{\Xi}_{i})-\rho^{\prime\prime\prime}(0)\right|\leq\sup_{1\leq i\leq n}\left|\widetilde{\bs{\lambda}}^{\trans}\bs{\Xi}_{i}\right|\leq\left\Vert \widehat{\bs{\lambda}}\right\Vert \sup_{1\leq i\leq n}\left\Vert \bs{\Xi}_{i}\right\Vert =o_{P}(1).
\]
By $\frac{1}{n}\sum_{i=1}^{n}\left\Vert \bs{\Xi}_{i}\right\Vert ^{2}\left|\widetilde{Y}_{i}\right|=O_{P}(1)$\footnote{This follows from 
\[
\frac{1}{n}\sum_{i=1}^{n}\left\Vert \bs{\Xi}_{i}\right\Vert ^{2}\left|\widetilde{Y}_{i}\right|\leq\sqrt{\frac{1}{n}\sum_{i=1}^{n}\left\Vert \bs{\Xi}_{i}\right\Vert ^{4}}\left\{ \frac{1}{n}\sum_{i=1}^{n}\left|\widetilde{Y}_{i}\right|^{4}\right\} ^{1/4}=O_{P}(1),
\]
where the last step is due to (\ref{eq:emprical 4th moment bounded})
and (\ref{eq:sample average of Y^4}).} and $\widehat{\bs{\lambda}}=\widehat{\bs{\lambda}}_{lin}+O_{P}(n^{-1})$,
we have 
\[
\frac{1}{n}\sum_{i=1}^{n}\frac{1}{2}\rho^{\prime\prime\prime}(\widetilde{\bs{\lambda}}^{\trans}\bs{\Xi}_{i})\widehat{\bs{\lambda}}^{\trans}\bs{\Xi}_{i}\bs{\Xi}_{i}^{\trans}\widehat{\bs{\lambda}}\widetilde{Y}_{i}=\frac{1}{n}\sum_{i=1}^{n}\frac{1}{2}\frac{\rho^{\prime\prime\prime}(0)}{\rho^{\prime\prime}(0)^{2}}\widehat{\bs{\lambda}}_{lin}^{\trans}\bs{\Xi}_{i}\bs{\Xi}_{i}^{\trans}\widehat{\bs{\lambda}}_{lin}\widetilde{Y}_{i}+o_{P}(n^{-1}).
\]
As a result, we have 
\begin{align}
 & \frac{1}{n}\sum_{i=1}^{n}(\rho^{\prime}(\widehat{\bs{\lambda}}^{\trans}\bs{\Xi}_{i})-1)\widetilde{Y}_{i}\nonumber \\
= & \frac{1}{n}\sum_{i=1}^{n}\left\{ \rho^{\prime\prime}(0)\widehat{\bs{\lambda}}^{\trans}\bs{\Xi}_{i}+\frac{1}{2}\frac{\rho^{\prime\prime\prime}(0)}{\rho^{\prime\prime}(0)^{2}}\widehat{\bs{\lambda}}_{lin}^{\trans}\bs{\Xi}_{i}\bs{\Xi}_{i}^{\trans}\widehat{\bs{\lambda}}_{lin}\right\} \left\{ \bs{\Xi}_{i}^{\trans}\bs{\beta}_{\mathcal{C}_{n}}+\epsilon_{i}\right\} +o_{P}(n^{-1})\nonumber \\
= & \widehat{\bs{\lambda}}_{lin}^{\trans}\frac{1}{n}\sum_{i=1}^{n}\bs{\Xi}_{i}\left\{ \bs{\Xi}_{i}^{\trans}\bs{\beta}_{\mathcal{C}_{n}}+\epsilon_{i}\right\} \nonumber \\
 & -\frac{1}{2n}\sum_{i=1}^{n}\frac{\rho^{\prime\prime\prime}(0)}{\rho^{\prime\prime}(0)^{2}}\widehat{\bs{\lambda}}_{lin}^{\trans}\bs{\Xi}_{i}\bs{\Xi}_{i}^{\trans}\bs{\Xi}_{i}^{\trans}\widehat{\bs{\lambda}}_{lin}\left\{ \frac{1}{n}\sum_{i=1}^{n}\bs{\Xi}_{i}\bs{\Xi}_{i}^{\trans}\right\} ^{-1}\frac{1}{n}\sum_{i=1}^{n}\bs{\Xi}_{i}\left\{ \bs{\Xi}_{i}^{\trans}\bs{\beta}_{\mathcal{C}_{n}}+\epsilon_{i}\right\} \nonumber \\
 & +\rho^{\prime\prime}(0)\frac{1}{n}\sum_{i=1}^{n}\bs{\Xi}_{i}\widetilde{Y}_{i}\times o_{P}(n^{-1})+\frac{1}{2n}\sum_{i=1}^{n}\frac{\rho^{\prime\prime\prime}(0)}{\rho^{\prime\prime}(0)^{2}}\widehat{\bs{\lambda}}_{lin}^{\trans}\bs{\Xi}_{i}\bs{\Xi}_{i}^{\trans}\widehat{\bs{\lambda}}_{lin}\left\{ \bs{\Xi}_{i}^{\trans}\bs{\beta}_{\mathcal{C}_{n}}+\epsilon_{i}\right\} +o_{P}(n^{-1})\nonumber \\
= & -\frac{1}{n}\sum_{i=1}^{n}\bs{\beta}_{\mathcal{C}_{n}}^{\trans}\bs{\Xi}_{i}-\frac{1}{2n}\sum_{i=1}^{n}\frac{\rho^{\prime\prime\prime}(0)}{\rho^{\prime\prime}(0)^{2}}\left\{ \widehat{\bs{\lambda}}_{lin}^{\trans}\bs{\Xi}_{i}\right\} ^{2}\bs{\Xi}_{i}^{\trans}\bs{\beta}_{\mathcal{C}_{n}}+\frac{1}{2n}\sum_{i=1}^{n}\frac{\rho^{\prime\prime\prime}(0)}{\rho^{\prime\prime}(0)^{2}}\left\{ \widehat{\bs{\lambda}}_{lin}^{\trans}\bs{\Xi}_{i}\right\} ^{2}\bs{\Xi}_{i}^{\trans}\bs{\beta}_{\mathcal{C}_{n}}\nonumber \\
 & +\widehat{\bs{\lambda}}_{lin}^{\trans}\frac{1}{n}\sum_{i=1}^{n}\bs{\Xi}_{i}\epsilon_{i}+\frac{1}{2n}\sum_{i=1}^{n}\frac{\rho^{\prime\prime\prime}(0)}{\rho^{\prime\prime}(0)^{2}}\widehat{\bs{\lambda}}_{lin}^{\trans}\bs{\Xi}_{i}\bs{\Xi}_{i}^{\trans}\widehat{\bs{\lambda}}_{lin}\epsilon_{i}\nonumber \\
 & -\frac{1}{2n}\sum_{i=1}^{n}\frac{\rho^{\prime\prime\prime}(0)}{\rho^{\prime\prime}(0)^{2}}\widehat{\bs{\lambda}}_{lin}^{\trans}\bs{\Xi}_{i}\bs{\Xi}_{i}^{\trans}\bs{\Xi}_{i}^{\trans}\widehat{\bs{\lambda}}_{lin}\left\{ \frac{1}{n}\sum_{i=1}^{n}\bs{\Xi}_{i}\bs{\Xi}_{i}^{\trans}\right\} ^{-1}\frac{1}{n}\sum_{i=1}^{n}\bs{\Xi}_{i}\epsilon_{i}+o_{P}(n^{-1})\nonumber \\
= & -\frac{1}{n}\sum_{i=1}^{n}\bs{\beta}_{\mathcal{C}_{n}}^{\trans}\bs{\Xi}_{i}+\widehat{\bs{\lambda}}_{lin}^{\trans}\frac{1}{n}\sum_{i=1}^{n}\bs{\Xi}_{i}\epsilon_{i}+\frac{1}{2n}\sum_{i=1}^{n}\frac{\rho^{\prime\prime\prime}(0)}{\rho^{\prime\prime}(0)^{2}}\widehat{\bs{\lambda}}_{lin}^{\trans}\bs{\Xi}_{i}\bs{\Xi}_{i}^{\trans}\widehat{\bs{\lambda}}_{lin}\epsilon_{i}+o_{P}(n^{-1}),\label{eq:(=00005Crho^=00007B=00005Cprime=00007D(=00005Cwidehat=00007B=00005Cbs=00007B=00005Clambda=00007D=00007D^=00007B=00005Ctrans=00007D=00005Cbs=00007B=00005CXi=00007D_=00007Bi=00007D)-1)=00005Cwidetilde=00007BY=00007D_=00007Bi=00007D}
\end{align}
where the second step follows from $\widehat{\bs{\lambda}}=\widehat{\bs{\lambda}}_{lin}+O_{P}(n^{-1})$,
the third step follows from $\frac{1}{n}\sum_{i=1}^{n}\bs{\Xi}_{i}\widetilde{Y}_{i}=O_{P}(1)$\footnote{This is because 
\[
\left\Vert \frac{1}{n}\sum_{i=1}^{n}\bs{\Xi}_{i}\widetilde{Y}_{i}\right\Vert \leq\left\{ \frac{1}{n}\sum_{i=1}^{n}\left\Vert \bs{\Xi}_{i}\right\Vert ^{4}\right\} ^{1/4}\left\{ \frac{1}{n}\sum_{i=1}^{n}\left|\widetilde{Y}_{i}\right|^{4}\right\} ^{1/4}=O_{P}(1),
\]
where the last step is due to (\ref{eq:emprical 4th moment bounded})
and (\ref{eq:sample average of Y^4}).} and the last step follows from $\frac{1}{n}\sum_{i=1}^{n}\bs{\Xi}_{i}\epsilon_{i}=o_{P}(1)$
(see (\ref{eq:xi epsilon =00003D o(1)}) and $\left\Vert \widehat{\bs{\lambda}}_{lin}\right\Vert =\left\Vert \widehat{\bs{\lambda}}\right\Vert +O_{P}(n^{-1})=O_{P}(n^{-1/2})$.
From (\ref{eq:sample ave of xi,k}), we have 
\[
\frac{1}{n}\sum_{i=1}^{n}\bs{\Xi}_{i}=\frac{1}{n}\sum_{i=1}^{n}\bs{\Xi}_{i}^{*}+\frac{1}{n}\sum_{i=1}^{n}\begin{pmatrix}\left\{ A_{i}-\pi_{n[1]}\right\} \1(B_{i}=1)\left\{ \bs{\xi}_{n}(\bs X_{i})-\bs{\xi}_{n}^{*}(\bs X_{i})\right\} \\
\left\{ A_{i}-\pi_{n[2]}\right\} \1(B_{i}=2)\left\{ \bs{\xi}_{n}(\bs X_{i})-\bs{\xi}_{n}^{*}(\bs X_{i})\right\} \\
\vdots\\
\left\{ A_{i}-\pi_{n[K]}\right\} \1(B_{i}=K)\left\{ \bs{\xi}_{n}(\bs X_{i})-\bs{\xi}_{n}^{*}(\bs X_{i})\right\} 
\end{pmatrix}.
\]
By Assumption \ref{assu:conditions on =00005Cxi-general D(v)}, we
have 
\begin{align*}
 & \frac{1}{n}\sum_{i=1}^{n}\left\{ A_{i}-\pi_{n[k]}\right\} \1(B_{i}=k)\left\{ \bs{\xi}_{n}(\bs X_{i})-\bs{\xi}_{n}^{*}(\bs X_{i})\right\} \\
= & \frac{n_{[k]}}{n}\pi_{n[k]}(1-\pi_{n[k]})\times\\
 & \left\{ \frac{1}{n_{1[k]}}\sum_{i=1}^{n}A_{i}\1(B_{i}=k)\left\{ \bs{\xi}_{n}(\bs X_{i})-\bs{\xi}_{n}^{*}(\bs X_{i})\right\} -\frac{1}{n_{0[k]}}\sum_{i=1}^{n}\left(1-A_{i}\right)\1(B_{i}=k)\left\{ \bs{\xi}_{n}(\bs X_{i})-\bs{\xi}_{n}^{*}(\bs X_{i})\right\} \right\} \\
= & o_{P}(n^{-1}).
\end{align*}
Thus, we have 
\begin{equation}
\frac{1}{n}\sum_{i=1}^{n}\bs{\Xi}_{i}=\frac{1}{n}\sum_{i=1}^{n}\bs{\Xi}_{i}^{*}+o_{P}(n^{-1}).\label{eq:sample ave of Xi to Xi*}
\end{equation}
Combining this with (\ref{eq:(=00005Crho^=00007B=00005Cprime=00007D(=00005Cwidehat=00007B=00005Cbs=00007B=00005Clambda=00007D=00007D^=00007B=00005Ctrans=00007D=00005Cbs=00007B=00005CXi=00007D_=00007Bi=00007D)-1)=00005Cwidetilde=00007BY=00007D_=00007Bi=00007D})
gives that 
\begin{align}
 & \frac{1}{n}\sum_{i=1}^{n}\widetilde{Y}_{i}^{*}+\frac{1}{n}\sum_{i=1}^{n}(\rho^{\prime}(\widehat{\bs{\lambda}}^{\trans}\bs{\Xi}_{i})-1)\widetilde{Y}_{i}\nonumber \\
= & \frac{1}{n}\sum_{i=1}^{n}\widetilde{Y}_{i}^{*}-\frac{1}{n}\sum_{i=1}^{n}\bs{\beta}_{\mathcal{C}_{n}}^{\trans}\bs{\Xi}_{i}^{*}+\widehat{\bs{\lambda}}_{lin}^{\trans}\frac{1}{n}\sum_{i=1}^{n}\bs{\Xi}_{i}\epsilon_{i}+\frac{1}{2n}\sum_{i=1}^{n}\frac{\rho^{\prime\prime\prime}(0)}{\rho^{\prime\prime}(0)^{2}}\widehat{\bs{\lambda}}_{lin}^{\trans}\bs{\Xi}_{i}\bs{\Xi}_{i}^{\trans}\widehat{\bs{\lambda}}_{lin}\epsilon_{i}+o_{P}(n^{-1}).\label{eq:R1+R2 diverging D(v)}
\end{align}

Now, we analyze the bias term $\widehat{\bs{\lambda}}_{lin}^{\trans}\frac{1}{n}\sum_{i=1}^{n}\bs{\Xi}_{i}\epsilon_{i}+\frac{1}{2n}\sum_{i=1}^{n}\frac{\rho^{\prime\prime\prime}(0)}{\rho^{\prime\prime}(0)^{2}}\widehat{\bs{\lambda}}_{lin}^{\trans}\bs{\Xi}_{i}\bs{\Xi}_{i}^{\trans}\widehat{\bs{\lambda}}_{lin}\epsilon_{i}$.

\textbf{Analysis of $\widehat{\bs{\lambda}}_{lin}^{\trans}\frac{1}{n}\sum_{i=1}^{n}\bs{\Xi}_{i}\epsilon_{i}$.}
By the definition of $\widehat{\bs{\lambda}}_{lin}$ we have 
\begin{equation}
\widehat{\bs{\lambda}}_{lin}^{\trans}\frac{1}{n}\sum_{i=1}^{n}\bs{\Xi}_{i}\epsilon_{i}=-\sum_{k=1}^{K}\frac{1}{n}\sum_{i=1}^{n}\bs{\Xi}_{i,[k]}^{\trans}\left\{ \frac{1}{n}\sum_{i=1}^{n}\bs{\Xi}_{i,[k]}\bs{\Xi}_{i,[k]}^{\trans}\right\} ^{-1}\frac{1}{n}\sum_{i=1}^{n}\bs{\Xi}_{i,[k]}\epsilon_{i,[k]}.\label{eq:first bias}
\end{equation}
We decompose $\frac{1}{n}\sum_{i=1}^{n}\bs{\Xi}_{i,[k]}\epsilon_{i,[k]}$
as 
\[
\frac{1}{n}\sum_{i=1}^{n}\bs{\Xi}_{i,[k]}\epsilon_{i,[k]}=\frac{1}{n}\sum_{i=1}^{n}\bs{\Xi}_{i,[k]}\widetilde{Y}_{i,[k]}-\frac{1}{n}\sum_{i=1}^{n}\bs{\Xi}_{i,[k]}\bs{\Xi}_{i,[k]}^{\trans}\bs{\beta}_{[k],\mathcal{C}_{n}}.
\]
We calculate $\frac{1}{n}\sum_{i=1}^{n}\bs{\Xi}_{i,[k]}\bs{\Xi}_{i,[k]}^{\trans}\bs{\beta}_{[k],\mathcal{C}_{n}}$
and $\frac{1}{n}\sum_{i=1}^{n}\bs{\Xi}_{i,[k]}\widetilde{Y}_{i,[k]}$
respectively.

From (\ref{eq:analyze empirical var}) we have {\footnotesize
\begin{align*}
 & \frac{1}{n}\sum_{i=1}^{n}\bs{\Xi}_{i,[k]}\bs{\Xi}_{i,[k]}^{\trans}\\
= & \frac{1}{n}\sum_{i=1}^{n}A_{i}\1(B_{i}=k)\left\{ \bs{\xi}_{n}^{*}(\bs X_{i})-\overline{\bs{\xi}}_{n[k]}^{*}\right\} \left\{ \bs{\xi}_{n}^{*}(\bs X_{i})-\overline{\bs{\xi}}_{n[k]}^{*}\right\} ^{\trans}-\\
 & \frac{2}{n}\sum_{i=1}^{n}\pi_{n[k]}A_{i}\1(B_{i}=k)\left\{ \bs{\xi}_{n}^{*}(\bs X_{i})-\overline{\bs{\xi}}_{n[k]}^{*}\right\} \left\{ \bs{\xi}_{n}^{*}(\bs X_{i})-\overline{\bs{\xi}}_{n[k]}^{*}\right\} ^{\trans}+\\
 & \frac{1}{n}\sum_{i=1}^{n}\pi_{n[k]}^{2}\1(B_{i}=k)\left\{ \bs{\xi}_{n}^{*}(\bs X_{i})-\overline{\bs{\xi}}_{n[k]}^{*}\right\} \left\{ \bs{\xi}_{n}^{*}(\bs X_{i})-\overline{\bs{\xi}}_{n[k]}^{*}\right\} ^{\trans}+\\
 & \underbrace{\frac{1}{n}\sum_{i=1}^{n}\left(A_{i}-\pi_{n[k]}\right)^{2}\1(B_{i}=k)\left[\left(\bs{\xi}_{n}(\bs X_{i})-\overline{\bs{\xi}}_{n[k]}\right)\left(\bs{\xi}_{n}(\bs X_{i})-\overline{\bs{\xi}}_{n[k]}\right)^{\trans}-\left(\bs{\xi}_{n}^{*}(\bs X_{i})-\overline{\bs{\xi}}_{n[k]}^{*}\right)\left(\bs{\xi}_{n}^{*}(\bs X_{i})-\overline{\bs{\xi}}_{n[k]}^{*}\right)^{\trans}\right]}_{R_{5}},
\end{align*}
}where $\overline{\bs{\xi}}_{n[k]}^{*}:=\frac{1}{n_{[k]}}\sum_{i=1}^{n}\1(B_{i}=k)\bs{\xi}_{n}^{*}(\bs X_{i})$
is the stratum-specific sample mean for $\bs{\xi}_{n}^{*}(\bs X_{i})$.
For every $1\leq k\leq K$ and $\bs{\alpha}\in\R^{d}$, by the Cauchy--Schwarz
inequality we have 
\begin{align*}
\bs{\alpha}^{\trans}R_{5}\bs{\alpha} & =\frac{1}{n}\sum_{i=1}^{n}\left(A_{i}-\pi_{n[k]}\right)^{2}\1(B_{i}=k)\left[\left(\bs{\alpha}^{\trans}\bs{\xi}_{n}(\bs X_{i})-\bs{\alpha}^{\trans}\overline{\bs{\xi}}_{n[k]}\right)^{2}-\left(\bs{\alpha}^{\trans}\bs{\xi}_{n}^{*}(\bs X_{i})-\bs{\alpha}^{\trans}\overline{\bs{\xi}}_{n[k]}^{*}\right)^{2}\right]\\
 & \leq\sqrt{\frac{1}{n}\sum_{i=1}^{n}\1(B_{i}=k)\left\{ \bs{\alpha}^{\trans}\bs{\xi}_{n}(\bs X_{i})-\bs{\alpha}^{\trans}\overline{\bs{\xi}}_{n[k]}-\left(\bs{\alpha}^{\trans}\bs{\xi}_{n}^{*}(\bs X_{i})-\bs{\alpha}^{\trans}\overline{\bs{\xi}}_{n[k]}^{*}\right)\right\} ^{2}}\times\\
 & \qquad\sqrt{\frac{1}{n}\sum_{i=1}^{n}\1(B_{i}=k)\left\{ \bs{\alpha}^{\trans}\bs{\xi}_{n}(\bs X_{i})-\bs{\alpha}^{\trans}\overline{\bs{\xi}}_{n[k]}+\left(\bs{\alpha}^{\trans}\bs{\xi}_{n}^{*}(\bs X_{i})-\bs{\alpha}^{\trans}\overline{\bs{\xi}}_{n[k]}^{*}\right)\right\} ^{2}}\\
 & \leq\left\Vert \bs{\alpha}\right\Vert ^{2}\sqrt{\underbrace{\frac{1}{n}\sum_{i=1}^{n}\1(B_{i}=k)\left\Vert \bs{\xi}_{n}(\bs X_{i})-\overline{\bs{\xi}}_{n[k]}-\left(\bs{\xi}_{n}^{*}(\bs X_{i})-\overline{\bs{\xi}}_{n[k]}^{*}\right)\right\Vert ^{2}}_{R_{6}}}\times\\
 & \qquad\sqrt{\underbrace{\frac{1}{n}\sum_{i=1}^{n}\1(B_{i}=k)\left\Vert \bs{\xi}_{n}(\bs X_{i})-\overline{\bs{\xi}}_{n[k]}+\left(\bs{\xi}_{n}^{*}(\bs X_{i})-\overline{\bs{\xi}}_{n[k]}^{*}\right)\right\Vert ^{2}}_{R_{7}}},
\end{align*}
which leads to 
\[
\left\Vert R_{5}\right\Vert \leq\sqrt{R_{6}}\times\sqrt{R_{7}}.
\]
We control $R_{6}$ and $R_{7}$ separately. By Assumption \ref{assu:conditions on =00005Cxi-general D(v)},
we have 
\begin{align}
\left\Vert \overline{\bs{\xi}}_{n[k]}-\overline{\bs{\xi}}_{n[k]}^{*}\right\Vert  & =\left\Vert \frac{1}{n_{[k]}}\sum_{i=1}^{n}\1(B_{i}=k)\left\{ \bs{\xi}_{n}(\bs X_{i})-\bs{\xi}_{n}^{*}(\bs X_{i})\right\} \right\Vert \nonumber \\
 & \leq\sqrt{\frac{1}{n_{[k]}}\sum_{i=1}^{n}\1(B_{i}=k)\left\Vert \bs{\xi}_{n}(\bs X_{i})-\bs{\xi}_{n}^{*}(\bs X_{i})\right\Vert ^{2}}=O_{P}(\Delta_{n})\label{eq:sample group mean-pop group mean-2}
\end{align}
for all $k=1,\ldots K$. Then it follows from Assumption \ref{assu:conditions on =00005Cxi-general D(v)},
weak law of large numbers and (\ref{eq:sample group mean-pop group mean-2})
that 
\begin{align*}
R_{6} & \leq2\times\frac{1}{n}\sum_{i=1}^{n}\1(B_{i}=k)\left\Vert \bs{\xi}_{n}(\bs X_{i})-\bs{\xi}_{n}^{*}(\bs X_{i})\right\Vert ^{2}+2\times\frac{1}{n}\sum_{i=1}^{n}\1(B_{i}=k)\left\Vert \overline{\bs{\xi}}_{n[k]}-\overline{\bs{\xi}}_{n[k]}^{*}\right\Vert ^{2}=O_{P}(\Delta_{n}^{2}).
\end{align*}
By Lemma \ref{lem:LLN for triangular array} and Assumption \ref{assu:conditions on =00005Cxi-general D(v)},
we have 
\[
\frac{1}{n}\sum_{i=1}^{n}\left\{ \1(B_{i}=k)\left\Vert \bs{\xi}_{n}^{*}(\bs X_{i})\right\Vert ^{2}-\e\left[\1(B_{i}=k)\left\Vert \bs{\xi}_{n}^{*}(\bs X_{i})\right\Vert ^{2}\right]\right\} =o_{P}(1)
\]
and 
\begin{align}
 & \overline{\bs{\xi}}_{n[k]}^{*}-\e\left[\bs{\xi}_{n}^{*}(\bs X_{i})\mid B_{i}=k\right]\nonumber \\
= & \frac{n}{n_{[k]}}\frac{1}{n}\sum_{i=1}^{n}\left\{ \1(B_{i}=k)\bs{\xi}_{n}^{*}(\bs X_{i})-\e\left[\1(B_{i}=k)\bs{\xi}_{n}^{*}(\bs X_{i})\right]\right\} +o_{P}(1)\nonumber \\
= & o_{P}(1).\label{eq:xin* mean - true mean}
\end{align}
As a result, we have
\begin{align*}
R_{7} & \leq2\times R_{6}+8\times\frac{1}{n}\sum_{i=1}^{n}\1(B_{i}=k)\left\Vert \bs{\xi}_{n}^{*}(\bs X_{i})-\overline{\bs{\xi}}_{n[k]}^{*}\right\Vert ^{2}\\
 & \leq O_{P}(\Delta_{n}^{2})+\frac{16}{n}\sum_{i=1}^{n}\1(B_{i}=k)\left\Vert \bs{\xi}_{n}^{*}(\bs X_{i})\right\Vert ^{2}+\frac{16}{n}\sum_{i=1}^{n}\1(B_{i}=k)\left\Vert \overline{\bs{\xi}}_{n[k]}^{*}\right\Vert ^{2}\\
 & =o_{P}(1)+\left\{ O(1)+o_{P}(1)\right\} +O_{P}(1)\left\{ O(1)+o_{P}(1)\right\} ^{2}=O_{P}(1).
\end{align*}
 Now, we can obtain that $\left\Vert R_{5}\right\Vert \leq\sqrt{R_{6}}\times\sqrt{R_{7}}=O_{P}(\Delta_{n})$
and thus
\begin{align}
 & \frac{1}{n}\sum_{i=1}^{n}\bs{\Xi}_{i,[k]}\bs{\Xi}_{i,[k]}^{\trans}\nonumber \\
= & \frac{1}{n}\sum_{i=1}^{n}\left(A_{i}-\pi_{n[k]}\right)^{2}\1(B_{i}=k)\left\{ \bs{\xi}_{n}^{*}(\bs X_{i})-\overline{\bs{\xi}}_{n[k]}^{*}\right\} \left\{ \bs{\xi}_{n}^{*}(\bs X_{i})-\overline{\bs{\xi}}_{n[k]}^{*}\right\} ^{\trans}+O_{P}(\Delta_{n})\nonumber \\
= & \frac{1}{n}\sum_{i=1}^{n}\left(A_{i}-\pi_{n[k]}\right)^{2}\1(B_{i}=k)\left\{ \widetilde{\bs{\xi}}_{n}^{*}(\bs X_{i})-\overline{\widetilde{\bs{\xi}}}_{n[k]}^{*}\right\} \left\{ \widetilde{\bs{\xi}}_{n}^{*}(\bs X_{i})-\overline{\widetilde{\bs{\xi}}}_{n[k]}^{*}\right\} ^{\trans}+O_{P}(\Delta_{n})\nonumber \\
= & \frac{1}{n}\sum_{i=1}^{n}\left(A_{i}-\pi_{n[k]}\right)^{2}\1(B_{i}=k)\widetilde{\bs{\xi}}_{n}^{*}(\bs X_{i})\widetilde{\bs{\xi}}_{n}^{*}(\bs X_{i})^{\trans}\nonumber \\
 & -2\overline{\widetilde{\bs{\xi}}}_{n[k]}^{*}\frac{1}{n}\sum_{i=1}^{n}\left(A_{i}-\pi_{n[k]}\right)^{2}\1(B_{i}=k)\widetilde{\bs{\xi}}_{n}^{*}(\bs X_{i})\nonumber \\
 & +\overline{\widetilde{\bs{\xi}}}_{n[k]}^{*}\overline{\widetilde{\bs{\xi}}}_{n[k]}^{*\trans}\frac{1}{n}\sum_{i=1}^{n}\left(A_{i}-\pi_{n[k]}\right)^{2}\1(B_{i}=k)+O_{P}(\Delta_{n}),\label{eq:second decom of empirical var}
\end{align}
where $\overline{\widetilde{\bs{\xi}}}_{n[k]}^{*}:=\frac{1}{n_{[k]}}\sum_{i=1}^{n}\1(B_{i}=k)\widetilde{\bs{\xi}}_{n}^{*}(\bs X_{i})$.
Note that $\e\left[\frac{n_{[k]}}{n}\overline{\widetilde{\bs{\xi}}}_{n[k]}^{*}\right]=0$
and 
\[
\var\left(\frac{n_{[k]}}{n}\overline{\widetilde{\bs{\xi}}}_{n[k]}^{*}\right)=\var\left(\frac{1}{n}\sum_{i=1}^{n}\1(B_{i}=k)\widetilde{\bs{\xi}}_{n}^{*}(\bs X_{i})\right)\leq\frac{1}{n}\e\left[\1(B_{i}=k)\left\Vert \widetilde{\bs{\xi}}_{n}^{*}(\bs X_{i})\right\Vert ^{2}\right]=O(1/n).
\]
By Markov's inequality, we can conclude that 
\begin{equation}
\overline{\widetilde{\bs{\xi}}}_{n[k]}^{*}=O_{P}(n^{-1/2}).\label{eq:sample group - true mean}
\end{equation}
Recall that $\mathcal{C}_{n}$ denotes the $\sigma$-algebra generated
by $(A^{(n)},B^{(n)})$. We have 
\[
\e_{\mathcal{C}_{n}}\left[\left(A_{i}-\pi_{n[k]}\right)^{2}\1(B_{i}=k)\widetilde{\bs{\xi}}_{n}^{*}(\bs X_{i})\right]=\left(A_{i}-\pi_{n[k]}\right)^{2}\1(B_{i}=k)\e\left[\widetilde{\bs{\xi}}_{n}^{*}(\bs X_{i})\mid B_{i}=k\right]=0
\]
and 
\[
\var_{\mathcal{C}_{n}}\left(\left(A_{i}-\pi_{n[k]}\right)^{2}\1(B_{i}=k)\widetilde{\bs{\xi}}_{n}^{*}(\bs X_{i})\right)\leq\e\left[\left\Vert \widetilde{\bs{\xi}}_{n}^{*}(\bs X_{i})\right\Vert ^{2}\mid B_{i}=k\right]=O_{P}(1).
\]
Then it follows from Lemma \ref{lem:LLN for triangular array} that
\[
\frac{1}{n}\sum_{i=1}^{n}\left(A_{i}-\pi_{n[k]}\right)^{2}\1(B_{i}=k)\widetilde{\bs{\xi}}_{n}^{*}(\bs X_{i})=O_{P}(n^{-1/2}).
\]
This, combined with (\ref{eq:second decom of empirical var}) and
$\overline{\widetilde{\bs{\xi}}}_{n[k]}^{*}=O_{P}(n^{-1/2})$ yields
that
\begin{align}
\frac{1}{n}\sum_{i=1}^{n}\bs{\Xi}_{i,[k]}\bs{\Xi}_{i,[k]}^{\trans} & =\frac{1}{n}\sum_{i=1}^{n}\left(A_{i}-\pi_{n[k]}\right)^{2}\1(B_{i}=k)\widetilde{\bs{\xi}}_{n}^{*}(\bs X_{i})\widetilde{\bs{\xi}}_{n}^{*}(\bs X_{i})^{\trans}+O_{P}(\Delta_{n})+O_{P}(n^{-1})\nonumber \\
 & =\frac{1}{n}\sum_{i=1}^{n}\bs{\Xi}_{i,[k]}^{*}\bs{\Xi}_{i,[k]}^{*\trans}+O_{P}(\Delta_{n})+O_{P}(n^{-1}).\label{eq:xixiT-xi*xi*T}
\end{align}
As a result, we have 
\begin{equation}
\frac{1}{n}\sum_{i=1}^{n}\bs{\Xi}_{i,[k]}\bs{\Xi}_{i,[k]}^{\trans}\bs{\beta}_{[k],\mathcal{C}_{n}}=\frac{1}{n}\sum_{i=1}^{n}\bs{\Xi}_{i,[k]}^{*}\bs{\Xi}_{i,[k]}^{*\trans}\bs{\beta}_{[k],\mathcal{C}_{n}}+O_{P}(\Delta_{n})+O_{P}(n^{-1}).\label{eq:=00005Cbs=00007B=00005CXi=00007D_=00007Bi,=00005Bk=00005D=00007D=00005Cbs=00007B=00005CXi=00007D_=00007Bi,=00005Bk=00005D=00007D^=00007B=00005Ctrans=00007D=00005Cbs=00007B=00005Cbeta=00007D^=00007B*=00007D}
\end{equation}

Now, we calculate $\frac{1}{n}\sum_{i=1}^{n}\bs{\Xi}_{i,[k]}\widetilde{Y}_{i,[k]}$.
From (\ref{eq:analyze empirical cov}) we have 
\begin{align}
 & \frac{1}{n}\sum_{i=1}^{n}\frac{1-\pi_{n[k]}}{\pi_{n[k]}}\1(B_{i}=k)A_{i}(\bs{\xi}_{n}(\bs X_{i})-\overline{\bs{\xi}}_{n[k]})(Y_{i}-\overline{Y}_{1[k]})\nonumber \\
= & \frac{1}{n}\sum_{i=1}^{n}\frac{1-\pi_{n[k]}}{\pi_{n[k]}}\1(B_{i}=k)A_{i}(\bs{\xi}_{n}^{*}(\bs X_{i})-\overline{\bs{\xi}}_{n[k]}^{*})(Y_{i}(1)-\overline{Y}_{1[k]})+R_{8}-R_{9}\nonumber \\
= & \frac{1}{n}\sum_{i=1}^{n}\frac{1-\pi_{n[k]}}{\pi_{n[k]}}\1(B_{i}=k)A_{i}(\bs{\xi}_{n}^{*}(\bs X_{i})-\overline{\bs{\xi}}_{n[k]}^{*})(Y_{i}(1)-\overline{Y}_{1[k]})+O_{P}(\Delta_{n})\nonumber \\
= & \frac{1}{n}\sum_{i=1}^{n}\frac{1-\pi_{n[k]}}{\pi_{n[k]}}\1(B_{i}=k)A_{i}(\widetilde{\bs{\xi}}_{n}^{*}(\bs X_{i})-\overline{\widetilde{\bs{\xi}}}_{n[k]}^{*})(\widetilde{Y}_{i}(1)-\overline{\widetilde{Y}}_{1[k]})+O_{P}(\Delta_{n})\nonumber \\
= & \frac{1}{n}\sum_{i=1}^{n}\frac{1-\pi_{n[k]}}{\pi_{n[k]}}\1(B_{i}=k)A_{i}\widetilde{\bs{\xi}}_{n}^{*}(\bs X_{i})\widetilde{Y}_{i}(1)-\overline{\widetilde{\bs{\xi}}}_{n[k]}^{*}\frac{1}{n}\sum_{i=1}^{n}\frac{1-\pi_{n[k]}}{\pi_{n[k]}}\1(B_{i}=k)A_{i}\widetilde{Y}_{i}(1)\nonumber \\
 & -\overline{\widetilde{Y}}_{1[k]}\frac{1}{n}\sum_{i=1}^{n}\frac{1-\pi_{n[k]}}{\pi_{n[k]}}\1(B_{i}=k)A_{i}\overline{\widetilde{\bs{\xi}}}_{n[k]}^{*}+\overline{\widetilde{\bs{\xi}}}_{n[k]}^{*}\overline{\widetilde{Y}}_{1[k]}\frac{1}{n}\sum_{i=1}^{n}\frac{1-\pi_{n[k]}}{\pi_{n[k]}}\1(B_{i}=k)A_{i}+O_{P}(\Delta_{n}),\label{eq:secon decomp of empirical cov}
\end{align}
where the second step follows from $\left|R_{8}\right|\leq O_{P}(1)\times\sqrt{R_{6}}=O_{P}(\Delta_{n})$
and $\left|R_{9}\right|\leq O_{P}(1)\times\sqrt{R_{6}}=O_{P}(\Delta_{n})$
and in the third step we define $\overline{\widetilde{Y}}_{1[k]}:=\frac{1}{n_{1[k]}}\sum_{i=1}^{n}A_{i}\1(B_{i}=k)\widetilde{Y}_{i}(a)$.
Note that $\e_{\mathcal{C}_{n}}\left[\frac{1-\pi_{n[k]}}{\pi_{n[k]}}\1(B_{i}=k)A_{i}\widetilde{Y}_{i}(1)\right]=0$
and 
\[
\var_{\mathcal{C}_{n}}\left(\frac{1-\pi_{n[k]}}{\pi_{n[k]}}\1(B_{i}=k)A_{i}\widetilde{Y}_{i}(1)\right)\leq\left\{ \frac{1-\pi_{n[k]}}{\pi_{n[k]}}\right\} ^{2}\e\left[Y_{i}(a)^{2}\mid B_{i}=k\right]=O_{P}(1).
\]
By Lemma \ref{lem:LLN for triangular array} we have 
\[
\frac{1}{n}\sum_{i=1}^{n}\frac{1-\pi_{n[k]}}{\pi_{n[k]}}\1(B_{i}=k)A_{i}\widetilde{Y}_{i}(1)=O_{P}(n^{-1/2}).
\]
Similarly, we also have 
\begin{equation}
\overline{\widetilde{Y}}_{1[k]}=\frac{1}{n_{1[k]}}\sum_{i=1}^{n}A_{i}\1(B_{i}=k)\widetilde{Y}_{i}(a)=O_{P}(n^{-1/2}).\label{eq:Ytildebar}
\end{equation}
The above two displays, combined with (\ref{eq:sample group - true mean})
and (\ref{eq:secon decomp of empirical cov}), give that 
\begin{align*}
 & \frac{1}{n}\sum_{i=1}^{n}\frac{1-\pi_{n[k]}}{\pi_{n[k]}}\1(B_{i}=k)A_{i}(\bs{\xi}_{n}(\bs X_{i})-\overline{\bs{\xi}}_{n[k]})(Y_{i}-\overline{Y}_{1[k]})\\
= & \frac{1}{n}\sum_{i=1}^{n}\frac{1-\pi_{n[k]}}{\pi_{n[k]}}\1(B_{i}=k)A_{i}\widetilde{\bs{\xi}}_{n}^{*}(\bs X_{i})\widetilde{Y}_{i}(1)+O_{P}(n^{-1})+O_{P}(\Delta_{n}).
\end{align*}
Similarly, we can also derive that
\begin{align*}
 & \frac{1}{n}\sum_{i=1}^{n}\frac{\pi_{n[k]}}{1-\pi_{n[k]}}\1(B_{i}=k)(1-A_{i})(\bs{\xi}_{n}(\bs X_{i})-\overline{\bs{\xi}}_{n[k]})(Y_{i}-\overline{Y}_{0[k]})\\
= & \frac{1}{n}\sum_{i=1}^{n}\frac{\pi_{n[k]}}{1-\pi_{n[k]}}\1(B_{i}=k)(1-A_{i})\widetilde{\bs{\xi}}_{n}^{*}(\bs X_{i})\widetilde{Y}_{i}(0)+O_{P}(n^{-1})+O_{P}(\Delta_{n}).
\end{align*}
Note that 
\begin{align*}
 & \frac{1}{n}\sum_{i=1}^{n}\bs{\Xi}_{i,[k]}\widetilde{Y}_{i,[k]}\\
= & \frac{1}{n}\sum_{i=1}^{n}\frac{1-\pi_{n[k]}}{\pi_{n[k]}}\1(B_{i}=k)A_{i}(\bs{\xi}_{n}(\bs X_{i})-\overline{\bs{\xi}}_{n[k]})(Y_{i}-\overline{Y}_{1[k]})+\\
 & \qquad+\frac{1}{n}\sum_{i=1}^{n}\frac{\pi_{n[k]}}{1-\pi_{n[k]}}\1(B_{i}=k)(1-A_{i})(\bs{\xi}_{n}(\bs X_{i})-\overline{\bs{\xi}}_{n[k]})(Y_{i}-\overline{Y}_{0[k]}),
\end{align*}
we have 
\begin{align*}
 & \frac{1}{n}\sum_{i=1}^{n}\bs{\Xi}_{i,[k]}\widetilde{Y}_{i,[k]}\\
= & \frac{1}{n}\sum_{i=1}^{n}\frac{1-\pi_{n[k]}}{\pi_{n[k]}}\1(B_{i}=k)A_{i}\widetilde{\bs{\xi}}_{n}^{*}(\bs X_{i})\widetilde{Y}_{i}(1)\\
 & \qquad+\frac{1}{n}\sum_{i=1}^{n}\frac{\pi_{n[k]}}{1-\pi_{n[k]}}\1(B_{i}=k)(1-A_{i})\widetilde{\bs{\xi}}_{n}^{*}(\bs X_{i})\widetilde{Y}_{i}(0)+O_{P}(n^{-1})+O_{P}(\Delta_{n})\\
= & \frac{1}{n}\sum_{i=1}^{n}\bs{\Xi}_{i,[k]}^{*}\widetilde{Y}_{i,[k]}^{*}+O_{P}(n^{-1})+O_{P}(\Delta_{n}).
\end{align*}

Combining this with (\ref{eq:=00005Cbs=00007B=00005CXi=00007D_=00007Bi,=00005Bk=00005D=00007D=00005Cbs=00007B=00005CXi=00007D_=00007Bi,=00005Bk=00005D=00007D^=00007B=00005Ctrans=00007D=00005Cbs=00007B=00005Cbeta=00007D^=00007B*=00007D}),
we have 
\begin{align*}
\frac{1}{n}\sum_{i=1}^{n}\bs{\Xi}_{i,[k]}\epsilon_{i,[k]} & =\frac{1}{n}\sum_{i=1}^{n}\bs{\Xi}_{i,[k]}^{*}\widetilde{Y}_{i,[k]}^{*}-\frac{1}{n}\sum_{i=1}^{n}\bs{\Xi}_{i,[k]}^{*}\bs{\Xi}_{i,[k]}^{*\trans}\bs{\beta}_{[k],\mathcal{C}_{n}}+O_{P}(\Delta_{n})+O_{P}(n^{-1})\\
 & =\frac{1}{n}\sum_{i=1}^{n}\bs{\Xi}_{i,[k]}^{*}\epsilon_{i,[k]}^{*}+O_{P}(\Delta_{n})+O_{P}(n^{-1}).
\end{align*}
Note that 
\begin{align*}
 & \e_{\mathcal{C}_{n}}\left[\left\Vert \bs{\Xi}_{i,[k]}^{*}\epsilon_{i,[k]}^{*}-\e_{\mathcal{C}_{n}}\left[\bs{\Xi}_{i,[k]}^{*}\epsilon_{i,[k]}^{*}\right]\right\Vert ^{2}\right]\leq\e_{\mathcal{C}_{n}}\left[\left\Vert \bs{\Xi}_{i,[k]}^{*}\epsilon_{i,[k]}^{*}\right\Vert ^{2}\right]\\
\leq & O(1)\times\left\{ \sqrt{\e\left[\left\Vert \widetilde{\bs{\xi}}_{n}^{*}\right\Vert ^{4}\mid B_{i}=k\right]}\sqrt{\max_{a\in\{0,1\}}\e\left[\left|Y_{i}(a)\right|^{4}\mid B_{i}=k\right]}+\e\left[\left\Vert \widetilde{\bs{\xi}}_{n}^{*}\right\Vert ^{4}\mid B_{i}=k\right]\right\} \\
= & O_{P}(1).
\end{align*}
By Lemma \ref{lem:LLN for triangular array} we have 
\[
\left\Vert \frac{1}{n}\sum_{i=1}^{n}\bs{\Xi}_{i,[k]}^{*}\epsilon_{i,[k]}^{*}-\frac{1}{n}\sum_{i=1}^{n}\e_{\mathcal{C}_{n}}\left[\bs{\Xi}_{i,[k]}^{*}\epsilon_{i,[k]}^{*}\right]\right\Vert =o_{P}(n^{-1/2}).
\]
Besides, from (\ref{eq:condi exp of Xiepsilon}) we have 
\begin{align*}
\frac{1}{n}\sum_{i=1}^{n}\e_{\mathcal{C}_{n}}\left[\bs{\Xi}_{i,[k]}^{*}\epsilon_{i,[k]}^{*}\right] & =0
\end{align*}
Therefore, we have $\frac{1}{n}\sum_{i=1}^{n}\bs{\Xi}_{i,[k]}^{*}\epsilon_{i,[k]}^{*}=o_{P}(n^{-1/2})$
and thus 
\begin{equation}
\frac{1}{n}\sum_{i=1}^{n}\bs{\Xi}_{i}\epsilon_{i}=\frac{1}{n}\sum_{i=1}^{n}\bs{\Xi}_{i,[k]}^{*}\epsilon_{i,[k]}^{*}+O_{P}(\Delta_{n})+O_{P}(n^{-1})=O_{P}(\Delta_{n})+o_{P}(n^{-1/2}).\label{eq:xi epsilon =00003D o(1)}
\end{equation}

Recall (\ref{eq:xixiT-xi*xi*T}), we have 
\[
\frac{1}{n}\sum_{i=1}^{n}\bs{\Xi}_{i,[k]}\bs{\Xi}_{i,[k]}^{\trans}=\frac{1}{n}\sum_{i=1}^{n}\bs{\Xi}_{i,[k]}^{*}\bs{\Xi}_{i,[k]}^{*\trans}+O_{P}(\Delta_{n})+O_{P}(n^{-1}).
\]
Note that 
\begin{align*}
 & \e_{\mathcal{C}_{n}}\left[\left\Vert \bs{\Xi}_{i,[k]}^{*}\bs{\Xi}_{i,[k]}^{*\trans}-\e_{\mathcal{C}_{n}}\left[\bs{\Xi}_{i,[k]}^{*}\bs{\Xi}_{i,[k]}^{*\trans}\right]\right\Vert _{F}^{2}\right]\leq2\e_{\mathcal{C}_{n}}\left[\left\Vert \bs{\Xi}_{i,[k]}^{*}\bs{\Xi}_{i,[k]}^{*\trans}\right\Vert _{F}^{2}\right]+2\left\Vert \e_{\mathcal{C}_{n}}\left[\bs{\Xi}_{i,[k]}^{*}\bs{\Xi}_{i,[k]}^{*\trans}\right]\right\Vert _{F}^{2}\\
 & \leq4\e_{\mathcal{C}_{n}}\left[\left\Vert \bs{\Xi}_{i,[k]}^{*}\right\Vert ^{4}\right]\leq4\e\left[\left\Vert \widetilde{\bs{\xi}}_{n}^{*}(\bs X_{i})\right\Vert ^{4}\mid B_{i}=k\right]=O_{P}(1),
\end{align*}
where the last step follows from Assumption \ref{assu:conditions on =00005Cxi-general D(v)}.
By Lemma \ref{lem:LLN for triangular array} and $\frac{1}{n}\sum_{i=1}^{n}\bs{\Xi}_{i,[k]}\bs{\Xi}_{i,[k]}^{\trans}=\frac{1}{n}\sum_{i=1}^{n}\bs{\Xi}_{i,[k]}^{*}\bs{\Xi}_{i,[k]}^{*\trans}+O_{P}(\Delta_{n})+O_{P}(n^{-1})$
we have 
\begin{equation}
\left\Vert \frac{1}{n}\sum_{i=1}^{n}\bs{\Xi}_{i,[k]}\bs{\Xi}_{i,[k]}^{\trans}-\mathbf{\Sigma}_{[k]}^{\mathcal{C}_{n}}\right\Vert =O_{P}(n^{-1/2})+O_{P}(\Delta_{n}),\label{eq:xixiT-condition xixiT}
\end{equation}
where we define $\mathbf{\Sigma}_{[k]}^{\mathcal{C}_{n}}:=\frac{1}{n}\sum_{i=1}^{n}\e_{\mathcal{C}_{n}}\left[\bs{\Xi}_{i,[k]}^{*}\bs{\Xi}_{i,[k]}^{*\trans}\right]$.

Now, it follows from (\ref{eq:first bias}), (\ref{eq:sample ave of Xi to Xi*}),
(\ref{eq:xi epsilon =00003D o(1)}) and (\ref{eq:xixiT-condition xixiT})
that 
\begin{align}
 & \widehat{\bs{\lambda}}_{lin}^{\trans}\frac{1}{n}\sum_{i=1}^{n}\bs{\Xi}_{i}\epsilon_{i}=-\frac{1}{n}\sum_{i=1}^{n}\bs{\Xi}_{i}^{\trans}\left\{ \frac{1}{n}\sum_{i=1}^{n}\bs{\Xi}_{i}\bs{\Xi}_{i}^{\trans}\right\} ^{-1}\frac{1}{n}\sum_{i=1}^{n}\bs{\Xi}_{i}\epsilon_{i}\nonumber \\
= & -\sum_{k=1}^{K}\left\{ \frac{1}{n}\sum_{i=1}^{n}\bs{\Xi}_{i,[k]}^{*\trans}+o_{P}(n^{-1})\right\} \left\{ \left(\mathbf{\Sigma}_{[k]}^{\mathcal{C}_{n}}\right)^{-1}+O_{P}(n^{-1/2})+O_{P}(\Delta_{n})\right\} \nonumber \\
 & \quad\times\left\{ \frac{1}{n}\sum_{i=1}^{n}\bs{\Xi}_{i,[k]}^{*}\epsilon_{i,[k]}^{*}+O_{P}(\Delta_{n})+O_{P}(n^{-1})\right\} \nonumber \\
= & \underbrace{-\sum_{k=1}^{K}\frac{1}{n}\sum_{i=1}^{n}\bs{\Xi}_{i,[k]}^{*\trans}\left(\mathbf{\Sigma}_{[k]}^{\mathcal{C}_{n}}\right)^{-1}\frac{1}{n}\sum_{i=1}^{n}\bs{\Xi}_{i,[k]}^{*}\epsilon_{i,[k]}^{*}}_{\text{bias}_{1}}+O_{P}(\Delta_{n}n^{-1/2})+o_{P}(n^{-1}),\label{eq:first bias final}
\end{align}
where the last step follows from $\frac{1}{n}\sum_{i=1}^{n}\bs{\Xi}_{i,[k]}^{*\trans}=O_{P}(n^{-1/2})$,
$\left\Vert \left\{ \frac{1}{n}\sum_{i=1}^{n}\bs{\Xi}_{i}\bs{\Xi}_{i}^{\trans}\right\} ^{-1}\right\Vert =O(1)$
and (\ref{eq:xi epsilon =00003D o(1)}).

\textbf{Analysis of $\frac{1}{2n}\sum_{i=1}^{n}\frac{\rho^{\prime\prime\prime}(0)}{\rho^{\prime\prime}(0)^{2}}\widehat{\bs{\lambda}}_{lin}^{\trans}\bs{\Xi}_{i}\bs{\Xi}_{i}^{\trans}\widehat{\bs{\lambda}}_{lin}\epsilon_{i}$.}
By the definition of $\widehat{\bs{\lambda}}_{lin}$ we have {\small
\begin{align}
 & \frac{1}{2n}\sum_{i=1}^{n}\frac{\rho^{\prime\prime\prime}(0)}{\rho^{\prime\prime}(0)^{2}}\widehat{\bs{\lambda}}_{lin}^{\trans}\bs{\Xi}_{i}\bs{\Xi}_{i}^{\trans}\widehat{\bs{\lambda}}_{lin}\epsilon_{i}\nonumber \\
= & \frac{\rho^{\prime\prime\prime}(0)}{2\rho^{\prime\prime}(0)^{2}}\frac{1}{n}\sum_{i=1}^{n}\bs{\Xi}_{i}^{\trans}\left(\frac{1}{n}\sum_{i=1}^{n}\bs{\Xi}_{i}\bs{\Xi}_{i}^{\trans}\right)^{-1}\frac{1}{n}\sum_{i=1}^{n}\bs{\Xi}_{i}\bs{\Xi}_{i}^{\trans}\epsilon_{i}\left(\frac{1}{n}\sum_{i=1}^{n}\bs{\Xi}_{i}\bs{\Xi}_{i}^{\trans}\right)^{-1}\frac{1}{n}\sum_{i=1}^{n}\bs{\Xi}_{i}\nonumber \\
= & \frac{\rho^{\prime\prime\prime}(0)}{2\rho^{\prime\prime}(0)^{2}}\sum_{k=1}^{K}\frac{1}{n}\sum_{i=1}^{n}\bs{\Xi}_{i,[k]}^{\trans}\left(\frac{1}{n}\sum_{i=1}^{n}\bs{\Xi}_{i,[k]}\bs{\Xi}_{i,[k]}^{\trans}\right)^{-1}\frac{1}{n}\sum_{i=1}^{n}\bs{\Xi}_{i,[k]}\bs{\Xi}_{i,[k]}^{\trans}\epsilon_{i,[k]}\left(\frac{1}{n}\sum_{i=1}^{n}\bs{\Xi}_{i,[k]}\bs{\Xi}_{i,[k]}^{\trans}\right)^{-1}\frac{1}{n}\sum_{i=1}^{n}\bs{\Xi}_{i,[k]}.\label{eq:decomposition of second bias}
\end{align}

}We deal with the term $\frac{1}{n}\sum_{i=1}^{n}\bs{\Xi}_{i,[k]}\bs{\Xi}_{i,[k]}^{\trans}\epsilon_{i,[k]}$.
Note that 
\[
\frac{1}{n}\sum_{i=1}^{n}\bs{\Xi}_{i,[k]}\bs{\Xi}_{i,[k]}^{\trans}\epsilon_{i,[k]}=\frac{1}{n}\sum_{i=1}^{n}\bs{\Xi}_{i,[k]}\bs{\Xi}_{i,[k]}^{\trans}\widetilde{Y}_{i,[k]}-\frac{1}{n}\sum_{i=1}^{n}\bs{\Xi}_{i,[k]}\bs{\Xi}_{i,[k]}^{\trans}\bs{\Xi}_{i,[k]}^{\trans}\bs{\beta}_{[k],\mathcal{C}_{n}}.
\]
We first analyze the term $\frac{1}{n}\sum_{i=1}^{n}\bs{\Xi}_{i,[k]}\bs{\Xi}_{i,[k]}^{\trans}\bs{\Xi}_{i,[k]}^{\trans}\bs{\beta}_{[k],\mathcal{C}_{n}}$.
By the Cauchy-Schwarz inequality, we have 
\begin{align*}
 & \sup_{\bs{\alpha}\in S^{d-1}}\left|\frac{1}{n}\sum_{i=1}^{n}\bs{\alpha}^{\trans}\bs{\Xi}_{i,[k]}\bs{\Xi}_{i,[k]}^{\trans}\bs{\alpha}\bs{\Xi}_{i,[k]}^{\trans}\bs{\beta}_{[k],\mathcal{C}_{n}}-\frac{1}{n}\sum_{i=1}^{n}\bs{\alpha}^{\trans}\bs{\Xi}_{i,[k]}^{*}\bs{\Xi}_{i,[k]}^{*\trans}\bs{\alpha}\bs{\Xi}_{i,[k]}^{*\trans}\bs{\beta}_{[k],\mathcal{C}_{n}}\right|\\
= & \sup_{\bs{\alpha}\in S^{d-1}}\left|\frac{1}{n}\sum_{i=1}^{n}\bs{\alpha}^{\trans}\left(\bs{\Xi}_{i,[k]}-\bs{\Xi}_{i,[k]}^{*}\right)\bs{\alpha}^{\trans}\left(\bs{\Xi}_{i,[k]}+\bs{\Xi}_{i,[k]}^{*}\right)\bs{\Xi}_{i,[k]}^{\trans}\bs{\beta}_{[k],\mathcal{C}_{n}}\right|+\\
 & \quad\sup_{\bs{\alpha}\in S^{d-1}}\left|\frac{1}{n}\sum_{i=1}^{n}\left(\bs{\alpha}^{\trans}\bs{\Xi}_{i,[k]}^{*}\right)^{2}\left\{ \bs{\Xi}_{i,[k]}^{\trans}-\bs{\Xi}_{i,[k]}^{*\trans}\right\} \bs{\beta}_{[k],\mathcal{C}_{n}}\right|\\
\leq & \sup_{\bs{\alpha}\in S^{d-1}}\sqrt{\frac{1}{n}\sum_{i=1}^{n}\left\{ \bs{\alpha}^{\trans}\bs{\Xi}_{i,[k]}-\bs{\alpha}^{\trans}\bs{\Xi}_{i,[k]}^{*}\right\} ^{2}}\times\sqrt{\frac{1}{n}\sum_{i=1}^{n}\left\{ \bs{\alpha}^{\trans}\bs{\Xi}_{i,[k]}+\bs{\alpha}^{\trans}\bs{\Xi}_{i,[k]}^{*}\right\} ^{2}\left\{ \bs{\Xi}_{i,[k]}^{\trans}\bs{\beta}_{[k],\mathcal{C}_{n}}\right\} ^{2}}\\
 & \quad+\sup_{\bs{\alpha}\in S^{d-1}}\sqrt{\frac{1}{n}\sum_{i=1}^{n}\left\{ \bs{\Xi}_{i,[k]}^{\trans}\bs{\beta}_{[k],\mathcal{C}_{n}}-\bs{\Xi}_{i,[k]}^{*\trans}\bs{\beta}_{[k],\mathcal{C}_{n}}\right\} ^{2}}\times\sqrt{\frac{1}{n}\sum_{i=1}^{n}\left(\bs{\alpha}^{\trans}\bs{\Xi}_{i,[k]}^{*}\right)^{4}}\\
\leq & O_{P}(1)\times\sqrt{\frac{1}{n}\sum_{i=1}^{n}\left\Vert \bs{\Xi}_{i,[k]}-\bs{\Xi}_{i,[k]}^{*}\right\Vert ^{2}}\times\sqrt{\frac{1}{n}\sum_{i=1}^{n}\left\{ \left\Vert \bs{\Xi}_{i,[k]}\right\Vert ^{4}+\left\Vert \bs{\Xi}_{i,[k]}^{*}\right\Vert ^{2}\left\Vert \bs{\Xi}_{i,[k]}\right\Vert ^{2}\right\} }\\
 & \quad+O_{P}(1)\times\sqrt{\frac{1}{n}\sum_{i=1}^{n}\left\Vert \bs{\Xi}_{i,[k]}-\bs{\Xi}_{i,[k]}^{*}\right\Vert ^{2}}\times\sqrt{\frac{1}{n}\sum_{i=1}^{n}\left\Vert \bs{\Xi}_{i,[k]}^{*}\right\Vert ^{4}}.
\end{align*}
First, we have 
\begin{align}
 & \frac{1}{n}\sum_{i=1}^{n}\left\Vert \bs{\Xi}_{i,[k]}-\bs{\Xi}_{i,[k]}^{*}\right\Vert ^{2}=\frac{1}{n}\sum_{i=1}^{n}(A_{i}-\pi_{n[k]})^{2}\1(B_{i}=k)\left\Vert \bs{\xi}_{n}(\bs X_{i})-\overline{\bs{\xi}}_{n[k]}-\widetilde{\bs{\xi}}_{n}^{*}(\bs X_{i})\right\Vert ^{2}\nonumber \\
\leq & \frac{1}{n}\sum_{i=1}^{n}\1(B_{i}=k)\left\Vert \bs{\xi}_{n}(\bs X_{i})-\overline{\bs{\xi}}_{n[k]}-\widetilde{\bs{\xi}}_{n}^{*}(\bs X_{i})\right\Vert ^{2}\nonumber \\
\leq & \frac{2}{n}\sum_{i=1}^{n}\1(B_{i}=k)\left\Vert \bs{\xi}_{n}(\bs X_{i})-\overline{\bs{\xi}}_{n[k]}-\bs{\xi}_{n}^{*}(\bs X_{i})-\overline{\bs{\xi}}_{n[k]}^{*}\right\Vert ^{2}+2\left\Vert \overline{\bs{\xi}}_{n[k]}^{*}-\e\left[\bs{\xi}_{n}^{*}(\bs X_{i})\mid B_{i}=k\right]\right\Vert ^{2}\nonumber \\
\leq & 2R_{6}+o_{P}(1)=o_{P}(1),\label{eq:xi-xi^*}
\end{align}
where the last two inequalities follows from (\ref{eq:xin* mean - true mean})
and $R_{6}=O_{P}(\Delta_{n}^{2})=o_{P}(1)$. Second, we have 
\[
\frac{1}{n}\sum_{i=1}^{n}\left\Vert \bs{\Xi}_{i,[k]}^{*}\right\Vert ^{4}\leq\frac{1}{n}\sum_{i=1}^{n}\1(B_{i}=k)\left\Vert \widetilde{\bs{\xi}}_{n}^{*}(\bs X_{i})\right\Vert ^{4}=O_{P}(1)
\]
by Markov's inequality and Assumption \ref{assu:conditions on =00005Cxi-general D(v)}.
Third, we have 
\begin{align}
\frac{1}{n}\sum_{i=1}^{n}\left\Vert \bs{\Xi}_{i,[k]}\right\Vert ^{4} & \leq\frac{1}{n}\sum_{i=1}^{n}\1(B_{i}=k)\left\Vert \bs{\xi}_{n}(\bs X_{i})-\overline{\bs{\xi}}_{n[k]}\right\Vert ^{4}\nonumber \\
 & \leq\frac{8}{n}\sum_{i=1}^{n}\1(B_{i}=k)\left\Vert \bs{\xi}_{n}(\bs X_{i})\right\Vert ^{4}+\frac{8}{n}\sum_{i=1}^{n}\1(B_{i}=k)\left\Vert \overline{\bs{\xi}}_{n[k]}\right\Vert ^{4}\nonumber \\
 & \leq\frac{8}{n}\sum_{i=1}^{n}\left\Vert \bs{\xi}_{n}(\bs X_{i})\right\Vert ^{4}+8\left|\frac{1}{n_{[k]}}\sum_{i=1}^{n}\1(B_{i}=k)\left\Vert \bs{\xi}_{n}(\bs X_{i})\right\Vert \right|^{4}\nonumber \\
 & \leq\frac{8}{n}\sum_{i=1}^{n}\left\Vert \bs{\xi}_{n}(\bs X_{i})\right\Vert ^{4}+\frac{8}{n_{[k]}}\sum_{i=1}^{n}\1(B_{i}=k)\left\Vert \bs{\xi}_{n}(\bs X_{i})\right\Vert ^{4}=O_{P}(1),\label{eq:emprical 4th moment bounded}
\end{align}
where the last inequality follows from the Power-Mean inequality and
the last equality follows from Assumption \ref{assu:conditions on =00005Cxi-general D(v)}.
Then, it follows from the Cauchy-Schwarz inequality that 
\[
\frac{1}{n}\sum_{i=1}^{n}\left\Vert \bs{\Xi}_{i,[k]}^{*}\right\Vert ^{2}\left\Vert \bs{\Xi}_{i,[k]}\right\Vert ^{2}\leq\sqrt{\frac{1}{n}\sum_{i=1}^{n}\left\Vert \bs{\Xi}_{i,[k]}^{*}\right\Vert ^{4}\frac{1}{n}\sum_{i=1}^{n}\left\Vert \bs{\Xi}_{i,[k]}\right\Vert ^{4}}=O_{P}(1).
\]
Combining these three results, we have 
\begin{align}
 & \left\Vert \frac{1}{n}\sum_{i=1}^{n}\bs{\Xi}_{i,[k]}\bs{\Xi}_{i,[k]}^{\trans}\bs{\Xi}_{i,[k]}^{\trans}\bs{\beta}_{[k],\mathcal{C}_{n}}-\frac{1}{n}\sum_{i=1}^{n}\bs{\Xi}_{i,[k]}^{*}\bs{\Xi}_{i,[k]}^{*\trans}\bs{\Xi}_{i,[k]}^{*\trans}\bs{\beta}_{[k],\mathcal{C}_{n}}\right\Vert \nonumber \\
= & \sup_{\bs{\alpha}\in S^{d-1}}\left|\frac{1}{n}\sum_{i=1}^{n}\bs{\alpha}^{\trans}\bs{\Xi}_{i,[k]}\bs{\Xi}_{i,[k]}^{\trans}\bs{\alpha}\bs{\Xi}_{i,[k]}^{\trans}\bs{\beta}_{[k],\mathcal{C}_{n}}-\frac{1}{n}\sum_{i=1}^{n}\bs{\alpha}^{\trans}\bs{\Xi}_{i,[k]}^{*}\bs{\Xi}_{i,[k]}^{*\trans}\bs{\alpha}\bs{\Xi}_{i,[k]}^{*\trans}\bs{\beta}_{[k],\mathcal{C}_{n}}\right|\nonumber \\
= & o_{P}(1).\label{eq:xixixi-xi*xi*xi*}
\end{align}

Next, we analyze the term $\frac{1}{n}\sum_{i=1}^{n}\bs{\Xi}_{i,[k]}\bs{\Xi}_{i,[k]}^{\trans}\widetilde{Y}_{i,[k]}$.
By the Cauchy-Schwarz inequality, we have 
\begin{align*}
 & \sup_{\bs{\alpha}\in S^{d-1}}\left|\frac{1}{n}\sum_{i=1}^{n}\bs{\alpha}^{\trans}\bs{\Xi}_{i,[k]}\bs{\Xi}_{i,[k]}^{\trans}\bs{\alpha}\widetilde{Y}_{i,[k]}-\frac{1}{n}\sum_{i=1}^{n}\bs{\alpha}^{\trans}\bs{\Xi}_{i,[k]}^{*}\bs{\Xi}_{i,[k]}^{*\trans}\bs{\alpha}\widetilde{Y}_{i,[k]}^{*}\right|\\
= & \sup_{\bs{\alpha}\in S^{d-1}}\left|\frac{1}{n}\sum_{i=1}^{n}\bs{\alpha}^{\trans}\left(\bs{\Xi}_{i,[k]}-\bs{\Xi}_{i,[k]}^{*}\right)\bs{\alpha}^{\trans}\left(\bs{\Xi}_{i,[k]}+\bs{\Xi}_{i,[k]}^{*}\right)\widetilde{Y}_{i,[k]}\right|\\
 & \quad+\sup_{\bs{\alpha}\in S^{d-1}}\left|\frac{1}{n}\sum_{i=1}^{n}\left(\bs{\alpha}^{\trans}\bs{\Xi}_{i,[k]}^{*}\right)^{2}\left(\widetilde{Y}_{i,[k]}-\widetilde{Y}_{i,[k]}^{*}\right)\right|\\
\leq & \sup_{\bs{\alpha}\in S^{d-1}}\sqrt{\frac{1}{n}\sum_{i=1}^{n}\left\{ \bs{\alpha}^{\trans}\bs{\Xi}_{i,[k]}-\bs{\alpha}^{\trans}\bs{\Xi}_{i,[k]}^{*}\right\} ^{2}}\times\sqrt{\frac{1}{n}\sum_{i=1}^{n}\left\{ \bs{\alpha}^{\trans}\bs{\Xi}_{i,[k]}+\bs{\alpha}^{\trans}\bs{\Xi}_{i,[k]}^{*}\right\} ^{2}\left|\widetilde{Y}_{i,[k]}\right|^{2}}\\
 & \quad+\sup_{\bs{\alpha}\in S^{d-1}}\sqrt{\frac{1}{n}\sum_{i=1}^{n}\left\{ \widetilde{Y}_{i,[k]}-\widetilde{Y}_{i,[k]}^{*}\right\} ^{2}}\times\sqrt{\frac{1}{n}\sum_{i=1}^{n}\left(\bs{\alpha}^{\trans}\bs{\Xi}_{i,[k]}^{*}\right)^{4}}\\
\leq & O(1)\times\sqrt{\frac{1}{n}\sum_{i=1}^{n}\left\Vert \bs{\Xi}_{i,[k]}-\bs{\Xi}_{i,[k]}^{*}\right\Vert ^{2}}\times\sqrt{\frac{1}{n}\sum_{i=1}^{n}\left\{ \left\Vert \bs{\Xi}_{i,[k]}\right\Vert ^{2}\left|\widetilde{Y}_{i,[k]}\right|^{2}+\left\Vert \bs{\Xi}_{i,[k]}^{*}\right\Vert ^{2}\left|\widetilde{Y}_{i,[k]}\right|^{2}\right\} }\\
 & \quad+O(1)\times\sqrt{\frac{1}{n}\sum_{i=1}^{n}\left|\widetilde{Y}_{i,[k]}-\widetilde{Y}_{i,[k]}^{*}\right|^{2}}\times\sqrt{\frac{1}{n}\sum_{i=1}^{n}\left\Vert \bs{\Xi}_{i,[k]}^{*}\right\Vert ^{4}}.
\end{align*}
First, we have 
\begin{align*}
 & \frac{1}{n}\sum_{i=1}^{n}\left|\widetilde{Y}_{i,[k]}-\widetilde{Y}_{i,[k]}^{*}\right|^{2}\\
= & \frac{1}{n}\sum_{i=1}^{n}\left\{ \frac{A_{i}}{\pi_{n[k]}}\1(B_{i}=k)(Y_{i}-\overline{Y}_{1[k]}-\widetilde{Y}_{i}(1))-\frac{1-A_{i}}{1-\pi_{n[k]}}\1(B_{i}=k)(Y_{i}-\overline{Y}_{0[k]}-\widetilde{Y}_{i}(0))\right\} ^{2}\\
\leq & \frac{1}{\pi_{n[k]}^{2}}\frac{2}{n}\sum_{i=1}^{n}\left\{ \overline{Y}_{1[k]}-\e\left[Y_{i}(1)\mid B_{i}=k\right]\right\} ^{2}+\frac{1}{(1-\pi_{n[k]})^{2}}\frac{2}{n}\sum_{i=1}^{n}\left\{ \overline{Y}_{0[k]}-\e\left[Y_{i}(0)\mid B_{i}=k\right]\right\} ^{2}\\
= & \frac{2}{\pi_{n[k]}^{2}}\overline{\widetilde{Y}}_{1[k]}^{2}+\frac{2}{(1-\pi_{n[k]})^{2}}\overline{\widetilde{Y}}_{0[k]}^{2}=O_{P}(n^{-1/2}),
\end{align*}
where the last step follows from (\ref{eq:Ytildebar}). Second, we
have 
\begin{align}
\frac{1}{n}\sum_{i=1}^{n}\left|\widetilde{Y}_{i,[k]}\right|^{4} & \leq\frac{8}{\pi_{n[k]}^{4}}\frac{1}{n}\sum_{i=1}^{n}(Y_{i}(1)-\overline{Y}_{1[k]})^{4}+\frac{8}{(1-\pi_{n[k]})^{4}}\frac{1}{n}\sum_{i=1}^{n}(Y_{i}(0)-\overline{Y}_{1[k]})^{4}\nonumber \\
 & \leq\frac{64}{\pi_{n[k]}^{4}}\frac{1}{n}\sum_{i=1}^{n}Y_{i}(1)^{4}+\frac{64}{\pi_{n[k]}^{4}}\overline{Y}_{1[k]}^{4}+\frac{64}{(1-\pi_{n[k]})^{4}}\frac{1}{n}\sum_{i=1}^{n}Y_{i}(0)^{4}+\frac{64}{(1-\pi_{n[k]})^{4}}\overline{Y}_{0[k]}^{4}\nonumber \\
 & =O_{P}(1),\label{eq:sample average of Y^4}
\end{align}
where the last step follows from $\max_{a\in\{0,1\}}\e\left[\left|Y_{i}(a)\right|^{4}\right]<\infty$,
$\overline{Y}_{1[k]}=O_{P}(1)$ and $\overline{Y}_{0[k]}=O_{P}(1)$.
Recall that 
\[
\frac{1}{n}\sum_{i=1}^{n}\left\Vert \bs{\Xi}_{i,[k]}^{*}\right\Vert ^{4}=O_{P}(1)\text{ and }\frac{1}{n}\sum_{i=1}^{n}\left\Vert \bs{\Xi}_{i,[k]}\right\Vert ^{4}=O_{P}(1),
\]
by the Cauchy-Schwarz inequality we have 
\begin{align*}
 & \sqrt{\frac{1}{n}\sum_{i=1}^{n}\left\{ \left\Vert \bs{\Xi}_{i,[k]}\right\Vert ^{2}\left|\widetilde{Y}_{i,[k]}\right|^{2}+\left\Vert \bs{\Xi}_{i,[k]}^{*}\right\Vert ^{2}\left|\widetilde{Y}_{i,[k]}\right|^{2}\right\} }\\
\leq & \sqrt{\sqrt{\frac{1}{n}\sum_{i=1}^{n}\left\Vert \bs{\Xi}_{i,[k]}\right\Vert ^{4}\frac{1}{n}\sum_{i=1}^{n}\left|\widetilde{Y}_{i,[k]}\right|^{4}}+\sqrt{\frac{1}{n}\sum_{i=1}^{n}\left\Vert \bs{\Xi}_{i,[k]}^{*}\right\Vert ^{4}\frac{1}{n}\sum_{i=1}^{n}\left|\widetilde{Y}_{i,[k]}\right|^{4}}}=O_{P}(1).
\end{align*}
Combining these results with (\ref{eq:xi-xi^*}), we can obtain that
\begin{align*}
 & \left\Vert \frac{1}{n}\sum_{i=1}^{n}\bs{\Xi}_{i,[k]}\bs{\Xi}_{i,[k]}^{\trans}\widetilde{Y}_{i,[k]}-\frac{1}{n}\sum_{i=1}^{n}\bs{\Xi}_{i,[k]}^{*}\bs{\Xi}_{i,[k]}^{*\trans}\widetilde{Y}_{i,[k]}^{*}\right\Vert \\
= & \sup_{\bs{\alpha}\in S^{d-1}}\left|\frac{1}{n}\sum_{i=1}^{n}\bs{\alpha}^{\trans}\bs{\Xi}_{i,[k]}\bs{\Xi}_{i,[k]}^{\trans}\bs{\alpha}\widetilde{Y}_{i,[k]}-\frac{1}{n}\sum_{i=1}^{n}\bs{\alpha}^{\trans}\bs{\Xi}_{i,[k]}^{*}\bs{\Xi}_{i,[k]}^{*\trans}\bs{\alpha}\widetilde{Y}_{i,[k]}^{*}\right|=o_{P}(1).
\end{align*}
Recall (\ref{eq:xixixi-xi*xi*xi*}), then we have 
\begin{equation}
\left\Vert \frac{1}{n}\sum_{i=1}^{n}\bs{\Xi}_{i,[k]}\bs{\Xi}_{i,[k]}^{\trans}\epsilon_{i,[k]}-\frac{1}{n}\sum_{i=1}^{n}\bs{\Xi}_{i,[k]}^{*}\bs{\Xi}_{i,[k]}^{*\trans}\epsilon_{i,[k]}^{*}\right\Vert =o_{P}(1).\label{eq:xixi=00005Cep-xi*/xi*=00005Cep*}
\end{equation}
Note that 
\begin{align*}
 & \e_{\mathcal{C}_{n}}\left[\left\Vert \bs{\Xi}_{i,[k]}^{*}\bs{\Xi}_{i,[k]}^{*\trans}\epsilon_{i,[k]}^{*}-\e_{\mathcal{C}_{n}}\left[\bs{\Xi}_{i,[k]}^{*}\bs{\Xi}_{i,[k]}^{*\trans}\epsilon_{i,[k]}^{*}\right]\right\Vert _{F}^{4/3}\right]\\
\leq & 2^{4/3-1}\e_{\mathcal{C}_{n}}\left[\left\Vert \bs{\Xi}_{i,[k]}^{*}\bs{\Xi}_{i,[k]}^{*\trans}\epsilon_{i,[k]}^{*}\right\Vert _{F}^{4/3}+\left\Vert \e_{\mathcal{C}_{n}}\left[\bs{\Xi}_{i,[k]}^{*}\bs{\Xi}_{i,[k]}^{*\trans}\epsilon_{i,[k]}^{*}\right]\right\Vert _{F}^{4/3}\right]\\
\leq & 2\e_{\mathcal{C}_{n}}\left[\left\Vert \bs{\Xi}_{i,[k]}^{*}\bs{\Xi}_{i,[k]}^{*\trans}\epsilon_{i,[k]}^{*}\right\Vert _{F}^{4/3}\right]+2\e_{\mathcal{C}_{n}}\left[\left\Vert \bs{\Xi}_{i,[k]}^{*}\bs{\Xi}_{i,[k]}^{*\trans}\epsilon_{i,[k]}^{*}\right\Vert _{F}\right]^{4/3}\leq4\e_{\mathcal{C}_{n}}\left[\left\Vert \bs{\Xi}_{i,[k]}^{*}\bs{\Xi}_{i,[k]}^{*\trans}\epsilon_{i,[k]}^{*}\right\Vert _{F}^{4/3}\right]\\
\leq & 8\e_{\mathcal{C}_{n}}\left[\left\Vert \bs{\Xi}_{i,[k]}^{*}\bs{\Xi}_{i,[k]}^{*\trans}\widetilde{Y}_{i,[k]}^{*}\right\Vert _{F}^{4/3}\right]+8\e_{\mathcal{C}_{n}}\left[\left\Vert \bs{\Xi}_{i,[k]}^{*}\bs{\Xi}_{i,[k]}^{*\trans}\bs{\Xi}_{i,[k]}^{*\trans}\bs{\beta}_{[k],\mathcal{C}_{n}}\right\Vert _{F}^{4/3}\right]\\
\leq & \frac{16}{\pi_{n[k]}^{4/3}}\e\left[\left\Vert \widetilde{\bs{\xi}}_{n}^{*}(\bs X_{i})\widetilde{\bs{\xi}}_{n}^{*}(\bs X_{i})^{\trans}\widetilde{Y}_{i}(1)\right\Vert _{F}^{4/3}\mid B_{i}=k\right]\\
 & +\frac{16}{(1-\pi_{n[k]})^{4/3}}\e\left[\left\Vert \widetilde{\bs{\xi}}_{n}^{*}(\bs X_{i})\widetilde{\bs{\xi}}_{n}^{*}(\bs X_{i})^{\trans}\widetilde{Y}_{i}(0)\right\Vert _{F}^{4/3}\mid B_{i}=k\right]\\
 & +8\e\left[\left\Vert \widetilde{\bs{\xi}}_{n}^{*}(\bs X_{i})\right\Vert ^{4}\mid B_{i}=k\right]\left\Vert \bs{\beta}_{[k],\mathcal{C}_{n}}\right\Vert ^{4/3}
\end{align*}
and
\begin{align*}
 & \e\left[\left\Vert \widetilde{\bs{\xi}}_{n}^{*}(\bs X_{i})\widetilde{\bs{\xi}}_{n}^{*}(\bs X_{i})^{\trans}\widetilde{Y}_{i}(a)\right\Vert _{F}^{4/3}\mid B_{i}=k\right]\\
= & \e\left[\tr\left\{ \widetilde{\bs{\xi}}_{n}^{*}(\bs X_{i})\widetilde{\bs{\xi}}_{n}^{*}(\bs X_{i})^{\trans}\widetilde{\bs{\xi}}_{n}^{*}(\bs X_{i})\widetilde{\bs{\xi}}_{n}^{*}(\bs X_{i})^{\trans}\widetilde{Y}_{i}(a)^{2}\right\} ^{2/3}\mid B_{i}=k\right]\\
= & \e\left[\left\Vert \widetilde{\bs{\xi}}_{n}^{*}(\bs X_{i})\right\Vert ^{8/3}\left|\widetilde{Y}_{i}(a)\right|^{4/3}\mid B_{i}=k\right]\\
\leq & \e\left[\left\Vert \widetilde{\bs{\xi}}_{n}^{*}(\bs X_{i})\right\Vert ^{(8/3)\times(3/2)}\mid B_{i}=k\right]^{2/3}\e\left[\left|\widetilde{Y}_{i}(a)\right|^{(4/3)\times3}\mid B_{i}=k\right]^{1/3}\\
< & \infty,
\end{align*}
where the last two steps follow from Hölder's inequality, Assumption
\ref{assu:conditions on =00005Cxi-general D(v)} and $\e\left[\left|\widetilde{Y}_{i}(a)\right|^{4}\right]<\infty$.
We can conclude that 
\[
\e_{\mathcal{C}_{n}}\left[\left\Vert \bs{\Xi}_{i,[k]}^{*}\bs{\Xi}_{i,[k]}^{*\trans}\epsilon_{i,[k]}^{*}-\e_{\mathcal{C}_{n}}\left[\bs{\Xi}_{i,[k]}^{*}\bs{\Xi}_{i,[k]}^{*\trans}\epsilon_{i,[k]}^{*}\right]\right\Vert _{F}^{4/3}\right]=O_{P}(1).
\]
By Lemma \ref{lem:LLN for triangular array}, we have 
\begin{align*}
 & \left\Vert \frac{1}{n}\sum_{i=1}^{n}\bs{\Xi}_{i,[k]}^{*}\bs{\Xi}_{i,[k]}^{*\trans}\epsilon_{i,[k]}^{*}-\frac{1}{n}\sum_{i=1}^{n}\e_{\mathcal{C}_{n}}\left[\bs{\Xi}_{i,[k]}^{*}\bs{\Xi}_{i,[k]}^{*\trans}\epsilon_{i,[k]}^{*}\right]\right\Vert \\
\leq & \left\Vert \frac{1}{n}\sum_{i=1}^{n}\bs{\Xi}_{i,[k]}^{*}\bs{\Xi}_{i,[k]}^{*\trans}\epsilon_{i,[k]}^{*}-\frac{1}{n}\sum_{i=1}^{n}\e_{\mathcal{C}_{n}}\left[\bs{\Xi}_{i,[k]}^{*}\bs{\Xi}_{i,[k]}^{*\trans}\epsilon_{i,[k]}^{*}\right]\right\Vert _{F}=o_{P}(1).
\end{align*}
This, combined with (\ref{eq:xixi=00005Cep-xi*/xi*=00005Cep*}) gives
that 
\begin{align*}
\frac{1}{n}\sum_{i=1}^{n}\bs{\Xi}_{i,[k]}\bs{\Xi}_{i,[k]}^{\trans}\epsilon_{i,[k]} & =\frac{1}{n}\sum_{i=1}^{n}\e_{\mathcal{C}_{n}}\left[\bs{\Xi}_{i,[k]}^{*}\bs{\Xi}_{i,[k]}^{*\trans}\epsilon_{i,[k]}^{*}\right]+o_{P}(1)\\
 & =\mathbf{\Sigma}_{\bs{\Xi}\bs{\Xi}\epsilon[k]}^{\mathcal{C}_{n}}+o_{P}(1).
\end{align*}

Now, we recall (\ref{eq:decomposition of second bias}). It follows
from (\ref{eq:sample ave of Xi to Xi*}), (\ref{eq:xixiT-condition xixiT})
and $\frac{1}{n}\sum_{i=1}^{n}\bs{\Xi}_{i,[k]}^{*\trans}=O_{P}(n^{-1/2})$
that
\begin{align}
 & \frac{1}{2n}\sum_{i=1}^{n}\frac{\rho^{\prime\prime\prime}(0)}{\rho^{\prime\prime}(0)^{2}}\widehat{\bs{\lambda}}_{lin}^{\trans}\bs{\Xi}_{i}\bs{\Xi}_{i}^{\trans}\widehat{\bs{\lambda}}_{lin}\epsilon_{i}\nonumber \\
= & \frac{\rho^{\prime\prime\prime}(0)}{2\rho^{\prime\prime}(0)^{2}}\sum_{k=1}^{K}\left\{ \frac{1}{n}\sum_{i=1}^{n}\bs{\Xi}_{i,[k]}^{*\trans}+o_{P}(n^{-1})\right\} \left(\mathbf{\Sigma}_{[k]}^{\mathcal{C}_{n}}+o_{P}(1)\right)^{-1}\nonumber \\
 & \qquad\times\left\{ \mathbf{\Sigma}_{\bs{\Xi}\bs{\Xi}\epsilon[k]}^{\mathcal{C}_{n}}+o_{P}(1)\right\} \left(\mathbf{\Sigma}_{[k]}^{\mathcal{C}_{n}}+o_{P}(1)\right)^{-1}\left\{ \frac{1}{n}\sum_{i=1}^{n}\bs{\Xi}_{i,[k]}^{*}+o_{P}(n^{-1})\right\} \nonumber \\
= & \underbrace{\frac{\rho^{\prime\prime\prime}(0)}{2\rho^{\prime\prime}(0)^{2}}\sum_{k=1}^{K}\frac{1}{n}\sum_{i=1}^{n}\bs{\Xi}_{i,[k]}^{*\trans}\left(\mathbf{\Sigma}_{[k]}^{\mathcal{C}_{n}}\right)^{-1}\mathbf{\Sigma}_{\bs{\Xi}\bs{\Xi}\epsilon[k]}^{\mathcal{C}_{n}}\left(\mathbf{\Sigma}_{[k]}^{\mathcal{C}_{n}}\right)^{-1}\frac{1}{n}\sum_{i=1}^{n}\bs{\Xi}_{i,[k]}^{*}}_{\text{bias}_{2}}+o_{P}(n^{-1}).\label{eq:second bias final}
\end{align}

Combining (\ref{eq:first bias final}) and (\ref{eq:second bias final}),
we have 
\begin{equation}
\widehat{\bs{\lambda}}_{lin}^{\trans}\frac{1}{n}\sum_{i=1}^{n}\bs{\Xi}_{i}\epsilon_{i}+\frac{1}{2n}\sum_{i=1}^{n}\frac{\rho^{\prime\prime\prime}(0)}{\rho^{\prime\prime}(0)^{2}}\widehat{\bs{\lambda}}_{lin}^{\trans}\bs{\Xi}_{i}\bs{\Xi}_{i}^{\trans}\widehat{\bs{\lambda}}_{lin}\epsilon_{i}=\text{bias}_{1}+\text{bias}_{2}+O_{P}(\Delta_{n}n^{-1/2})+o_{P}(n^{-1}).\label{eq:second order bias decomposition}
\end{equation}
Note that $(\bs{\Xi}_{i,[k]}^{*},\epsilon_{i,[k]}^{*})$, $i=1,\ldots,n$,
are independent conditional on $\mathcal{C}_{n}$, we have 
\begin{align*}
 & \e_{\mathcal{C}_{n}}\left[\text{bias}_{1}\right]=-\sum_{k=1}^{K}\frac{1}{n^{2}}\e_{\mathcal{C}_{n}}\left[\sum_{1\leq i,j\leq n}\bs{\Xi}_{i,[k]}^{*\trans}\left(\mathbf{\Sigma}_{[k]}^{\mathcal{C}_{n}}\right)^{-1}\bs{\Xi}_{j,[k]}^{*}\epsilon_{j,[k]}^{*}\right]\\
= & -\sum_{k=1}^{K}\frac{1}{n^{2}}\sum_{i=1}^{n}\e_{\mathcal{C}_{n}}\left[\bs{\Xi}_{i,[k]}^{*\trans}\left(\mathbf{\Sigma}_{[k]}^{\mathcal{C}_{n}}\right)^{-1}\bs{\Xi}_{i,[k]}^{*}\epsilon_{i,[k]}^{*}\right]\\
= & -\sum_{k=1}^{K}\frac{1}{n^{2}}\sum_{i=1}^{n}\e_{\mathcal{C}_{n}}\left[\tr\left\{ \bs{\Xi}_{i,[k]}^{*\trans}\left(\mathbf{\Sigma}_{[k]}^{\mathcal{C}_{n}}\right)^{-1}\bs{\Xi}_{i,[k]}^{*}\epsilon_{i,[k]}^{*}\right\} \right]\\
= & -\sum_{k=1}^{K}\frac{1}{n^{2}}\sum_{i=1}^{n}\e_{\mathcal{C}_{n}}\left[\tr\left\{ \left(\mathbf{\Sigma}_{[k]}^{\mathcal{C}_{n}}\right)^{-1}\bs{\Xi}_{i,[k]}^{*}\bs{\Xi}_{i,[k]}^{*\trans}\epsilon_{i,[k]}^{*}\right\} \right]\\
= & -\sum_{k=1}^{K}\frac{1}{n^{2}}\sum_{i=1}^{n}\tr\left\{ \e_{\mathcal{C}_{n}}\left[\left(\mathbf{\Sigma}_{[k]}^{\mathcal{C}_{n}}\right)^{-1}\bs{\Xi}_{i,[k]}^{*}\bs{\Xi}_{i,[k]}^{*\trans}\epsilon_{i,[k]}^{*}\right]\right\} \\
= & -\sum_{k=1}^{K}\frac{1}{n^{2}}\sum_{i=1}^{n}\tr\left\{ \left(\mathbf{\Sigma}_{[k]}^{\mathcal{C}_{n}}\right)^{-1}\e_{\mathcal{C}_{n}}\left[\bs{\Xi}_{i,[k]}^{*}\bs{\Xi}_{i,[k]}^{*\trans}\epsilon_{i,[k]}^{*}\right]\right\} \\
= & -\sum_{k=1}^{K}\frac{1}{n}\tr\left\{ \frac{1}{n}\sum_{i=1}^{n}\left(\mathbf{\Sigma}_{[k]}^{\mathcal{C}_{n}}\right)^{-1}\e_{\mathcal{C}_{n}}\left[\bs{\Xi}_{i,[k]}^{*}\bs{\Xi}_{i,[k]}^{*\trans}\epsilon_{i,[k]}^{*}\right]\right\} \\
= & -\sum_{k=1}^{K}\frac{1}{n}\tr\left\{ \left(\mathbf{\Sigma}_{[k]}^{\mathcal{C}_{n}}\right)^{-1}\frac{1}{n}\sum_{i=1}^{n}\e_{\mathcal{C}_{n}}\left[\bs{\Xi}_{i,[k]}^{*}\bs{\Xi}_{i,[k]}^{*\trans}\epsilon_{i,[k]}^{*}\right]\right\} \\
= & -\frac{1}{n}\sum_{k=1}^{K}\tr\left\{ \left(\mathbf{\Sigma}_{[k]}^{\mathcal{C}_{n}}\right)^{-1}\mathbf{\Sigma}_{\bs{\Xi}\bs{\Xi}\epsilon[k]}^{\mathcal{C}_{n}}\right\} 
\end{align*}
and 
\begin{align*}
\e_{\mathcal{C}_{n}}\left[\text{bias}_{2}\right] & =\e_{\mathcal{C}_{n}}\left[\frac{\rho^{\prime\prime\prime}(0)}{2\rho^{\prime\prime}(0)^{2}}\sum_{k=1}^{K}\frac{1}{n}\sum_{i=1}^{n}\bs{\Xi}_{i,[k]}^{*\trans}\left(\mathbf{\Sigma}_{[k]}^{\mathcal{C}_{n}}\right)^{-1}\frac{1}{n}\sum_{i=1}^{n}\bs{\Xi}_{i,[k]}^{*}\bs{\Xi}_{i,[k]}^{*\trans}\epsilon_{i,[k]}^{*}\left(\mathbf{\Sigma}_{[k]}^{\mathcal{C}_{n}}\right)^{-1}\frac{1}{n}\sum_{i=1}^{n}\bs{\Xi}_{i,[k]}^{*}\right]\\
 & =\frac{\rho^{\prime\prime\prime}(0)}{2\rho^{\prime\prime}(0)^{2}}\sum_{k=1}^{K}\frac{1}{n^{2}}\e_{\mathcal{C}_{n}}\left[\sum_{1\leq i,j\leq n}\bs{\Xi}_{i,[k]}^{*\trans}\left(\mathbf{\Sigma}_{[k]}^{\mathcal{C}_{n}}\right)^{-1}\mathbf{\Sigma}_{\bs{\Xi}\bs{\Xi}\epsilon[k]}^{\mathcal{C}_{n}}\left(\mathbf{\Sigma}_{[k]}^{\mathcal{C}_{n}}\right)^{-1}\bs{\Xi}_{j,[k]}^{*}\right]\\
 & =\frac{\rho^{\prime\prime\prime}(0)}{2\rho^{\prime\prime}(0)^{2}}\sum_{k=1}^{K}\frac{1}{n^{2}}\sum_{i=1}^{n}\e_{\mathcal{C}_{n}}\left[\tr\left\{ \bs{\Xi}_{i,[k]}^{*\trans}\left(\mathbf{\Sigma}_{[k]}^{\mathcal{C}_{n}}\right)^{-1}\mathbf{\Sigma}_{\bs{\Xi}\bs{\Xi}\epsilon[k]}^{\mathcal{C}_{n}}\left(\mathbf{\Sigma}_{[k]}^{\mathcal{C}_{n}}\right)^{-1}\bs{\Xi}_{i,[k]}^{*}\right\} \right]\\
 & =\frac{\rho^{\prime\prime\prime}(0)}{2\rho^{\prime\prime}(0)^{2}}\sum_{k=1}^{K}\frac{1}{n^{2}}\sum_{i=1}^{n}\e_{\mathcal{C}_{n}}\left[\tr\left\{ \left(\mathbf{\Sigma}_{[k]}^{\mathcal{C}_{n}}\right)^{-1}\mathbf{\Sigma}_{\bs{\Xi}\bs{\Xi}\epsilon[k]}^{\mathcal{C}_{n}}\left(\mathbf{\Sigma}_{[k]}^{\mathcal{C}_{n}}\right)^{-1}\bs{\Xi}_{i,[k]}^{*}\bs{\Xi}_{i,[k]}^{*\trans}\right\} \right]\\
 & =\frac{\rho^{\prime\prime\prime}(0)}{2\rho^{\prime\prime}(0)^{2}}\sum_{k=1}^{K}\frac{1}{n^{2}}\sum_{i=1}^{n}\tr\left\{ \left(\mathbf{\Sigma}_{[k]}^{\mathcal{C}_{n}}\right)^{-1}\mathbf{\Sigma}_{\bs{\Xi}\bs{\Xi}\epsilon[k]}^{\mathcal{C}_{n}}\left(\mathbf{\Sigma}_{[k]}^{\mathcal{C}_{n}}\right)^{-1}\e_{\mathcal{C}_{n}}\left[\bs{\Xi}_{i,[k]}^{*}\bs{\Xi}_{i,[k]}^{*\trans}\right]\right\} \\
 & =\frac{\rho^{\prime\prime\prime}(0)}{2\rho^{\prime\prime}(0)^{2}}\sum_{k=1}^{K}\frac{1}{n}\tr\left\{ \left(\mathbf{\Sigma}_{[k]}^{\mathcal{C}_{n}}\right)^{-1}\mathbf{\Sigma}_{\bs{\Xi}\bs{\Xi}\epsilon[k]}^{\mathcal{C}_{n}}\left(\mathbf{\Sigma}_{[k]}^{\mathcal{C}_{n}}\right)^{-1}\frac{1}{n}\sum_{i=1}^{n}\e_{\mathcal{C}_{n}}\left[\bs{\Xi}_{i,[k]}^{*}\bs{\Xi}_{i,[k]}^{*\trans}\right]\right\} \\
 & =\frac{\rho^{\prime\prime\prime}(0)}{2\rho^{\prime\prime}(0)^{2}}\frac{1}{n}\sum_{k=1}^{K}\tr\left\{ \left(\mathbf{\Sigma}_{[k]}^{\mathcal{C}_{n}}\right)^{-1}\mathbf{\Sigma}_{\bs{\Xi}\bs{\Xi}\epsilon[k]}^{\mathcal{C}_{n}}\right\} .
\end{align*}
As a result, we have 
\[
\e\left[\text{bias}_{1}+\text{bias}_{2}\right]=\e\left[\e_{\mathcal{C}_{n}}\left[\text{bias}_{1}+\text{bias}_{2}\right]\right]=\left\{ \frac{\rho^{\prime\prime\prime}(0)}{2\rho^{\prime\prime}(0)^{2}}-1\right\} \frac{1}{n}\sum_{k=1}^{K}\tr\left\{ \e\left[\left(\mathbf{\Sigma}_{[k]}^{\mathcal{C}_{n}}\right)^{-1}\mathbf{\Sigma}_{\bs{\Xi}\bs{\Xi}\epsilon[k]}^{\mathcal{C}_{n}}\right]\right\} .
\]

Let 
\begin{align*}
\psi_{1,i} & :=\sum_{k=1}^{K}\left\{ \left(\frac{A_{i}}{\pi_{n[k]}}-\frac{1-A_{i}}{1-\pi_{n[k]}}\right)\1(B_{i}=k)\cdot Y_{i}-\tau-\bs{\beta}_{[k],\mathcal{C}_{n}}^{\trans}\bs{\Xi}_{i,[k]}^{*}\right\} ,
\end{align*}
\[
\phi_{1,i}:=\sum_{k=1}^{K}\left\{ \e\left[Y_{i}(1)\mid B_{i}=k\right]-\e\left[Y_{i}(0)\mid B_{i}=k\right]\right\} \{\1(B_{i}=k)-p_{[k]}\}+\sum_{k=1}^{K}\left(\widetilde{Y}_{i,[k]}^{*}-\bs{\beta}_{[k],\mathcal{C}_{n}}^{\trans}\bs{\Xi}_{i}^{*}\right)
\]
and 
\[
\psi_{2}:=n\left\{ \text{bias}_{1}+\text{bias}_{2}\right\} ,
\]
then it follows from (\ref{eq:decomposition of tau-tau*}), (\ref{eq:R1+R2 diverging D(v)})
and (\ref{eq:second order bias decomposition}) that 
\begin{align*}
\widehat{\tau}_{\mathrm{cal}}-\tau & =\frac{1}{n}\sum_{i=1}^{n}\sum_{k=1}^{K}\left\{ \e\left[Y_{i}(1)\mid B_{i}=k\right]-\e\left[Y_{i}(0)\mid B_{i}=k\right]\right\} \{\1(B_{i}=k)-p_{[k]}\}\\
 & \qquad+\frac{1}{n}\sum_{i=1}^{n}\sum_{k=1}^{K}\left(\widetilde{Y}_{i,[k]}^{*}-\bs{\beta}_{[k],\mathcal{C}_{n}}^{\trans}\bs{\Xi}_{i,[k]}^{*}\right)\\
 & \qquad+\text{bias}_{1}+\text{bias}_{2}+O_{P}(\Delta_{n}n^{-1/2})+o_{P}(n^{-1})\\
 & =\frac{1}{n}\sum_{i=1}^{n}\phi_{1,i}+\frac{1}{n}\psi_{2}+O_{P}(\Delta_{n}n^{-1/2})+o_{P}(n^{-1}).\\
 & =\frac{1}{n}\sum_{i=1}^{n}\psi_{1,i}+\frac{1}{n}\psi_{2}+O_{P}(\Delta_{n}n^{-1/2})+o_{P}(n^{-1})+\frac{1}{n}\sum_{i=1}^{n}\phi_{1,i}-\frac{1}{n}\sum_{i=1}^{n}\psi_{1,i}
\end{align*}
Now, it remains to show that 
\[
\frac{1}{n}\sum_{i=1}^{n}\phi_{1,i}=\frac{1}{n}\sum_{i=1}^{n}\psi_{1,i}
\]
and 
\[
\e\left[\frac{1}{n}\sum_{i=1}^{n}\psi_{1,i}\right]=0.
\]
First, by $\pi_{n[k]}=\frac{1}{n_{[k]}}\sum_{i=1}^{n}\1(B_{i}=k)A_{i}$
we have 
\begin{align*}
 & \frac{1}{n}\sum_{i=1}^{n}\psi_{1,i}\\
= & \sum_{k=1}^{K}\frac{1}{n}\sum_{i=1}^{n}\frac{A_{i}}{\pi_{n[k]}}\1(B_{i}=k)\cdot\left\{ Y_{i}(1)-\e\left[Y_{i}(1)\mid B_{i}=k\right]+\e\left[Y_{i}(1)\mid B_{i}=k\right]\right\} \\
 & \quad-\sum_{k=1}^{K}\frac{1}{n}\sum_{i=1}^{n}\frac{1-A_{i}}{1-\pi_{n[k]}}\1(B_{i}=k)\cdot\left\{ Y_{i}(1)-\e\left[Y_{i}(0)\mid B_{i}=k\right]+\e\left[Y_{i}(0)\mid B_{i}=k\right]\right\} \\
 & \quad-\sum_{k=1}^{K}p_{[k]}\left\{ \e\left[Y_{i}(1)\mid B_{i}=k\right]-\e\left[Y_{i}(0)\mid B_{i}=k\right]\right\} -\sum_{k=1}^{K}\frac{1}{n}\sum_{i=1}^{n}\bs{\beta}_{[k],\mathcal{C}_{n}}^{\trans}\bs{\Xi}_{i,[k]}^{*}\\
= & \sum_{k=1}^{K}\frac{1}{n}\sum_{i=1}^{n}\frac{A_{i}}{\pi_{n[k]}}\1(B_{i}=k)\widetilde{Y}_{i}(1)-\sum_{k=1}^{K}\frac{1}{n}\sum_{i=1}^{n}\frac{1-A_{i}}{1-\pi_{n[k]}}\1(B_{i}=k)\widetilde{Y}_{i}(0)\\
 & \quad+\sum_{k=1}^{K}\frac{n_{[k]}}{n}\left\{ \e\left[Y_{i}(1)\mid B_{i}=k\right]-\e\left[Y_{i}(0)\mid B_{i}=k\right]\right\} \\
 & \quad-\sum_{k=1}^{K}p_{[k]}\left\{ \e\left[Y_{i}(1)\mid B_{i}=k\right]-\e\left[Y_{i}(0)\mid B_{i}=k\right]\right\} -\sum_{k=1}^{K}\frac{1}{n}\sum_{i=1}^{n}\bs{\beta}_{[k],\mathcal{C}_{n}}^{\trans}\bs{\Xi}_{i,[k]}^{*}\\
= & \frac{1}{n}\sum_{i=1}^{n}\phi_{1,i}.
\end{align*}
Second, it follows from 
\[
\e\left[\frac{1}{n}\sum_{i=1}^{n}\sum_{k=1}^{K}\left\{ \e\left[Y_{i}(1)\mid B_{i}=k\right]-\e\left[Y_{i}(0)\mid B_{i}=k\right]\right\} \{\1(B_{i}=k)-p_{[k]}\}\right]=0
\]
and 
\[
\e\left[\frac{1}{n}\sum_{i=1}^{n}\sum_{k=1}^{K}\left(\widetilde{Y}_{i,[k]}^{*}-\bs{\beta}_{[k]}^{*\trans}\bs{\Xi}_{i,[k]}^{*}\right)\right]=\sum_{k=1}^{K}\frac{1}{n}\sum_{i=1}^{n}\e\left[\e_{\mathcal{C}_{n}}\left[\left(\widetilde{Y}_{i,[k]}^{*}-\bs{\beta}_{[k],\mathcal{C}_{n}}^{\trans}\bs{\Xi}_{i,[k]}^{*}\right)\right]\right]=0
\]
that 
\[
\e\left[\frac{1}{n}\sum_{i=1}^{n}\phi_{1,i}\right]=0.
\]
As a result, $\e\left[\frac{1}{n}\sum_{i=1}^{n}\psi_{1,i}\right]=\e\left[\frac{1}{n}\sum_{i=1}^{n}\phi_{1,i}\right]=0$.
The proof is completed.$\hfill\qedsymbol$

\section{\label{sec:Additional-simulation-results}Additional simulation results}

We consider the following model, which imposes the heterogeneity of
the conditional mean functions $g_{0}(\bs X)$ and $g_{1}(\bs X)$
across different strata.
\begin{table}[!tb]
\caption{\label{tab:Model4}The comparison of the performance of different
estimators under Model 4.}

\resizebox{\textwidth}{!}{%
\begin{threeparttable}
\begin{centering}
\begin{tabular}{clrrrrrrrrrrrr}
\toprule 
\multirow{2}{*}{$n$} & \multirow{2}{*}{Estimator} & \multicolumn{4}{c}{Simple Rand.} & \multicolumn{4}{c}{Stratified Block Rand.} & \multicolumn{4}{c}{Minimization}\tabularnewline
\cmidrule{3-14}
 &  & Bias & SD & SE & CP & Bias & SD & SE & CP & Bias & SD & SE & CP\tabularnewline
\midrule 
\multirow{11}{*}{500} & \texttt{cal\_rf} & 0.08 & 4.67 & 5.13 & 0.957 & 0.80 & 5.07 & 5.08 & 0.940 & 0.24 & 5.14 & 5.11 & 0.937\tabularnewline
 & \texttt{cal\_nn} & 0.18 & 4.96 & 5.55 & 0.977 & 0.58 & 5.37 & 5.49 & 0.953 & 0.08 & 5.50 & 5.51 & 0.937\tabularnewline
 & \texttt{cal\_rfnn} & 0.13 & 4.65 & 5.14 & 0.967 & 0.97 & 5.02 & 5.09 & 0.933 & 0.40 & 5.11 & 5.12 & 0.930\tabularnewline
 & \texttt{cal\_rflin} & 0.08 & 4.67 & 5.11 & 0.960 & 0.73 & 5.02 & 5.05 & 0.943 & 0.15 & 5.14 & 5.08 & 0.933\tabularnewline
 & \texttt{cal\_rf\_g} & 0.09 & 4.64 & 5.14 & 0.953 & 0.82 & 5.10 & 5.08 & 0.937 & 0.21 & 5.16 & 5.11 & 0.947\tabularnewline
 & \texttt{cal\_nn\_g} & 0.18 & 5.03 & 5.46 & 0.970 & 0.72 & 5.32 & 5.41 & 0.963 & 0.15 & 5.34 & 5.43 & 0.940\tabularnewline
 & \texttt{cal\_lin\_EL} & 0.37 & 4.72 & 5.29 & 0.967 & 0.41 & 5.14 & 5.24 & 0.953 & 0.07 & 5.43 & 5.27 & 0.933\tabularnewline
 & \texttt{aipw\_rf} & 0.42 & 4.82 & 5.15 & 0.950 & 0.35 & 5.07 & 5.10 & 0.947 & 0.25 & 5.17 & 5.13 & 0.937\tabularnewline
 & \texttt{aipw\_nn} & 0.47 & 6.27 & 7.02 & 0.967 & 0.36 & 6.80 & 6.96 & 0.950 & 0.59 & 6.90 & 6.92 & 0.953\tabularnewline
 & \texttt{aipw\_lin} & 0.34 & 5.08 & 5.49 & 0.973 & 0.57 & 5.29 & 5.42 & 0.950 & 0.18 & 5.68 & 5.46 & 0.933\tabularnewline
 & \texttt{sdim} & 0.33 & 5.58 & 5.84 & 0.960 & 0.31 & 5.79 & 5.79 & 0.963 & 0.15 & 6.00 & 5.81 & 0.927\tabularnewline
\midrule 
\multirow{11}{*}{1000} & \texttt{cal\_rf} & 0.00 & 3.40 & 3.56 & 0.970 & 0.19 & 3.59 & 3.56 & 0.940 & 0.28 & 3.58 & 3.55 & 0.953\tabularnewline
 & \texttt{cal\_nn} & 0.05 & 3.45 & 3.68 & 0.973 & 0.11 & 3.76 & 3.69 & 0.933 & 0.19 & 3.74 & 3.68 & 0.947\tabularnewline
 & \texttt{cal\_rfnn} & 0.15 & 3.40 & 3.55 & 0.963 & 0.35 & 3.59 & 3.55 & 0.933 & 0.42 & 3.57 & 3.54 & 0.943\tabularnewline
 & \texttt{cal\_rflin} & 0.00 & 3.39 & 3.54 & 0.963 & 0.18 & 3.61 & 3.54 & 0.937 & 0.27 & 3.58 & 3.54 & 0.953\tabularnewline
 & \texttt{cal\_rf\_g} & 0.00 & 3.42 & 3.56 & 0.960 & 0.18 & 3.59 & 3.56 & 0.943 & 0.30 & 3.59 & 3.55 & 0.947\tabularnewline
 & \texttt{cal\_nn\_g} & 0.09 & 3.40 & 3.65 & 0.970 & 0.12 & 3.71 & 3.66 & 0.933 & 0.25 & 3.71 & 3.65 & 0.950\tabularnewline
 & \texttt{cal\_lin\_EL} & 0.23 & 3.43 & 3.71 & 0.980 & 0.10 & 3.71 & 3.71 & 0.940 & 0.01 & 3.69 & 3.70 & 0.957\tabularnewline
 & \texttt{aipw\_rf} & 0.23 & 3.40 & 3.57 & 0.970 & 0.09 & 3.65 & 3.57 & 0.933 & 0.05 & 3.59 & 3.56 & 0.947\tabularnewline
 & \texttt{aipw\_nn} & 0.28 & 3.83 & 4.02 & 0.963 & 0.13 & 4.07 & 4.04 & 0.943 & 0.05 & 4.16 & 4.04 & 0.957\tabularnewline
 & \texttt{aipw\_lin} & 0.23 & 3.45 & 3.75 & 0.977 & 0.15 & 3.76 & 3.75 & 0.937 & 0.05 & 3.73 & 3.73 & 0.957\tabularnewline
 & \texttt{sdim} & 0.30 & 3.89 & 4.11 & 0.980 & 0.15 & 4.29 & 4.11 & 0.930 & 0.01 & 4.11 & 4.10 & 0.937\tabularnewline
\midrule 
\multirow{11}{*}{2000} & \texttt{cal\_rf} & 0.37 & 2.57 & 2.49 & 0.933 & 0.35 & 2.61 & 2.48 & 0.940 & 0.20 & 2.44 & 2.50 & 0.953\tabularnewline
 & \texttt{cal\_nn} & 0.37 & 2.56 & 2.51 & 0.933 & 0.36 & 2.66 & 2.50 & 0.933 & 0.21 & 2.49 & 2.52 & 0.953\tabularnewline
 & \texttt{cal\_rfnn} & 0.46 & 2.55 & 2.47 & 0.930 & 0.45 & 2.62 & 2.47 & 0.933 & 0.12 & 2.42 & 2.49 & 0.957\tabularnewline
 & \texttt{cal\_rflin} & 0.37 & 2.57 & 2.48 & 0.920 & 0.35 & 2.61 & 2.48 & 0.940 & 0.23 & 2.44 & 2.49 & 0.950\tabularnewline
 & \texttt{cal\_rf\_g} & 0.39 & 2.57 & 2.49 & 0.933 & 0.37 & 2.62 & 2.48 & 0.940 & 0.17 & 2.43 & 2.50 & 0.953\tabularnewline
 & \texttt{cal\_nn\_g} & 0.37 & 2.57 & 2.51 & 0.933 & 0.37 & 2.66 & 2.50 & 0.930 & 0.22 & 2.50 & 2.52 & 0.957\tabularnewline
 & \texttt{cal\_lin\_EL} & 0.29 & 2.64 & 2.61 & 0.920 & 0.14 & 2.71 & 2.60 & 0.947 & 0.35 & 2.62 & 2.62 & 0.953\tabularnewline
 & \texttt{aipw\_rf} & 0.26 & 2.58 & 2.49 & 0.940 & 0.19 & 2.61 & 2.49 & 0.937 & 0.33 & 2.48 & 2.50 & 0.950\tabularnewline
 & \texttt{aipw\_nn} & 0.29 & 2.57 & 2.56 & 0.940 & 0.21 & 2.65 & 2.56 & 0.947 & 0.35 & 2.58 & 2.57 & 0.957\tabularnewline
 & \texttt{aipw\_lin} & 0.32 & 2.63 & 2.61 & 0.930 & 0.09 & 2.73 & 2.61 & 0.943 & 0.38 & 2.62 & 2.63 & 0.960\tabularnewline
 & \texttt{sdim} & 0.33 & 2.86 & 2.90 & 0.937 & 0.10 & 3.07 & 2.90 & 0.940 & 0.36 & 2.94 & 2.91 & 0.950\tabularnewline
\bottomrule
\end{tabular}
\par\end{centering}
\begin{tablenotes}[flushleft]
      \footnotesize 
      \item \textit{Abbreviations:} Rand., Randomization; SD, standard deviation, SE: standard error; CP, coverage probability. 
    \end{tablenotes}
  \end{threeparttable}
}
\end{table}

\textbf{Model 4.} In this model, we set
\begin{align*}
g_{0}(\bs X_{i}) & =\mu_{0}+\left(\beta_{01}X_{i1}+\beta_{02}X_{i2}\right)S_{i}+\beta_{03}\log(X_{i1}+1)\1(S_{i}=1)\\
g_{1}(\bs X_{i}) & =\mu_{1}+\left(\beta_{11}X_{i1}+\beta_{12}X_{i2}\right)S_{i}+\beta_{13}\exp(X_{i2})\1(S_{i}=-1),
\end{align*}
with $\mu_{0}=5$, $\mu_{1}=5$, $(\beta_{01},\beta_{02},\beta_{03})=(20,30,50)$,
and $(\beta_{11},\beta_{12},\beta_{13})=(20,30,65)$. Additionally,
the variables are specified as follows: $\epsilon_{0,i}\sim N(0,1)$,
$\epsilon_{1,i}\sim N(0,9)$, $X_{i1}\sim\text{Beta}(3,4)$, $X_{i2}\sim\text{Uniform}(-2,2)$
with these two variables being independent of each other. The additional
covariates $X_{i3},\dots,X_{ip}$ are first generated as in Model
1. Then we randomly select $\left\lfloor p/3\right\rfloor $ covariates
from the additional covariates and multiply them by either $X_{i1}$
or $X_{i2}$ with equal probability to form the final additional covariates.
The randomization variable $S_{i}$ takes values in $\{1,-1\}$ with
equal probability, and is independent of $X_{ij}$ for $j=1,\dots,p$.
The simulation results for Model 4 are presented in Table~\ref{tab:Model4}.
When sample size is large ($n=2000$), the performances of the adjusted
estimators are nearly identical, with \texttt{cal\_rf\_g} performing
slightly better and \texttt{aipw\_nn} performing slightly worse. However,
with smaller sample sizes ($n=500,1000$), two key observations emerge:
(i) the calibration estimators outperform the AIPW estimators, and
(ii) the random forest-based estimators perform better than the neural
network-based estimators, with the linear regression-based estimators
performing intermediate to the random forest and neural network-based
estimators.

\putbib
\end{bibunit}


@misc{gu2024incorporatingexternaldataanalyzing,
  title  = {Incorporating External Data for Analyzing Randomized Clinical Trials: A Transfer Learning Approach},
  author = {Gu, Yujia and Liu, Hanzhong and Ma, Wei},
  year   = {2024},
  note   = {arXiv preprint arXiv:2409.04126}
}

@article{jiang2025Adjustments,
  title = {Adjustments with Many Regressors under Covariate-Adaptive Randomizations},
  author = {Jiang, Liang and Li, Liyao and Miao, Ke and Zhang, Yichong},
  year = {2025},
  journal = {Journal of Econometrics},
  volume = {249},
  pages = {105991}
}

@article{ma2026Integrating,
title = {Supplementary materials for ``{{Integrating}} Heterogeneous Information in Randomized Experiments: {{A}} Unified Calibration Framework''},
author={Ma, Wei and Wu, Zeqi and Zhang, Zheng},
year = {2026}

}

@misc{FDA2019RareDiseases,
  title        = {Rare Diseases: Natural History Studies for Drug Development: Draft Guidance for Industry},
  author       = {{FDA}},
  year         = {2019},
  howpublished = {U.S. Department of Health and Human Services},
  url          = {https://www.fda.gov/media/122425/download}
}

@article{chi2022Asymptotic,
  title   = {Asymptotic Properties of High-Dimensional Random Forests},
  author  = {Chi, Chien-Ming and Vossler, Patrick and Fan, Yingying and Lv, Jinchi},
  year    = {2022},
  journal = {The Annals of Statistics},
  volume  = {50},
  number  = {6},
  pages   = {3415--3438}
}

@misc{gu2025assumptionleancovariateadjustmentcovariate,
  title  = {Assumption-lean Covariate Adjustment under Covariate Adaptive Randomization when $p = o(n)$},
  author = {Gu, Yujia and Liu, Lin and Ma, Wei},
  year   = {2025},
  note   = {arXiv preprint arXiv:2512.20046}
}

@article{lu2025Debiased,
  title   = {Debiased Regression Adjustment in Completely Randomized Experiments with Moderately High-Dimensional Covariates},
  author  = {Lu, Xin and Yang, Fan and Wang, Yuhao},
  year    = {2025},
  journal = {The Annals of Statistics},
  volume  = {53},
  number  = {4},
  pages   = {1535--1558}
}

@article{gu2023RegressionBased,
  title   = {Regression-based Multiple Treatment Effect Estimation under Covariate-Adaptive Randomization},
  author  = {Gu, Yujia and Liu, Hanzhong and Ma, Wei},
  year    = {2023},
  journal = {Biometrics},
  volume  = {79},
  number  = {4},
  pages   = {2869--2880}
}

@article{callegaro2023Historical,
  title   = {Historical Controls in Clinical Trials: A Note on Linking {{Pocock}}'s Model with the Robust Mixture Priors},
  author  = {Callegaro, Andrea and Galwey, Nicholas and Abellan, Juan J.},
  year    = {2023},
  journal = {Biostatistics},
  volume  = {24},
  number  = {2},
  pages   = {443--448}
}

@book{imbens2015Causal,
  title     = {Causal Inference for Statistics, Social, and Biomedical Sciences: An Introduction},
  author    = {Imbens, Guido W. and Rubin, Donald B.},
  year      = {2015},
  publisher = {Cambridge University Press},
  address   = {New York}
}

@article{ye2023Better,
  title   = {Toward Better Practice of Covariate Adjustment in Analyzing Randomized Clinical Trials},
  author  = {Ye, Ting and Shao, Jun and Yi, Yanyao and Zhao, Qingyuan},
  year    = {2023},
  journal = {Journal of the American Statistical Association},
  volume  = {118},
  number  = {544},
  pages   = {2370--2382}
}

@article{hobbs2011Hierarchical,
  title   = {Hierarchical Commensurate and Power Prior Models for Adaptive Incorporation of Historical Information in Clinical Trials},
  author  = {Hobbs, Brian P. and Carlin, Bradley P. and Mandrekar, Sumithra J. and Sargent, Daniel J.},
  year    = {2011},
  journal = {Biometrics},
  volume  = {67},
  number  = {3},
  pages   = {1047--1056}
}

@article{ibrahim2015Power,
  title   = {The Power Prior: Theory and Applications},
  author  = {Ibrahim, Joseph G. and Chen, Ming-Hui and Gwon, Yeongjin and Chen, Fang},
  year    = {2015},
  journal = {Statistics in Medicine},
  volume  = {34},
  number  = {28},
  pages   = {3724--3749}
}

@article{lei2021Regression,
  title   = {Regression Adjustment in Completely Randomized Experiments with a Diverging Number of Covariates},
  author  = {Lei, Lihua and Ding, Peng},
  year    = {2021},
  journal = {Biometrika},
  volume  = {108},
  number  = {4},
  pages   = {815--828}
}

@article{qin2007EmpiricalLikelihoodBased,
  title   = {Empirical-Likelihood-Based Inference in Missing Response Problems and Its Application in Observational Studies},
  author  = {Qin, Jing and Zhang, Biao},
  year    = {2007},
  journal = {Journal of the Royal Statistical Society Series B: Statistical Methodology},
  volume  = {69},
  number  = {1},
  pages   = {101--122}
}

@article{cohen2024Noharm,
  title   = {No-Harm Calibration for Generalized {{Oaxaca}}--{{Blinder}} Estimators},
  author  = {Cohen, P. L. and Fogarty, C. B.},
  year    = {2024},
  journal = {Biometrika},
  volume  = {111},
  number  = {1},
  pages   = {331--338}
}

@article{chan2016Globally,
  title   = {Globally Efficient Non-Parametric Inference of Average Treatment Effects by Empirical Balancing Calibration Weighting},
  author  = {Chan, Kwun Chuen Gary and Yam, Sheung Chi Phillip and Zhang, Zheng},
  year    = {2016},
  journal = {Journal of the Royal Statistical Society Series B: Statistical Methodology},
  volume  = {78},
  number  = {3},
  pages   = {673--700}
}

@article{deville1992Calibration,
  title   = {Calibration Estimators in Survey Sampling},
  author  = {Deville, Jean-Claude and S{\"a}rndal, Carl-Erik},
  year    = {1992},
  journal = {Journal of the American Statistical Association},
  volume  = {87},
  number  = {418},
  pages   = {376--382}
}

@article{kwon2025Debiased,
  title   = {Debiased Calibration Estimation Using Generalized Entropy in Survey Sampling},
  author  = {Kwon, Yonghyun and Kim, Jae Kwang and Qiu, Yumou},
  year    = {2025},
  journal = {Journal of the American Statistical Association},
  volume  = {0},
  number  = {0},
  pages   = {1--12},
  note    = {Just Accepted}
}

@article{tsiatis2008Covariate,
  title   = {Covariate Adjustment for Two-Sample Treatment Comparisons in Randomized Clinical Trials: A Principled yet Flexible Approach},
  author  = {Tsiatis, Anastasios A. and Davidian, Marie and Zhang, Min and Lu, Xiaomin},
  year    = {2008},
  journal = {Statistics in Medicine},
  volume  = {27},
  number  = {23},
  pages   = {4658--4677}
}

@article{zhang2008Improving,
  title   = {Improving Efficiency of Inferences in Randomized Clinical Trials Using Auxiliary Covariates},
  author  = {Zhang, Min and Tsiatis, Anastasios A. and Davidian, Marie},
  year    = {2008},
  journal = {Biometrics},
  volume  = {64},
  number  = {3},
  pages   = {707--715}
}

@article{lin2013Agnostic,
  title   = {Agnostic Notes on Regression Adjustments to Experimental Data: Reexamining {{Freedman}}'s Critique},
  author  = {Lin, Winston},
  year    = {2013},
  journal = {The Annals of Applied Statistics},
  volume  = {7},
  number  = {1},
  pages   = {295--318}
}

@article{qin1994Empirical,
  title   = {Empirical Likelihood and General Estimating Equations},
  author  = {Qin, Jing and Lawless, Jerry},
  year    = {1994},
  journal = {The Annals of Statistics},
  volume  = {22},
  number  = {1},
  pages   = {300--325}
}

@article{tibshirani1996Regression,
  title   = {Regression Shrinkage and Selection via the {{Lasso}}},
  author  = {Tibshirani, Robert},
  year    = {1996},
  journal = {Journal of the Royal Statistical Society Series B: Statistical Methodology},
  volume  = {58},
  number  = {1},
  pages   = {267--288}
}

@article{bai2024Covariate,
  title   = {Covariate Adjustment in Experiments with Matched Pairs},
  author  = {Bai, Yuehao and Jiang, Liang and Romano, Joseph P. and Shaikh, Azeem M. and Zhang, Yichong},
  year    = {2024},
  journal = {Journal of Econometrics},
  volume  = {241},
  number  = {1},
  pages   = {105740}
}

@article{kitamura1997Informationtheoretic,
  title   = {An Information-Theoretic Alternative to Generalized Method of Moments Estimation},
  author  = {Kitamura, Yuichi and Stutzer, Michael},
  year    = {1997},
  journal = {Econometrica},
  volume  = {65},
  number  = {4},
  pages   = {861--874}
}

@article{dupas2018Bankinga,
  title   = {Banking the Unbanked? Evidence from Three Countries},
  author  = {Dupas, Pascaline and Karlan, Dean and Robinson, Jonathan and Ubfal, Diego},
  year    = {2018},
  journal = {American Economic Journal: Applied Economics},
  volume  = {10},
  number  = {2},
  pages   = {257--297}
}

@article{ye2022Inference,
  title   = {Inference on the Average Treatment Effect under Minimization and Other Covariate-Adaptive Randomization Methods},
  author  = {Ye, Ting and Yi, Yanyao and Shao, Jun},
  year    = {2022},
  journal = {Biometrika},
  volume  = {109},
  number  = {1},
  pages   = {33--47}
}

@article{pocock1975Sequential,
  title   = {Sequential Treatment Assignment with Balancing for Prognostic Factors in the Controlled Clinical Trial},
  author  = {Pocock, Stuart J. and Simon, Richard},
  year    = {1975},
  journal = {Biometrics},
  volume  = {31},
  number  = {1},
  pages   = {103--115}
}

@misc{xin2024inferencecovariateadaptiverandomizationstrata,
  title  = {Inference under Covariate-Adaptive Randomization with Many Strata},
  author = {Xin, Jiahui and Liu, Hanzhong and Ma, Wei},
  year   = {2024},
  note   = {arXiv preprint arXiv:2405.18856}
}

@article{tan2014Secondorder,
  title   = {Second-Order Asymptotic Theory for Calibration Estimators in Sampling and Missing-Data Problems},
  author  = {Tan, Zhiqiang},
  year    = {2014},
  journal = {Journal of Multivariate Analysis},
  volume  = {131},
  pages   = {240--253}
}

@article{robins1994Estimation,
  title   = {Estimation of Regression Coefficients When Some Regressors Are Not Always Observed},
  author  = {Robins, James M. and Rotnitzky, Andrea and Zhao, Lue Ping},
  year    = {1994},
  journal = {Journal of the American Statistical Association},
  volume  = {89},
  number  = {427},
  pages   = {846--866}
}

@article{zelen1974Randomization,
  title   = {The Randomization and Stratification of Patients to Clinical Trials},
  author  = {Zelen, Marvin},
  year    = {1974},
  journal = {Journal of Chronic Diseases},
  volume  = {27},
  number  = {7},
  pages   = {365--375}
}

@article{newey2004Higher,
  title   = {Higher Order Properties of {{GMM}} and Generalized Empirical Likelihood Estimators},
  author  = {Newey, Whitney K. and Smith, Richard J.},
  year    = {2004},
  journal = {Econometrica},
  volume  = {72},
  number  = {1},
  pages   = {219--255}
}

@article{taves1974Minimization,
  title   = {Minimization: A New Method of Assigning Patients to Treatment and Control Groups},
  author  = {Taves, Donald R.},
  year    = {1974},
  journal = {Clinical Pharmacology and Therapeutics},
  volume  = {15},
  number  = {5},
  pages   = {443--453}
}

@article{efron1971Forcing,
  title   = {Forcing a Sequential Experiment to Be Balanced},
  author  = {Efron, Bradley},
  year    = {1971},
  journal = {Biometrika},
  volume  = {58},
  number  = {3},
  pages   = {403--417}
}

@book{boyd2004Convex,
  title     = {Convex Optimization},
  author    = {Boyd, Stephen P. and Vandenberghe, Lieven},
  year      = {2004},
  publisher = {Cambridge University Press},
  address   = {Cambridge, UK; New York}
}

@article{jiao2023Deep,
  title   = {Deep Nonparametric Regression on Approximate Manifolds: Nonasymptotic Error Bounds with Polynomial Prefactors},
  author  = {Jiao, Yuling and Shen, Guohao and Lin, Yuanyuan and Huang, Jian},
  year    = {2023},
  journal = {The Annals of Statistics},
  volume  = {51},
  number  = {2},
  pages   = {691--716}
}

@article{bannick2025General,
  title   = {A General Form of Covariate Adjustment in Clinical Trials under Covariate-Adaptive Randomization},
  author  = {Bannick, Marlena S. and Shao, Jun and Liu, Jingyi and Du, Yu and Yi, Yanyao and Ye, Ting},
  year    = {2025},
  journal = {Biometrika},
  volume  = {112},
  number  = {3},
  pages   = {asaf029}
}

@article{bugni2018Inference,
  title   = {Inference Under Covariate-Adaptive Randomization},
  author  = {Bugni, Federico A. and Canay, Ivan A. and Shaikh, Azeem M.},
  year    = {2018},
  journal = {Journal of the American Statistical Association},
  volume  = {113},
  number  = {524},
  pages   = {1784--1796}
}

@article{bugni2019Inference,
  title   = {Inference under Covariate-Adaptive Randomization with Multiple Treatments},
  author  = {Bugni, Federico A. and Canay, Ivan A. and Shaikh, Azeem M.},
  year    = {2019},
  journal = {Quantitative Economics},
  volume  = {10},
  number  = {4},
  pages   = {1747--1785}
}

@book{davidson2021Stochastic,
  title     = {Stochastic Limit Theory: An Introduction for Econometricians},
  author    = {Davidson, James},
  year      = {2021},
  edition   = {2nd},
  publisher = {Oxford University Press},
  address   = {Oxford}
}

@article{jiang2023Regressionadjusted,
  title   = {Regression-Adjusted Estimation of Quantile Treatment Effects under Covariate-Adaptive Randomizations},
  author  = {Jiang, Liang and Phillips, Peter C. B. and Tao, Yubo and Zhang, Yichong},
  year    = {2023},
  journal = {Journal of Econometrics},
  volume  = {234},
  number  = {2},
  pages   = {758--776}
}

@book{billingsley1968convergence,
  title     = {Convergence of Probability Measures},
  author    = {Billingsley, Patrick},
  year      = {1968},
  publisher = {Wiley},
  address   = {New York}
}

@article{bulinski2017Conditional,
  title   = {Conditional Central Limit Theorem},
  author  = {Bulinski, A. V.},
  year    = {2017},
  journal = {Theory of Probability \& Its Applications},
  volume  = {61},
  number  = {4},
  pages   = {613--631}
}

@article{liu2023Lassoadjusted,
  title   = {{{Lasso}}-Adjusted Treatment Effect Estimation under Covariate-Adaptive Randomization},
  author  = {Liu, Hanzhong and Tu, Fuyi and Ma, Wei},
  year    = {2023},
  journal = {Biometrika},
  volume  = {110},
  number  = {2},
  pages   = {431--447}
}

@book{dennis2009matrix,
  title     = {Matrix Mathematics: Theory, Facts, and Formulas},
  author    = {Bernstein, Dennis S.},
  year      = {2009},
  edition   = {2nd},
  publisher = {Princeton University Press},
  address   = {Princeton, NJ}
}

@article{burkholder1988sharp,
  title   = {Sharp Inequalities for Martingales and Stochastic Integrals},
  author  = {Burkholder, Donald L.},
  year    = {1988},
  journal = {Ast{\'e}risque},
  volume  = {157},
  number  = {158},
  pages   = {75--94}
}

@article{tu2024Unified,
  title   = {A Unified Framework for Covariate Adjustment Under Stratified Randomisation},
  author  = {Tu, Fuyi and Ma, Wei and Liu, Hanzhong},
  year    = {2024},
  journal = {Stat},
  volume  = {13},
  number  = {4},
  pages   = {e70016}
}

@misc{rafi2023Efficient,
  title  = {Efficient Semiparametric Estimation of Average Treatment Effects Under Covariate Adaptive Randomization},
  author = {Rafi, Ahnaf},
  year   = {2023},
  note   = {arXiv preprint arXiv:2305.08340}
}

@article{ma2022Regression,
  title   = {Regression Analysis for Covariate-Adaptive Randomization: A Robust and Efficient Inference Perspective},
  author  = {Ma, Wei and Tu, Fuyi and Liu, Hanzhong},
  year    = {2022},
  journal = {Statistics in Medicine},
  volume  = {41},
  number  = {29},
  pages   = {5645--5661}
}

@article{stewart1977Perturbation,
  title   = {On the Perturbation of Pseudo-Inverses, Projections and Linear Least Squares Problems},
  author  = {Stewart, G. W.},
  year    = {1977},
  journal = {SIAM Review},
  volume  = {19},
  number  = {4},
  pages   = {634--661}
}

@article{tropp2015Introduction,
  title   = {An Introduction to Matrix Concentration Inequalities},
  author  = {Tropp, Joel A.},
  year    = {2015},
  journal = {Foundations and Trends{\textregistered} in Machine Learning},
  volume  = {8},
  number  = {1-2},
  pages   = {1--230}
}

@book{vershynin2018HighDimensional,
  title     = {High-Dimensional Probability: An Introduction with Applications in Data Science},
  author    = {Vershynin, Roman},
  year      = {2018},
  publisher = {Cambridge University Press},
  address   = {Cambridge}
}

@article{tseng1991Relaxation,
  title   = {Relaxation Methods for Problems with Strictly Convex Costs and Linear Constraints},
  author  = {Tseng, Paul and Bertsekas, Dimitri P.},
  year    = {1991},
  journal = {Mathematics of Operations Research},
  volume  = {16},
  number  = {3},
  pages   = {462--481}
}

@book{vaart1998Asymptotic,
  title     = {Asymptotic Statistics},
  author    = {van der Vaart, Aad W.},
  year      = {1998},
  publisher = {Cambridge University Press},
  address   = {Cambridge}
}

@article{breiman2001random,
  title   = {Random Forests},
  author  = {Breiman, Leo},
  year    = {2001},
  journal = {Machine Learning},
  volume  = {45},
  number  = {1},
  pages   = {5--32}
}

@article{lecun2015deep,
  title   = {Deep Learning},
  author  = {LeCun, Yann and Bengio, Yoshua and Hinton, Geoffrey},
  year    = {2015},
  journal = {Nature},
  volume  = {521},
  number  = {7553},
  pages   = {436--444}
}

@article{wager2018estimation,
  title   = {Estimation and Inference of Heterogeneous Treatment Effects Using Random Forests},
  author  = {Wager, Stefan and Athey, Susan},
  year    = {2018},
  journal = {Journal of the American Statistical Association},
  volume  = {113},
  number  = {523},
  pages   = {1228--1242}
}

@article{farrell2021Deepa,
  title   = {Deep Neural Networks for Estimation and Inference},
  author  = {Farrell, Max H. and Liang, Tengyuan and Misra, Sanjog},
  year    = {2021},
  journal = {Econometrica},
  volume  = {89},
  number  = {1},
  pages   = {181--213}
}

@article{chernozhukov2018Double,
  title   = {Double/Debiased Machine Learning for Treatment and Structural Parameters},
  author  = {Chernozhukov, Victor and Chetverikov, Denis and Demirer, Mert and Duflo, Esther and Hansen, Christian and Newey, Whitney and Robins, James},
  year    = {2018},
  journal = {The Econometrics Journal},
  volume  = {21},
  number  = {1},
  pages   = {C1--C68}
}

@article{bai2022Optimality,
  title   = {Optimality of Matched-Pair Designs in Randomized Controlled Trials},
  author  = {Bai, Yuehao},
  year    = {2022},
  journal = {American Economic Review},
  volume  = {112},
  number  = {12},
  pages   = {3911--3940}
}
\end{document}